\documentclass[12pt]{article}
\newcommand{\blind}{0}

\usepackage{color}

\usepackage{fdsymbol}
\usepackage{amsthm}
\usepackage{amsmath}
\usepackage{siunitx}

\usepackage{booktabs}
\usepackage{multirow}
\usepackage{adjustbox}
\usepackage{enumerate}
\usepackage{float}
\usepackage{hyperref}

\newtheorem{theorem}{Theorem}
\newtheorem{corollary}{Corollary}
\newtheorem{lemma}{Lemma}
\newtheorem{proposition}{Proposition}

\newtheorem{remark}{Remark}

\usepackage{natbib}
\usepackage{bibunits}
\defaultbibliographystyle{chicago}
\defaultbibliography{ref}

\makeatletter

\def\BU@hyperprefix{}

\let\BU@orig@startbibunit\@startbibunit
\def\@startbibunit{%
  \BU@orig@startbibunit
  \xdef\BU@hyperprefix{\@bibunitname.}%
}

\let\BU@orig@endbibunit\endbibunit
\def\endbibunit{%
  \BU@orig@endbibunit
  \gdef\BU@hyperprefix{}%
}

\let\BU@orig@hyper@natlinkstart\hyper@natlinkstart
\def\hyper@natlinkstart#1{%
  \BU@orig@hyper@natlinkstart{\BU@hyperprefix#1}%
}

\let\BU@orig@hyper@natlinkbreak\hyper@natlinkbreak
\def\hyper@natlinkbreak#1#2{%
  \BU@orig@hyper@natlinkbreak{#1}{\BU@hyperprefix#2}%
}

\let\BU@orig@hyper@natanchorstart\hyper@natanchorstart
\def\hyper@natanchorstart#1{%
  \BU@orig@hyper@natanchorstart{\BU@hyperprefix#1}%
}

\makeatother

\begin{document}

\def\spacingset#1{\renewcommand{\baselinestretch}
{#1}\small\normalsize} \spacingset{1}

\if0\blind
{
  \title{\bf Testing for Heterogeneous Treatment Effects in Regression Discontinuity Designs}
  \author{Xiaojun Song\thanks{Corresponding author. Email: \texttt{sxj@gsm.pku.edu.cn}. This work was supported by the National Natural Science Foundation of China [Grant Numbers 72373007, 72333001, and 72621002]. The author also gratefully acknowledges the research support from the Center for Statistical Science of Peking University and the Key Laboratory of Mathematical Economics and Quantitative Finance (Peking University) of the Ministry of Education, China.} 
    \hspace{.2cm} \\
    Department of Business Statistics and Econometrics \\
    Guanghua School of Management, Peking University \\    \\ 
    Haojiao Zhao\thanks{Email: \texttt{zhaohaojiao@stu.pku.edu.cn}.} \\
    Department of Business Statistics and Econometrics \\
    Guanghua School of Management, Peking University}
  \maketitle
} \fi

\if1\blind
{
  \bigskip
  \bigskip
  \bigskip
  \begin{center}
    {\LARGE\bf Testing for Heterogeneous Treatment Effects in Regression Discontinuity Designs}
\end{center}
  \medskip
} \fi

\bigskip
\begin{abstract}
We propose a nonparametric test for unobserved treatment effect heterogeneity in regression discontinuity designs. 
Under the null of no unobserved heterogeneity, a transformed outcome that imputes treated potential outcomes for untreated units must have a continuous conditional distribution at the cutoff. 
We convert this implication into an integrated conditional-moment restriction using characteristic functions, thereby allowing the conditional local average treatment effect to be an unrestricted function of covariates. 
We derive the asymptotic distribution of the test statistics via a $U$-process and establish the validity of a multiplier bootstrap procedure for calculating critical values. 
Monte Carlo experiments show well-controlled size and increasing power. Two empirical applications illustrate how the test distinguishes between heterogeneity explained by observables and that explained by unobserved factors.
\end{abstract}

\vfill

\newpage
\spacingset{1.45} 

\begin{bibunit}
\section{Introduction} \label{sec:introduction}
Empirical studies increasingly document heterogeneous treatment effects, motivating targeted policy interventions, yet a fundamental question remains unanswered: does the observed variation reflect heterogeneity explained by observable covariates, or a residual component driven by unobserved factors? 
Treatment effects may vary with observed covariates or with latent characteristics that persist even after conditioning on those covariates. 
Existing inference methods primarily address the first source. The second source, whether treatment effects still vary within subgroups defined by observed covariates, has received far less formal attention.
This gap matters for policy: if unobserved heterogeneity exists, interventions designed solely on the basis of observable covariates will be incomplete. A formal diagnostic is therefore needed.

Current approaches to exploring treatment effect heterogeneity suffer from two limitations. 
First, applied studies typically rely on ad hoc analyses, which are sensitive to model misspecification and multiple-testing concerns. 
Second, existing methods are mainly designed for instrumental-variable settings and cannot readily be extended to regression discontinuity (RD) designs, where identification is inherently local at the cutoff. 
Moreover, existing methods focus on detecting heterogeneity attributable to observables, leaving formal tests of unobserved treatment effect heterogeneity scarce. 
Consequently, applied researchers currently lack a principled diagnostic for assessing unobserved treatment effect heterogeneity, especially in RD designs.

This paper fills this gap by developing a nonparametric test for unobserved treatment effect heterogeneity in RD designs. 
The key insight is that if all treatment effect variation is explained by the covariates, then imputing treated potential outcomes for untreated units should yield a transformed outcome with a continuous conditional distribution at the cutoff. 
We convert this continuity implication into an integrated conditional moment (ICM) restriction, thereby avoiding parametric restrictions on the functional form of the conditional local average treatment effect (CLATE), denoted by $\tau(X)$ with $X$ the observed covariates. 
Unlike existing methods that focus on whether $\tau(X)$ is constant across subgroups, our framework directly targets the presence of unobserved heterogeneity beyond the covariates and is valid regardless of the functional form of $\tau(X)$. 
In doing so, the proposed test provides a principled way to assess whether observed covariates can plausibly account for treatment effect variation without imposing restrictions on the functional form.

The proposed test has two theoretical advantages. 
First, the proposed process converges at the rate $1/\sqrt{nh}$ (with $n$ the sample size and $h$ the testing bandwidth), the natural rate associated with RD designs. 
The ICM approach integrates over the covariates, reducing the testing problem to an effectively one-dimensional smoothing exercise and thereby alleviating the curse of dimensionality. 
Second, it exhibits favorable power against local alternatives. 
The global-smoothing structure of the ICM statistic can detect Pitman-type alternatives shrinking at the nonparametric rate $1/\sqrt{nh}$. 
From a practical standpoint, the test also serves as a diagnostic for covariate sufficiency. 
A failure to reject means that the data do not reveal residual heterogeneity within the class of alternatives detectable under the maintained assumptions. 
A rejection indicates a violation of the homogeneity restriction, which can be interpreted as residual treatment effect heterogeneity under the maintained assumptions. 
Moreover, the test adapts to both sharp and fuzzy RD designs and accommodates both local constant and local linear estimators, enhancing the applicability of the proposed testing framework.

The work most closely related to ours is \cite{Hsu2019}, which tests whether the CLATE is constant across covariate-defined subgroups. 
In contrast, our test examines whether the heterogeneity can be plausibly explained by observed covariates. 
The two procedures are therefore complementary. One can first use \cite{Hsu2019} to assess whether the CLATE varies with the observed covariates, and then apply our procedure to determine whether those covariates exhaust the heterogeneity. 
Together, the two tests provide a systematic diagnostic: the first assesses whether $\tau(X)$ varies with $X$, and the second assesses whether residual treatment effect heterogeneity persists. 
In this respect, our testing framework explicitly investigates the existence of unobserved treatment effect heterogeneity, which is difficult to examine and has received limited attention. 

The rest of the paper is organized as follows. 
Section \ref{sec:literature} reviews the related literature.
Section \ref{sec:framework} introduces the testing framework and the $U$-process formulated under the null hypothesis.
Section \ref{sec:asymptotic} establishes the asymptotic properties of the test statistics.
Section \ref{sec:bootstrap} develops a multiplier bootstrap procedure to compute critical values and proves its validity.
Sections \ref{sec:simulation} and \ref{sec:empirical} examine the finite-sample performance through numerical simulations and empirical applications, respectively. 
Section \ref{sec:conclusion} concludes. 
Proofs of the theoretical results and supplementary simulation results are provided in the Appendix. 

\section{Related Literature} \label{sec:literature}
A growing number of empirical studies have examined heterogeneous treatment effects across a wide range of contexts, such as institutional reforms and college enrollment. 
Substantial variation in treatment effects has been reported along observable dimensions including demographic characteristics, employment histories, and family status.
A common feature of these studies is their reliance on ad hoc strategies---splitting the sample into subgroups \citep{LopesdeFonseca2020, Brock2023, Giupponi2023, Chetty2026, Mountjoy2026} or adding interaction terms to regression models \citep{AbouDaher2025, Sorrenti2025}---to probe heterogeneity. 
While such approaches are informative, none provides statistical evidence on whether the observable covariates exhaust the sources of treatment effect variation or whether significant unobserved heterogeneity persists. 
This methodological gap motivates the formal testing procedure developed in the present paper.

A substantial literature studies treatment effect heterogeneity across various identification frameworks. 
Early contributions focus on identification in selection models \citep{Heckman1997, Heckman2001, Heckman2005} and non-separable models \citep{Chesher2003, Imbens2009, Torgovitsky2015}. 
A growing econometric literature investigates treatment effect heterogeneity under unconfoundedness \citep{Crump2008, Hsu2017, SantAnna2021, Cai2024} and IV identification \citep{Chang2015, Hsu2023}. 
Most existing tests assess heterogeneity attributable to observable covariates, whereas \cite{Hsu2023} directly studies unobserved heterogeneity in an IV model. 
Nevertheless, none formally addresses unobserved heterogeneity in fuzzy RD designs. 
This paper fills this gap by developing the corresponding testing framework for fuzzy RD designs.

Since the seminal work of \cite{Thistlethwaite1960} and \cite{Angrist1999}, RD designs have been widely used to identify causal treatment effects in observational studies. 
Within the RD literature, two strands are most relevant to our work.
The first concerns estimation of (heterogeneous) treatment effects in RD designs, including RD estimation methods with covariates \citep{Calonico2019, Frolich2019, kreiss2023, Calonico2025, Noack2025} and continuous treatments \citep{Dong2023, Xie2024}. 
The second strand develops specification tests for the identifying assumptions in RD designs, such as continuity of the running variable density \citep{McCrary2008, Otsu2013, Bugni2021}, monotonicity \citep{Hsu2021, Arai2022, Hsu2024}, and other identification conditions \citep{Dong2018, Bertanha2020a}.
Neither strand, however, offers a formal test for unobserved treatment effect heterogeneity, the question addressed in this paper.

Methodologically, our test contributes to the literature on ICM tests \citep{Bierens1982, Bierens1990, Bierens1997, Delgado2001} by developing an ICM test to detect unobserved treatment effect heterogeneity in fuzzy RD designs. 
We adapt and extend the ICM framework to fuzzy RD designs, in which both discontinuity-based identification and compliance behavior pose additional challenges.
The ICM framework allows us to test for unobserved heterogeneity without imposing parametric restrictions, making the test attractive in RD designs.

\section{Testing Framework} \label{sec:framework}
\subsection{Model} \label{subsec:model}
Consider the following nonparametric and nonseparable RD model: 
$$Y=g(D, R, X, \epsilon),$$ 
where $Y\in\mathbb{R}$ is the observed outcome, $D\in\{0,1\}$ is the binary treatment, $R\in\mathbb{R}$ is the running variable, $X\in\mathcal{X} \subset \mathbb{R}^d$ collects observed covariates, and $\epsilon\in\mathbb{R}^{d_\epsilon}$ is a general disturbance vector. 
The structural function $g$ is a measurable function that determines $Y$ given $D$, $R$, $X$, and $\epsilon$. 
Under the potential outcome framework, the observed outcome $Y$ can be written as 
$$ Y = DY_1 + (1-D) Y_0,$$ 
where the potential outcomes $Y_1 = g(1,R,X,\epsilon)$ and $Y_0 = g(0,R,X,\epsilon)$ cannot be observed simultaneously for the same individual. 
In the function $g$, $R$ may affect $Y$ directly, but RD identification exploits only the discontinuous jump in $D$ at the cutoff. 
Without loss of generality, define $R=0$ as the cutoff. 
Let the propensity score be $p(x,r) = \mathbb{P}(D=1|X=x, R=r) = \mathbb{E}[D|X=x, R=r]$, i.e., the conditional probability of treatment, and let $\mu(x,r) = \mathbb{E}[Y|X=x, R=r]$ denote the conditional expectation function of $Y$. 
For any such function $\psi(x,r)$, denote its one-sided limits by $\psi(x,0^-) = \lim_{r\uparrow 0} \psi(x,r)$ and $\psi(x,0^+) = \lim_{r\downarrow 0} \psi(x,r)$.
Under Assumption A3 below, the propensity score $p(x,r)$ is continuous in $r$ for $r\ne 0$, exhibiting a jump at the cutoff, i.e., $p(x,0^-) \ne p(x,0^+)$ for all $x\in\mathcal{X}$. 
In fuzzy RD designs, the propensity score exhibits a discontinuous but incomplete jump at the cutoff.
Our identification naturally focuses on compliers, whose treatment status switches from untreated to treated as the running variable crosses the cutoff.\footnote{In sharp RD designs, the treatment $D$ is determined by the running variable $R$: $D=1$ whenever $R\ge0$. The propensity score jumps at the cutoff with $p(x,0^+) - p(x,0^-) = 1$ exactly for all $x\in\mathcal{X}$, indicating a deterministic treatment assignment at the cutoff. The conditional average treatment effect at the cutoff can be identified as $\mu(x,0^+)-\mu(x,0^-)$. In this sense, sharp RD is a degenerate special case in which all units near the cutoff are compliers. The theory of this paper is established for the more general fuzzy RD designs, but the testing procedure still accommodates sharp RD with proper estimators.} 
The following assumptions guarantee identification of the CLATE in fuzzy RD designs.

\textbf{Assumption A1 (Continuity)} The conditional density $f_{R|X}(r|x)$ is continuous at $r=0$ uniformly in $x\in\mathcal X$ and strictly positive in a neighborhood of $r=0$.

\textbf{Assumption A2 (Probability)} The probability $\mathbb{P}(R\ge0 \mid X=x)$ is bounded away from zero and one for all $x\in \mathcal{X}$.

\textbf{Assumption A3 (Monotonicity)} The treatment is determined by a threshold function $D = \mathbf{1}(\theta(X,R)\ge\eta)$ with an unknown function $\theta$ that is continuous in $r$ on each side of the cutoff and satisfies $\theta(x,0^+)>\theta(x,0^-)$ for all $x\in\mathcal{X}$.
The latent variable satisfies $(\eta,\epsilon) \perp R \mid X$ and the conditional CDF $F_{\eta|X}(\cdot|x)$ is strictly increasing for each $x\in\mathcal X$.

Assumptions A1--A3 combine standard RD overlap and continuity requirements with the stronger latent-index restriction in Assumption A3 for the structural characterization. 
For Assumption A1, the continuity of the conditional density $f_{R|X}(r|x)$ at $r=0$ ensures that units on either side of the cutoff are locally comparable for every $x\in\mathcal{X}$.
Assumption A2 ensures that comparable groups exist on both sides of the cutoff for every $x\in\mathcal{X}$. 
Assumption A3 rules out defiers, units whose treatment status switches from treated to untreated as $R$ crosses the cutoff.
Under the monotone selection mechanism in Assumption A3, units can be partitioned into always-takers, never-takers, and compliers via $\eta$. 
The three groups are defined as $\mathcal{A}_x = \left\{ \eta\in\mathbb{R}: \eta \le \theta(x,0^-) \right\}$, $\mathcal{N}_x = \left\{ \eta\in\mathbb{R}: \eta > \theta(x,0^+) \right\}$, and $\mathcal{C}_x = \left\{ \eta\in\mathbb{R}: \theta(x,0^-) < \eta \le \theta(x,0^+) \right\}$.
This mechanism exhausts the support of $\eta$ and inherently rules out defiers.
Consequently, the propensity score is $p(x,r) = F_{\eta|X}(\theta(x,r)|x)$, and the strict increase of $F_{\eta|X}(\cdot|x)$ together with $\theta(x,0^+)>\theta(x,0^-)$ implies $p(x,0^+)>p(x,0^-)$ for every $x\in\mathcal X$.
The discontinuity in the propensity score at the cutoff confirms the identifiability of the CLATE $\tau(x)$ with a nonzero denominator.
Under Assumptions A1--A3, the CLATE is identified at the cutoff: 
$$\tau(x) = \frac{\mu(x,0^+)-\mu(x,0^-)}{p(x,0^+)-p(x,0^-)} = \mathbb{E}[Y_1 - Y_0 | X=x, R=0, \eta\in\mathcal{C}_X]. $$
As is standard, the denominator is the proportion of compliers, while the numerator corresponds to the intent-to-treat effect. 
The CLATE $\tau(X)$ is the average treatment gain for the compliers conditional on the covariate $X$, excluding the always-takers and never-takers from the identification. 
The above identification formula is inherently local: it is valid only at the cutoff, where the discontinuity in $p(x,r)$ provides exogenous variation. 
It loses validity outside the cutoff, where the discontinuity disappears, and the denominator becomes zero. 
The remainder of this section develops a formal test for unobserved heterogeneity.

\subsection{Hypothesis of interest} \label{subsec:hypothesis}
Building on the CLATE identification above, we now formulate the hypothesis of interest. 
The CLATE summarizes the average treatment effect within each covariate cell but is silent about the dispersion of individual treatment effects within the cell. 
Detecting unobserved heterogeneity, therefore, requires examining the distribution of treatment effects, not merely their mean. 
In this paper, we test whether the heterogeneity in treatment effects is driven solely by observable covariates $X$, or, alternatively, whether the unobservable factor $\epsilon$ also plays a role. 
Formally, the null hypothesis of no unobserved heterogeneity is 
$$\mathbb{H}_0: \mathbb{P}[ Y_{1i} - Y_{0i} = \tau(X_i) | X_i, R_i = 0] = 1 \quad a.s., $$ 
i.e., conditional on $X_i$, the difference in potential outcomes at the cutoff is deterministic. 
Conversely, the alternative hypothesis posits that unobserved disturbances also contribute to treatment effect heterogeneity. 
Within the structural function $g$, the null hypothesis $\mathbb{H}_0$ is equivalent to an additively separable structure at the cutoff: the unobserved term does not interact with $D$ to produce heterogeneity. 
The following proposition provides a structural characterization of the null hypothesis, thereby formalizing a testable implication.

\begin{proposition} \label{prop1}
    The null hypothesis of no unobserved heterogeneity in treatment effects, i.e., for some measurable function $\tau(\cdot): \mathcal{X} \mapsto \mathbb{R}$, 
    $$ \mathbb{H}_0: g(1,0,X,\cdot) - g(0,0,X,\cdot) = \tau(X)$$ 
    holds if and only if the function $g$ at $R=0$ is additive in $D$ and $\epsilon$, i.e., 
    $$ g(D,0,X,\epsilon) = m(D,X) + \nu(X,\epsilon), $$
    where $m: \mathcal{S}_{DX} \mapsto \mathbb{R}$ and $\nu: \mathcal{S}_{X\epsilon} \mapsto \mathbb{R}$ are measurable functions with $\mathcal{S}$ denoting the support of the relevant variables.
\end{proposition} 

Proposition \ref{prop1} adapts the result of \cite{Lu2014} to RD designs: the absence of unobserved heterogeneity is structurally equivalent to additive separability in the function $g$ at the cutoff. 
Crucially, the separability property is required only at the cutoff, which is a weaker restriction than the global structural conditions widely used in the literature. 
At the cutoff ($R=0$), $Y$ is determined by $D$ and $\epsilon$ given the covariate $X$. If the structural function $g$ satisfies the separability condition in Proposition \ref{prop1}, the treatment effect at the cutoff is $g(1,0,X,\epsilon) - g(0,0,X,\epsilon) = m(1,X) - m(0,X) = \tau(X)$. 
Proposition~\ref{prop1} establishes a structural equivalence, but additive separability of $g$ is not directly testable: the structural function $g$ is latent and cannot be identified from data alone. 
Assumptions A4 and A5 below translate this structural condition into a testable restriction that can be verified from the observable data.

\textbf{Assumption A4 (Single-index error)} 
The structural function $g(d,r,x,\epsilon)$ is continuous in $r$ in a neighborhood of $r=0$ for all $d,x,\epsilon$. 
There exists a measurable function $\tilde{g}: \mathcal{S}_{DX} \times \mathbb{R} \mapsto \mathbb{R}$ and a scalar-valued function $\nu: \mathcal{S}_{X\epsilon} \mapsto \mathbb{R}$ such that $g(D,0,X,\epsilon) = \tilde{g}(D,X,\nu(X,\epsilon))$, where $\tilde{g}$ strictly increases in $\nu$. 

\textbf{Assumption A5 (Support invariance)} For each $d\in\{0,1\}$ and $x\in\mathcal{X}$, the support of $g(d,0,x,\epsilon)$ conditional on compliers $\mathcal{C}_x$ equals the support in the overall population conditional on the covariate $X$, i.e., $\mathcal{S}_{g(d,0,x,\epsilon) | X=x, \eta\in\mathcal{C}_x} = \mathcal{S}_{g(d,0,x,\epsilon) | X=x }$. 

Assumptions A4 and A5 jointly transform the structural property of Proposition \ref{prop1} into an observable restriction at the cutoff. 
Assumption A4 is a direct primitive condition on $g$, which complements the continuity of the density in Assumption A1. 
When $d_\epsilon>1$, Assumption A4 reduces the multidimensional unobserved heterogeneity to a scalar index, which is the economically relevant source of heterogeneity for our test. 
Note that Assumption A4 is naturally satisfied under $\mathbb{H}_0$ since $g(D,0,X,\epsilon) = m(D,X) + \nu(X,\epsilon) $ provides an explicit single-index representation that is monotone in $\nu$. 
Monotonicity is essential for identifying the latent scalar index $\nu$ from the observable distribution, and it is a standard regularity condition in nonseparable models \citep{Chesher2003, Matzkin2003}. 
Under the alternative, however, the single-index and monotonicity restrictions jointly determine the class of detectable departures from the null. 
The proposed test has power against alternatives where the treatment effect depends on a scalar unobservable in a monotone fashion; for instance, when the structural function takes the form $g(d,0,x,\epsilon) = m(d,x) + h(d,x) \cdot \nu(x,\epsilon)$ with $h(0,x)$ and $h(1,x)$ nonzero, of the same sign, and different from each other, so that treatment effects vary systematically with $\nu$ in one direction. 
Conversely, the test may lack power when unobserved heterogeneity arises from separate latent factors entering the treated and untreated potential outcomes, as in a sharp RD design with $g(d,0,x,\epsilon) = m(d,x) + d \cdot \epsilon_1 + (1-d) \cdot \epsilon_2$, where $\epsilon_1$ and $\epsilon_2$ are independent and identically distributed. 
In this case, the treatment effect $Y_1 - Y_0$ depends on $(\epsilon_1, \epsilon_2)$ and cannot be compressed into a single scalar index while preserving monotonicity.\footnote{Because $\epsilon_1$ and $\epsilon_2$ share the same distribution, however, the conditional distribution of the transformed outcome $W=Y+(1-D)\tau(X)$ defined below remains continuous at the cutoff, so the test may have limited power against this form of heterogeneity.} 
Thus, Assumption A4 narrows the alternatives, translating the structural null hypothesis into a testable implication and distinguishing the alternatives from the null. 
An important consequence is that the test controls size under $\mathbb{H}_0$ without requiring Assumption A4 to hold under the alternative: the additively separable structure implied by the null automatically satisfies the single-index condition. 
This creates a useful asymmetry: size control is robust to violations of A4, while power is concentrated against alternatives that respect the scalar-index structure. 

Assumption A5 is a support assumption first introduced in \cite{Vuong2017} in the IV setting. 
It ensures that the range of potential outcomes realized by compliers covers that of the whole population.
It also implies that $\mathcal{S}_{g(d,0,x,\epsilon) | X=x, \eta\in\mathcal{C}_x} = \mathcal{S}_{g(d,0,x,\epsilon) | D=d, X=x }$ since $ \mathcal{S}_{g(d,0,x,\epsilon) | X=x, \eta\in\mathcal{C}_x} \subseteq \mathcal{S}_{g(d,0,x,\epsilon) | D=d, X=x} \subseteq \mathcal{S}_{g(d,0,x,\epsilon) | X=x}. $
This condition guarantees that testing the continuity condition for the compliers is equivalent to testing the continuity condition for the whole sample, thereby enabling a global diagnostic of unobserved heterogeneity.
The distribution of $g(d,0,x,\epsilon)$ among compliers can be identified as 
\begin{align*}
    & \mathbb{P}\left[ g(d,0,x,\epsilon)\le y | X=x, \eta\in\mathcal{C}_x \right] \\
    =& \frac{\lim_{r\downarrow0} \mathbb{P}\left[ Y\le y, D=d | X=x, R=r \right] - \lim_{r\uparrow0} \mathbb{P}\left[ Y\le y, D=d | X=x, R=r \right]}{\lim_{r\downarrow0} \mathbb{P}\left[ D=d | X=x, R=r \right] - \lim_{r\uparrow0} \mathbb{P}\left[ D=d | X=x, R=r \right]} \,\,\text{ for any } y.     
\end{align*}
The support $\mathcal{S}_{g(d,0,x,\epsilon) | X=x, \eta\in\mathcal{C}_x}$ can be identified from this distribution. Some primitive conditions may be sufficient to imply Assumption A5, e.g.~$\mathcal{S}_{\epsilon|X=x,\eta\in\mathcal{C}_x} = \mathcal{S}_{\epsilon|X=x}$, which holds if the selection margin $\eta$ affects treatment assignment but does not truncate the support of the unobserved determinants relative to the full population. 
This condition is plausible when the complier subpopulation is sufficiently representative of the overall population within each covariate cell, for instance, when the instrument induces a shift in treatment take-up without systematically excluding individuals from the extremes of the unobserved distribution.

In practice, the plausibility of Assumption A5 should be assessed on a case-by-case basis. 
If compliers systematically differ from the full population in ways that affect the range of potential outcomes (e.g., only individuals with moderate unobserved gains self-select into treatment), then the support equality may be violated. 
When A5 fails, the test may produce spurious rejections: a discontinuity in the conditional distribution of $W$ could emerge not from genuine unobserved treatment effect heterogeneity, but from the truncation of the complier support relative to the full population. 
Conversely, in some cases the test could become conservative if support differences mask genuine heterogeneity. 
A useful diagnostic is to compare the empirical support of the complier distribution (identified via the RD ratio above) with that of the full sample, and assess whether discrepancies are substantively large. 
Developing a formal sensitivity analysis for violations of Assumption A5 remains an important avenue for future research.

Together, Assumptions A4 and A5 yield a valid, observable implication regarding the absence of unobserved treatment effect heterogeneity. 
In particular, they are required only at the cutoff, rather than across the entire support of $R$.

To derive a testable implication from Proposition \ref{prop1}, we define the transformed outcome $W = Y + (1-D)\tau(X)$. 
This construction has a natural interpretation: adding $(1-D)\tau(X)$ to $Y$ imputes the treated potential outcome $Y_1$ for untreated units. 
For treated units ($D=1$), $W = Y = Y_1$ directly; for untreated units ($D=0$), $W = Y_0 + \tau(X) = Y_1$ under $\mathbb{H}_0$ at the cutoff. 
Intuitively, if $W$ were discontinuous at the cutoff, its density function would exhibit a jump. 
Based on Proposition \ref{prop1} and Assumptions A1--A5, the following proposition further connects the null hypothesis and the continuity condition in RD designs.
\begin{proposition} \label{prop2}
    Suppose Assumptions A1--A5 hold. Then $\mathbb{H}_0$ holds if and only if the conditional distribution of $W$ is continuous at the cutoff for all $(w,x)$, i.e., $$f_{W|XR}(w|x,0^+) = f_{W|XR}(w|x,0^-) \text{ for all } (w,x).$$
\end{proposition} 
Proposition~\ref{prop2} transforms the additive structure into a testable condition in RD designs. 
The ``only if'' direction (from $\mathbb{H}_0$ to continuity) follows from Proposition~\ref{prop1} together with Assumptions A1--A3 and the continuity of $g$ in $r$ imposed in the first part of Assumption A4; the latter rules out a jump in the distribution of $W$ that is unrelated to unobserved heterogeneity. 
The ``if'' direction (from continuity back to $\mathbb{H}_0$) additionally requires the single-index monotonicity in the second part of Assumption A4 and the support invariance of Assumption A5, as detailed in the proof in the Appendix. 

Under $\mathbb{H}_0$, $W$ equals $Y_1$ at the cutoff. 
The continuity of the conditional density function is equivalent to the continuity of $Y_1$'s conditional density, which is stricter than the continuity of the conditional expectation commonly used in the literature since the unobserved heterogeneity cannot be captured by the expectation.
Proposition \ref{prop2} exploits such distributional continuity as a characterization of the null implied by the absence of unobserved heterogeneity.
Furthermore, denote the conditional characteristic function of $W$ by $\varphi_{W|XR}(w|x,r)$, which is the Fourier transform of the conditional density function $f_{W|XR}(w|x,r)$: 
$$\varphi_{W|XR}(w|x,r) = \mathbb{E}[e^{\mathrm{i} wW}|X=x, R=r] = \int e^{\mathrm{i}w\bar{w}} f_{W|XR}(\bar{w}|x,r) d\bar{w}, $$ 
with $\mathrm{i}$ denoting the imaginary unit in the characteristic function. 
Because characteristic functions uniquely determine distributions, the continuity of the conditional density function is equivalent to the continuity of the conditional characteristic function. 
Therefore, by Proposition \ref{prop2}, $\mathbb{H}_0$ is equivalent to the continuity of the conditional characteristic functions: 
$$\varphi_{W|XR}(w|x,0^+) = \varphi_{W|XR}(w|x,0^-) \text{ for all } (w,x).$$
A natural approach would be to estimate and compare the conditional functions on either side of the cutoff. 
However, the required nonparametric estimation introduces the random-denominator problem and suffers from the curse of dimensionality when $X$ is multivariate.
Rather than using local smoothing, we aggregate the discontinuity into a global moment by integrating it against an exponential weight over the distribution of the covariates: 
$$U(w,x) = \int e^{\mathrm{i} x'\bar{x}} \left( \varphi_{W|XR}(w|\bar{x}, 0^+) - \varphi_{W|XR}(w|\bar{x}, 0^-) \right) f_{XR}(\bar{x}, 0) d\bar{x}=0 \text{ for all } (w,x).$$ 
In the spirit of \cite{Bierens1982}, $U(w,x)$ is a global distance that can be recognized as a weighted aggregation of the discontinuity in integral form with an exponential weighting function.\footnote{Readers are referred to \cite{stinchcombe1998} and \cite{Escanciano2006}, among others, for alternative weighting functions frequently used in the ICM literature.} 
Under the null, $\varphi_{W|XR}(w|x,r)$ is continuous at $r=0$ for all $(w,x)$ and the discontinuity at the cutoff equals zero everywhere, leading to $U(w,x)=0$ for all $(w,x)$. 
Under the alternatives, the discontinuity of $\varphi_{W|XR}(w|x,r)$ is aggregated in the integral, making $U(w,x)\neq0$ for some $(w,x)$. 

Under the ICM approach, $U(w,x)$ captures any discontinuity and avoids the random-denominator problem in estimating the conditional density, reducing the testing problem to an effectively one-dimensional smoothing exercise. 
This property preserves all the distributional information required for testing while mitigating the curse of dimensionality.
Note that, by applying the Fourier representation of the conditional density and interchanging the order of integration via Fubini's theorem (the integrand is bounded by the integrable density $f_{WX|R}f_R$), $U(w,x)$ admits the following equivalent representations:
\begin{align*}
    U(w,x) =& \int e^{\mathrm{i} x'\bar{x}} \left( \varphi_{W|XR}(w|\bar{x}, 0^+) - \varphi_{W|XR}(w|\bar{x}, 0^-) \right) f_{XR}(\bar{x}, 0) d\bar{x} \\
    =& \int e^{\mathrm{i} x'\bar{x}} e^{\mathrm{i} w\bar{w}} \left( f_{W|XR}(\bar{w}|\bar{x}, 0^+) - f_{W|XR}(\bar{w}|\bar{x}, 0^-) \right) f_{XR}(\bar{x},0) d\bar{w}d\bar{x} \\
    =& \int e^{\mathrm{i} x'\bar{x}} e^{\mathrm{i} w\bar{w}} \left( f_{WX|R}(\bar{w}, \bar{x}| 0^+) f_R(0^+) - f_{WX|R}(\bar{w}, \bar{x}| 0^-) f_R(0^-) \right) d\bar{w}d\bar{x} \\
    =& f_R(0) \left[ \lim_{r\downarrow0} \mathbb{E} \left[ e^{\mathrm{i} \left( x'X + wW \right)} | R=r \right] - \lim_{r\uparrow0} \mathbb{E} \left[ e^{\mathrm{i} \left( x'X + wW \right)} | R=r \right] \right]. 
\end{align*}
The passage from the third to the fourth line uses the law of iterated expectation, the identity $f_{WX|R}(w,x|r)f_R(r)=f_{WXR}(w,x,r)$, and the fact that $f_R(0^+)=f_R(0^-)=f_R(0)$ under Assumption A1. 
This equality of the marginal density of $R$ at the cutoff follows from integrating the continuity of the conditional density $f_{R|X}$ over the compact support $\mathcal{X}$: $f_R(r) = \int f_{R|X}(r|x) f_X(x) dx$, which is continuous at $r=0$ whenever $f_{R|X}(\cdot|x)$ is continuous at $0$ uniformly in $x$.
The final expression in fact shows that $U(w,x)$ equals the jump in the conditional characteristic function of $(X,W)$ at $R=0$, multiplied by $f_R(0)$. 
Thus, $U(w,x)=0$ for all $(w,x)$ if and only if the discontinuity is zero almost everywhere, which is the completeness property behind the ICM transformation.

Up to the one-sided kernel normalization, a natural sample analog of $U(w,x)$ is 
$$ U_n(w,x) = \frac{1}{n} \sum_{i=1}^n e^{\mathrm{i} \left( x'X_i + wW_i \right)} K_h(R_i) \delta_i,$$ 
where $\delta_i = \mathbf{1}(R_i\ge 0) - \mathbf{1}(R_i<0)$ encodes the signed one-sided difference and $K_h(\cdot) = h^{-1} K(\cdot/h)$ is the kernel function with bandwidth $h=h_n\to 0$ as $n\to\infty$. 
Note that $U_n(w,x)$ is infeasible because $W_i$ cannot be observed for all individuals. 
Specifically, for untreated units ($D_i=0$), $W_i = Y_i + \tau(X_i)$ requires the CLATE $\tau(X_i)$, which is unknown and must be estimated. 
We therefore construct a feasible counterpart $\hat{U}_n(w,x)$ by plugging in $\hat{W}_i = Y_i + (1-D_i) \hat\tau(X_i)$ with a consistent estimator 
$$ \hat\tau(X_i) = \frac{\hat\mu(X_i, 0^+) - \hat\mu(X_i, 0^-)}{\hat p(X_i, 0^+) - \hat p(X_i, 0^-)},$$ 
where $\hat\mu(X_i,0^\pm)$ and $\hat{p}(X_i,0^\pm)$ are leave-one-out estimators from local constant or local linear regression for the functions $\mu(x,0^\pm)$ and $p(x,0^\pm)$ at the sample points using bandwidth parameters $h_x$ and $h_r$. 
Using the estimated values $\hat{W}_i$, we construct the test statistics from the following feasible $U$-process: 
$$ \hat{U}_n(w,x) = \frac{1}{n} \sum_{i=1}^n e^{\mathrm{i} \left( x'X_i + w \hat{W}_i \right)} K_h(R_i) \delta_i, $$ 
whose asymptotic behavior will be investigated in Section \ref{sec:asymptotic}.

\begin{remark}
    When the covariates are absent, the testing problem reduces to a test of pure homogeneity of the treatment effect.
    The null hypothesis reduces to $$\widetilde{\mathbb{H}}_0: \mathbb{P}[ Y_{1i} - Y_{0i} = \tau | R_i = 0] = 1 \quad a.s., $$
    i.e., the treatment effect is homogeneous among individuals at the cutoff. 
    Let $Y=g(D,R,\epsilon)$ denote the structural function of $Y$, where $\epsilon$ may include some unobservable factors. 
    The null hypothesis $\widetilde{\mathbb{H}}_0$ can be stated as $g(1,0,\cdot) - g(0,0,\cdot) = \tau$ for some constant $\tau\in\mathbb{R}$. 
    It holds if and only if the function $g$ is additively separable in $D$ and $\epsilon$ at the cutoff, i.e., 
    $$ g(D,0,\epsilon) = m(D) + \nu(\epsilon), $$
    where the mean function $m: \{0,1\} \mapsto \mathbb{R}$ and the single-index error function $\nu: \mathcal{S}_{\epsilon} \mapsto \mathbb{R}$ are measurable.
    If the analogues of Assumptions A1--A5 hold in the unconditional sense, $\widetilde{\mathbb{H}}_0$ is equivalent to the continuity of the density function $f_{W|R}(w|r)$ at $r=0$ for all $w$ using the imputed treated outcome $W=Y+(1-D)\tau$, which can be assessed by the test statistics constructed from the $U$-process $\hat{U}_n(w) = n^{-1}\sum_{i=1}^n e^{\mathrm{i}w\hat{W}_i} K_h(R_i) \delta_i$.
\end{remark}

\subsection{Test Statistics}\label{subsec:statistics}
Based on the feasible $U$-process $\hat{U}_n(w,x)$, one can construct various types of test statistics to assess how close the process is to 0. 
In this paper, we consider two popular choices: Kolmogorov--Smirnov (KS) statistics based on the coordinatewise supremum and Cram\'{e}r--von Mises (CvM) statistics based on the $L_2$ norm. 
Both statistics aggregate deviations over the appropriate index set: the KS statistic via a supremum, and the CvM statistic via integration. 
They therefore retain sensitivity to a broad class of alternatives while avoiding pointwise nonparametric estimation of the conditional functions. 
Neither statistic imposes a parametric model on $\tau(\cdot)$ or on the conditional distribution of $W$.

In practice, we evaluate the $U$-process $\hat{U}_n(w,x)$ on a finite grid whose points are typically the sample points $(\hat{W}_i, X_i)_{i=1}^n$ or a fixed grid $(w_l,x_l)_{l=1}^L$. 
Based on this specification, the KS statistic is $\sqrt{nh}$ times the supremum of the maximum of the absolute real and imaginary parts of the process $\hat{U}_n(w,x)$ over a given grid. 
For the CvM statistic, a closed-form expression can be derived with a weighting function $g(w,x) =(2\pi)^{-\frac{1+d}{2}} e^{-\frac{w^2 + \Vert x \Vert^2}{2}}$. 
Specifically, the formulas for the two types of test statistics are
\begin{align*}
    \mathrm{KS}_n = \sqrt{nh} \sup_{(w,x)\in\Omega} \max\left\{ \left\vert \mathrm{Re} (\hat{U}_n(w,x)) \right\vert, \left\vert \mathrm{Im} (\hat{U}_n(w,x)) \right\vert \right\}, 
\end{align*}
where, here and henceforth, $\Omega$ denotes any compact subset of $\mathbb R\times\mathcal X$, and
\begin{align*}
    \mathrm{CvM}_n = \int \left| \sqrt{nh} \hat{U}_n(w,x) \right|^2 g(w,x)dwdx = \frac{h}{n} \sum_{i=1}^n \sum_{j=1}^n e^{-\frac{(\hat{W}_i - \hat{W}_j)^2 + \lVert X_i-X_j \rVert^2}{2}} K_h(R_i)K_h(R_j) \delta_i\delta_j.
\end{align*}
These specifications can simplify computation and improve efficiency. 
Alternative norms or weighting functions could also be employed to construct other forms of test statistics. We note a practical consideration: the double-sum expression for $\mathrm{CvM}_n$ above is an exact evaluation of the integrated squared process for any bandwidth configuration, because it is obtained by interchanging the integral with the double sum over the sample; in particular, it does not rely on $h=o(h_r)$. 
By contrast, when the estimation effect is non-negligible (i.e., $h=O(h_r)$), the limiting null distribution of $\mathrm{CvM}_n$ depends on the estimation effect, and the corresponding bootstrap statistic $\mathrm{CvM}_n^\ast$ involves quadruple summations (see Section~\ref{sec:bootstrap}), which become computationally expensive for moderate to large sample sizes. 
For this reason, we recommend the KS statistic as the primary test when $h=O(h_r)$ and report only KS-based results in the numerical studies. 
The CvM statistic remains available as a complementary option whenever $h=o(h_r)$ or when computational resources permit. 
In the next section, we give the asymptotic properties of the proposed test statistics.

\section{Asymptotic Theory}\label{sec:asymptotic}
In this section, we study the asymptotic behavior of the proposed test statistics. 
First, we investigate the estimation effect of the $U$-process $\hat{U}_n(w,x)$ and obtain its linear representation. 
We then analyze the limiting behavior of the $U$-process under the null and the alternatives in Theorems \ref{thm1}--\ref{thm3}, and derive the limiting distributions of the proposed test statistics. 
The asymptotic results are organized into three cases reflecting the different regimes of the bandwidth ratio $h/h_r$: cases (i) and (ii) cover the local constant and local linear estimators under $h=O(h_r)$, respectively, where the estimation effect is non-negligible and enters the limiting distribution; case (iii) covers either estimator under $h=o(h_r)$, where the estimation effect is asymptotically dominated by the main term and drops out, yielding a simpler limiting process.
Before presenting these asymptotic results, several technical assumptions are introduced as follows:

\textbf{Assumption B1 (Bandwidth)} As $n\to\infty$, $h, h_r, h_x\to0$, $h/h_r\to c\in[0, \infty)$, $nh\to\infty$, $nh^3\to0$, $nhh_r^2\to0$, $nhh_x^{2l}\to0$, $nh_x^{2d}h_r^2\to\infty$. 

\textbf{Assumption B1' (Bandwidth)} As $n\to\infty$, $h, h_r, h_x\to0$, $h/h_r\to c\in[0, \infty)$, $nh\to\infty$, $nh^3\to0$, $nhh_r^4\to0$, $nhh_x^{2l}\to0$, $nh_x^{2d}h_r^2\to\infty$. 

\textbf{Assumption B2 (Kernel)} The kernel function $K$ is defined on a compact support. It is nonnegative, symmetric, and of order $l$ (that is $\int v^\alpha K(v)dv =0$ for $\alpha = 1, \cdots, l-1$, and $0 < \int v^l K(v)dv < \infty$). Its roughness is finite: $R_K = \int K^2(v)dv < \infty$.

The bandwidth conditions in Assumptions B1 and B1' control bias and variance for the local constant and local linear CLATE estimators, respectively. 
Recall that $h_x$ and $h_r$ are bandwidth parameters for estimating the CLATE $\tau(X)$, which can be flexibly chosen in the regression and can be different from the testing bandwidth $h$ in $U_n(w,x)$. 
In particular, under-smoothing is used so that the smoothing bias is asymptotically negligible. 
The bandwidth condition $nh_x^{2d}h_r^2\to\infty$ for variance replaces the standard rate $nh_x^d h_r\to\infty$ from conditional nonparametric estimation because the density-weighted kernel estimator in $\hat\tau(X_i)$ incorporates an additional $h_x^d h_r$ factor from the density weighting term. 
The condition $h/h_r\to c\in[0,\infty)$ allows the testing bandwidth $h$ to be asymptotically proportional to or smaller than the estimating bandwidth $h_r$, preventing the estimation effect from dominating the main term $U_n(w,x)$. 
Assumption B1 is more restrictive than Assumption B1' because the local constant estimator has a larger boundary bias at the cutoff. This is the cost of the simple local constant estimator. 
For transparency, write $h\asymp n^{-a}$, $h_r\asymp n^{-b_r}$, and $h_x\asymp n^{-b_x}$. 
A sufficient exponent formulation of Assumption B1 is 
$$ 0<a<1, \quad 3a>1, \quad a+2b_r>1, \quad a+2lb_x>1, \quad 2db_x+2b_r<1, \quad a\ge b_r,$$ 
with $a+4b_r>1$ replacing $a+2b_r>1$ under Assumption B1'. 
These inequalities are sufficient, and the feasible set they define depends on the covariate dimension $d$ and the kernel order $l$ in Assumption B2. 
When no exponent triple $(a,b_r,b_x)$ satisfies the system with $a=b_r$ ($c\in(0,\infty)$), the bandwidths may instead be selected from a feasible region with $a>b_r$, i.e., $h=o(h_r)$. 
That case is covered by Theorem~\ref{thm1}(iii), under which the estimation effect is asymptotically negligible. 
In practice, the bandwidth can be selected following \cite{Imbens2012a}, \cite{Calonico2014}, and \cite{Calonico2019} when Assumption B1 or B1' is satisfied. Assumption B2 is standard and is satisfied by commonly used kernels.

\textbf{Assumption B3 (Bound)} The support $\mathcal{X}$ is compact and there exists a positive constant $C$ such that $\inf_{x\in\mathcal{X}} |f_{XR}(x,0)(p(x,0^+) - p(x,0^-))| \ge C^{-1}$.

\textbf{Assumption B4 (Smoothness)} Functions $\varphi_{WD|XR}(w,d|x,r)$ (defined in Lemma \ref{Lemma1}), $p(x,r)$ and $\mu(x,r)$ are continuously differentiable to order $l$ in $x$ and up to second order in $r$ on both sides of $r=0$, uniformly in $x\in\mathcal{X}$.

\textbf{Assumption B5 (Moment)} The outcome has a finite second moment: $\mathbb{E}[Y^2]<\infty$.

Assumption B3 ensures that $f_{XR}(x,0)>0$ almost everywhere in $\mathcal{X}$ and that the denominator of $\tau(x)$ is uniformly bounded away from 0, making the treatment effect $\tau(x)$ identifiable. 
The compact support restriction can be relaxed at the cost of stronger conditions on the distribution of $X$.
Assumption B4 is a smoothness condition that guarantees the validity of a Taylor expansion around the cutoff. 
Assumption B5 imposes a mild moment condition to ensure that the second moments required in the analysis of the $U$-process are finite.

Although our test is motivated by the infeasible process $U_n(w,x)$, we can only obtain $\hat{U}_n(w,x)$. 
Therefore, it is necessary to consider the difference between them when investigating the asymptotic behavior of the test statistics. 
The (nonparametric) estimation effect is given by 
$$\hat{U}_n(w,x) - U_n(w,x)= \frac{1}{n} \sum_{i=1}^n e^{\mathrm{i} x'X_i} \left( e^{\mathrm{i} w\hat{W}_i} - e^{\mathrm{i} wW_i} \right) K_h(R_i) \delta_i. $$
Its asymptotic properties are given in the following lemma.

\begin{lemma} \label{Lemma1}
    (i) Suppose Assumptions A1--A5 and B1--B5 hold. For the local constant estimator $\hat\tau(X_i)$, the estimation effect admits the expression 
    $$ \hat{U}_n(w,x) - U_n(w,x) = \mathrm{i}w \frac{1}{n} \sum_{j=1}^n e^{\mathrm{i} x'X_j} \kappa(w,X_j) (Y_j-D_j\tau(X_j)) K_{h_r}(R_j) \delta_j + o_p\left(\frac{1}{\sqrt{nh}}\right) $$ 
    uniformly in $(w,x)\in\Omega$, where 
    $$ \kappa(w,x) = \frac{\varphi_{WD|XR}(w, 0|x, 0^+)-\varphi_{WD|XR}(w, 0|x, 0^-)}{p(x,0^+)-p(x,0^-)}$$ 
    with $\varphi_{WD|XR}(w, d|x, r) = \mathbb{E}[e^{\mathrm{i}wW} \mathbf{1}(D=d) | X=x, R=r]$.

    (ii) For the local linear estimator $\hat\tau(X_i)$, under Assumption B1' (instead of Assumption B1), the estimation effect admits the expression 
    \begin{align*}
        \hat{U}_n(w,x) - U_n(w,x) =& \mathrm{i}w \frac{1}{2n} \sum_{j=1}^n  e^{\mathrm{i}x'X_j} \kappa(w,X_j) f_{XR}(X_j,0) (u_{Yj} - u_{Dj}\tau(X_j)) K_{h_r}(R_j) \\
        & e_0' \left( \mathbf{1}(R_j\ge0)B_+^{-1}(X_j) - \mathbf{1}(R_j<0) B_-^{-1}(X_j) \right) \mathbf{r}(R_j)  + o_p\left( \frac{1}{\sqrt{nh}}\right), 
    \end{align*}
     uniformly in $(w,x)\in\Omega$, where $u_{Yj}=Y_j-\mu(X_j,R_j)$ and $u_{Dj}=D_j-p(X_j,R_j)$ are the nonparametric residuals of $Y_j$ and $D_j$. The matrices $B_+(x)$ and $B_-(x)$ are defined as $B_+(x)=f_X(x)\mathbb{E}[\mathbf{r}(R)\mathbf{r}(R)'K_{h_r}(R)\mathbf{1}(R\ge0)|X=x]$ and $B_-(x)=f_X(x)\mathbb{E}[\mathbf{r}(R)\mathbf{r}(R)'K_{h_r}(R)\mathbf{1}(R<0)|X=x]$ with $\mathbf{r}(R) = (1, R/h_r)'$.
\end{lemma}

Lemma \ref{Lemma1} provides an explicit representation of the estimation effect $\hat{U}_n(w,x) - U_n(w,x)$ with either a local constant or a local linear estimator for the CLATE $\tau(\cdot)$. 
The order of the estimation effect $(nh_r)^{-1/2}$ is slightly different from the order of the infeasible $U$-process $U_n(w,x)$, namely $(nh)^{-1/2}$. The estimation effect may affect the asymptotic distribution, depending on the relative rates of $h$ and $h_r$. 
When $h=h_r$, the estimation effect and $U_n(w,x)$ are of the same order, and thus the estimation effect contributes to the asymptotic distribution of $\hat{U}_n(w,x)$. If $h/h_r \to 0$ (i.e., $h=o(h_r)$), the estimation effect is asymptotically negligible relative to $U_n(w,x)$. 
In this sense, Lemma \ref{Lemma1} bridges the feasible $U$-process $\hat{U}_n(w,x)$ and the infeasible $U$-process $U_n(w,x)$, guiding the investigation of the limiting distribution under the null hypothesis $\mathbb{H}_0$, the fixed alternative $\mathbb{H}_1$, and a sequence of local alternatives $\mathbb{H}_{1n}$.

\subsection{Asymptotic Consistency}\label{subsec:null}
\begin{theorem} \label{thm1}
Suppose Assumptions A1--A5 and B2--B5 are satisfied. Under the null hypothesis $\mathbb{H}_0: \varphi_{W|XR}(w|x,0^+) = \varphi_{W|XR}(w|x,0^-), \,\,\, \forall (w,x),$

(i) when Assumption B1 holds for the local constant estimator $\hat\tau(X_i)$, $\hat{U}_n(w,x)$ converges weakly on $\Omega$: 
$$ \sqrt{nh} \hat{U}_n(w,x) \Longrightarrow U_{\infty,1}(w,x),$$
where, here and henceforth, $\Longrightarrow$ denotes weak convergence;

(ii) when Assumption B1' holds for the local linear estimator $\hat\tau(X_i)$, $\hat{U}_n(w,x)$ converges weakly on $\Omega$: $$ \sqrt{nh} \hat{U}_n(w,x) \Longrightarrow U_{\infty,2}(w,x);$$ 

(iii) when Assumption B1 holds for the local constant estimator $\hat\tau(X_i)$ or Assumption B1' holds for the local linear estimator $\hat\tau(X_i)$, if $h=o(h_r)$, $\hat{U}_n(w,x)$ converges weakly on $\Omega$: 
$$ \sqrt{nh} \hat{U}_n(w,x) \Longrightarrow U_{\infty,0}(w,x).$$ 
Here, $U_{\infty,1}(w,x)$ is a zero-mean Gaussian process with covariance structure $\mathcal{K}_1(w_1, x_1; w_2, x_2) = \mathbb{E} [U_{\infty,1}(w_1,x_1) \bar{U}_{\infty,1}(w_2,x_2)]$, with $\bar{U}_{\infty,1}(w,x)$ denoting the conjugate process of $U_{\infty,1}(w,x)$. 
In addition, $U_{\infty,2}(w,x)$ and $U_{\infty,0}(w,x)$ are zero-mean Gaussian processes with covariance structures $\mathcal{K}_2(w_1, x_1; w_2, x_2) = \mathbb{E} [U_{\infty,2}(w_1,x_1) \bar{U}_{\infty,2}(w_2,x_2)]$ and $\mathcal{K}_0(w_1, x_1; w_2, x_2) = \mathbb{E} [U_{\infty,0}(w_1,x_1) \bar{U}_{\infty,0}(w_2,x_2)]$. 
In particular, when $h=o(h_r)$ so that the estimation effect is asymptotically negligible, the covariance structure $\mathcal{K}_0(w_1, x_1; w_2, x_2)$ simplifies to
\begin{align*}
    \mathcal{K}_0(w_1, x_1; w_2, x_2) = R_K \int e^{\mathrm{i} (x_1-x_2)'x} \varphi_{W|XR}(w_1-w_2|x, 0) f_{XR}(x, 0) dx.
\end{align*}
The expressions for $\mathcal{K}_1(w_1, x_1; w_2, x_2)$ and $\mathcal{K}_2(w_1, x_1; w_2, x_2)$ additionally incorporate terms arising from the estimation effect and are provided in the Appendix.
\end{theorem}

Under $\mathbb{H}_0$, the conditional characteristic function $\varphi_{W|XR}$ is continuous at $r=0$ for all $(w,x)\in\Omega$ and the one-sided limits coincide. 
Consequently, $\sqrt{nh} \hat{U}_n(w,x)$ converges weakly to the centered Gaussian process $U_{\infty, j}$ for $j=0,1,2$ under the corresponding under-smoothing conditions. 
The covariance structure of the limiting Gaussian process $\mathcal{K}_j$ for $j=0,1,2$ is complicated. 
In particular, the covariance structures $\mathcal{K}_1$ and $\mathcal{K}_2$ additionally incorporate the first-order estimation effect and the intrinsic stochastic variation in the RD model. 
The convergence rate in Theorem \ref{thm1} is $1/\sqrt{nh}$, instead of the parametric rate $1/\sqrt{n}$, reflecting the local identification of RD designs. 
This rate is analogous to that of a one-dimensional boundary kernel estimator, in which the effective sample size is $nh$ rather than $n$. 
We note, however, that the first-stage estimation of $\tau(X)$ still requires $d$-dimensional smoothing over the covariates, so the overall procedure shares the curse-of-dimensionality limitations common to all RD methods that involve covariates in the estimation step. 
In practice, this restricts the feasible dimension of $X$ to $d=2$ or $3$ with realistic sample sizes, as discussed in Assumptions B1 and B1'. 
Thus, procedures for accommodating higher-dimensional covariate vectors, such as semiparametric index restrictions on $\tau(X)$, constitute a useful extension. 

Theorem \ref{thm1} and the continuous mapping theorem yield the asymptotic null distributions of continuous functionals of $\sqrt{nh} \hat{U}_n(w,x)$, including our test statistics $\mathrm{KS}_n$ and $\mathrm{CvM}_n$ in the following corollary. 
\begin{corollary} \label{corollary}
Under the assumptions of Theorem \ref{thm1} and $\mathbb{H}_0$, for any continuous functional $\mathcal{T}(\cdot)$ (with respect to the supremum norm), 
$$\mathcal{T}(\sqrt{nh} \hat{U}_n) \overset{d}{\rightarrow} \mathcal{T}(U_{\infty,j}), $$ 
for $j=0,1,2$, where $\overset{d}{\rightarrow}$ denotes convergence in distribution. In particular, for the KS and CvM statistics,
\begin{align*}
    \mathrm{KS}_n &\overset{d}{\rightarrow} \sup_{(w,x)\in\Omega} \max \left\{ \left| \mathrm{Re} (U_{\infty,j}(w,x))\right|, \left| \mathrm{Im} (U_{\infty,j}(w,x)) \right| \right\}, \\
    \mathrm{CvM}_n &\overset{d}{\rightarrow} \int \left| U_{\infty,j}(w,x) \right|^2 g(w,x)dwdx. 
\end{align*}
\end{corollary}
The test statistics $\mathrm{KS}_n$ and $\mathrm{CvM}_n$ converge to the images of $U_{\infty,j}(w,x)$ under continuous mappings for $j=0,1,2$, depending on the bandwidth choice and the estimator. Their limits are obtained via the continuous mapping theorem (see, e.g., Theorem 1.3.6 in \cite{vdv1996}), and a detailed proof is provided in the Appendix. 
Importantly, the limiting distributions of the test statistics $\mathrm{KS}_n$ and $\mathrm{CvM}_n$ under the null depend in a complex way on the underlying data-generating process through the covariance structure $\mathcal{K}_j(w_1, x_1; w_2, x_2)$ of the centered Gaussian process $U_{\infty,j}$. 
Therefore, critical values cannot be tabulated, which motivates the multiplier bootstrap procedure developed in Section \ref{sec:bootstrap}.

\subsection{Asymptotic Power} \label{subsec:power}
\subsubsection{Against the fixed alternative} \label{subsubsec:H1}
The following theorem delivers the asymptotic power against the fixed alternative 
$$\mathbb{H}_1: \varphi_{W|XR}(w|x,0^+) - \varphi_{W|XR}(w|x,0^-) =\gamma(w,x), $$ 
where $\gamma(w,x)\ne0$ for some $(w,x) \in \Omega$.
\begin{theorem} \label{thm2}
Suppose Assumptions A1--A5 and B2--B5 hold. Under $\mathbb{H}_1$, when Assumption B1 holds for the local constant estimator or Assumption B1' holds for the local linear estimator, 
$$\sup_{(w,x) \in \Omega} \left| \hat{U}_n(w,x) - \Gamma(w,x) \right| = o_p(1), $$ 
where $\Gamma(w,x) = \int e^{\mathrm{i} x'\bar{x}} \gamma(w,\bar{x}) f_{XR}(\bar{x},0) d\bar{x}$ is not equal to 0 for some $(w,x)\in\Omega$. 
\end{theorem}

Theorem \ref{thm2} indicates that $\sqrt{nh} \hat{U}_n(w,x)$ diverges when $\Gamma(w,x)\ne 0$ for some $(w,x)\in\Omega$.
Under $\mathbb{H}_1$, the uniform law of large numbers guarantees that $\hat{U}_n(w,x)$ converges to the deterministic function $\Gamma(w,x)$ in probability. 
Since $\Gamma(w,x) \ne 0$ whenever the characteristic function of $W$ is discontinuous at the cutoff for some $(w,x)$, $\sqrt{nh} \hat{U}_n(w,x)$ diverges at some points $(w,x)\in\Omega$ with $\Gamma(w,x)\ne0$.
Consequently, the test is consistent against fixed alternatives that generate a nonzero ICM signal under the maintained assumptions, with rejection probability approaching 1 as $n$ increases.

\subsubsection{Against the sequence of local alternatives} \label{subsubsec:H1n}
When the alternatives are close to the null hypothesis, the performance of the test may be sharply different from that under the fixed alternative $\mathbb{H}_1$. 
We therefore study the asymptotic power under a sequence of local alternatives that approach the null, thereby exploring the theoretical limit of the test. 
Consider a Pitman-type sequence of alternatives: 
$$\mathbb{H}_{1n}: \varphi_{W|XR}(w|x,0^+) - \varphi_{W|XR}(w|x,0^-) = \frac{\lambda(w,x) }{\sqrt{nh}}, $$ 
where $\lambda(w,x)\ne 0$ for some $(w,x)\in\Omega$. This sequence of alternatives converges to the null at the rate $1/\sqrt{nh}$, which is the detection boundary of our procedure in the RD context.
The following theorem reveals the asymptotic behavior of $\sqrt{nh}\hat{U}_n(w,x)$ under the sequence $\mathbb{H}_{1n}$ and guarantees non-trivial power against local alternatives of this type.

\begin{theorem} \label{thm3} 
Suppose Assumptions A1--A5 and B2--B5 hold. Under the sequence of local alternatives $\mathbb{H}_{1n}$, 

(i) when Assumption B1 holds for the local constant estimator $\hat\tau(X_i)$, $$ \sqrt{nh} \hat{U}_n(w,x) \Longrightarrow U_{\infty,1}(w,x) + \Lambda(w,x);$$ 

(ii) when Assumption B1' holds for the local linear estimator $\hat\tau(X_i)$, $$ \sqrt{nh} \hat{U}_n(w,x) \Longrightarrow U_{\infty,2}(w,x) + \Lambda(w,x);$$ 

(iii) when Assumption B1 holds for the local constant estimator $\hat\tau(X_i)$ or Assumption B1' holds for the local linear estimator $\hat\tau(X_i)$, if $h=o(h_r)$, 
$$ \sqrt{nh} \hat{U}_n(w,x) \Longrightarrow U_{\infty,0}(w,x) + \Lambda(w,x),$$ 
where the drift term $\Lambda(w,x) = \int e^{\mathrm{i} x'\bar{x}} \lambda(w,\bar{x}) f_{XR}(\bar{x},0) d\bar{x} \ne 0$ for some $(w,x)\in\Omega$.
\end{theorem}

Under $\mathbb{H}_{1n}$, the limiting distribution of $\sqrt{nh} \hat{U}_n(w,x)$ differs from the null limit by the deterministic drift $\Lambda(w,x)\neq0$ for some $(w,x)\in\Omega$. 
While the stochastic fluctuation remains asymptotically identical to that under $\mathbb{H}_0$, the additional drift term $\Lambda(w,x)$ introduces a systematic deviation and shifts the limiting distribution away from the centered Gaussian process defined in Theorem \ref{thm1}. 
Consequently, the test statistics exhibit non-negligible deviations from $\mathbb{H}_0$ under $\mathbb{H}_{1n}$. 
This behavior is characteristic of global-smoothing tests, whose integrated structure allows weak but systematic discrepancies from the null hypothesis to accumulate asymptotically. 
In contrast, local-smoothing tests evaluate departures pointwise and require the alternative to converge to the null at a slower rate than the nonparametric smoothing rate. 
This difference in mechanism explains the power advantage of global-smoothing tests against local alternatives. In light of the continuous mapping theorem, under $\mathbb{H}_{1n}$, the test statistics satisfy
$$\operatorname{KS}_{n} \overset{d}{\rightarrow} \sup_{(w,x)\in\Omega} \max \left\{ \left\vert \mathrm{Re} \left( U_{\infty,j}(w,x) + \Lambda(w,x) \right) \right\vert, \left\vert \mathrm{Im} \left( U_{\infty,j}(w,x) + \Lambda(w,x) \right) \right\vert \right\}$$ 
and $$\operatorname{CvM}_{n} \overset{d}{\rightarrow} \int\left| U_{\infty,j}(w,x) + \Lambda(w,x) \right|^2 g(w,x)dwdx $$ for $j=0,1,2$.
The test statistics converge to a non‑centered Gaussian limit, guaranteeing non‑trivial local power at the $1/\sqrt{nh}$ detection boundary.

\section{Critical Values} \label{sec:bootstrap}
In this section, we introduce a multiplier bootstrap procedure to approximate the limiting distributions of the proposed test statistics and calculate their critical values. 
In Section \ref{sec:asymptotic}, Lemma \ref{Lemma1} provides linear representations of the $U$-process $\hat{U}_n(w,x)$, and Theorems \ref{thm1}--\ref{thm3} establish its limiting behavior under different hypotheses. 
However, the limiting distributions of the test statistics are difficult to evaluate analytically because they depend on the underlying data-generating process in a complex way. 
This motivates the use of a multiplier bootstrap procedure, which is both computationally efficient and theoretically valid. 
The resulting bootstrap sample allows us to compute the critical values of the test statistics.
In practice, we provide the following multiplier bootstrap procedure for the proposed tests.
\begin{itemize}
    \item \textbf{Step 1} \quad 
    Compute the feasible $U$-process $\hat U_n(w,x)$ using $\hat\tau(X_i)$ and obtain the test statistics $\mathrm{KS}_n$ and $\mathrm{CvM}_n$ as described in Section \ref{sec:framework}.
    \item \textbf{Step 2} \quad Draw multipliers $\{V_i\}_{i=1}^n$ and compute the bootstrap statistics $\mathrm{KS}_n^\ast$ and $\mathrm{CvM}_n^\ast$ based on the bootstrap process $\hat{U}_n^\ast(w,x)$. 
    \item \textbf{Step 3} \quad Repeat Step 2 $B$ times to generate the bootstrap statistics $\{\mathrm{KS}^\ast_{n,b}\}_{b=1}^B$ and $\{\mathrm{CvM}^\ast_{n,b}\}_{b=1}^B$. 
    \item \textbf{Step 4} \quad Reject the null hypothesis at significance level $\alpha$ if the test statistic exceeds the $(1-\alpha)$ quantile of the bootstrap statistics.
\end{itemize}
In Step 2, the multipliers $\{V_i\}_{i=1}^n$ are i.i.d. random variables with zero mean, unit variance, and bounded support, independent of the original sample path $\{(Y_i, D_i, X'_i, R_i)'\}_{i=1}^n$. 
Following \cite{Mammen1993}, one can adopt the i.i.d. two-point Bernoulli random variables $\{V_i\}_{i=1}^n$ with $\mathbb{P}(V_i = 1- \iota) = \iota/\sqrt{5}$ and $\mathbb{P}(V_i = \iota) = 1-\iota/\sqrt{5}$, where $\iota = (\sqrt{5} + 1)/2$. This sequence $\{V_i\}_{i=1}^n$ ensures that $\mathbb{E}[V_i]=0$, $\mathbb{E}[V_i^2]=1$, and $\mathbb{E}[V_i^3]=1$. 
As discussed in \cite{Escanciano2014} and \cite{santanna2019}, the multiplier bootstrap avoids repeated estimation in each bootstrap repetition. 
This feature makes the multiplier bootstrap straightforward to implement and computationally efficient, compared with alternative resampling methods. 

With the multipliers $\{V_i\}_{i=1}^n$, we define the bootstrap process to construct the bootstrap statistics. 
Note that the asymptotic linear representation varies across different estimators and bandwidth conditions. 
We introduce the bootstrap process and the test statistics here for different scenarios. The bootstrap process based on the local constant estimator is 
$$\hat{U}_n^\ast(w,x) = \frac{1}{n} \sum_{i=1}^n V_i e^{\mathrm{i} x'X_i} \left( e^{\mathrm{i} w\hat{W}_i } K_h(R_i) \delta_i + \mathrm{i} w \hat\kappa(w,X_i) (Y_i-D_i\hat\tau(X_i)) K_{h_r}(R_i) \delta_i \right), $$
with the leave-one-out estimator of $\kappa(w,X_i)$ given by 
$$ \hat\kappa(w,X_i) = \frac{\frac{1}{n-1} \sum_{j\ne i}^n e^{\mathrm{i}w\hat{W}_j} (1-D_j) K_{h_x}(X_j-X_i) K_{h_r}(R_j) \delta_j}{\frac{1}{n-1} \sum_{j\ne i}^n D_j K_{h_x}(X_j-X_i) K_{h_r}(R_j) \delta_j}.$$
As in $\hat\tau(X_i)$, the product kernel $K_{h_x}(x) = h_x^{-d} \prod_{k=1}^d K(x_k/h_x)$ is employed for the $d$-dimensional vector $x$ in conditional local smoothing.

When $h=o(h_r)$, the estimation effect is asymptotically dominated by the main term, and the bootstrap process based on either the local constant or the local linear estimator reduces to 
$$\hat{U}_n^\ast(w,x) = \frac{1}{n} \sum_{i=1}^n V_i e^{\mathrm{i} \left( x'X_i + w\hat{W}_i \right) } K_h(R_i) \delta_i.$$ 
Based on the bootstrap process $\hat{U}_n^\ast(w,x)$, the bootstrap statistics, $\mathrm{KS}_n^\ast$ and $\mathrm{CvM}_n^\ast$, can be computed as follows:
\begin{align*}
    \mathrm{KS}_n^\ast =& \sqrt{nh} \sup_{(w,x)\in\Omega} \max \left\{ \left| \mathrm{Re} (\hat{U}_n^\ast(w,x)) \right|, \left| \mathrm{Im} (\hat{U}_n^\ast(w,x)) \right| \right\}, \\
    \mathrm{CvM}_n^\ast =& \frac{h}{n} \sum_{i=1}^n \sum_{j=1}^n V_iV_j e^{-\frac{(\hat{W}_i - \hat{W}_j)^2 + \lVert X_i-X_j \rVert^2}{2}} K_h(R_i)K_h(R_j) \delta_i\delta_j.
\end{align*}
Note that the above expression for $\mathrm{CvM}_n^\ast$ applies when $h=o(h_r)$ and the estimation effect is asymptotically negligible. 
Otherwise, the corresponding bootstrap expression of the CvM statistic involves quadruple summations, which are computationally expensive.
The bootstrap samples closely mimic the theoretical distribution of the test statistics under the null, facilitating robust inference in finite samples. 
To establish the validity of the multiplier bootstrap, we first investigate the estimation effect of the bootstrap process in the following lemma.

\begin{lemma} \label{Lemma2} 
Suppose Assumptions A1--A5 and B2--B5 hold.

(i) Suppose Assumption B1 holds. Then the infeasible bootstrap process with the estimation effect of the local constant estimator is 
$$U_n^\ast(w,x) = \frac{1}{n} \sum_{i=1}^n V_i e^{\mathrm{i} x'X_i} \left( e^{\mathrm{i} wW_i } K_h(R_i) \delta_i + \mathrm{i} w \kappa(w,X_i) (Y_i-D_i\tau(X_i)) K_{h_r}(R_i) \delta_i \right), $$ 
where $W, \kappa$ and $\tau$ in the formula are the true values.

(ii) Suppose Assumption B1 holds for the local constant estimator or Assumption B1' holds for the local linear estimator. When $h=o(h_r)$, the infeasible bootstrap process is 
$$U_n^\ast(w,x) = \frac{1}{n} \sum_{i=1}^n V_i e^{\mathrm{i} \left( x'X_i + wW_i \right) } K_h(R_i) \delta_i. $$
In either case (i) or (ii), the difference between the bootstrap process $\hat{U}_n^\ast(w,x)$ and its infeasible version $U_n^\ast(w,x)$ is uniformly bounded: 
$$ \sqrt{nh} \sup_{(w,x)\in\Omega} \left| \hat{U}_n^\ast(w,x) - U_n^\ast(w,x)  \right| = o_p(1). $$
\end{lemma}

\textbf{Remark} \quad 
Lemma \ref{Lemma2} establishes the asymptotic equivalence between the feasible and infeasible bootstrap processes. 
The equivalence holds for the local constant estimator under Assumption B1 and, when $h=o(h_r)$, for both the local constant and local linear estimators. 
The corresponding result for the local linear estimator when $h=O(h_r)$ can be established, but is not demonstrated here to avoid excessive technical complexity.

The next theorem establishes the asymptotic validity of the multiplier bootstrap procedure above.

\begin{theorem} \label{thm4}
Suppose Assumptions A1--A5 and B2--B5 hold.

(i) Suppose Assumption B1 holds. Then the bootstrap process $\hat{U}_n^\ast(w,x)$ with the local constant estimator converges weakly on $\Omega$ under $\mathbb{H}_0$, $\mathbb{H}_1$, or $\mathbb{H}_{1n}$: 
$$\sqrt{nh} \hat{U}_n^\ast(w,x) \Longrightarrow_\ast U_{\infty,1}(w,x), $$ 
where $\Longrightarrow_\ast$ denotes weak convergence under the bootstrap law. 

(ii) Suppose Assumption B1 holds for the local constant estimator or Assumption B1' holds for the local linear estimator. When $h=o(h_r)$, the bootstrap process $\hat{U}_n^\ast(w,x)$ converges weakly on $\Omega$ under $\mathbb{H}_0$, $\mathbb{H}_1$, or $\mathbb{H}_{1n}$: 
$$\sqrt{nh} \hat{U}_n^\ast(w,x) \Longrightarrow_\ast U_{\infty,0}(w,x). $$ 
\end{theorem}  

The bootstrap validity established above is guaranteed by the conditional multiplier central limit theorem for $U$-processes. 
An interesting feature of Theorem~\ref{thm4} is that the bootstrap process converges weakly to the same centered Gaussian process under $\mathbb{H}_0$, $\mathbb{H}_1$, and $\mathbb{H}_{1n}$. 
This unification is uncommon in conventional settings, where the bootstrap process under the alternative typically converges to a different limit, although still a centered Gaussian process. 
The distinctive behavior here follows from the nonparametric convergence rate $1/\sqrt{nh}$. 
Under $\mathbb{H}_1$, the feasible $U$-process $\hat{U}_n(w,x)$ contains a nonzero deterministic component $\Gamma(w,x)$, which diverges when scaled by $\sqrt{nh}$. 
However, the key insight is that the multiplier bootstrap $U$-process perturbs each summand by a mean-zero multiplier $V_i$; the signal $\Gamma(w,x)$ resides in the empirical mean of the summands, and because $\mathbb{E}[V_i]=0$ it does not enter the bootstrap mean. 
Formally, conditional on the original sample, the expectation of the bootstrap process $\hat{U}_n^\ast(w,x)$ is approximately zero regardless of whether $\mathbb{H}_0$ or $\mathbb{H}_1$ holds, because $V_i$ is independent of the sample path with zero mean. 
Therefore, the signal $\Gamma(w,x)$ does not enter its mean. The component of the bootstrap process that carries the signal is of order $\sqrt{h}$ (and thus asymptotically negligible). 
The detailed calculation in the Appendix shows that the bootstrap limit is governed by the same functional form as the limiting null process, with the covariance kernel $\mathcal{K}_j$ of Theorem~\ref{thm1}. 
Consequently, $\sqrt{nh}\,\hat{U}_n^\ast(w,x)$ converges to the same centered Gaussian limit under $\mathbb{H}_1$ as under $\mathbb{H}_0$. 

A practical consequence is that the bootstrap critical values remain stochastically bounded regardless of which hypothesis holds, while the test statistics diverge at the rate $\sqrt{nh}$ under $\mathbb{H}_1$, ensuring that the rejection probability of the proposed test approaches 1 against any fixed alternative as $n$ increases. 
This phenomenon highlights a distinct feature of nonparametric inference in RD designs: the local identification at the cutoff forces a slow convergence rate, which in turn produces a multiplier bootstrap procedure that is uniformly valid across the null, the alternative, and the local alternatives.

\section{Numerical Study} \label{sec:simulation} 
In this section, Monte Carlo experiments are conducted to evaluate the finite-sample performance of the proposed test. 
We consider a set of data-generating processes (DGPs) adapted from \cite{Hsu2019, Hsu2021}.
In all DGPs, the running variable $R$, covariate $X$, and error term $\epsilon$ are independently generated from the following distributions: 
$$ R \sim 2\text{Beta}(2,2)-1, \quad\quad X \sim \text{Unif}[0,1], \quad\quad \epsilon \sim N(0,1). $$
For each DGP, the experiments are repeated 1000 times. 
We follow the finite-sample guidance on bandwidth selection in \cite{Calonico2014} and \cite{Calonico2019}, which provide MSE-optimal bandwidths for the RD estimator. 
For the estimating bandwidths $h_r$ and $h_x$, we start with the MSE-optimal bandwidth $\hat{h}_{\text{MSE}}$ obtained from the procedure of \cite{Calonico2014} and apply under-smoothing: we set $h_r = h_x = \hat{h}_{\text{MSE}} \times n^{1/5-1/k_1}$ with $k_1 \in [4, 4.5]$, where larger $k_1$ produces a larger bandwidth. 
For the testing bandwidth $h$, we set $h = k_2 h_r$ with $k_2 \in [0.9, 1.1]$. When $k_2 = 1$, the estimation effect and the main $U$-process contribute equally to the limiting distribution; when $k_2 < 1$, the estimation effect is partially attenuated relative to the main term. 
For the configurations used in our simulations, the effective sample size $nh$ ranges from approximately 100 to 700. 
The triangular kernel function $K(u) = (1-|u|) \mathbf{1}(|u|\le 1)$ is used for both the testing and estimating procedures, following the recommendations in \cite{Imbens2012a}, \cite{Calonico2014}, \cite{Calonico2019} and \cite{Hsu2019}. 
A total of 1000 bootstrap simulations are conducted within each repetition to compute the critical values. 

The test statistics are calculated using the $U$-process, which incorporates the estimation effect of the local constant estimator. 
The KS statistic is calculated over a grid of sample points. 
Additional results using the local linear estimator are reported in the Appendix. 
The outcome $Y$ and binary treatment $D$ are generated according to DGP-specific mechanisms listed below:

DGP1: \textit{Homogeneous Zero Treatment Effects}. 
\begin{align*}
    Y &= -0.555 - 0.553X + 0.581R + 0.060XR - 0.058R^2 + 1.074X^2 + 0.1\epsilon. \\
    D &= \mathbf{1}(R\ge0) \mathbf{1}(0.596 - 2.103X + 0.128R + 0.352XR + 0.013R^2 + 2.454X^2 + \epsilon > 0).
\end{align*}
The treatment effect under DGP1 is $\tau(x)=0$. 

DGP2: \textit{Homogeneous Nonzero Constant Treatment Effects}.
\begin{align*}
    Y &= \begin{cases}
         -0.373 + 0.545R - 0.056R^2 + 0.1\epsilon, & R\ge0, \\
         -0.531 + 0.556R - 0.192R^2 + 0.1\epsilon, & R<0.        \end{cases} \\
    D &= \mathbf{1}(R\ge0) \mathbf{1}(0.331 + 0.277R + 0.049R^2 + \epsilon > 0).
\end{align*}
The nonzero constant treatment effect in DGP2 is 
$$\tau(x) = \frac{\mu(x,0^+)-\mu(x,0^-)}{p(x,0^+)-p(x,0^-)}= \frac{0.158}{1-\Phi(-0.331)} \approx 0.25,$$ 
where $\Phi(\cdot)$ denotes the CDF of the standard normal distribution. 

DGP3: \textit{Homogeneous Functional Treatment Effects}. 
\begin{align*}
    Y &= \begin{cases}
         -0.921 - 4X + 0.584R - 0.054R^2 + 5X^2 + 0.1\epsilon, & R\ge0, \\
         -0.705 + 0.264X + 0.580R + 0.191R^2 + 0.1\epsilon, & R<0.          \end{cases} \\
    D &= \mathbf{1}(R\ge0) \mathbf{1}(0.331 + 0.277R + 0.049R^2 + \epsilon > 0).
\end{align*}

The functional treatment effect in DGP3 is 
$$\tau(x)=\frac{5x^2-4.264x-0.216}{1-\Phi(-0.331)}.$$ 

DGP4: \textit{Heterogeneous Treatment Effects}: a 1:1 mixture of DGP1 (zero treatment effect) and DGP3 (functional treatment effect).

DGP5: \textit{Heterogeneous Treatment Effects}: a 1:1 mixture of DGP2 (constant treatment effect) and DGP3 (functional treatment effect). 

DGP6: \textit{Heterogeneous Treatment Effects}: a 1:1:1 mixture of DGP1 (zero treatment effect), DGP2 (nonzero constant treatment effect) and DGP3 (functional treatment effect). 

DGPs 1--3, adapted from \cite{Hsu2019, Hsu2021}, represent scenarios with zero, nonzero constant, and functional treatment effects, respectively, all in the absence of unobserved heterogeneity. 
These DGPs are used to evaluate the size of the proposed test under the null hypothesis of no unobserved heterogeneity. 
DGPs 4--6 are constructed as mixtures of DGPs 1--3 to introduce unobserved heterogeneity and assess the test's power under complex scenarios. 
Specifically, for each individual $i$, a latent group membership variable $G_i$ is drawn independently of $(R_i, X_i, \epsilon_i)$, with the appropriate probabilities to match the stated mixture proportions (e.g., $G_i \in \{1,3\}$ each with probability $1/2$ for DGP4). 
For a given value of $G_i$, the pair $(Y_i, D_i)$ is generated from the corresponding component DGP using a common draw of the base variables $(R_i, X_i, \epsilon_i)$. In all mixture DGPs, the treatment effect at the cutoff is $Y_{1i} - Y_{0i} = \tau_{G_i}(X_i)$, where $\tau_g(X_i)$ denotes the treatment effect under component DGP $g$. Because $G_i$ is independent of the observables, the conditional local average treatment effect $\tau(X_i) = \mathbb{E}[Y_{1i} - Y_{0i} \mid X_i, R_i=0]$ is the weighted average of the component-specific conditional treatment effects. 
Nevertheless, conditional on $X_i$, the individual-level treatment effect $Y_{1i} - Y_{0i}$ still varies with the unobserved $G_i$, thus violating $\mathbb{H}_0$. Note that each component $g$ specifies $g(d,0,x,\epsilon) = m_g(d,x) + 0.1\epsilon$, which is additively separable in $\epsilon$ with a monotone index $\nu(x,\epsilon) = \epsilon$. 
Across the mixture, however, the group indicator $G$ acts as an additional latent factor that shifts both the outcome level $m_G(0,x)$ and the treatment effect $\tau_G(x)$, so the mixture as a whole need not satisfy the single-index restriction in Assumption A4. 
As discussed in Section~\ref{subsec:hypothesis}, Assumption A4 is used to characterize the null hypothesis and to delineate the class of alternatives against which the test has power; size control under $\mathbb{H}_0$ does not require it to hold under the alternative, and the mixture DGPs are therefore used purely as alternatives that violate the continuity restriction on the imputed outcome $W$.
DGP4, which mixes zero and functional treatment effects, creates the sharpest contrast between components and is expected to yield the highest power; DGP5, which mixes constant and functional effects, represents a subtler form of heterogeneity; DGP6 combines all three components, further diluting the signal and posing the greatest challenge to the test.
These mixture scenarios are challenging but can provide insights for evaluating the power of the proposed test statistics.
The numerical results are presented in Tables \ref{tab:DGP_null_lc} and \ref{tab:DGP_alt_lc}.

\begin{table}[H]
    \centering
    \caption{Rejection rates under the null hypothesis (DGPs 1--3)}
    \label{tab:DGP_null_lc}
    \setlength{\extrarowheight}{-1pt}
    \begin{adjustbox}{max width=0.99\textwidth, max height=\textheight}
    \begin{tabular}{@{}cccccccccccccccccc@{}}
    \toprule
    \multirow{2}{*}{DGP} & \multirow{2}{*}{$k_1$} & $k_2$ & \multicolumn{3}{c}{0.9} & \multicolumn{3}{c}{0.95} & \multicolumn{3}{c}{1} & \multicolumn{3}{c}{1.05} & \multicolumn{3}{c}{1.1} \\
    \cmidrule{4-18}
     & & $n$ & 1\% & 5\% & 10\% & 1\% & 5\% & 10\% & 1\% & 5\% & 10\% & 1\% & 5\% & 10\% & 1\% & 5\% & 10\% \\
    \midrule
    \multirow{12}{*}{1} & \multirow{4}{*}{4} & 500 & 0.004 & 0.031 & 0.057 & 0.004 & 0.032 & 0.056 & 0.005 & 0.032 & 0.059 & 0.005 & 0.032 & 0.062 & 0.004 & 0.036 & 0.065 \\
     &  & 1000 & 0.006 & 0.031 & 0.065 & 0.006 & 0.030 & 0.068 & 0.008 & 0.031 & 0.068 & 0.008 & 0.036 & 0.070 & 0.011 & 0.038 & 0.072 \\
     &  & 2000 & 0.011 & 0.048 & 0.089 & 0.011 & 0.051 & 0.091 & 0.013 & 0.055 & 0.093 & 0.013 & 0.053 & 0.095 & 0.012 & 0.055 & 0.097 \\
     &  & 4000 & 0.015 & 0.057 & 0.092 & 0.017 & 0.056 & 0.092 & 0.016 & 0.057 & 0.095 & 0.018 & 0.061 & 0.098 & 0.018 & 0.065 & 0.100 \\
    \cmidrule{2-18}
     & \multirow{4}{*}{4.25} & 500 & 0.005 & 0.029 & 0.061 & 0.004 & 0.031 & 0.064 & 0.005 & 0.033 & 0.066 & 0.004 & 0.036 & 0.068 & 0.006 & 0.037 & 0.068 \\
     &  & 1000 & 0.009 & 0.029 & 0.076 & 0.010 & 0.031 & 0.076 & 0.010 & 0.035 & 0.074 & 0.010 & 0.038 & 0.074 & 0.012 & 0.041 & 0.077 \\
     &  & 2000 & 0.011 & 0.055 & 0.097 & 0.013 & 0.054 & 0.097 & 0.012 & 0.053 & 0.095 & 0.012 & 0.052 & 0.097 & 0.012 & 0.054 & 0.104 \\
     &  & 4000 & 0.014 & 0.057 & 0.097 & 0.015 & 0.058 & 0.098 & 0.016 & 0.061 & 0.100 & 0.017 & 0.062 & 0.100 & 0.019 & 0.061 & 0.106 \\
    \cmidrule{2-18}
     & \multirow{4}{*}{4.5} & 500 & 0.004 & 0.032 & 0.064 & 0.005 & 0.031 & 0.064 & 0.006 & 0.036 & 0.067 & 0.006 & 0.037 & 0.069 & 0.006 & 0.039 & 0.070 \\
     &  & 1000 & 0.010 & 0.031 & 0.074 & 0.009 & 0.035 & 0.076 & 0.010 & 0.039 & 0.078 & 0.011 & 0.041 & 0.079 & 0.012 & 0.044 & 0.079 \\
     &  & 2000 & 0.010 & 0.055 & 0.096 & 0.011 & 0.057 & 0.099 & 0.012 & 0.056 & 0.101 & 0.013 & 0.054 & 0.104 & 0.015 & 0.055 & 0.108 \\
     &  & 4000 & 0.013 & 0.056 & 0.100 & 0.016 & 0.055 & 0.104 & 0.017 & 0.059 & 0.102 & 0.018 & 0.062 & 0.107 & 0.018 & 0.064 & 0.107 \\
    \midrule
    \multirow{12}{*}{2} & \multirow{4}{*}{4} & 500 & 0.004 & 0.023 & 0.046 & 0.005 & 0.025 & 0.047 & 0.005 & 0.023 & 0.047 & 0.007 & 0.024 & 0.051 & 0.007 & 0.025 & 0.051 \\
     &  & 1000 & 0.005 & 0.021 & 0.051 & 0.005 & 0.021 & 0.052 & 0.005 & 0.023 & 0.051 & 0.007 & 0.024 & 0.050 & 0.008 & 0.025 & 0.053 \\
     &  & 2000 & 0.013 & 0.042 & 0.096 & 0.012 & 0.042 & 0.095 & 0.011 & 0.043 & 0.093 & 0.012 & 0.044 & 0.088 & 0.011 & 0.045 & 0.089 \\
     &  & 4000 & 0.017 & 0.045 & 0.089 & 0.017 & 0.046 & 0.086 & 0.013 & 0.047 & 0.090 & 0.013 & 0.049 & 0.092 & 0.014 & 0.051 & 0.093 \\
    \cmidrule{2-18}
     & \multirow{4}{*}{4.25} & 500 & 0.005 & 0.023 & 0.050 & 0.005 & 0.023 & 0.052 & 0.007 & 0.025 & 0.054 & 0.006 & 0.024 & 0.056 & 0.007 & 0.025 & 0.056 \\
     &  & 1000 & 0.006 & 0.025 & 0.053 & 0.007 & 0.026 & 0.058 & 0.009 & 0.026 & 0.058 & 0.009 & 0.027 & 0.057 & 0.010 & 0.028 & 0.057 \\
     &  & 2000 & 0.010 & 0.046 & 0.097 & 0.011 & 0.044 & 0.092 & 0.011 & 0.048 & 0.090 & 0.011 & 0.048 & 0.089 & 0.011 & 0.049 & 0.089 \\
     &  & 4000 & 0.013 & 0.048 & 0.092 & 0.013 & 0.050 & 0.090 & 0.012 & 0.052 & 0.094 & 0.014 & 0.054 & 0.096 & 0.014 & 0.057 & 0.100 \\
    \cmidrule{2-18}
     & \multirow{4}{*}{4.5} & 500 & 0.006 & 0.026 & 0.054 & 0.006 & 0.028 & 0.054 & 0.006 & 0.027 & 0.054 & 0.006 & 0.027 & 0.057 & 0.006 & 0.025 & 0.058 \\
     &  & 1000 & 0.008 & 0.026 & 0.059 & 0.009 & 0.026 & 0.058 & 0.008 & 0.029 & 0.057 & 0.009 & 0.029 & 0.060 & 0.010 & 0.029 & 0.060 \\
     &  & 2000 & 0.009 & 0.046 & 0.095 & 0.010 & 0.050 & 0.096 & 0.010 & 0.051 & 0.097 & 0.010 & 0.048 & 0.096 & 0.010 & 0.048 & 0.097 \\
     &  & 4000 & 0.015 & 0.050 & 0.093 & 0.015 & 0.053 & 0.090 & 0.014 & 0.054 & 0.092 & 0.013 & 0.055 & 0.098 & 0.014 & 0.058 & 0.099 \\
    \midrule
    \multirow{12}{*}{3} & \multirow{4}{*}{4} & 500 & 0.003 & 0.019 & 0.046 & 0.003 & 0.021 & 0.046 & 0.003 & 0.023 & 0.049 & 0.004 & 0.024 & 0.052 & 0.005 & 0.023 & 0.050 \\
     &  & 1000 & 0.007 & 0.027 & 0.052 & 0.007 & 0.026 & 0.054 & 0.009 & 0.026 & 0.055 & 0.009 & 0.028 & 0.056 & 0.008 & 0.029 & 0.057 \\
     &  & 2000 & 0.009 & 0.038 & 0.082 & 0.009 & 0.038 & 0.085 & 0.009 & 0.041 & 0.085 & 0.008 & 0.041 & 0.085 & 0.008 & 0.041 & 0.085 \\
     &  & 4000 & 0.015 & 0.045 & 0.087 & 0.014 & 0.044 & 0.083 & 0.016 & 0.046 & 0.084 & 0.016 & 0.047 & 0.086 & 0.017 & 0.047 & 0.089 \\
    \cmidrule{2-18}
     & \multirow{4}{*}{4.25} & 500 & 0.005 & 0.022 & 0.050 & 0.004 & 0.024 & 0.053 & 0.004 & 0.025 & 0.053 & 0.005 & 0.027 & 0.052 & 0.005 & 0.029 & 0.054 \\
     &  & 1000 & 0.008 & 0.028 & 0.056 & 0.008 & 0.029 & 0.059 & 0.008 & 0.029 & 0.057 & 0.009 & 0.029 & 0.060 & 0.009 & 0.031 & 0.060 \\
     &  & 2000 & 0.008 & 0.047 & 0.087 & 0.008 & 0.047 & 0.089 & 0.009 & 0.049 & 0.088 & 0.009 & 0.047 & 0.086 & 0.010 & 0.048 & 0.086 \\
     &  & 4000 & 0.016 & 0.050 & 0.086 & 0.015 & 0.048 & 0.091 & 0.015 & 0.047 & 0.094 & 0.015 & 0.049 & 0.096 & 0.015 & 0.052 & 0.097 \\
    \cmidrule{2-18}
     & \multirow{4}{*}{4.5} & 500 & 0.004 & 0.026 & 0.050 & 0.004 & 0.029 & 0.053 & 0.004 & 0.032 & 0.054 & 0.004 & 0.030 & 0.057 & 0.005 & 0.031 & 0.056 \\
     &  & 1000 & 0.008 & 0.026 & 0.063 & 0.008 & 0.028 & 0.060 & 0.008 & 0.031 & 0.064 & 0.009 & 0.029 & 0.062 & 0.010 & 0.029 & 0.067 \\
     &  & 2000 & 0.008 & 0.052 & 0.097 & 0.009 & 0.053 & 0.096 & 0.009 & 0.050 & 0.097 & 0.011 & 0.049 & 0.098 & 0.013 & 0.049 & 0.097 \\
     &  & 4000 & 0.014 & 0.048 & 0.098 & 0.012 & 0.049 & 0.096 & 0.013 & 0.048 & 0.096 & 0.013 & 0.053 & 0.097 & 0.015 & 0.057 & 0.098 \\
    \bottomrule
    \end{tabular}
    \end{adjustbox}
\end{table}

Table \ref{tab:DGP_null_lc} reports the rejection rates for DGPs 1--3 under the null hypothesis $\mathbb{H}_0$. 
The results indicate that the proposed test maintains an appropriate size across varying values of the tuning parameters $k_1$ for the estimating bandwidth and $k_2$ for the testing bandwidth. 
The rejection rates approach the nominal levels as the sample size grows under all DGPs, demonstrating size control and robustness of the proposed test. 
As is typical for nonparametric tests based on kernel methods, the performance is affected by the choice of bandwidth.

Table \ref{tab:DGP_alt_lc} reports the rejection rates for DGPs 4--6 under the alternative hypothesis $\mathbb{H}_1$. 
In all cases, the rejection rates increase with sample size. The rejection rate is relatively insensitive to the choice of $k_2$, but depends on $k_1$. 
As $k_1$ increases, the rejection rates increase. 
As established in Lemma \ref{Lemma1}, $k_1$ determines the magnitude of the estimation effect, thereby affecting the performance of the test.

A comparison of empirical power across DGPs 4--6 shows that the rejection rates under DGP4 are consistently the highest given the same bandwidth setting, because it creates the sharpest contrast between components. 
DGP5 (constant versus functional) and DGP6 (three-way mixture) exhibit lower power, consistent with their more diluted heterogeneity signals, as anticipated in the DGP designs.
Overall, the results in Tables \ref{tab:DGP_null_lc}--\ref{tab:DGP_alt_lc} demonstrate that the proposed test maintains appropriate size under $\mathbb{H}_0$ and achieves increasing power under $\mathbb{H}_1$, which aligns closely with the theoretical properties in Section \ref{sec:asymptotic}.

\begin{table}[H]
    \centering
    \caption{Rejection rates under the alternative hypothesis (DGPs 4--6)}
    \label{tab:DGP_alt_lc}
    \setlength{\extrarowheight}{-1pt}
    \begin{adjustbox}{max width=0.99\textwidth, max height=\textheight}
    \begin{tabular}{@{}cccccccccccccccccc@{}}
    \toprule
    \multirow{2}{*}{DGP} & \multirow{2}{*}{$k_1$} & $k_2$ & \multicolumn{3}{c}{0.9} & \multicolumn{3}{c}{0.95} & \multicolumn{3}{c}{1} & \multicolumn{3}{c}{1.05} & \multicolumn{3}{c}{1.1} \\
    \cmidrule{4-18}
    & & $n$ & 1\% & 5\% & 10\% & 1\% & 5\% & 10\% & 1\% & 5\% & 10\% & 1\% & 5\% & 10\% & 1\% & 5\% & 10\% \\
    \midrule
    \multirow{12}{*}{4} & \multirow{4}{*}{4} & 500 & 0.002 & 0.025 & 0.067 & 0.001 & 0.028 & 0.069 & 0.002 & 0.033 & 0.072 & 0.002 & 0.039 & 0.077 & 0.002 & 0.041 & 0.085 \\
     &  & 1000 & 0.034 & 0.235 & 0.406 & 0.034 & 0.252 & 0.422 & 0.043 & 0.274 & 0.435 & 0.055 & 0.287 & 0.443 & 0.057 & 0.296 & 0.449 \\
     &  & 2000 & 0.398 & 0.748 & 0.830 & 0.429 & 0.762 & 0.833 & 0.445 & 0.767 & 0.838 & 0.457 & 0.769 & 0.842 & 0.471 & 0.767 & 0.845 \\
     &  & 4000 & 0.925 & 0.982 & 0.990 & 0.927 & 0.982 & 0.989 & 0.930 & 0.983 & 0.988 & 0.932 & 0.984 & 0.987 & 0.931 & 0.985 & 0.987 \\
    \cmidrule{2-18}
     & \multirow{4}{*}{4.25} & 500 & 0.002 & 0.034 & 0.087 & 0.002 & 0.038 & 0.095 & 0.002 & 0.043 & 0.103 & 0.003 & 0.048 & 0.112 & 0.004 & 0.057 & 0.124 \\
     &  & 1000 & 0.053 & 0.323 & 0.505 & 0.060 & 0.345 & 0.517 & 0.076 & 0.366 & 0.528 & 0.087 & 0.378 & 0.532 & 0.101 & 0.393 & 0.537 \\
     &  & 2000 & 0.525 & 0.821 & 0.888 & 0.562 & 0.820 & 0.899 & 0.571 & 0.823 & 0.899 & 0.591 & 0.828 & 0.906 & 0.607 & 0.833 & 0.906 \\
     &  & 4000 & 0.960 & 0.991 & 0.996 & 0.962 & 0.991 & 0.996 & 0.960 & 0.992 & 0.996 & 0.964 & 0.991 & 0.995 & 0.963 & 0.991 & 0.995 \\
    \cmidrule{2-18}
     & \multirow{4}{*}{4.5} & 500 & 0.002 & 0.040 & 0.120 & 0.003 & 0.046 & 0.131 & 0.003 & 0.056 & 0.147 & 0.004 & 0.066 & 0.155 & 0.005 & 0.071 & 0.173 \\
     &  & 1000 & 0.079 & 0.397 & 0.600 & 0.093 & 0.417 & 0.604 & 0.108 & 0.444 & 0.613 & 0.122 & 0.457 & 0.616 & 0.149 & 0.472 & 0.621 \\
     &  & 2000 & 0.646 & 0.868 & 0.929 & 0.662 & 0.875 & 0.933 & 0.675 & 0.879 & 0.936 & 0.688 & 0.887 & 0.939 & 0.698 & 0.889 & 0.939 \\
     &  & 4000 & 0.979 & 0.996 & 0.999 & 0.983 & 0.996 & 0.999 & 0.982 & 0.996 & 0.999 & 0.981 & 0.996 & 0.998 & 0.983 & 0.996 & 0.998 \\
    \midrule
    \multirow{12}{*}{5} & \multirow{4}{*}{4} & 500 & 0.001 & 0.020 & 0.050 & 0.001 & 0.020 & 0.055 & 0.001 & 0.023 & 0.059 & 0.003 & 0.027 & 0.063 & 0.003 & 0.027 & 0.068 \\
     &  & 1000 & 0.015 & 0.129 & 0.273 & 0.016 & 0.142 & 0.285 & 0.021 & 0.150 & 0.299 & 0.027 & 0.154 & 0.313 & 0.030 & 0.168 & 0.317 \\
     &  & 2000 & 0.228 & 0.585 & 0.727 & 0.244 & 0.597 & 0.735 & 0.263 & 0.608 & 0.735 & 0.280 & 0.615 & 0.740 & 0.294 & 0.630 & 0.740 \\
     &  & 4000 & 0.836 & 0.951 & 0.972 & 0.841 & 0.954 & 0.972 & 0.843 & 0.956 & 0.974 & 0.846 & 0.956 & 0.978 & 0.852 & 0.956 & 0.978 \\
    \cmidrule{2-18}
     & \multirow{4}{*}{4.25} & 500 & 0.001 & 0.025 & 0.063 & 0.002 & 0.028 & 0.073 & 0.003 & 0.030 & 0.079 & 0.004 & 0.031 & 0.086 & 0.004 & 0.034 & 0.094 \\
     &  & 1000 & 0.024 & 0.185 & 0.354 & 0.028 & 0.200 & 0.372 & 0.035 & 0.211 & 0.382 & 0.044 & 0.222 & 0.391 & 0.048 & 0.243 & 0.402 \\
     &  & 2000 & 0.347 & 0.685 & 0.815 & 0.361 & 0.701 & 0.825 & 0.382 & 0.709 & 0.827 & 0.397 & 0.712 & 0.826 & 0.414 & 0.714 & 0.833 \\
     &  & 4000 & 0.900 & 0.971 & 0.989 & 0.903 & 0.972 & 0.990 & 0.909 & 0.974 & 0.991 & 0.910 & 0.976 & 0.989 & 0.911 & 0.976 & 0.989 \\
    \cmidrule{2-18}
     & \multirow{4}{*}{4.5} & 500 & 0.003 & 0.029 & 0.083 & 0.003 & 0.036 & 0.095 & 0.003 & 0.040 & 0.104 & 0.003 & 0.041 & 0.114 & 0.004 & 0.050 & 0.118 \\
     &  & 1000 & 0.034 & 0.250 & 0.439 & 0.039 & 0.260 & 0.468 & 0.047 & 0.275 & 0.472 & 0.061 & 0.297 & 0.478 & 0.067 & 0.316 & 0.490 \\
     &  & 2000 & 0.456 & 0.771 & 0.866 & 0.473 & 0.779 & 0.874 & 0.486 & 0.791 & 0.881 & 0.503 & 0.800 & 0.884 & 0.512 & 0.808 & 0.885 \\
     &  & 4000 & 0.939 & 0.987 & 0.994 & 0.941 & 0.990 & 0.994 & 0.945 & 0.991 & 0.995 & 0.946 & 0.990 & 0.996 & 0.947 & 0.989 & 0.996 \\
    \midrule
    \multirow{12}{*}{6} & \multirow{4}{*}{4} & 500 & 0.001 & 0.016 & 0.053 & 0.001 & 0.019 & 0.054 & 0.001 & 0.022 & 0.060 & 0.001 & 0.024 & 0.069 & 0.002 & 0.029 & 0.072 \\
     &  & 1000 & 0.016 & 0.105 & 0.184 & 0.018 & 0.110 & 0.197 & 0.020 & 0.117 & 0.214 & 0.030 & 0.118 & 0.224 & 0.036 & 0.132 & 0.237 \\
     &  & 2000 & 0.126 & 0.447 & 0.641 & 0.137 & 0.461 & 0.658 & 0.153 & 0.464 & 0.676 & 0.169 & 0.478 & 0.692 & 0.185 & 0.494 & 0.703 \\
     &  & 4000 & 0.693 & 0.918 & 0.958 & 0.715 & 0.925 & 0.962 & 0.725 & 0.932 & 0.968 & 0.737 & 0.938 & 0.969 & 0.744 & 0.940 & 0.971 \\
    \cmidrule{2-18}
     & \multirow{4}{*}{4.25} & 500 & 0.001 & 0.021 & 0.069 & 0.002 & 0.024 & 0.076 & 0.002 & 0.026 & 0.081 & 0.003 & 0.032 & 0.087 & 0.003 & 0.036 & 0.093 \\
     &  & 1000 & 0.017 & 0.128 & 0.250 & 0.026 & 0.136 & 0.268 & 0.031 & 0.141 & 0.278 & 0.040 & 0.155 & 0.284 & 0.046 & 0.167 & 0.297 \\
     &  & 2000 & 0.170 & 0.554 & 0.738 & 0.193 & 0.566 & 0.754 & 0.214 & 0.578 & 0.767 & 0.228 & 0.597 & 0.775 & 0.237 & 0.612 & 0.785 \\
     &  & 4000 & 0.806 & 0.952 & 0.974 & 0.825 & 0.958 & 0.979 & 0.832 & 0.963 & 0.979 & 0.842 & 0.966 & 0.980 & 0.848 & 0.969 & 0.982 \\
    \cmidrule{2-18}
     & \multirow{4}{*}{4.5} & 500 & 0.002 & 0.028 & 0.082 & 0.003 & 0.033 & 0.084 & 0.003 & 0.037 & 0.092 & 0.003 & 0.043 & 0.099 & 0.005 & 0.053 & 0.113 \\
     &  & 1000 & 0.023 & 0.150 & 0.319 & 0.034 & 0.160 & 0.329 & 0.045 & 0.178 & 0.337 & 0.051 & 0.194 & 0.347 & 0.059 & 0.206 & 0.364 \\
     &  & 2000 & 0.240 & 0.651 & 0.825 & 0.261 & 0.669 & 0.836 & 0.268 & 0.682 & 0.848 & 0.290 & 0.696 & 0.855 & 0.297 & 0.710 & 0.857 \\
     &  & 4000 & 0.877 & 0.968 & 0.991 & 0.880 & 0.971 & 0.992 & 0.890 & 0.974 & 0.993 & 0.900 & 0.976 & 0.995 & 0.903 & 0.977 & 0.996 \\
    \bottomrule
    \end{tabular}
    \end{adjustbox}
\end{table}

\section{Empirical Applications} \label{sec:empirical}
In addition to the numerical evidence in Section \ref{sec:simulation}, we illustrate the performance of the proposed test in empirical applications. 
The two empirical applications considered here, Islamic election and high-school attendance, represent canonical sharp and fuzzy RD designs in political and education economics. 
Applying our testing procedure to these well-studied datasets allows us to assess whether treatment effect heterogeneity is driven by unobserved characteristics. In each application, we implement the test in two steps. 
We first examine unconditional homogeneity at the cutoff and then test the conditional restriction after introducing the key covariates from the original study. 
Comparing the two results reveals whether the observed heterogeneity is attributable to covariates or to unobserved factors. 
Throughout this section, rejection and non-rejection are interpreted in terms of the identifying and regularity conditions that the test maintains.

\subsection{The treatment effect of the election of an Islamic party} \label{subsec:Meyersson} 
In this subsection, we revisit \cite{meyersson2014}, which studies the consequences of the election of an Islamic party on high-school completion via an RD design.
\cite{meyersson2014} finds that the election of the Islamic party in Turkey causally increased female secular high school completion and reveals an underlying mechanism. 
The paper further demonstrates treatment effect heterogeneity across subpopulations defined by proxy variables for poverty and Islamic conservatism and shows that the treatment effect is larger in more religiously conservative communities.

In this application, the running variable $R$ is the Islamic win margin---the vote share difference between the largest Islamic party and the largest secular party in the 1994 municipal elections---with a cutoff at zero; the binary treatment $D$ denotes assignment to an Islamic mayor. 
The RD design in this case is sharp since $D=1$ when $R\ge0$. The primary outcome $Y$ is the share of women aged 15--20 who had completed a secular high school education in 2000. 
A vector of covariates $X = (X_1, X_2, X_3)'$ includes the Islamic vote share in the 1994 election ($X_1$), the illiteracy rate in 2000 ($X_2$), and the proportion of religious buildings in 1990--2000 ($X_3$).
$X_1$ is a proxy for Islamic conservatism of the municipality; $X_2$ is a proxy variable for poverty; and $X_3$ measures the religious piety of the municipality. 

We use a dataset containing election and census information for Turkish municipalities from \cite{meyersson2014}, comprising 2630 observations. 
The proposed testing procedure is implemented with both local constant and local linear estimators. 
The estimating bandwidth $h_r$ is selected based on \cite{Calonico2014} with an undersmoothing factor $n^{1/5-1/k_1}$, while the testing bandwidth $h = h_r \times n^{1/5-1/k_2}$ is undersmoothed. 
The KS statistic is computed as the supremum over a grid of $L=1000$ or $2000$ points drawn from a normal distribution. 
The number of bootstrap replications and the kernel function are the same as those in Section \ref{sec:simulation}. 
The test results are presented in Table \ref{tab:Meyersson2014}.

\begin{table}[H]
  \centering
  \caption{The bootstrap $p$-values of the proposed test}
  \label{tab:Meyersson2014}
  \setlength{\extrarowheight}{-1pt}
  \begin{adjustbox}{max width=\textwidth, max height=\textheight}
  \begin{tabular}{ccccccccccccc}
  \toprule & \multicolumn{6}{c}{Local Constant Estimator} & \multicolumn{6}{c}{Local Linear Estimator} \\
  \cmidrule(lr){2-7} \cmidrule(lr){8-13} 
   & \multicolumn{3}{c}{$L=1000$} & \multicolumn{3}{c}{$L=2000$} & \multicolumn{3}{c}{$L=1000$} & \multicolumn{3}{c}{$L=2000$} \\
  \cmidrule(lr){2-4} \cmidrule(lr){5-7} \cmidrule(lr){8-10} \cmidrule(lr){11-13}
  $k_1 \backslash k_2$ & 4.25 & 4.50 & 4.75 & 4.25 & 4.50 & 4.75 & 4.25 & 4.50 & 4.75 & 4.25 & 4.50 & 4.75 \\
  \midrule & \multicolumn{12}{l}{\textit{Panel A: Testing without conditioning on covariates}} \\
  4.00 & 0.001 & 0.000 & 0.000 & 0.001 & 0.000 & 0.000 & 0.001 & 0.000 & 0.000 & 0.001 & 0.000 & 0.000 \\
  4.25 & 0.000 & 0.000 & 0.000 & 0.000 & 0.000 & 0.000 & 0.000 & 0.000 & 0.000 & 0.000 & 0.000 & 0.000 \\
  4.50 & 0.000 & 0.000 & 0.000 & 0.000 & 0.000 & 0.000 & 0.000 & 0.000 & 0.000 & 0.000 & 0.000 & 0.000 \\
  \midrule & \multicolumn{12}{l}{\textit{Panel B: Testing conditional on $X_1$}} \\
  4.00 & 0.000 & 0.000 & 0.000 & 0.000 & 0.000 & 0.000 & 0.000 & 0.000 & 0.000 & 0.000 & 0.000 & 0.000 \\
  4.25 & 0.000 & 0.000 & 0.000 & 0.000 & 0.000 & 0.000 & 0.000 & 0.000 & 0.000 & 0.000 & 0.000 & 0.000 \\
  4.50 & 0.000 & 0.000 & 0.000 & 0.000 & 0.000 & 0.000 & 0.000 & 0.000 & 0.000 & 0.000 & 0.000 & 0.000 \\
  \midrule & \multicolumn{12}{l}{\textit{Panel C: Testing conditional on the covariate vector $X$}} \\
  4.00 & 0.012 & 0.003 & 0.002 & 0.012 & 0.003 & 0.002 & 0.012 & 0.003 & 0.002 & 0.012 & 0.003 & 0.002 \\
  4.25 & 0.002 & 0.001 & 0.001 & 0.002 & 0.001 & 0.001 & 0.002 & 0.001 & 0.001 & 0.002 & 0.001 & 0.001 \\
  4.50 & 0.001 & 0.001 & 0.000 & 0.001 & 0.001 & 0.001 & 0.001 & 0.001 & 0.001 & 0.001 & 0.001 & 0.001 \\
  \bottomrule
  \end{tabular}
  \end{adjustbox}
\end{table}

Table \ref{tab:Meyersson2014} summarizes the test results. 
At the 1\% significance level, Panel A rejects unconditional homogeneity, while Panel B rejects the conditional null given the Islamic vote share. 
Thus, detectable treatment effect heterogeneity remains after conditioning on $X_1$.
This conclusion is stable across the reported tuning-parameter choices and complements the heterogeneity analysis in \cite{meyersson2014}. 
Notably, the null is rejected even in Panel C after controlling for the full covariate vector, indicating that those covariates alone do not fully account for the heterogeneous response to the election of an Islamic party. 
Because the conditioning sets differ in dimension across panels, the corresponding statistics are not on a common scale and cannot be used to rank how much of the heterogeneity each covariate set accounts for. 
The results establish that unobserved heterogeneity remains under every conditioning set considered, not that the covariates contribute nothing. 
Other characteristics appear to play a role in determining female secular high school completion rates, suggesting a more complex underlying mechanism. 
We also examine specifications conditioning on each covariate individually (results available upon request): controlling for $X_2$ (illiteracy rate) or $X_3$ (religious buildings) alone yields qualitative conclusions similar to those in Panel B, with $p$-values near zero, confirming that no single observable dimension plausibly exhausts the treatment effect heterogeneity in this setting.

\subsection{The treatment effect of attending a better high school} \label{subsec:popeleches} 
In this subsection, we revisit \cite{Pop-Eleches2013}, which studies Romanian administrative and survey data. 
In this study, students can be admitted to high school if their transition score exceeds the admission cutoff.
\cite{Pop-Eleches2013} concludes that attending a more selective high school improves average academic performance on the Baccalaureate exam. 

Our focus is on whether the treatment effect of attending a better school varies across students and whether it remains heterogeneous after accounting for peer quality. 
In this case, the running variable $R$ is the distance between a student's score and the high school admission cutoff, and the binary treatment $D$ is attendance at the more selective high school with a higher cutoff in a given town. 
The covariate $X$ is peer quality (average score of the student's class), and the outcome $Y$ is the student's Baccalaureate exam grade upon graduation. 
This setting constitutes a fuzzy RD design because not all students whose scores exceed the admission cutoff choose to attend the more selective school. 
Following \cite{Pop-Eleches2013}, we restrict the sample to towns with only two schools to avoid the complexity of multiple cutoffs, resulting in 23507 observations. 
The test settings remain the same as those in Section \ref{subsec:Meyersson}, and results are reported in Table \ref{tab:popeleches2013}. 

\begin{table}[H]
    \centering
    \caption{The bootstrap $p$-values of the proposed test}
    \label{tab:popeleches2013}
    \setlength{\extrarowheight}{-1pt}
    \begin{adjustbox}{max width=\textwidth, max height=\textheight}
    \begin{tabular}{ccccccccccccc}
    \toprule & \multicolumn{6}{c}{Local Constant Estimator} & \multicolumn{6}{c}{Local Linear Estimator} \\
    \cmidrule(lr){2-7} \cmidrule(lr){8-13} 
     & \multicolumn{3}{c}{$L=1000$} & \multicolumn{3}{c}{$L=2000$} & \multicolumn{3}{c}{$L=1000$} & \multicolumn{3}{c}{$L=2000$} \\
    \cmidrule(lr){2-4} \cmidrule(lr){5-7} \cmidrule(lr){8-10} \cmidrule(lr){11-13}
    $k_1 \backslash k_2$ & 4.25 & 4.50 & 4.75 & 4.25 & 4.50 & 4.75 & 4.25 & 4.50 & 4.75 & 4.25 & 4.50 & 4.75 \\
    \midrule & \multicolumn{12}{l}{\textit{Panel A: Testing without conditioning on $X$}} \\
    4.00 & 0.146 & 0.056 & 0.023 & 0.147 & 0.056 & 0.023 & 0.146 & 0.055 & 0.023 & 0.146 & 0.056 & 0.023 \\
    4.25 & 0.052 & 0.016 & 0.003 & 0.052 & 0.016 & 0.003 & 0.052 & 0.016 & 0.003 & 0.052 & 0.016 & 0.003 \\
    4.50 & 0.016 & 0.003 & 0.001 & 0.016 & 0.003 & 0.001 & 0.016 & 0.003 & 0.001 & 0.016 & 0.003 & 0.001 \\
    \midrule & \multicolumn{12}{l}{\textit{Panel B: Testing conditional on $X$}} \\
    4.00 & 0.186 & 0.239 & 0.220 & 0.282 & 0.347 & 0.318 & 0.943 & 0.932 & 0.960 & 0.944 & 0.934 & 0.956 \\
    4.25 & 0.672 & 0.616 & 0.588 & 0.513 & 0.398 & 0.360 & 0.870 & 0.915 & 0.946 & 0.926 & 0.945 & 0.973 \\
    4.50 & 0.497 & 0.416 & 0.421 & 0.481 & 0.378 & 0.344 & 0.984 & 0.990 & 0.992 & 0.976 & 0.985 & 0.988 \\
    \bottomrule
    \end{tabular}
    \end{adjustbox}
\end{table}

First, we test whether the treatment effect is homogeneous among all individuals in Panel A of Table \ref{tab:popeleches2013}. 
The proposed test rejects the null of homogeneity at the 10\% significance level in most bandwidth configurations, indicating heterogeneity in individual treatment effects.
We then test whether such heterogeneity can be fully explained by the covariate $X$ in Panel B. 
The local constant estimator yields $p$-values ranging from 0.186 to 0.672, so the conditional null is not rejected at the 10\% level, while the local linear estimator yields $p$-values from 0.870 to 0.992, a clear failure to reject at any conventional level. 
This discrepancy between the two estimators warrants discussion. The local constant estimator is known to have a larger boundary bias at the cutoff in RD settings, which enters the first-stage estimation of $\tau(X)$ and propagates into the test statistic through the estimation effect characterized in Lemma~\ref{Lemma1}. 
The local linear estimator, by contrast, achieves a smaller boundary bias, yielding a more accurate estimate of the CLATE. 
The substantially larger $p$-values obtained under the local linear estimator are therefore more reliable for inference in this application.
Under either estimator, the heterogeneity detected in the unconditional analysis (Panel A) is no longer detected after conditioning on $X$ (Panel B). 
Peer quality, therefore, accounts for the treatment effect heterogeneity that the unconditional test detects. 
This finding indicates that peer quality is a relevant summary of detectable heterogeneity in this application and is informative for targeting.

Taken together, the two applications illustrate two qualitatively distinct conclusions that the proposed test can support. 
In the Islamic-election application, significant unobserved heterogeneity persists even after conditioning on covariates, suggesting that latent characteristics are additional drivers of heterogeneous treatment response. 
In the school-quality application, conditioning on peer quality accounts for the detectable heterogeneity in the unconditional test, consistent with peer quality serving as a summary of that heterogeneity. 
These conclusions demonstrate the practical value of the test: it distinguishes settings in which observables leave residual heterogeneity from those in which they account for the heterogeneity. 
We emphasize that the empirical results are interpreted as diagnostics under the maintained assumptions rather than as proof that a covariate set is sufficient.

\section{Conclusion}\label{sec:conclusion}
This paper proposes a novel testing method to detect unobserved heterogeneity in RD designs. 
The proposed test transforms the null of no unobserved heterogeneity into the continuity of an imputed treated outcome at the cutoff. 
An ICM restriction based on characteristic functions is then used to construct the test statistics.
We establish the limiting behavior of the proposed test under the null, the fixed alternative, and a sequence of local alternatives approaching the null at the rate $1/\sqrt{nh}$. 
A straightforward multiplier bootstrap procedure provides feasible critical values that are theoretically valid and computationally efficient.
Numerical simulations and empirical applications show that the test controls size under the null and has increasing power against the alternatives considered.

\let\oldthebibliography\thebibliography
\renewenvironment{thebibliography}[1]{%
  \oldthebibliography{#1}%
  \linespread{1.1}\selectfont   
}{%
  \endlist
}

\setlength{\bibsep}{3pt plus 1pt minus 1pt}
\putbib
\end{bibunit}

\newpage
\begin{bibunit}
\section{Appendix}
\subsection{Notations}
Let $\mathcal{F}_i$ denote the $\sigma$-field generated by $\{X_i,Y_i,D_i,R_i,W_i\}$ including information on individual $i$ and let $\mathcal{F}^\ast_i$ denote the $\sigma$-field generated by $\{ X_i,Y_i,D_i,R_i,W_i,V_i\}$, which is a larger $\sigma$-field containing $\mathcal{F}_i$ and the external multiplier variable $V_i$. Let $\mathcal{F}_{ij}$ be the $\sigma$-field generated by $\mathcal{F}_i$ and $\mathcal{F}_j$, and $\mathcal{F}^\ast_{ij}$ be the $\sigma$-field generated by $\mathcal{F}^\ast_i$ and $\mathcal{F}^\ast_j$. Before the proofs of the lemmas and theorems, we introduce some notation. 

The Hoeffding projection is defined as follows: let $\mathcal{G} = \{g(\zeta_i, \cdots, \zeta_m)\}$ be a class of functions of $m$ variables where i.i.d. variables $\{\zeta_i\}_{i=1}^n$ are defined on the common probability space $(B,\mathcal{B}, Q)$, such that its envelope $G = \sup_{g\in\mathcal{G}}|g|$ is measurable (otherwise take $G$ as the least measurable majorant of $\sup_{g\in\mathcal{G}}|g|$). Below, $P$ denotes the same probability $Q$; the Dirac notation $\delta_{\zeta_i}$ is unrelated to the cutoff sign $\delta_i$.
With the notation $Q_1 \times \cdots \times Q_m g = \int g d\left( Q_1 \times \cdots \times Q_m  \right)$, the Hoeffding projections of $g: B^m \mapsto \mathbb{R}$ are defined as $$ \pi_k g = (\delta_{\zeta_1} - Q ) \times \cdots \times (\delta_{\zeta_k} - Q ) \times Q^{m-k} g \quad k=1, \cdots, m, $$ where $\delta_{\zeta_i}(g) = g(\zeta_i)$ denotes the Dirac measure at $\zeta_i$.
If $g$ is symmetric in its entries, these projections induce the Hoeffding decomposition 
\begin{align*}
    U_n^{(m)}g - Q^m g =& \sum_{k=1}^m \binom{m}{k} U_n^{(k)}(\pi_k g), \\
    \text{where } \quad U_n^{(r)}(g) =& \frac{1}{n(n-1)\cdots(n-r+1)} \sum_{1\le i_1\ne\cdots\ne i_r \le n} g(\zeta_{i_1}, \cdots, \zeta_{i_r}).
\end{align*}
Asymmetric functions can be symmetrized by rearranging the entries. Denote the symmetrized function of an asymmetric function $C_{ij}$ by $g_{C_{ij}}$. Then $U_n^{(2)}(C_{ij}) = U_n^{(2)}(g_{C_{ij}})$ because $$ U_n^{(2)}(C_{ij}) = \frac{1}{n(n-1)} \sum_{i\ne j} C_{ij} = \frac{1}{n(n-1)} \sum_{i\ne j} \frac{C_{ij}+C_{ji}}{2} = \frac{1}{n(n-1)} \sum_{i\ne j} g_{C_{ij}} = U_n^{(2)}(g_{C_{ij}}). $$
Then the Hoeffding decomposition of $U_n^{(2)}(C_{ij})$ is $$ U_n^{(2)}(C_{ij}) = U_n^{(2)}(g_{C_{ij}}) = P^2 g_{C_{ij}} + 2 U_n^{(1)}(\pi_1 g_{C_{ij}}) + U_n^{(2)}(\pi_2 g_{C_{ij}}). $$ The first component is $P^2 g_{C_{ij}} = \mathbb{E} G_{C_{ij}} = \mathbb{E} ( C_{ij}+C_{ji} )/2 = \mathbb{E} C_{ij}$. The second component is 
\begin{align*}
    2 U_n^{(1)}(\pi_1 g_{C_{ij}}) =& 2 \frac{1}{n}\sum_{i=1}^n \left[ \mathbb{E} ( g_{C_{ij}} | \mathcal{F}_i ) - P^2 g_{C_{ij}} \right]\\
    =& \frac{1}{n}\sum_{i=1}^n 2\mathbb{E} \left[ \frac{C_{ij}+C_{ji}}{2} \bigg| \mathcal{F}_i \right] - 2 \mathbb{E} C_{ij} \\
    =& \frac{1}{n}\sum_{i=1}^n \mathbb{E} ( C_{ij} | \mathcal{F}_i ) + \frac{1}{n}\sum_{j=1}^n \mathbb{E} ( C_{ij} | \mathcal{F}_j ) - 2 \mathbb{E} C_{ij}.
\end{align*}
Then the Hoeffding decomposition of $U_n^{(2)}(C_{ij})$ can be written as $$U_n^{(2)}(C_{ij}) = \frac{1}{n}\sum_{i=1}^n\mathbb{E} ( C_{ij} | \mathcal{F}_i ) + \frac{1}{n}\sum_{j=1}^n \mathbb{E} ( C_{ij} | \mathcal{F}_j ) - \mathbb{E} C_{ij} + U_n^{(2)}(\pi_2 g_{C_{ij}}). $$
In a similar way, the Hoeffding decomposition of $U_n^{(3)}(C_{ijk})$ is 
\begin{align*}
    U_n^{(3)}(C_{ijk}) =& \frac{1}{n}\sum_{i=1}^n \mathbb{E}(C_{ijk}|\mathcal{F}_i) + \frac{1}{n}\sum_{j=1}^n \mathbb{E}(C_{ijk}|\mathcal{F}_j) + \frac{1}{n}\sum_{k=1}^n \mathbb{E}(C_{ijk}|\mathcal{F}_k) \\ &- 2\mathbb{E} C_{ijk} + 3 U_n^{(2)}(\pi_2 g_{C_{ijk}}) + U_n^{(3)}(\pi_3 g_{C_{ijk}}). 
\end{align*}
Expanding $U_n^{(2)}(\pi_2 g_{C_{ijk}})$, $U_n^{(3)}(C_{ijk})$ can be further written as
\begin{align*}
    U_n^{(3)}(C_{ijk}) =& \frac{1}{n(n-1)}\sum_{i\ne j}^n\mathbb{E} ( C_{ijk} | \mathcal{F}_{ij} ) + \frac{1}{n(n-1)}\sum_{i\ne k}^n \mathbb{E} ( C_{ijk} | \mathcal{F}_{ik} ) + \frac{1}{n(n-1)}\sum_{j\ne k}^n \mathbb{E} ( C_{ijk} | \mathcal{F}_{jk} ) \\
    &- \frac{1}{n}\sum_{i=1}^n \mathbb{E}(C_{ijk}|\mathcal{F}_i) - \frac{1}{n}\sum_{j=1}^n \mathbb{E} (C_{ijk}|\mathcal{F}_j) - \frac{1}{n}\sum_{k=1}^n \mathbb{E} (C_{ijk}|\mathcal{F}_k)
    + \mathbb{E} C_{ijk} + U_n^{(3)}(\pi_3 g_{C_{ijk}}),
\end{align*}
where $\mathbb{E}(\cdot|\mathcal{F}_{ij})$, $\mathbb{E}(\cdot|\mathcal{F}_{ik})$ and $\mathbb{E}(\cdot|\mathcal{F}_{jk})$ are abbreviations for $\mathbb{E}(\cdot|\mathcal{F}_i, \mathcal{F}_j)$, $\mathbb{E}(\cdot|\mathcal{F}_i, \mathcal{F}_k)$ and $\mathbb{E}(\cdot|\mathcal{F}_j, \mathcal{F}_k)$.

Moreover, we use the notation $A\lesssim B$ to express that the left-hand side term $A$ is bounded by a positive constant times the right-hand side $B$. The abbreviations for indicator functions $\mathbf{1}_+(R) = \mathbf{1}(R\ge0)$ and $\mathbf{1}_-(R) = \mathbf{1}(R<0)$ are used for notational simplicity throughout the Appendix.  

The local constant and local linear estimators of $$\tau(x) = \frac{\mu(x,0^+)-\mu(x,0^-)}{p(x,0^+)-p(x,0^-)}$$ and $$\kappa(w,x) = \frac{\varphi_{WD|XR}(w, 0|x, 0^+)-\varphi_{WD|XR}(w, 0|x, 0^-)}{p(x,0^+)-p(x,0^-)}$$ are introduced here.

By Assumption A1, $f(x,0^+) = f(x,0^-) = f(x,0)$ for any $x\in\mathcal{X}$. Then the local constant estimators $\hat{p}(X_i,0^\pm)$ and $\hat\mu(X_i,0^\pm)$ can be written as
\begin{align*}
    \hat\mu(X_i,0^+) =& \left[ \frac{1}{n-1} \sum_{j\ne i} Y_j K_{h_x}(X_j-X_i) K_{h_r}(R_j) \mathbf{1}_+(R_j) \right] \bigg/ \left[ \frac{1}{n-1} \sum_{j\ne i} K_{h_x}(X_j-X_i) K_{h_r}(R_j) \right], \\
    \hat\mu(X_i,0^-) =& \left[ \frac{1}{n-1} \sum_{j\ne i} Y_j K_{h_x}(X_j-X_i) K_{h_r}(R_j) \mathbf{1}_-(R_j) \right] \bigg/ \left[ \frac{1}{n-1} \sum_{j\ne i} K_{h_x}(X_j-X_i) K_{h_r}(R_j) \right]; \\
    \hat{p}(X_i,0^+) =& \left[ \frac{1}{n-1} \sum_{j\ne i} D_j K_{h_x}(X_j-X_i) K_{h_r}(R_j) \mathbf{1}_+(R_j) \right] \bigg/ \left[ \frac{1}{n-1} \sum_{j\ne i} K_{h_x}(X_j-X_i) K_{h_r}(R_j) \right], \\
    \hat{p}(X_i,0^-) =& \left[ \frac{1}{n-1} \sum_{j\ne i} D_j K_{h_x}(X_j-X_i) K_{h_r}(R_j) \mathbf{1}_-(R_j) \right] \bigg/ \left[ \frac{1}{n-1} \sum_{j\ne i} K_{h_x}(X_j-X_i) K_{h_r}(R_j) \right].
\end{align*}
Plugging them into $\hat\tau(X_i) = \frac{\hat\mu(X_i,0^+)-\hat\mu(X_i,0^-)}{\hat{p}(X_i,0^+)-\hat{p}(X_i,0^-)}$, the local constant estimator $\hat\tau(X_i)$ can be written as $$\hat\tau(X_i) = \frac{\frac{1}{n-1} \sum_{j\ne i} Y_j K_{h_x}(X_j-X_i) K_{h_r}(R_j) \delta_j}{\frac{1}{n-1} \sum_{j\ne i} D_j K_{h_x}(X_j-X_i) K_{h_r}(R_j)\delta_j} $$

Local linear estimators for the conditional means are
\begin{align*}
    \hat\mu(X_i, 0^+) =& e_0' \arg\min_{\mathbf{b}} \sum_{j\ne i} \left( Y_j - \mathbf{b}' \mathbf{r}(R_j) \right)^2 K_{h_x}(X_j-X_i) K_{h_r}(R_j) \mathbf{1}_+(R_j), \\
    \hat\mu(X_i, 0^-) =& e_0' \arg\min_{\mathbf{b}} \sum_{j\ne i} \left( Y_j - \mathbf{b}' \mathbf{r}(R_j) \right)^2 K_{h_x}(X_j-X_i) K_{h_r}(R_j) \mathbf{1}_-(R_j); \\
    \hat p(X_i, 0^+) =& e_0' \arg\min_{\mathbf{b}} \sum_{j\ne i} \left( D_j - \mathbf{b}' \mathbf{r}(R_j) \right)^2 K_{h_x}(X_j-X_i) K_{h_r}(R_j) \mathbf{1}_+(R_j), \\
    \hat p(X_i, 0^-) =& e_0' \arg\min_{\mathbf{b}} \sum_{j\ne i} \left( D_j - \mathbf{b}' \mathbf{r}(R_j) \right)^2 K_{h_x}(X_j-X_i) K_{h_r}(R_j) \mathbf{1}_-(R_j),
\end{align*}
where $e_0 = (1,0)'$ and $\mathbf{r}(R) = (1, R/h_r)'$. Denote the following matrices and vectors:
\begin{align*}
    B_{n+}(X_i) =& \frac{1}{n-1} \sum_{j\ne i} \mathbf{r}(R_j)\mathbf{r}(R_j)' K_{h_x}(X_j-X_i)K_{h_r}(R_j) \mathbf{1}_+(R_j), \\
    B_{n-}(X_i) =& \frac{1}{n-1} \sum_{j\ne i} \mathbf{r}(R_j)\mathbf{r}(R_j)' K_{h_x}(X_j-X_i)K_{h_r}(R_j) \mathbf{1}_-(R_j); \\
    A_{Zn+}(X_i) =& \frac{1}{n-1} \sum_{j\ne i} \mathbf{r}(R_j)Z_j K_{h_x}(X_j-X_i)K_{h_r}(R_j) \mathbf{1}_+(R_j), \\
    A_{Zn-}(X_i) =& \frac{1}{n-1} \sum_{j\ne i} \mathbf{r}(R_j)Z_j K_{h_x}(X_j-X_i)K_{h_r}(R_j) \mathbf{1}_-(R_j),
\end{align*}
where $Z$ is a generic variable for $Y$ or $D$. Then $\hat\mu$ and $\hat{p}$ can be written as $$ \hat\mu(X_i, 0^\pm) = e_0' B_{n\pm}^{-1}(X_i) A_{Yn\pm}(X_i), \quad \hat{p}(X_i, 0^\pm) = e_0' B_{n\pm}^{-1}(X_i) A_{Dn\pm}(X_i). $$

\subsection{Preliminary results}
The basic asymptotic properties of the local constant and the local linear estimators, which will be frequently used in the proofs, are collected below.
Before that, we introduce the expansion $B_n^{-1} A_n$:
\begin{align*}
    B_n^{-1} A_n &= \left( B^{-1} - B^{-1} (B_n-B) B_n^{-1} \right) \left( A + A_n - A \right) \\ 
    &= B^{-1}A + B^{-1} (A_n-A) - B^{-1} (B_n-B) B_n^{-1}A_n.
\end{align*}
Iterating this equation yields
\begin{align*}
    & B_n^{-1} A_n - B^{-1}A \\ 
    =& B^{-1} (A_n-A) - B^{-1} (B_n-B) B_n^{-1}A_n \\
    =& B^{-1} (A_n-A) - B^{-1} (B_n-B) \left( B^{-1}A + B^{-1} (A_n-A) - B^{-1} (B_n-B) B_n^{-1}A_n \right) \\
    =& B^{-1} (A_n-A) - B^{-1} (B_n-B) B^{-1}A - B^{-1} (B_n-B) B^{-1} (A_n-A) \\
    &+ B^{-1} (B_n-B) B^{-1} (B_n-B) \left( B^{-1}A + B^{-1} (A_n-A) - B^{-1} (B_n-B) B_n^{-1}A_n \right) \\
    =& B^{-1} A_n - B^{-1} B_n B^{-1}A - B^{-1} (B_n-B) B^{-1} (A_n-A) + B^{-1} (B_n-B) B^{-1} (B_n-B) B^{-1}A \\
    &+ B^{-1} (B_n-B) B^{-1} (B_n-B) B^{-1} (A_n-A) - B^{-1} (B_n-B) B^{-1} (B_n-B) B^{-1} (B_n-B) B_n^{-1}A_n.
\end{align*}
The scalar version of the expansion $\hat{a}/\hat{b}$ is $$ \frac{\hat{a}}{\hat{b}} - \frac{a}{b} = \frac{\hat{a}-a}{b} - \frac{a}{b} \frac{\hat{b}-b}{b} - \frac{\hat{a}-a}{b} \frac{\hat{b}-b}{b} + \frac{a}{b} \left( \frac{\hat{b}-b}{b} \right)^2 + \frac{\hat{a}-a}{b} \left( \frac{\hat{b}-b}{b} \right)^2 - \frac{\hat{a}}{\hat{b}} \left( \frac{\hat{b}-b}{b} \right)^3, $$ where the first-order term can also be written as $\frac{\hat{a}-a}{b} - \frac{a}{b} \frac{\hat{b}-b}{b} = \frac{\hat{a}}{b} - \frac{a}{b} \frac{\hat{b}}{b}$. 

\textbf{Lemma S.A1} \quad 
For the local constant estimator, the numerator and denominator have the following convergence results:
\begin{align*}
    m_\mu(x) =& \mathbb{E}[Y K_{h_x}(X-x) K_{h_r}(R) \delta] - \frac{1}{2} f(x,0) \Delta_\mu(x) = O(h_r) + O(h_x^l), \\
    m_p(x) =& \mathbb{E}[D K_{h_x}(X-x) K_{h_r}(R) \delta] - \frac{1}{2} f(x,0) \Delta_p(x) = O(h_r) + O(h_x^l), \\
    m_\varphi(w,x) =& \mathbb{E}[(1-D) e^{\mathrm{i}wW} K_{h_x}(X-x) K_{h_r}(R) \delta] - \frac{1}{2} f(x,0) \Delta_\varphi(w,x) = O(h_r) + O(h_x^l),
\end{align*}
where the bias term is of order $O(h_r) + O(h_x^l)$, $\Delta_g(x) := g(x,0^+) - g(x,0^-)$ for a generic function $g(x,r)$, and $\Delta_\varphi(w,x) := \varphi(w,0|x,0^+) - \varphi(w,0|x,0^-)$ denotes the difference in the function $\varphi_{WD|XR}(w,0|x,r)$; henceforth, $\varphi(w,0|x,r)$ abbreviates $\varphi_{WD|XR}(w,0|x,r)$.
$f(x,0)$ is the abbreviation for $f_{XR}(x,0)$.
The result can be further expanded as $$ \mathbb{E}[g(X) Z K_{h_x}(X-x) K_{h_r}(R) \delta] - \frac{1}{2} f(x,0) g(x) \Delta_{m_Z}(x) = O(h_r) + O(h_x^l) $$ for a generic variable $Z$.
\begin{proof}
$$\mathbb{E} \left[ Y K_{h_x}(X-x) K_{h_r}(R)\delta \right] = \mathbb{E} \left[ Y K_{h_x}(X-x) K_{h_r}(R) \mathbf{1}_+(R) \right] - \mathbb{E} \left[ Y K_{h_x}(X-x) K_{h_r}(R)\mathbf{1}_-(R) \right].$$ By the law of iterated expectations, the first term on the right-hand side of the equation can be written as
\begin{align*}
    \mathbb{E}\left[ Y K_{h_x}(X-x) K_{h_r}(R)\mathbf{1}_+(R) \right] =& \mathbb{E}\left[ \mathbb{E} \left[ Y K_{h_x}(X-x) K_{h_r}(R)\mathbf{1}_+(R) |X,R \right] \right] \\
    =& \mathbb{E}\left[ \mu(X,R) K_{h_r}(R)\mathbf{1}_+(R) K_{h_x}(X-x) \right] 
\end{align*}
Expand this conditional expectation in integral form with the change of variables $X-x = uh_x$ and $R = vh_r$, and define the one-sided kernel constants $\nu_\rho^+ = \int_0^\infty v^\rho K(v)dv$, $\nu_\rho^- = \int_{-\infty}^0 v^\rho K(v)dv$ for $\rho\ge0$, and $\nu_l = \int v^l K(v)dv$. Then the expectation is
\begin{align*}
    & \mathbb{E}\left[ Y K_{h_x}(X-x) K_{h_r}(R)\mathbf{1}_+(R) \right] \\
    =& \iint_0^\infty \mu(x+uh_x,vh_r) K(v) K(u) f(x+uh_x,vh_r) dudv \\
    =& \frac{1}{2} \mu(x,0^+) f(x,0^+) + h_r \nu_1^+ \frac{\partial}{\partial r}\left( \mu(x,0^+)f(x,0^+)\right) + \frac{h_x^l \nu_l}{l!} \frac{\partial^l}{\partial x^l} \left( \mu(x,0^+)f(x,0^+)\right) + o(h_r) + o(h_x^l). 
\end{align*}
The equation comes from a Taylor expansion of the function $\mu(x,r)f(x,r)$ and the property of the kernel function $K(\cdot)$. Following the same steps, 
\begin{align*}
    & \mathbb{E}\left[ Y K_{h_x}(X-x) K_{h_r}(R)\mathbf{1}_-(R) \right] \\
    =& \iint_{-\infty}^0 \mu(x+uh_x,vh_r) K(v) K(u) f(x+uh_x,vh_r) dudv \\
    =& \frac{1}{2} \mu(x,0^-) f(x,0^-) + h_r\nu_1^- \frac{\partial}{\partial r}\left( \mu(x,0^-)f(x,0^-)\right) + \frac{h_x^l \nu_l}{l!} \frac{\partial^l}{\partial x^l} \left( \mu(x,0^-)f(x,0^-)\right) + o(h_r) + o(h_x^l). 
\end{align*}
By Assumption A1, $f(x,0^+) = f(x,0^-) = f(x,0)$ and $\mu(x,0^+) f(x,0^+) - \mu(x,0^-) f(x,0^-) = f(x,0) \Delta_\mu(x)$. Then 
\begin{align*}
    & \mathbb{E}\left[ Y K_{h_x}(X-x) K_{h_r}(R) \delta \right] - \frac{1}{2} f(x,0) \Delta_\mu(x) \\
    =& h_r \left[ \nu_1^+ \frac{\partial}{\partial r}\left( \mu(x,0^+)f(x,0^+)\right) - \nu_1^- \frac{\partial}{\partial r}\left( \mu(x,0^-)f(x,0^-)\right) \right] + o(h_r) \\&+ \left[ \frac{h_x^l \nu_l}{l!} \frac{\partial^l}{\partial x^l} \left( \mu(x,0^+)f(x,0^+)\right) - \frac{h_x^l \nu_l}{l!} \frac{\partial^l}{\partial x^l} \left( \mu(x,0^-)f(x,0^-)\right) \right] + o(h_x^l), 
\end{align*}
and the bias term is of order $h_r$ and $h_x^l$. Applying this result to the conditional expectation scenario, 
\begin{align*}
    m_\mu(X_i) =& \mathbb{E}\left[ Y_j K_{h_x}(X_j-X_i) K_{h_r}(R_j) \delta_j | X_i \right] - \frac{1}{2} f(X_i,0) \Delta_\mu(X_i) \\
    =& h_r \left[ \nu_1^+ \frac{\partial}{\partial r}\left( \mu(X_i,0^+)f(X_i,0^+)\right) - \nu_1^- \frac{\partial}{\partial r}\left( \mu(X_i,0^-)f(X_i,0^-)\right) \right] + o_p(h_r) \\&+ \frac{h_x^l \nu_l}{l!} \left[ \frac{\partial^l}{\partial x^l} \left( \mu(X_i,0^+)f(X_i,0^+)\right) - \frac{\partial^l}{\partial x^l} \left( \mu(X_i,0^-)f(X_i,0^-)\right) \right] + o_p(h_x^l).
\end{align*}
Replacing $Y$ with $D$ yields the analogous result:
\begin{align*}
    m_p(X_i) =& \mathbb{E}\left[ D_j K_{h_x}(X_j-X_i) K_{h_r}(R_j) \delta_j | X_i \right] - \frac{1}{2} f(X_i,0) \Delta_p(X_i) \\
    =& h_r \left[ \nu_1^+ \frac{\partial}{\partial r}\left( p(X_i,0^+)f(X_i,0^+)\right) - \nu_1^- \frac{\partial}{\partial r}\left( p(X_i,0^-)f(X_i,0^-)\right) \right] + o_p(h_r) \\&+ \frac{h_x^l \nu_l}{l!} \left[ \frac{\partial^l}{\partial x^l} \left( p(X_i,0^+)f(X_i,0^+)\right) - \frac{\partial^l}{\partial x^l} \left( p(X_i,0^-)f(X_i,0^-)\right) \right] + o_p(h_x^l).
\end{align*}
Replacing $Y$ with $(1-D) e^{\mathrm{i}wW}$ yields the analogous result:
\begin{align*}
    m_\varphi(w, X_i) =& \mathbb{E}\left[ (1-D_j) e^{\mathrm{i}wW_j} K_{h_x}(X_j-X_i) K_{h_r}(R_j) \delta_j | X_i \right] - \frac{1}{2} f(X_i,0) \Delta_\varphi(w, X_i) \\
    =& h_r \left[ \nu_1^+ \frac{\partial}{\partial r}\left( \varphi(w,0|X_i,0^+)f(X_i,0^+)\right) - \nu_1^- \frac{\partial}{\partial r}\left( \varphi(w,0|X_i,0^-)f(X_i,0^-)\right) \right] + o_p(h_r) \\&+ \frac{h_x^l \nu_l}{l!} \left[ \frac{\partial^l}{\partial x^l} \left( \varphi(w,0|X_i,0^+)f(X_i,0^+)\right) - \frac{\partial^l}{\partial x^l} \left( \varphi(w,0|X_i,0^-)f(X_i,0^-)\right) \right] + o_p(h_x^l).
\end{align*}
For the last result in Lemma S.A1,
\begin{align*}
    & \mathbb{E}\left[g(X) Z K_{h_x}(X-x) K_{h_r}(R)\mathbf{1}_+(R) \right] \\
    =& \iint_0^\infty g(x+uh_x) m_Z(x+uh_x,vh_r) K(v) K(u) f(x+uh_x,vh_r) dudv \\
    =& \frac{1}{2} g(x) m_Z(x,0^+) f(x,0^+) + h_r \nu_1^+ \frac{\partial}{\partial r}\left( g(x) m_Z(x,0^+)f(x,0^+)\right) \\&+ \frac{h_x^l \nu_l}{l!} \frac{\partial^l}{\partial x^l} \left( g(x) m_Z(x,0^+)f(x,0^+)\right) + o(h_r) + o(h_x^l). 
\end{align*}
Then
\begin{align*}
    & \mathbb{E}\left[g(X) Z K_{h_x}(X-x) K_{h_r}(R) \delta \right] - \frac{1}{2} f(x,0) g(x) \Delta_{m_Z}(x) \\
    =& h_r \left[ \nu_1^+ \frac{\partial}{\partial r}\left( g(x) m_Z(x,0^+)f(x,0^+)\right) - \nu_1^- \frac{\partial}{\partial r}\left( g(x) m_Z(x,0^-)f(x,0^-)\right) \right] + o(h_r) \\&+ \left[ \frac{h_x^l \nu_l}{l!} \frac{\partial^l}{\partial x^l} \left( g(x) m_Z(x,0^+)f(x,0^+)\right) - \frac{h_x^l \nu_l}{l!} \frac{\partial^l}{\partial x^l} \left( g(x) m_Z(x,0^-)f(x,0^-)\right) \right] + o(h_x^l), 
\end{align*}
and the bias term is of order $h_r$ and $h_x^l$. Applying this result to the conditional expectation scenario, 
\begin{align*}
    & \mathbb{E}\left[g(X_j) Z_j K_{h_x}(X_j-X_i) K_{h_r}(R_j) \delta_j | X_i \right] - \frac{1}{2} f(X_i,0) g(X_i) \Delta_{m_Z}(X_i) \\
    =& h_r \left[ \nu_1^+ \frac{\partial}{\partial r}\left( g(X_i) m_Z(X_i,0^+)f(X_i,0^+)\right) - \nu_1^- \frac{\partial}{\partial r}\left( g(X_i) m_Z(X_i,0^-)f(X_i,0^-)\right) \right] + o_p(h_r) \\&+ \frac{h_x^l \nu_l}{l!} \left[ \frac{\partial^l}{\partial x^l} \left( g(X_i) m_Z(X_i,0^+)f(X_i,0^+)\right) - \frac{\partial^l}{\partial x^l} \left( g(X_i) m_Z(X_i,0^-)f(X_i,0^-)\right) \right] + o_p(h_x^l).
\end{align*}
\end{proof}

For the local linear estimator, denote 
\begin{align*}
    B_\pm(x) =& f_X(x) \mathbb{E}[\mathbf{r}(R) \mathbf{r}(R)' K_{h_r}(R) \mathbf{1}_\pm(R)|X=x], \\
    A_{Z\pm}(x) =& f_X(x) \mathbb{E}[\mathbf{r}(R) Z K_{h_r}(R) \mathbf{1}_\pm(R)|X=x].
\end{align*}

\textbf{Lemma S.A2} \quad 
For a generic variable $Z$, the local linear estimator of its conditional expectation implies that $$ e_0' B_\pm^{-1}(X_i) A_{Z\pm}(X_i) = e_0' \left( \begin{array}{c} m_Z(X_i, 0^\pm) \\ h_r \frac{\partial}{\partial r}m_Z(X_i, 0^\pm) \end{array} \right) + O_p(h_r^2) = m_Z(X_i, 0^\pm) + O_p(h_r^2). $$
\begin{proof}
The definition of $B_\pm(x)$ and $A_{Z\pm}(x)$ shows that the smoothing bias of $X$ is $O(h_x^l)$:
\begin{align*}
    \mathbb{E}[\mathbf{r}(R) \mathbf{r}(R)' K_{h_x}(X-x) K_{h_r}(R) \mathbf{1}_\pm(R)] =& B_\pm(x) + O(h_x^l), \\
    \mathbb{E}[\mathbf{r}(R) Z K_{h_x}(X-x) K_{h_r}(R) \mathbf{1}_\pm(R)] =& A_{Z\pm}(x) + O(h_x^l).
\end{align*}
Let $m_Z(x,r)$ denote $\mathbb{E}[Z|X=x, R=r]$. Its Taylor expansion at the point $(x, 0^+)$ is $$ m_Z(x, R_i) = m_Z(x, 0^+) + R_i \frac{\partial}{\partial r}m_Z(x, 0^+) + \frac{1}{2} R_i^2 \frac{\partial^2}{\partial r^2}m_Z(x, 0^+) + \frac{1}{3!} R_i^3 \frac{\partial^3}{\partial r^3}m_Z(x, \bar{R}_i), $$ where $\bar{R}_i$ lies between $0$ and $R_i$. Then $A_{Z+}(X_i)$ can be rewritten as
\begin{align*}
    A_{Z+}(x) =& f_X(x) \mathbb{E} \left[ \mathbf{r}(R) m_Z(X,R) K_{h_r}(R) \mathbf{1}_+(R) | X=x \right] \\
    =& f_X(x) \mathbb{E} \left[ \mathbf{r}(R) \mathbf{r}(R)' \left( \begin{array}{c} m_Z(x, 0^+) \\ h_r \frac{\partial}{\partial r}m_Z(x, 0^+) \end{array} \right) K_{h_r}(R) \mathbf{1}_+(R) | X=x \right] \\
    &+ \frac{1}{2} f_X(x) \mathbb{E} \left[ \mathbf{r}(R) R^2 \frac{\partial^2}{\partial r^2}m_Z(x, 0^+) K_{h_r}(R) \mathbf{1}_+(R) | X=x \right] (1+o(1)). 
\end{align*}
The first term of the conditional expectation equals $$\mathbb{E} \left[ \mathbf{r}(R) \mathbf{r}(R)' K_{h_r}(R) \mathbf{1}_+(R) | X=x \right] \left( \begin{array}{c} m_Z(x, 0^+) \\ h_r \frac{\partial}{\partial r}m_Z(x, 0^+) \end{array} \right) = B_+(x) \left( \begin{array}{c} m_Z(x, 0^+) \\ h_r \frac{\partial}{\partial r}m_Z(x, 0^+) \end{array} \right). $$ The second term in the formula of $A_{Z+}(x)$ is $O(h_r^2)$ since it equals
\begin{align*}
& \frac{1}{2} \frac{\partial^2}{\partial r^2}m_Z(x, 0^+) \mathbb{E} \left[ \mathbf{r}(R) R^2 K_{h_r}(R) \mathbf{1}_+(R) | X=x \right] \\
=& \frac{h_r^2}{2} \frac{\partial^2}{\partial r^2}m_Z(x, 0^+) \int_0^\infty (1, v)' f(x,vh_r) v^2K(v)dv \\
=& \frac{h_r^2}{2} \frac{\partial^2}{\partial r^2}m_Z(x, 0^+) f(x, 0) \left(  \nu_2^+, \nu_3^+ \right)' \left( 1+o(1) \right) = O(h_r^2),
\end{align*}
In summary, $A_{Z+}(X_i)$ can be written as $ A_{Z+}(X_i) = B_+(X_i) \left( \begin{array}{c} m_Z(X_i, 0^+) \\ h_r \frac{\partial}{\partial r}m_Z(X_i, 0^+) \end{array} \right) + O_p(h_r^2)$. Furthermore, on either side of the cutoff, the estimation effect of the conditional expectation function is $$e_0' B_\pm^{-1}(X_i) A_{Z\pm}(X_i) = e_0' \left( \begin{array}{c} m_Z(X_i, 0^\pm) \\ h_r \frac{\partial}{\partial r}m_Z(X_i, 0^\pm) \end{array} \right) + O_p(h_r^2) = m_Z(X_i, 0^\pm) + O_p(h_r^2). $$
\end{proof}

By the definition of $\Delta_g(x) = g(x,0^+) - g(x,0^-)$ for a generic function $g(x,r)$, including its estimator $\hat{g}$, $\tau(x)$ can be written as $\tau(x) = \Delta_\mu(x)/\Delta_p(x)$ and $\hat\tau(x)$ can be written as $\hat\tau(x) = \hat\Delta_\mu(x)/\hat\Delta_p(x)$.

\subsection{Proof of Proposition \ref{prop1}}
The null hypothesis of no unobserved treatment effect heterogeneity is that, at the cutoff, $$ \mathbb{H}_0: g(1,0,X,\cdot) -g(0,0,X,\cdot) = \tau(X). $$
It holds if and only if $g$ at the cutoff is additively separable in $\epsilon$ with respect to $D$, i.e., $$ g(D,0,X,\epsilon) = m(D,X) + \nu(X,\epsilon) ,$$ where $m: \mathcal{S}_{DX} \mapsto \mathbb{R}$ and $\nu: \mathcal{S}_{X\epsilon} \mapsto \mathbb{R}$.
The validity of this statement is straightforward. For the ``if" part, $\tau(x) = m(1,x) - m(0,x)$. For the ``only if" part, $g(d,0,x,\epsilon) = d \times \left[ g(1,0,x,\epsilon) -g(0,0,x,\epsilon) \right] + g(0,0,x,\epsilon)$. The first term $d\times\tau(x)$ can be identified as $m(d,x)$ and the second term $g(0,0,x,\epsilon)$ can be identified as $\nu(x,\epsilon)$.

\subsection{Proof of Proposition \ref{prop2}}
Assumptions A1-A5 establish the ``only if" part of the proposition. By definition, $W = Y + (1-D)\tau(X)$ with $\tau(x) = \frac{\mu(x,0^+) - \mu(x,0^-)}{p(x,0^+) - p(x,0^-)}$. Since $Y = g(D,R,X,\epsilon) = g(0,R,X,\epsilon) + D[ g(1,R,X,\epsilon) - g(0,R,X,\epsilon) ]$, 
\begin{align*}
    \mu(x,0^+) =& \mathbb{E} \left[ g(0,R,X,\epsilon) \mid X=x, R=0^+ \right] \\
    &+ \mathbb{E} \left[ D[g(1,R,X,\epsilon) - g(0,R,X,\epsilon)] \mid X=x, R=0^+ \right], \\
    \text{and } \quad \mu(x,0^-) =& \mathbb{E} \left[ g(0,R,X,\epsilon) \mid X=x, R=0^- \right] \\
    &+ \mathbb{E} \left[ D[g(1,R,X,\epsilon) - g(0,R,X,\epsilon)] \mid X=x, R=0^- \right].
\end{align*}
By Assumption A4, the function $g$ is continuous at the cutoff $R=0$. Then $$\mathbb{E} \left[ g(0,R,X,\epsilon) \mid X=x, R=0^+ \right] = \mathbb{E} \left[ g(0,R,X,\epsilon) \mid X=x, R=0^- \right] = \mathbb{E} \left[ g(0,0,X,\epsilon) | X=x \right]. $$
Furthermore, 
\begin{align*}
    \mu(x,0^+) - \mu(x,0^-) =& \mathbb{E} \left[ D[g(1,0,X,\epsilon) - g(0,0,X,\epsilon)] \mid X=x, R=0^+ \right] \\&- \mathbb{E} \left[ D[g(1,0,X,\epsilon) - g(0,0,X,\epsilon)] \mid X=x, R=0^- \right].
\end{align*}
Under the null $\mathbb{H}_0$, $g(1,0,X,\cdot) - g(0,0,X,\cdot) = \tau(X)$.
The above expectation can be expressed as
\begin{align*}
    \mu(x,0^+) - \mu(x,0^-) =& \mathbb{E} \left[ D \tau(X) \mid X=x, R=0^+ \right] -\mathbb{E} \left[ D \tau(X) \mid X=x, R=0^- \right] \\
    =& \tau(x) \left[ \mathbb{E} \left[ D \mid X=x, R=0^+ \right] -\mathbb{E} \left[ D \mid X=x, R=0^- \right] \right] \\
    =& \tau(x) \left[ p(x,0^+) - p(x,0^-) \right].
\end{align*}
At the cutoff, $Y_1-Y_0=\frac{\mu(X,0^+) - \mu(X,0^-)}{p(X,0^+) - p(X,0^-)}$, and $W = Y + (1-D) (Y_1-Y_0) = Y_1$. Therefore, its conditional density function is continuous at the cutoff.

For the ``if" part of the proposition, suppose the conditional density function of $W$ given $X$ and $R$ is continuous at $r=0$. 
Then its conditional CDF is continuous at the cutoff: $$ \lim_{r\downarrow 0} \mathbb{P}(W \le y \mid X, R=r) = \lim_{r\uparrow 0} \mathbb{P}(W \le y \mid X, R=r). $$ Since the joint probability can be decomposed as the sum of the sub-event probabilities, the equation can be written as
\begin{align*}
    & \lim_{r\downarrow 0}\mathbb{P}(Y\le y, D=1 \mid X=x, R=r) - \lim_{r\uparrow 0}\mathbb{P}(Y\le y, D=1 \mid X=x, R=r) \\
    =& \lim_{r\uparrow 0} \mathbb{P}(Y\le y-\tau(X), D=0 \mid X=x, R=r) - \lim_{r\downarrow 0} \mathbb{P}(Y\le y-\tau(X), D=0 \mid X=x, R=r). 
\end{align*}
Define
\begin{align*}
    \Delta_0(u,x) =& \lim_{r\uparrow 0} \mathbb{P}(\nu(X,\epsilon)\le u, D=0 \mid X=x, R=r) - \lim_{r\downarrow 0} \mathbb{P}(\nu(X,\epsilon)\le u, D=0 \mid X=x, R=r), \\
    \Delta_1(u,x) =& \lim_{r\downarrow 0} \mathbb{P}(\nu(X,\epsilon)\le u, D=1 \mid X=x, R=r) - \lim_{r\uparrow 0} \mathbb{P}(\nu(X,\epsilon)\le u, D=1 \mid X=x, R=r). 
\end{align*}
By Assumptions A1 and A3, we have 
$$\Delta_0(u,x) = \Delta_1(u,x) = \mathbb{P}(\nu(X,\epsilon)\le u, \eta\in\mathcal{C}_x \mid X=x, R=0) .$$
It is strictly monotone in $u\in\mathcal{S}_{\nu(X,\epsilon) \mid X=x, \eta\in\mathcal{C}_x }$. Furthermore, Assumption A5 guarantees the invariance of the support $\mathcal{S}_{\nu(X,\epsilon) \mid X=x, \eta\in\mathcal{C}_x } = \mathcal{S}_{\nu(X,\epsilon) \mid X=x}$. Combining all the results, we have 
\begin{align*}
    & \lim_{r\downarrow 0} \mathbb{P}(Y\le y, D=1 \mid X=x, R=r) - \lim_{r\uparrow 0} \mathbb{P}(Y\le y, D=1 \mid X=x, R=r) \\
    =& \Delta_1( \tilde{g}^{-1}(1,X,y), x) = \Delta_0( \tilde{g}^{-1}(1,X,y), x) \\
    =& \lim_{r\uparrow 0} \mathbb{P}\left( Y\le \tilde{g}(0,X,\tilde{g}^{-1}(1,X,y)), D=0 \mid X=x, R=r \right) \\
    &- \lim_{r\downarrow 0} \mathbb{P} \left( Y\le \tilde{g}(0,X,\tilde{g}^{-1}(1,X,y)), D=0 \mid X=x, R=r \right),
\end{align*}
where $\tilde{g}^{-1}(1,x,\cdot)$ is the inverse function of $\tilde{g}(1,x,\cdot)$, which is strictly monotonic given Assumption A4. Both sides of the equation are strictly monotone in $y \in \mathcal{S}_{\tilde{g}(1,X,V) \mid X=x}$ since $\Delta_d(\cdot, x)$ is strictly monotone. Combining all the results, we have
$$ y-\tau(x) = \tilde{g}(0,x,\tilde{g}^{-1}(1,x,y) ), \quad\quad \forall x\in\mathcal{X}, y\in \mathcal{S}_{\tilde{g}(1,x,\cdot) \mid X=x} .$$
Using $y=\tilde{g}(1,x,u)$ for some $u\in \mathcal{S}_{\nu(X,\epsilon) \mid X=x}$ as a change of variables, the equation becomes 
$$\tilde{g}(1,x,u) - \tau(x) = \tilde{g}(0,x,u). $$
That means, at the cutoff, the function $g$ satisfies $$ g(1,0,x,\epsilon) - \tau(x) = g(0,0,x,\epsilon). $$ The ``if" part of the proposition is proved.

\subsection{Proof of Lemma \ref{Lemma1}}
The estimation effect of $\hat{U}_n(w,x)$ is $$ \hat{U}_n(w,x) - U_n(w,x) = \frac{1}{n} \sum_{i=1}^n (e^{\mathrm{i} w\hat{W}_i}-e^{\mathrm{i} wW_i}) e^{\mathrm{i} x'X_i} K_h(R_i) \delta_i. $$
By a Taylor expansion, with $\bar{W}_i \in (\hat{W}_i, W_i)$ (or $(W_i, \hat{W}_i)$), $$ e^{\mathrm{i} w\hat{W}_i} - e^{\mathrm{i} wW_i} = \mathrm{i} we^{\mathrm{i} wW_i} (\hat{W}_i-W_i) + \frac{(\mathrm{i} w)^2}{2} e^{\mathrm{i} wW_i} (\hat{W}_i-W_i)^2 + \frac{(\mathrm{i} w)^3}{3!} e^{\mathrm{i} w\bar{W}_i} (\hat{W}_i-W_i)^3. $$
By the definition of $W$, $\hat{W}_i-W_i = (1-D_i) (\hat\tau(X_i)-\tau(X_i))$ where $1-D_i$ is binary. Then
\begin{align*}
    \hat{U}_n(w,x) - U_n(w,x) =& \frac{\mathrm{i}w}{n} \sum_{i=1}^n e^{\mathrm{i} (wW_i+x'X_i)} (1-D_i) (\hat\tau(X_i) - \tau(X_i)) K_h(R_i) \delta_i \\
    &- \frac{w^2}{2} \frac{1}{n} \sum_{i=1}^n e^{\mathrm{i} (wW_i+x'X_i)} (1-D_i) (\hat\tau(X_i) - \tau(X_i))^2 K_h(R_i) \delta_i \\
    &- \frac{\mathrm{i} w^3}{6} \frac{1}{n} \sum_{i=1}^n e^{\mathrm{i} (w\bar{W}_i+x'X_i)} (1-D_i) (\hat\tau(X_i) - \tau(X_i))^3 K_h(R_i) \delta_i.
\end{align*}
The third term is of smaller order than the first and second terms due to the cubic term $(\hat\tau(X_i) - \tau(X_i))^3$. Denote the first and second terms as
\begin{align*}
    C_n(w,x) =& \frac{\mathrm{i}w}{n} \sum_{i=1}^n e^{\mathrm{i} (wW_i+x'X_i)} (1-D_i) (\hat\tau(X_i) - \tau(X_i)) K_h(R_i) \delta_i \\
    &- \frac{w^2}{2} \frac{1}{n} \sum_{i=1}^n e^{\mathrm{i} (wW_i+x'X_i)} (1-D_i) (\hat\tau(X_i) - \tau(X_i))^2 K_h(R_i) \delta_i,
\end{align*}
where the treatment effect $\tau(X_i) = \Delta_\mu(X_i) / \Delta_p(X_i)$ can be estimated by either the local constant or the local linear method.

\subsubsection{The local constant estimator}
For the local constant estimator $$ \hat\tau(X_i) = \frac{\frac{1}{n-1}\sum_{j\ne i}Y_j K_{h_x}(X_j-X_i) K_{h_r}(R_j)\delta_j}{\frac{1}{n-1}\sum_{j\ne i}D_j K_{h_x}(X_j-X_i) K_{h_r}(R_j)\delta_j} = \frac{\hat{f}(X_i,0) \Delta_{\hat\mu}(X_i)}{\hat{f}(X_i,0) \Delta_{\hat{p}}(X_i)}, $$ where $\hat{f}(X_i,0)$ denotes the leave-one-out kernel estimator of $f_{XR}(X_i,0)$ based on $K_{h_x}K_{h_r}$. Since $\tau(x)$ can be written as $$\tau(x) = \frac{\mu(x,0^+)-\mu(x,0^-)}{p(x,0^+)-p(x,0^-)} = \frac{\frac{1}{2}f(x,0)\left(\mu(x,0^+)-\mu(x,0^-)\right)}{\frac{1}{2}f(x,0)\left(p(x,0^+)-p(x,0^-)\right)} = \frac{\frac{1}{2}f(x,0) \Delta_\mu(x)}{\frac{1}{2}f(x,0) \Delta_p(x)}, $$ using the expansion of $\hat{a}/\hat{b}$ in the preliminary results, the estimation effect of $\tau(X)$ can be decomposed in terms of $\hat{f}(X_i,0) \Delta_{\hat\mu}(X_i) - \frac{1}{2} f(X_i,0) \Delta_\mu(X_i)$ and $\hat{f}(X_i,0) \Delta_{\hat{p}}(X_i) - \frac{1}{2} f(X_i,0) \Delta_p(X_i)$:
\begin{align*}
    \hat\tau(X_i) - \tau(X_i) =& \frac{\hat{f}(X_i,0) \Delta_{\hat\mu}(X_i) - \frac{1}{2} f(X_i,0) \Delta_\mu(X_i)}{\frac{1}{2} f(X_i,0) \Delta_p(X_i)} - \tau(X_i) \left( \frac{\hat{f}(X_i,0) \Delta_{\hat{p}}(X_i)}{\frac{1}{2} f(X_i,0) \Delta_p(X_i)} -1 \right) \\
    &- \frac{\hat{f}(X_i,0) \Delta_{\hat\mu}(X_i) - \frac{1}{2} f(X_i,0) \Delta_\mu(X_i)}{\frac{1}{2} f(X_i,0) \Delta_p(X_i)} \left( \frac{\hat{f}(X_i,0) \Delta_{\hat{p}}(X_i)}{\frac{1}{2} f(X_i,0) \Delta_p(X_i)} -1 \right) \\
    &+ \tau(X_i) \left( \frac{\hat{f}(X_i,0) \Delta_{\hat{p}}(X_i)}{\frac{1}{2} f(X_i,0) \Delta_p(X_i)} -1 \right)^2 - \hat\tau(X_i) \left( \frac{\hat{f}(X_i,0) \Delta_{\hat{p}}(X_i)}{\frac{1}{2} f(X_i,0) \Delta_p(X_i)} -1 \right)^3 \\
    &+ \frac{\hat{f}(X_i,0) \Delta_{\hat\mu}(X_i) - \frac{1}{2} f(X_i,0) \Delta_\mu(X_i)}{\frac{1}{2} f(X_i,0) \Delta_p(X_i)} \left( \frac{\hat{f}(X_i,0) \Delta_{\hat{p}}(X_i)}{\frac{1}{2} f(X_i,0) \Delta_p(X_i)} -1 \right)^2,
\end{align*}
where the last two terms are of smaller order than the first four terms, and the leading terms can be written as
\begin{align*}
    & \frac{\hat{f}(X_i,0) \Delta_{\hat\mu}(X_i)}{\frac{1}{2} f(X_i,0) \Delta_p(X_i)} - \tau(X_i) \frac{\hat{f}(X_i,0) \Delta_{\hat p}(X_i)}{\frac{1}{2} f(X_i,0) \Delta_p(X_i)} \\&- \left( \frac{\hat{f}(X_i,0) \Delta_{\hat\mu}(X_i)}{\frac{1}{2} f(X_i,0) \Delta_p(X_i)} - \tau(X_i) \right) \left( \frac{\hat{f}(X_i,0) \Delta_{\hat{p}}(X_i)}{\frac{1}{2} f(X_i,0) \Delta_p(X_i)} - 1 \right) + \tau(X_i) \left( \frac{\hat{f}(X_i,0) \Delta_{\hat{p}}(X_i)}{\frac{1}{2} f(X_i,0) \Delta_p(X_i)} - 1 \right)^2.
\end{align*}
Similarly, the leading terms in the expansion of $(\hat\tau(X_i) - \tau(X_i))^2$ are
\begin{align*}
    & \left( \frac{\hat{f}(X_i,0) \Delta_{\hat\mu}(X_i)}{\frac{1}{2} f(X_i,0) \Delta_p(X_i)} - \tau(X_i) \right)^2 + \tau^2(X_i) \left( \frac{\hat{f}(X_i,0) \Delta_{\hat{p}}(X_i)}{\frac{1}{2} f(X_i,0) \Delta_p(X_i)} - 1 \right)^2 \\
    &-2 \tau(X_i) \left( \frac{\hat{f}(X_i,0) \Delta_{\hat\mu}(X_i)}{\frac{1}{2} f(X_i,0) \Delta_p(X_i)} - \tau(X_i) \right) \left( \frac{\hat{f}(X_i,0) \Delta_{\hat{p}}(X_i)}{\frac{1}{2} f(X_i,0) \Delta_p(X_i)} - 1 \right)
\end{align*}

With these results, the leading terms in $C_n(w,x)$ can be decomposed as $C_n(w,x) = C_{1n}(w,x) - C_{2n}(w,x) - C_{3n}(w,x) + C_{4n}(w,x) - C_{5n}(w,x) - C_{6n}(w,x) + C_{7n}(w,x)$, where the elements are 
\begin{align*}
    C_{1n}(w,x) =& \frac{\mathrm{i}w}{n} \sum_{i=1}^n e^{\mathrm{i} (wW_i+x'X_i)} (1-D_i) \frac{\hat f(X_i,0) \Delta_{\hat\mu}(X_i)}{\frac{1}{2} f(X_i,0) \Delta_p(X_i)} K_h(R_i) \delta_i \\
    C_{2n}(w,x) =& \frac{\mathrm{i}w}{n} \sum_{i=1}^n e^{\mathrm{i} (wW_i+x'X_i)} (1-D_i) \frac{\hat f(X_i,0) \Delta_{\hat p}(X_i)}{\frac{1}{2} f(X_i,0) \Delta_p(X_i)} \tau(X_i) K_h(R_i) \delta_i \\
    C_{3n}(w,x) =& \frac{\mathrm{i}w}{n} \sum_{i=1}^n e^{\mathrm{i} (wW_i+x'X_i)} (1-D_i) \left( \frac{\hat{f}(X_i,0) \Delta_{\hat\mu}(X_i)}{\frac{1}{2} f(X_i,0) \Delta_p(X_i)} - \tau(X_i) \right) \left( \frac{\hat{f}(X_i,0) \Delta_{\hat{p}}(X_i)}{\frac{1}{2} f(X_i,0) \Delta_p(X_i)} - 1 \right) \\& K_h(R_i) \delta_i \\
    C_{4n}(w,x) =& \frac{\mathrm{i}w}{n} \sum_{i=1}^n e^{\mathrm{i} (wW_i+x'X_i)} (1-D_i) \left( \frac{\hat{f}(X_i,0) \Delta_{\hat{p}}(X_i)}{\frac{1}{2} f(X_i,0) \Delta_p(X_i)} - 1 \right)^2 \tau(X_i) K_h(R_i) \delta_i \\
    C_{5n}(w,x) =& \frac{w^2}{2n} \sum_{i=1}^n e^{\mathrm{i} (wW_i+x'X_i)} (1-D_i) \left( \frac{\hat{f}(X_i,0) \Delta_{\hat\mu}(X_i)}{\frac{1}{2} f(X_i,0) \Delta_p(X_i)} - \tau(X_i) \right)^2 K_h(R_i) \delta_i \\
    C_{6n}(w,x) =& \frac{w^2}{2n} \sum_{i=1}^n e^{\mathrm{i} (wW_i+x'X_i)} (1-D_i) \tau^2(X_i) \left( \frac{\hat{f}(X_i,0) \Delta_{\hat{p}}(X_i)}{\frac{1}{2} f(X_i,0) \Delta_p(X_i)} - 1 \right)^2 K_h(R_i) \delta_i \\
    C_{7n}(w,x) =& \frac{w^2}{n} \sum_{i=1}^n e^{\mathrm{i} (wW_i+x'X_i)} (1-D_i) \left( \frac{\hat{f}(X_i,0) \Delta_{\hat\mu}(X_i)}{\frac{1}{2} f(X_i,0) \Delta_p(X_i)} - \tau(X_i) \right) \left( \frac{\hat{f}(X_i,0) \Delta_{\hat{p}}(X_i)}{\frac{1}{2} f(X_i,0) \Delta_p(X_i)} - 1 \right) \\& \tau(X_i) K_h(R_i) \delta_i.
\end{align*}
The first two terms $C_{1n}(w,x)$ and $C_{2n}(w,x)$ are second-order $U$-processes. The other terms $C_{3n}(w,x)$--$C_{7n}(w,x)$ are combinations of second-order $U$-processes and third-order $U$-processes. $U$-processes $C_{1n}(w,x) = U_n^{(2)}(C_{1n,ij})$ and $C_{2n}(w,x) = U_n^{(2)}(C_{2n,ij})$ are composed of elements 
\begin{align*}
    C_{1n,ij} =& \mathrm{i}w e^{\mathrm{i} (wW_i+x'X_i)} \frac{1-D_i}{\frac{1}{2} f(X_i,0) \Delta_p(X_i)} Y_j K_{h_x}(X_j-X_i) K_{h_r}(R_j)\delta_j K_h(R_i) \delta_i, \\
    C_{2n,ij} =& \mathrm{i}w e^{\mathrm{i} (wW_i+x'X_i)} \frac{\tau(X_i)(1-D_i)}{\frac{1}{2} f(X_i,0) \Delta_p(X_i)} D_j K_{h_x}(X_j-X_i) K_{h_r}(R_j)\delta_j K_h(R_i) \delta_i. 
\end{align*}
For notational simplicity, henceforth, we omit the indices $(w,x)$ from elements in the $U$-processes. $C_{1n}(w,x)$ and $C_{2n}(w,x)$ have Hoeffding decompositions
\begin{align*}
    C_{1n}(w,x) =& \frac{1}{n}\sum_{i=1}^n \mathbb{E} (C_{1n,ij}|\mathcal{F}_i) + \frac{1}{n}\sum_{j=1}^n \mathbb{E}(C_{1n,ij}|\mathcal{F}_j) - \mathbb{E}C_{1n,ij} + U_n^{(2)}(\pi_2 g_{C_{1n}}), \\
    C_{2n}(w,x) =& \frac{1}{n}\sum_{i=1}^n \mathbb{E}(C_{2n,ij}|\mathcal{F}_i) + \frac{1}{n}\sum_{j=1}^n \mathbb{E}(C_{2n,ij}|\mathcal{F}_j) - \mathbb{E}C_{2n,ij} + U_n^{(2)}(\pi_2 g_{C_{2n}}).
\end{align*}
The conditional expectation terms $\mathbb{E} (C_{1n,ij}|\mathcal{F}_i)$ and $\mathbb{E} (C_{2n,ij}|\mathcal{F}_i)$ are 
\begin{align*}
    \mathbb{E}(C_{1n,ij}|\mathcal{F}_i) = \mathrm{i}w e^{\mathrm{i} (wW_i+x'X_i)} \frac{1-D_i}{\frac{1}{2} f(X_i,0) \Delta_p(X_i)} K_h(R_i) \delta_i \mathbb{E}\left[ Y_j K_{h_x}(X_j-X_i) K_{h_r}(R_j)\delta_j |X_i \right], \\
    \mathbb{E}(C_{2n,ij}|\mathcal{F}_i) = \mathrm{i}w e^{\mathrm{i} (wW_i+x'X_i)} \frac{\tau(X_i)(1-D_i)}{\frac{1}{2} f(X_i,0) \Delta_p(X_i)} K_h(R_i) \delta_i \mathbb{E}\left[ D_j K_{h_x}(X_j-X_i) K_{h_r}(R_j)\delta_j |X_i \right], 
\end{align*}
Following Lemma S.A1, $\mathbb{E}(C_{1n,ij}|\mathcal{F}_i)$ and $\mathbb{E}(C_{2n,ij}|\mathcal{F}_i)$ can be expressed as
\begin{align*}
    \mathbb{E}(C_{1n,ij}|\mathcal{F}_i) =& \mathrm{i}w e^{\mathrm{i} (wW_i+x'X_i)} (1-D_i) \frac{\frac{1}{2} f(X_i,0) \Delta_\mu(X_i)}{\frac{1}{2} f(X_i,0) \Delta_p(X_i)} K_h(R_i) \delta_i + O_p(h_r) + O_p(h_x^l) \\
    =& \mathrm{i}w e^{\mathrm{i} (wW_i+x'X_i)} (1-D_i) \tau(X_i) K_h(R_i) \delta_i + O_p(h_r) + O_p(h_x^l); \\
    \mathbb{E}(C_{2n,ij}|\mathcal{F}_i) =& \mathrm{i}w e^{\mathrm{i} (wW_i+x'X_i)} (1-D_i) \tau(X_i) \frac{\frac{1}{2} f(X_i,0) \Delta_p(X_i)}{\frac{1}{2} f(X_i,0) \Delta_p(X_i)} K_h(R_i) \delta_i + O_p(h_r) + O_p(h_x^l) \\
    =& \mathrm{i}w e^{\mathrm{i} (wW_i+x'X_i)} (1-D_i) \tau(X_i) K_h(R_i) \delta_i + O_p(h_r) + O_p(h_x^l).
\end{align*}
The sums of the small-order terms are bounded asymptotically since 
\begin{align*}
    & \mathbb{E}\sup_{(w,x)\in\Omega} \bigg| \frac{1}{n}\sum_{i=1}^n \mathrm{i}w e^{\mathrm{i} (wW_i+x'X_i)} \frac{1-D_i}{\frac{1}{2} f(X_i,0) \Delta_p(X_i)} h_r \bigg( \nu_1^+ \frac{\partial}{\partial r}\left( \mu(X_i,0^+)f(X_i,0^+)\right) \\&- \nu_1^- \frac{\partial}{\partial r}\left( \mu(X_i,0^-)f(X_i,0^-)\right) \bigg) K_h(R_i) \delta_i \bigg| \\
    \le& \mathbb{E}\sup_{(w,x)\in\Omega} \frac{1}{n}\sum_{i=1}^n \left| \mathrm{i}w e^{\mathrm{i} (wW_i+x'X_i)} \frac{1-D_i}{\frac{1}{2} f(X_i,0) \Delta_p(X_i)} h_r \nu_1^+ \frac{\partial}{\partial r}\left( \mu(X_i,0^+)f(X_i,0^+)\right) K_h(R_i) \delta_i \right| \\
    &+ \mathbb{E}\sup_{(w,x)\in\Omega} \frac{1}{n}\sum_{i=1}^n \left| \mathrm{i}w e^{\mathrm{i} (wW_i+x'X_i)} \frac{1-D_i}{\frac{1}{2} f(X_i,0) \Delta_p(X_i)} h_r \nu_1^- \frac{\partial}{\partial r}\left( \mu(X_i,0^-)f(X_i,0^-)\right) K_h(R_i) \delta_i \right| \\
    \le& \mathbb{E}\sup_{(w,x)\in\Omega} \left| \mathrm{i}w e^{\mathrm{i} (wW+x'X)} \frac{1-D}{\frac{1}{2} f(X,0) \Delta_p(X)} h_r \nu_1^+ \frac{\partial}{\partial r}\left( \mu(X,0^+)f(X,0^+)\right) K_h(R) \delta \right| \\
    &+ \mathbb{E}\sup_{(w,x)\in\Omega} \left| \mathrm{i}w e^{\mathrm{i} (wW+x'X)} \frac{1-D}{\frac{1}{2} f(X,0) \Delta_p(X)} h_r \nu_1^- \frac{\partial}{\partial r}\left( \mu(X,0^-)f(X,0^-)\right) K_h(R) \delta \right| \\
    \lesssim& \mathbb{E} \left| h_r \frac{\partial}{\partial r}\left( \mu(X,0^+)f(X,0^+)\right) K_h(R) \right| + \mathbb{E} \left| h_r \frac{\partial}{\partial r}\left( \mu(X,0^-)f(X,0^-)\right) K_h(R) \right| = O(h_r); \\
    & \mathbb{E}\sup_{(w,x)\in\Omega} \bigg| \frac{1}{n}\sum_{i=1}^n \mathrm{i}w e^{\mathrm{i} (wW_i+x'X_i)} \frac{1-D_i}{\frac{1}{2} f(X_i,0) \Delta_p(X_i)} \frac{h_x^l\nu_l}{l!} \frac{\partial^l}{\partial x^l} \left( f(X_i,0) \Delta_\mu(X_i) \right) K_h(R_i) \delta_i \bigg| \\
    \le& \mathbb{E}\sup_{(w,x)\in\Omega} \left| \mathrm{i}w e^{\mathrm{i} (wW+x'X)} \frac{1-D}{\frac{1}{2} f(X,0) \Delta_p(X)} \frac{h_x^l\nu_l}{l!} \frac{\partial^l}{\partial x^l} \left( f(X,0) \Delta_\mu(X) \right) K_h(R) \delta \right| \\
    \lesssim& \mathbb{E} \left| h_x^l \frac{\partial^l}{\partial x^l} \left( f(X,0) \Delta_\mu(X) \right) K_h(R) \right| = O(h_x^l);
\end{align*}
The proof also applies to the other side of the cutoff and applies to the residual term from a Taylor expansion of the function $p(x,r)f(x,r)$. 
Then $\frac{1}{n} \sum_{i=1}^n \mathbb{E}(C_{1n,ij}|\mathcal{F}_i) - \frac{1}{n} \sum_{i=1}^n \mathbb{E}(C_{2n,ij}|\mathcal{F}_i)$ is uniformly bounded by $\frac{1}{\sqrt{nh}}$ given Assumption B1. The total expectations are asymptotically equal with $\mathbb{E}C_{1n,ij} - \mathbb{E}C_{2n,ij} = O(h_r) + O(h_x^l)$, which is also uniformly bounded by $\frac{1}{\sqrt{nh}}$ given Assumption B1.

The conditional expectations given $\mathcal{F}_j$ are
\begin{align*}
    \mathbb{E}(C_{1n,ij}|\mathcal{F}_j) =& \mathrm{i}w \mathbb{E}\left[ e^{\mathrm{i} x'X_i} \frac{(1-D_i)e^{\mathrm{i}wW_i}}{\frac{1}{2} f(X_i,0) \Delta_p(X_i)} K_{h_x}(X_j-X_i) K_h(R_i) \delta_i \bigg| X_j \right] Y_j K_{h_r}(R_j)\delta_j, \\
    \mathbb{E}(C_{2n,ij}|\mathcal{F}_j) =& \mathrm{i}w \mathbb{E}\left[ e^{\mathrm{i} x'X_i} \frac{(1-D_i)e^{\mathrm{i}wW_i}}{\frac{1}{2} f(X_i,0) \Delta_p(X_i)} \tau(X_i) K_{h_x}(X_j-X_i) K_h(R_i) \delta_i \bigg| X_j \right] D_j K_{h_r}(R_j)\delta_j .
\end{align*}
Applying Lemma S.A1, the conditional expectation terms for variable $Z_i=(1-D_i)e^{\mathrm{i}wW_i}$ with functions $g_1(X_i)=\frac{e^{\mathrm{i} x'X_i}}{\frac{1}{2} f(X_i,0) \Delta_p(X_i)}$ and $g_2(X_i)=\frac{\tau(X_i) e^{\mathrm{i} x'X_i}}{\frac{1}{2} f(X_i,0) \Delta_p(X_i)}$ are 
\begin{align*}
    \mathbb{E}\left[ g_1(X_i) (1-D_i)e^{\mathrm{i}wW_i} K_{h_x}(X_j-X_i) K_h(R_i) \delta_i \bigg| X_j \right] =& e^{\mathrm{i} x'X_j} \frac{\frac{1}{2} f(X_j,0) \Delta_\varphi(w,X_j)}{\frac{1}{2} f(X_j,0) \Delta_p(X_j)} + O_p(h + h_x^l), \\
    \mathbb{E}\left[ g_2(X_i) (1-D_i)e^{\mathrm{i}wW_i} K_{h_x}(X_j-X_i) K_h(R_i) \delta_i \bigg| X_j \right] =& e^{\mathrm{i} x'X_j} \frac{\frac{1}{2} f(X_j,0) \Delta_\varphi(w,X_j)}{\frac{1}{2} f(X_j,0) \Delta_p(X_j)} \tau(X_j) + O_p(h + h_x^l).
\end{align*}
Therefore, $\mathbb{E}(C_{1n,ij}|\mathcal{F}_j)$ and $\mathbb{E}(C_{2n,ij}|\mathcal{F}_j)$ can be rewritten as
\begin{align*}
    \mathbb{E}(C_{1n,ij}|\mathcal{F}_j) =& \mathrm{i}w e^{\mathrm{i}x'X_j} \kappa(w,X_j) Y_j K_{h_r}(R_j)\delta_j + O_p(h) + O_p(h_x^l), \\
    \mathbb{E}(C_{2n,ij}|\mathcal{F}_j) =& \mathrm{i}w e^{\mathrm{i}x'X_j} \kappa(w,X_j) \tau(X_j) D_j K_{h_r}(R_j)\delta_j + O_p(h) + O_p(h_x^l) .
\end{align*}
The summations of the $O_p(h)$ and $O_p(h_x^l)$ terms are uniformly bounded by $\frac{1}{\sqrt{nh}}$ since 
\begin{align*}
    & \mathbb{E}\sup_{(w,x)\in\Omega} \bigg| \frac{1}{n}\sum_{j=1}^n \mathrm{i}w \frac{h_x^l \nu_l}{l!} \frac{\partial^l}{\partial x^l} \left( g(X_j) \frac{1}{2} f(X_j,0) \Delta_\varphi(w,X_j) \right) Y_j K_{h_r}(R_j) \delta_j \bigg| \\
    \le & \mathbb{E}\sup_{(w,x)\in\Omega} \bigg| \mathrm{i}w \frac{h_x^l \nu_l}{l!} \frac{\partial^l}{\partial x^l} \left( g(X) \frac{1}{2} f(X,0) \Delta_\varphi(w,X) \right) Y K_{h_r}(R) \delta \bigg| \\
    \lesssim& \mathbb{E} \left| h_x^l \frac{\partial^l}{\partial x^l} \left( g(X) f(X,0) \Delta_\varphi(w,X) \right) Y K_{h_r}(R) \right| = O(h_x^l); \\
    & \mathbb{E}\sup_{(w,x)\in\Omega} \bigg| \frac{1}{n}\sum_{j=1}^n \mathrm{i}w h \bigg( \nu_1^+ \frac{\partial}{\partial r} \left( g(X_j) f(X_j,0) \varphi(w,0|X_j,0^+) \right) \\&- \nu_1^- \frac{\partial}{\partial r} \left( g(X_j) f(X_j,0) \varphi(w,0|X_j,0^-) \right) \bigg)  Y_j K_{h_r}(R_j) \delta_j \bigg| \\
    \le& \mathbb{E}\sup_{(w,x)\in\Omega} \left| \mathrm{i}w h \nu_1^+ \frac{\partial}{\partial r} \left( g(X) f(X,0) \varphi(w,0|X,0^+) \right) Y K_{h_r}(R) \delta \right| \\
    &+ \mathbb{E}\sup_{(w,x)\in\Omega} \left| \mathrm{i}w h \nu_1^- \frac{\partial}{\partial r} \left( g(X) f(X,0) \varphi(w,0|X,0^-) \right) Y K_{h_r}(R) \delta \right| \\
    \lesssim& \mathbb{E} \left| h \frac{\partial}{\partial r} \left( g(X) f(X,0) \varphi(w,0|X,0^+) \right) Y K_{h_r}(R)\right| \\&+ \mathbb{E} \left| h \frac{\partial}{\partial r} \left( g(X) f(X,0) \varphi(w,0|X,0^-) \right) Y K_{h_r}(R) \right| = O(h); 
\end{align*}

For the high-order term in the decomposition of $C_{1n}(w,x)$, $U_n^{(2)}(\pi_2 g_{C_{1n}})$ is composed of the symmetrized function $g_{C_{1n}} := \frac{1}{2} (C_{1n,ij} + C_{1n,ji})$. $g_{C_{1n}}$ is of VC-type and has envelope $G_{C_{1n}} = \sup_{(w,x)\in\Omega}|g_{C_{1n}}|$. The second moment of the envelope $G_{C_{1n}}$ can be derived by a fundamental inequality
\begin{align*}
    G_{C_{1n}} =& \sup_{(w,x)\in\Omega} \left| \frac{1}{2} (C_{1n,ij} + C_{1n,ji}) \right| \le \frac{1}{2} \sup_{(w,x)\in\Omega} \left|C_{1n,ij}\right| + \frac{1}{2} \sup_{(w,x)\in\Omega}\left|C_{1n,ji}\right| \\
    G_{C_{1n}}^2 \le& \left( \frac{1}{2} \sup_{(w,x)\in\Omega} \left|C_{1n,ij}\right| + \frac{1}{2} \sup_{(w,x)\in\Omega}\left|C_{1n,ji}\right| \right)^2 \\
    =& \frac{1}{2} \sup_{(w,x)\in\Omega} \left|C_{1n,ij}\right|^2 + \frac{1}{2} \sup_{(w,x)\in\Omega}\left|C_{1n,ji}\right|^2 = \sup_{(w,x)\in\Omega} \left|C_{1n,ij}\right|^2. \\
    \mathbb{E} G_{C_{1n}}^2 \le& \mathbb{E} \sup_{(w,x)\in\Omega} \left|C_{1n,ij}\right|^2 \\
    =& \mathbb{E} \sup_{(w,x)\in\Omega} \bigg| \mathrm{i}w e^{\mathrm{i} (wW_i+x'X_i)} \frac{1-D_i}{\frac{1}{2} f(X_i,0) \Delta_p(X_i)} Y_j K_{h_x}(X_j-X_i) K_{h_r}(R_j) K_h(R_i) \delta_i \bigg|^2 \\
    \lesssim& \mathbb{E} \left[ Y_j K_{h_x}(X_j-X_i) K_{h_r}(R_j) K_h(R_i) \right]^2 = O(\frac{1}{hh_rh_x^d}). 
\end{align*}
The last equation is guaranteed by Assumption B5. Following Proposition 4 in \cite{Delgado2001}:
\begin{align*}
    \mathbb{E} \sup_{(w,x)\in\Omega} \left| \frac{1}{n^{2/2}} \sum_{i\ne j} (\pi_2 g_{C_{1n}}) \right|^2 &= \mathbb{E} \sup_{(w,x)\in\Omega} \left| (n-1) U_n^{(2)}(\pi_2 g_{C_{1n}}) \right|^2 \lesssim \mathbb{E} G_{C_{1n}}^2. \\
    \mathbb{E} \sup_{(w,x)\in\Omega} \left| U_n^{(2)}(\pi_2 g_{C_{1n}}) \right|^2 &\lesssim \frac{\mathbb{E} G_{C_{1n}}^2}{(n-1)^2} = O(\frac{1}{n^2hh_rh_x^d}) = o(\frac{1}{nh}). 
\end{align*}
Thus, $\sup_{(w,x)\in\Omega} \left| U_n^{(2)}(\pi_2 g_{C_{1n}}) \right| = o_p(\frac{1}{\sqrt{nh}})$. Similarly, the high-order term $U_n^{(2)}(\pi_2 g_{C_{2n}})$ in the decomposition of $C_{2n}(w,x)$ is bounded given the second moment of the envelope $\mathbb{E} G_{C_{2n}}^2 \le \mathbb{E} \sup_{(w,x)\in\Omega} \left|C_{2n,ij}\right|^2$:
\begin{align*}
    \mathbb{E} G_{C_{2n}}^2 \le& \mathbb{E} \sup_{(w,x)\in\Omega} \bigg| \mathrm{i}w e^{\mathrm{i} (wW_i+x'X_i)} \frac{\tau(X_i)(1-D_i)}{\frac{1}{2} f(X_i,0) \Delta_p(X_i)} D_j K_{h_x}(X_j-X_i) K_{h_r}(R_j)\delta_j K_h(R_i) \delta_i \bigg|^2 \\
    \lesssim& \mathbb{E} \left[ \tau(X_i) K_{h_x}(X_j-X_i) K_{h_r}(R_j) K_h(R_i) \right]^2 = O(\frac{1}{hh_rh_x^d}).
\end{align*}
Following Proposition 4 in \cite{Delgado2001}, $U_n^{(2)}(\pi_2 g_{C_{2n}})$ is uniformly bounded by $\frac{1}{\sqrt{nh}}$ over $(w,x)\in\Omega$ since $$\mathbb{E} \sup_{(w,x)\in\Omega} \left| U_n^{(2)}(\pi_2 g_{C_{2n}}) \right|^2 \lesssim \frac{\mathbb{E} G_{C_{2n}}^2}{(n-1)^2} = O(\frac{1}{n^2hh_rh_x^d}) = o(\frac{1}{nh}). $$

Summarizing the aforementioned results for $C_{1n}(w,x)$ and $C_{2n}(w,x)$, 
\begin{align*}
    & C_{1n}(w,x) - C_{2n}(w,x) \\
    =& \frac{1}{n} \sum_{j=1}^n \mathrm{i}w e^{\mathrm{i}x'X_j} \kappa(w,X_j) \left( Y_j - D_j\tau(X_j) \right) K_{h_r}(R_j)\delta_j + O_p(h_r) + O_p(h_x^l) + o_p(\frac{1}{\sqrt{nh}}) \\
    =& \frac{1}{n} \sum_{j=1}^n \mathrm{i}w e^{\mathrm{i}x'X_j} \kappa(w,X_j) \left( Y_j - D_j\tau(X_j) \right) K_{h_r}(R_j)\delta_j + o_p(\frac{1}{\sqrt{nh}}).
\end{align*}

For the third term in the leading terms of $C_n(w,x)$, $$ C_{3n}(w,x) = \frac{\mathrm{i}w}{n} \sum_{i=1}^n e^{\mathrm{i} (wW_i+x'X_i)} (1-D_i) \left( \frac{\hat{f}(X_i,0) \Delta_{\hat\mu}(X_i)}{\frac{1}{2} f(X_i,0) \Delta_p(X_i)} - \tau(X_i) \right) \left( \frac{\hat{f}(X_i,0) \Delta_{\hat{p}}(X_i)}{\frac{1}{2} f(X_i,0) \Delta_p(X_i)} - 1 \right) K_h(R_i) \delta_i, $$ by expanding the components $\frac{\hat{f}(X_i,0) \Delta_{\hat\mu}(X_i)}{\frac{1}{2} f(X_i,0) \Delta_p(X_i)} - \tau(X_i)$ and $\frac{\hat{f}(X_i,0) \Delta_{\hat{p}}(X_i)}{\frac{1}{2} f(X_i,0) \Delta_p(X_i)} - 1$ in the summation, it can be written as $$C_{3n}(w,x) = \frac{1}{n-1} C_{31n}(w,x) + \frac{n-2}{n-1} C_{32n}(w,x), $$ where $C_{31n}(w,x) = U_n^{(2)}(C_{31n,ij})$ is a second-order $U$-process and $C_{32n}(w,x) = U_n^{(3)}(C_{32n,ijk})$ is a third-order $U$-process with elements 
\begin{align*}
    C_{31n,ij} =& \mathrm{i}w e^{\mathrm{i} (wW_i+x'X_i)} (1-D_i) \left( \frac{Y_j K_{h_x}(X_j-X_i) K_{h_r}(R_j) \delta_j}{\frac{1}{2} f(X_i,0) \Delta_p(X_i)} - \tau(X_i) \right) \\& \left( \frac{D_j K_{h_x}(X_j-X_i) K_{h_r}(R_j) \delta_j}{\frac{1}{2} f(X_i,0) \Delta_p(X_i)} -1 \right) K_h(R_i) \delta_i, \\
    C_{32n,ijk} =& \mathrm{i}w e^{\mathrm{i} (wW_i+x'X_i)} (1-D_i) \left( \frac{Y_j K_{h_x}(X_j-X_i) K_{h_r}(R_j) \delta_j}{\frac{1}{2} f(X_i,0) \Delta_p(X_i)} - \tau(X_i) \right) \\& \left( \frac{D_k K_{h_x}(X_k-X_i) K_{h_r}(R_k) \delta_k}{\frac{1}{2} f(X_i,0) \Delta_p(X_i)} -1 \right) K_h(R_i) \delta_i. 
\end{align*}
The Hoeffding decomposition of $U$-process $C_{31n}(w,x)$ is $$ C_{31n}(w,x) = \mathbb{E} C_{31n,ij} + 2U_n^{(1)}(\pi_1 g_{C_{31n}}) + U_n^{(2)}(\pi_2 g_{C_{31n}}). $$ The expectation $\mathbb{E} C_{31n,ij}$ is $O(\frac{1}{h_x^dh_r})$ and $\frac{1}{n-1} \mathbb{E} C_{31n,ij} = O(\frac{1}{nh_x^dh_r}) = o(\frac{1}{\sqrt{nh}})$. For the high-order terms, the second-order moment of envelope $G_{C_{31n}}$ is bounded by
\begin{align*}
    \mathbb{E} G_{C_{31n}}^2 \le& \mathbb{E} \sup_{(w,x)\in\Omega} \left| \mathrm{i}w e^{\mathrm{i} (wW_i+x'X_i)} (1-D_i) \left( \frac{Y_j K_{h_x}(X_j-X_i) K_{h_r}(R_j) \delta_j}{\frac{1}{2} f(X_i,0) \Delta_p(X_i)} - \tau(X_i) \right) \right. \\& \left. \left( \frac{D_j K_{h_x}(X_j-X_i) K_{h_r}(R_j) \delta_j}{\frac{1}{2} f(X_i,0) \Delta_p(X_i)} -1 \right) K_h(R_i) \delta_i \right|^2 \\
    \lesssim& \mathbb{E} \left| Y_j K_{h_x}^2(X_j-X_i) K_{h_r}^2(R_j) K_h(R_i) \right|^2 + \mathbb{E} \left| Y_j K_{h_x}(X_j-X_i) K_{h_r}(R_j) K_h(R_i) \right|^2 \\&+ \mathbb{E} \left| K_{h_x}(X_j-X_i) K_{h_r}(R_j) K_h(R_i) \right|^2 + \mathbb{E} \left| \tau(X_i) K_h(R_i) \right|^2 \\
    =& O(h_x^{-3d} h_r^{-3} h^{-1}) + O(h_x^{-d} h_r^{-1} h^{-1}) + O(h^{-1}) = O(h_x^{-3d} h_r^{-3} h^{-1}).
\end{align*}
Following Proposition 4 in \cite{Delgado2001}, the first- and second-order terms in the $U$-process are uniformly bounded since
\begin{align*}
    \mathbb{E} \sup_{(w,x)\in\Omega} \left\Vert \frac{2}{n-1} U_n^{(1)}(\pi_1 g_{C_{31n}}) \right\Vert^2 &\lesssim \frac{\mathbb{E} G_{C_{31n}}^2}{n(n-1)^2} = O(\frac{1}{n^3h_x^{3d}h_r^3h}) = o(\frac{1}{nh}). \\
    \mathbb{E} \sup_{(w,x)\in\Omega} \left\Vert \frac{1}{n-1} U_n^{(2)}(\pi_2 g_{C_{31n}}) \right\Vert^2 &\lesssim \frac{\mathbb{E} G_{C_{31n}}^2}{(n-1)^4} = O(\frac{1}{n^4h_x^{3d}h_r^3h}) = o(\frac{1}{nh}).
\end{align*}
Then it can be concluded that $\sup_{(w,x)\in\Omega} \left\Vert \frac{C_{31n}(w,x)}{n-1} \right\Vert = o_p\left( (nh)^{-1/2} \right)$. 

The Hoeffding decomposition of $U$-process $C_{32n}(w,x)$ is \begin{align*}
    C_{32n}(w,x) =& \frac{1}{n} \sum_{i=1}^n \mathbb{E} \left( C_{32n,ijk} | \mathcal{F}_i \right) + \frac{1}{n} \sum_{j=1}^n \mathbb{E} \left( C_{32n,ijk} | \mathcal{F}_j \right) + \frac{1}{n} \sum_{k=1}^n \mathbb{E} \left( C_{32n,ijk} | \mathcal{F}_k \right) \\ &-2 \mathbb{E} C_{32n,ijk} + 3 U_n^{(2)} (\pi_2 g_{C_{32n}}) + U_n^{(3)} (\pi_3 g_{C_{32n}}).
\end{align*}
The conditional expectation terms are
\begin{align*}
    & \mathbb{E} \left( C_{32n,ijk} | \mathcal{F}_i \right) = \mathrm{i}w e^{\mathrm{i} (wW_i+x'X_i)} (1-D_i) \mathbb{E} \left[ \left( \frac{Y_j K_{h_x}(X_j-X_i) K_{h_r}(R_j) \delta_j}{\frac{1}{2} f(X_i,0) \Delta_p(X_i)} - \tau(X_i) \right) \right. \\& \left. \left( \frac{D_k K_{h_x}(X_k-X_i) K_{h_r}(R_k) \delta_k}{\frac{1}{2} f(X_i,0) \Delta_p(X_i)} -1 \right) \bigg|X_i \right] K_h(R_i) \delta_i \\
    =& \mathrm{i}w e^{\mathrm{i} (wW_i+x'X_i)} \frac{(1-D_i) m_\mu(X_i) m_p(X_i)}{\left( \frac{1}{2} f(X_i,0) \Delta_p(X_i) \right)^2} K_h(R_i) \delta_i = [O_p(h_x^l) + O_p(h_r)]^2 = O_p(h_x^{2l}) + O_p(h_r^2); 
\end{align*}
\begin{align*}
    & \mathbb{E} \left( C_{32n,ijk} | \mathcal{F}_j \right) = \mathbb{E} \left( \mathbb{E} (C_{32n,ijk}|\mathcal{F}_i, \mathcal{F}_j) | \mathcal{F}_j \right) \\
    =& \mathbb{E} \bigg[ \mathrm{i}w e^{\mathrm{i} (wW_i+x'X_i)} (1-D_i) K_h(R_i) \delta_i \left( \frac{Y_j K_{h_x}(X_j-X_i) K_{h_r}(R_j) \delta_j}{\frac{1}{2} f(X_i,0) \Delta_p(X_i)} - \tau(X_i) \right) \\& \mathbb{E} \left( \frac{D_k K_{h_x}(X_k-X_i) K_{h_r}(R_k) \delta_k}{\frac{1}{2} f(X_i,0) \Delta_p(X_i)} -1 \bigg|X_i \right) \bigg| \mathcal{F}_j \bigg] \\
    =& \mathbb{E} \bigg[ \mathrm{i}w e^{\mathrm{i} (wW_i+x'X_i)} (1-D_i) \left( \frac{Y_j K_{h_x}(X_j-X_i) K_{h_r}(R_j) \delta_j}{\frac{1}{2} f(X_i,0) \Delta_p(X_i)} - \tau(X_i) \right) \frac{m_p(X_i) K_h(R_i) \delta_i}{\frac{1}{2} f(X_i,0) \Delta_p(X_i)} \bigg| \mathcal{F}_j \bigg] \\=& O_p(h_x^l) + O_p(h_r); \\
    & \mathbb{E} \left( C_{32n,ijk} | \mathcal{F}_k \right) = \mathbb{E} \left( \mathbb{E} (C_{32n,ijk}|\mathcal{F}_i, \mathcal{F}_k) | \mathcal{F}_k \right) \\
    =& \mathbb{E} \bigg[ \mathrm{i}w e^{\mathrm{i} (wW_i+x'X_i)} (1-D_i) \left( \frac{D_k K_{h_x}(X_k-X_i) K_{h_r}(R_k) \delta_k}{\frac{1}{2} f(X_i,0) \Delta_p(X_i)} -1 \right) \frac{m_\mu(X_i) K_h(R_i) \delta_i}{\frac{1}{2} f(X_i,0) \Delta_p(X_i)} \bigg| \mathcal{F}_k \bigg] \\=& O_p(h_x^l) + O_p(h_r).
\end{align*}
Then the total expectation is $\mathbb{E}[\mathbb{E}(C_{32n,ijk}|\mathcal{F}_i)] = O(h_x^{2l}) + O(h_r^2)$. The second-order moment of envelope $G_{C_{32n}}$ is bounded by
\begin{align*}
    \mathbb{E} G_{C_{32n}}^2 \le& \mathbb{E} \sup_{(w,x)\in\Omega} \left| \mathrm{i}w e^{\mathrm{i} (wW_i+x'X_i)} (1-D_i) \left( \frac{Y_j K_{h_x}(X_j-X_i) K_{h_r}(R_j) \delta_j}{\frac{1}{2} f(X_i,0) \Delta_p(X_i)} - \tau(X_i) \right) \right. \\& \left. \left( \frac{D_k K_{h_x}(X_k-X_i) K_{h_r}(R_k) \delta_k}{\frac{1}{2} f(X_i,0) \Delta_p(X_i)} -1 \right) K_h(R_i) \delta_i \right|^2 \\
    \lesssim& \mathbb{E} \left| Y_j K_{h_x}(X_j-X_i) K_{h_x}(X_k-X_i) K_{h_r}(R_j) K_{h_r}(R_k) K_h(R_i) \right|^2 + \mathbb{E} \left| \tau(X_i) K_h(R_i) \right|^2 \\
    &+ \mathbb{E} \left| Y_j K_{h_x}(X_j-X_i) K_{h_r}(R_j) K_h(R_i) \right|^2 + \mathbb{E} \left| K_{h_x}(X_k-X_i) K_{h_r}(R_k) K_h(R_i) \right|^2 \\
    =& O(h_x^{-2d} h_r^{-2} h^{-1}) + O(h_x^{-d} h_r^{-1} h^{-1}) + O(h^{-1}) = O(h_x^{-2d} h_r^{-2} h^{-1}).
\end{align*}
Following Proposition 4 in \cite{Delgado2001}, the second- and third-order terms in the $U$-process are uniformly bounded since
\begin{align*}
    \mathbb{E} \sup_{(w,x)\in\Omega} \left\Vert 3 \frac{n-2}{n-1} U_n^{(2)}(\pi_2 g_{C_{32n}}) \right\Vert^2 &\lesssim \frac{(n-2)^2 \mathbb{E} G_{C_{32n}}^2}{(n-1)^4} = O(\frac{1}{n^2h_x^{2d}h_r^2h}) = o(\frac{1}{nh}), \\
    \mathbb{E} \sup_{(w,x)\in\Omega} \left\Vert \frac{n-2}{n-1} U_n^{(3)}(\pi_3 g_{C_{32n}}) \right\Vert^2 &\lesssim \frac{n \mathbb{E} G_{C_{32n}}^2}{(n-1)^4} = O(\frac{1}{n^3h_x^{2d}h_r^2h}) = o(\frac{1}{nh}).
\end{align*}
Then it can be concluded that $\sup_{(w,x)\in\Omega} \left\Vert \frac{n-2}{n-1} C_{32n}(w,x) \right\Vert = o_p\left( (nh)^{-1/2} \right)$. Furthermore, $\sup_{(w,x)\in\Omega} \left\Vert C_{3n}(w,x) \right\Vert = o_p\left( (nh)^{-1/2} \right)$.

For the fourth term in the leading terms of $C_n(w,x)$, $$ C_{4n}(w,x) = \frac{\mathrm{i}w}{n} \sum_{i=1}^n e^{\mathrm{i} (wW_i+x'X_i)} (1-D_i) \left( \frac{\hat{f}(X_i,0) \Delta_{\hat{p}}(X_i)}{\frac{1}{2} f(X_i,0) \Delta_p(X_i)} - 1 \right)^2 \tau(X_i) K_h(R_i) \delta_i, $$ by expanding $\frac{\hat{f}(X_i,0) \Delta_{\hat{p}}(X_i)}{\frac{1}{2} f(X_i,0) \Delta_p(X_i)} - 1$ in the summation, it can be written as $$C_{4n}(w,x) = \frac{1}{n-1} C_{41n}(w,x) + \frac{n-2}{n-1} C_{42n}(w,x), $$ where $C_{41n}(w,x) = U_n^{(2)}(C_{41n,ij})$ is a second-order $U$-process and $C_{42n}(w,x) = U_n^{(3)}(C_{42n,ijk})$ is a third-order $U$-process with elements
\begin{align*}
    C_{41n,ij} =& \mathrm{i}w e^{\mathrm{i} (wW_i+x'X_i)} (1-D_i) \left( \frac{D_j K_{h_x}(X_j-X_i) K_{h_r}(R_j) \delta_j}{\frac{1}{2} f(X_i,0) \Delta_p(X_i)} -1 \right)^2 \tau(X_i) K_h(R_i) \delta_i, \\
    C_{42n,ijk} =& \mathrm{i}w e^{\mathrm{i} (wW_i+x'X_i)} (1-D_i) \left( \frac{D_j K_{h_x}(X_j-X_i) K_{h_r}(R_j) \delta_j}{\frac{1}{2} f(X_i,0) \Delta_p(X_i)} - 1 \right) \\& \left( \frac{D_k K_{h_x}(X_k-X_i) K_{h_r}(R_k) \delta_k}{\frac{1}{2} f(X_i,0) \Delta_p(X_i)} -1 \right) \tau(X_i) K_h(R_i) \delta_i. 
\end{align*}
The Hoeffding decomposition of $U$-process $C_{41n}(w,x)$ is $$ C_{41n}(w,x) = \mathbb{E} C_{41n,ij} + 2U_n^{(1)}(\pi_1 g_{C_{41n}}) + U_n^{(2)}(\pi_2 g_{C_{41n}}). $$ The expectation $\mathbb{E} C_{41n,ij}$ is $O(\frac{1}{h_x^dh_r})$ and $\frac{1}{n-1} \mathbb{E} C_{41n,ij} = O(\frac{1}{nh_x^dh_r}) = o(\frac{1}{\sqrt{nh}})$. For the high-order terms, the second-order moment of envelope $G_{C_{41n}}$ is bounded by
\begin{align*}
    \mathbb{E} G_{C_{41n}}^2 \le& \mathbb{E} \sup_{(w,x)\in\Omega} \left| \mathrm{i}w e^{\mathrm{i} (wW_i+x'X_i)} (1-D_i) \left( \frac{D_j K_{h_x}(X_j-X_i) K_{h_r}(R_j) \delta_j}{\frac{1}{2} f(X_i,0) \Delta_p(X_i)} -1 \right)^2 \tau(X_i) K_h(R_i) \delta_i \right|^2 \\
    \lesssim& \mathbb{E} \left| \tau(X_i) K_{h_x}^2(X_j-X_i) K_{h_r}^2(R_j) K_h(R_i) \right|^2 + \mathbb{E} \left| \tau(X_i) K_{h_x}(X_j-X_i) K_{h_r}(R_j) K_h(R_i) \right|^2 \\&+ \mathbb{E} \left| \tau(X_i) K_h(R_i) \right|^2 \\
    =& O(h_x^{-3d} h_r^{-3} h^{-1}) + O(h_x^{-d} h_r^{-1} h^{-1}) + O(h^{-1}) = O(h_x^{-3d} h_r^{-3} h^{-1}).
\end{align*}
Following Proposition 4 in \cite{Delgado2001}, the first- and second-order terms in the $U$-process are uniformly bounded since
\begin{align*}
    \mathbb{E} \sup_{(w,x)\in\Omega} \left\Vert \frac{2}{n-1} U_n^{(1)}(\pi_1 g_{C_{41n}}) \right\Vert^2 &\lesssim \frac{\mathbb{E} G_{C_{41n}}^2}{n(n-1)^2} = O(\frac{1}{n^3h_x^{3d}h_r^3h}) = o(\frac{1}{nh}). \\
    \mathbb{E} \sup_{(w,x)\in\Omega} \left\Vert \frac{1}{n-1} U_n^{(2)}(\pi_2 g_{C_{41n}}) \right\Vert^2 &\lesssim \frac{\mathbb{E} G_{C_{41n}}^2}{(n-1)^4} = O(\frac{1}{n^4h_x^{3d}h_r^3h}) = o(\frac{1}{nh}).
\end{align*}
Then it can be concluded that $\sup_{(w,x)\in\Omega} \left\Vert \frac{C_{41n}(w,x)}{n-1} \right\Vert = o_p\left( (nh)^{-1/2} \right)$. 

The Hoeffding decomposition of $U$-process $C_{42n}(w,x)$ is \begin{align*}
    C_{42n}(w,x) =& \frac{1}{n} \sum_{i=1}^n \mathbb{E} \left( C_{42n,ijk} | \mathcal{F}_i \right) + \frac{1}{n} \sum_{j=1}^n \mathbb{E} \left( C_{42n,ijk} | \mathcal{F}_j \right) + \frac{1}{n} \sum_{k=1}^n \mathbb{E} \left( C_{42n,ijk} | \mathcal{F}_k \right) \\ &-2 \mathbb{E} C_{42n,ijk} + 3 U_n^{(2)} (\pi_2 g_{C_{42n}}) + U_n^{(3)} (\pi_3 g_{C_{42n}}).
\end{align*}
The conditional expectation terms are
\begin{align*}
    & \mathbb{E} \left( C_{42n,ijk} | \mathcal{F}_i \right) = \mathrm{i}w e^{\mathrm{i} (wW_i+x'X_i)} (1-D_i) \mathbb{E} \left[ \left( \frac{D_j K_{h_x}(X_j-X_i) K_{h_r}(R_j) \delta_j}{\frac{1}{2} f(X_i,0) \Delta_p(X_i)} - 1 \right) \right. \\& \left. \left( \frac{D_k K_{h_x}(X_k-X_i) K_{h_r}(R_k) \delta_k}{\frac{1}{2} f(X_i,0) \Delta_p(X_i)} -1 \right) \bigg|X_i \right] \tau(X_i) K_h(R_i) \delta_i \\
    =& \mathrm{i}w e^{\mathrm{i} (wW_i+x'X_i)} \frac{(1-D_i) m_p^2(X_i)}{\left( \frac{1}{2} f(X_i,0) \Delta_p(X_i) \right)^2} \tau(X_i) K_h(R_i) \delta_i \\
    =& [O_p(h_x^l) + O_p(h_r)] [O_p(h_x^l) + O_p(h_r)] = O_p(h_x^{2l}) + O_p(h_r^2); \\
    & \mathbb{E} \left( C_{42n,ijk} | \mathcal{F}_j \right) = \mathbb{E} \left( \mathbb{E} (C_{42n,ijk}|\mathcal{F}_i, \mathcal{F}_j) | \mathcal{F}_j \right) \\
    =& \mathbb{E} \bigg[ \mathrm{i}w e^{\mathrm{i} (wW_i+x'X_i)} (1-D_i) \tau(X_i) K_h(R_i) \delta_i \left( \frac{D_j K_{h_x}(X_j-X_i) K_{h_r}(R_j) \delta_j}{\frac{1}{2} f(X_i,0) \Delta_p(X_i)} - 1 \right) \\& \mathbb{E} \left( \frac{D_k K_{h_x}(X_k-X_i) K_{h_r}(R_k) \delta_k}{\frac{1}{2} f(X_i,0) \Delta_p(X_i)} -1 \bigg|X_i \right) \bigg| \mathcal{F}_j \bigg] \\
    =& \mathbb{E} \bigg[ \mathrm{i}w e^{\mathrm{i} (wW_i+x'X_i)} (1-D_i) \left( \frac{D_j K_{h_x}(X_j-X_i) K_{h_r}(R_j) \delta_j}{\frac{1}{2} f(X_i,0) \Delta_p(X_i)} - 1 \right) \frac{m_p(X_i) \tau(X_i) K_h(R_i) \delta_i}{\frac{1}{2} f(X_i,0) \Delta_p(X_i)} \bigg| \mathcal{F}_j \bigg] \\=& O_p(h_x^l) + O_p(h_r); \\
    & \mathbb{E} \left( C_{42n,ijk} | \mathcal{F}_k \right) = \mathbb{E} \left( \mathbb{E} (C_{42n,ijk}|\mathcal{F}_i, \mathcal{F}_k) | \mathcal{F}_k \right) \\
    =& \mathbb{E} \bigg[ \mathrm{i}w e^{\mathrm{i} (wW_i+x'X_i)} (1-D_i) \left( \frac{D_k K_{h_x}(X_k-X_i) K_{h_r}(R_k) \delta_k}{\frac{1}{2} f(X_i,0) \Delta_p(X_i)} -1 \right) \frac{m_p(X_i) K_h(R_i) \delta_i}{\frac{1}{2} f(X_i,0) \Delta_p(X_i)} \bigg| \mathcal{F}_k \bigg] \\=& O_p(h_x^l) + O_p(h_r).
\end{align*}
Then the total expectation is $\mathbb{E}[\mathbb{E}(C_{42n,ijk}|\mathcal{F}_i)] = O(h_x^{2l}) + O(h_r^2)$. The second-order moment of envelope $G_{C_{42n}}$ is bounded by

\begin{align*}
    \mathbb{E} G_{C_{42n}}^2 \le& \mathbb{E} \sup_{(w,x)\in\Omega} \left| \mathrm{i}w e^{\mathrm{i} (wW_i+x'X_i)} (1-D_i) \left( \frac{D_j K_{h_x}(X_j-X_i) K_{h_r}(R_j) \delta_j}{\frac{1}{2} f(X_i,0) \Delta_p(X_i)} - 1 \right) \right. \\& \left. \left( \frac{D_k K_{h_x}(X_k-X_i) K_{h_r}(R_k) \delta_k}{\frac{1}{2} f(X_i,0) \Delta_p(X_i)} -1 \right) \tau(X_i) K_h(R_i) \delta_i \right|^2 \\
    \lesssim& \mathbb{E} \left| \tau(X_i) K_{h_x}(X_j-X_i) K_{h_x}(X_k-X_i) K_{h_r}(R_j) K_{h_r}(R_k) K_h(R_i) \right|^2 \\
    &+ \mathbb{E} \left| \tau(X_i) \tau(X_i) K_h(R_i) \right|^2 + \mathbb{E} \left| \tau(X_i) K_{h_x}(X_k-X_i) K_{h_r}(R_k) K_h(R_i) \right|^2 \\
    =& O(h_x^{-2d} h_r^{-2} h^{-1}) + O(h_x^{-d} h_r^{-1} h^{-1}) + O(h^{-1}) = O(h_x^{-2d} h_r^{-2} h^{-1}).
\end{align*}
Following Proposition 4 in \cite{Delgado2001}, the second- and third-order terms in the $U$-process are uniformly bounded since
\begin{align*}
    \mathbb{E} \sup_{(w,x)\in\Omega} \left\Vert 3 \frac{n-2}{n-1} U_n^{(2)}(\pi_2 g_{C_{42n}}) \right\Vert^2 &\lesssim \frac{(n-2)^2 \mathbb{E} G_{C_{42n}}^2}{(n-1)^4} = O(\frac{1}{n^2h_x^{2d}h_r^2h}) = o(\frac{1}{nh}), \\
    \mathbb{E} \sup_{(w,x)\in\Omega} \left\Vert \frac{n-2}{n-1} U_n^{(3)}(\pi_3 g_{C_{42n}}) \right\Vert^2 &\lesssim \frac{n \mathbb{E} G_{C_{42n}}^2}{(n-1)^4} = O(\frac{1}{n^3h_x^{2d}h_r^2h}) = o(\frac{1}{nh}).
\end{align*}
Then it can be concluded that $\sup_{(w,x)\in\Omega} \left\Vert \frac{n-2}{n-1} C_{42n}(w,x) \right\Vert = o_p\left( (nh)^{-1/2} \right)$. Furthermore, $\sup_{(w,x)\in\Omega} \left\Vert C_{4n}(w,x) \right\Vert = o_p\left( (nh)^{-1/2} \right)$.

$C_{5n}(w,x)$--$C_{7n}(w,x)$ can be written as $$C_{sn}(w,x) = \frac{1}{n-1} C_{s1n}(w,x) + \frac{n-2}{n-1} C_{s2n}(w,x), \quad s=5,6,7, $$ where $C_{s1n}(w,x) = U_n^{(2)}(C_{s1n,ij})$ is a second-order $U$-process and $C_{s2n}(w,x) = U_n^{(3)}(C_{s2n,ijk})$ is a third-order $U$-process. 

The Hoeffding decompositions of the $U$-processes $C_{s1n}(w,x)$ and $C_{s2n}(w,x)$ are
\begin{align*}
    C_{s1n}(w,x) =& \mathbb{E} C_{s1n,ij} + 2U_n^{(1)}(\pi_1 g_{C_{s1n}}) + U_n^{(2)}(\pi_2 g_{C_{s1n}}), \\
    C_{s2n}(w,x) =& \frac{1}{n} \sum_{i=1}^n \mathbb{E} \left( C_{s2n,ijk} | \mathcal{F}_i \right) + \frac{1}{n} \sum_{j=1}^n \mathbb{E} \left( C_{s2n,ijk} | \mathcal{F}_j \right) + \frac{1}{n} \sum_{k=1}^n \mathbb{E} \left( C_{s2n,ijk} | \mathcal{F}_k \right) \\ &-2 \mathbb{E} C_{s2n,ijk} + 3 U_n^{(2)} (\pi_2 g_{C_{s2n}}) + U_n^{(3)} (\pi_3 g_{C_{s2n}}).    
\end{align*}
The conditional expectation terms in the Hoeffding decomposition of $C_{s2n}(w,x)$ are
\begin{align*}
    & \mathbb{E} \left( C_{s2n,ijk} | \mathcal{F}_i \right) = [O_p(h_x^l) + O_p(h_r)] [O_p(h_x^l) + O_p(h_r)] = O_p(h_x^{2l}) + O_p(h_r^2); \\
    & \mathbb{E} \left( C_{s2n,ijk} | \mathcal{F}_j \right) = \mathbb{E} \left( \mathbb{E} (C_{s2n,ijk}|\mathcal{F}_i, \mathcal{F}_j) | \mathcal{F}_j \right) = O_p(h_x^l) + O_p(h_r); \\
    & \mathbb{E} \left( C_{s2n,ijk} | \mathcal{F}_k \right) = \mathbb{E} \left( \mathbb{E} (C_{s2n,ijk}|\mathcal{F}_i, \mathcal{F}_k) | \mathcal{F}_k \right) = O_p(h_x^l) + O_p(h_r).
\end{align*}
Then $\mathbb{E} C_{s2n,ijk} = \mathbb{E}[\mathbb{E}(C_{s2n,ijk}|\mathcal{F}_i)] = O(h_x^{2l}) + O(h_r^2) = o(\frac{1}{\sqrt{nh}})$. Also, the other expectation $\mathbb{E} C_{s1n,ij} = O(\frac{1}{h_x^dh_r})$ and $\frac{1}{n-1}\mathbb{E} C_{s1n,ij} = O(\frac{1}{nh_x^dh_r}) = o(\frac{1}{\sqrt{nh}})$.
For the high-order terms in the Hoeffding decompositions, we investigate the second moments of their envelopes, $\mathbb{E}G_{C_{s1n}}^2$ and $\mathbb{E}G_{C_{s2n}}^2$, which are bounded by 
\begin{align*}
    \mathbb{E} G_{C_{s1n}}^2 \lesssim& O(h_x^{-3d} h_r^{-3} h^{-1}) + O(h_x^{-d} h_r^{-1} h^{-1}) + O(h^{-1}) = O(h_x^{-3d} h_r^{-3} h^{-1}), \\
    \mathbb{E} G_{C_{s2n}}^2 \lesssim& O(h_x^{-2d} h_r^{-2} h^{-1}) + O(h_x^{-d} h_r^{-1} h^{-1}) + O(h^{-1}) = O(h_x^{-2d} h_r^{-2} h^{-1}).
\end{align*}
Following Proposition 4 in \cite{Delgado2001}, the high-order terms in the Hoeffding decompositions are uniformly bounded since
\begin{align*}
    \mathbb{E} \sup_{(w,x)\in\Omega} \left| \frac{2}{n-1} U_n^{(1)}(\pi_1 g_{C_{s1n}}) \right|^2 &\lesssim \frac{\mathbb{E} G_{C_{s1n}}^2}{n(n-1)^2} = O(\frac{1}{n^3h_x^{3d}h_r^3h}) = o(\frac{1}{nh}), \\
    \mathbb{E} \sup_{(w,x)\in\Omega} \left| \frac{1}{n-1} U_n^{(2)}(\pi_2 g_{C_{s1n}}) \right|^2 &\lesssim \frac{\mathbb{E} G_{C_{s1n}}^2}{(n-1)^4} = O(\frac{1}{n^4h_x^{3d}h_r^3h}) = o(\frac{1}{nh}); \\
    \mathbb{E} \sup_{(w,x)\in\Omega} \left| 3 \frac{n-2}{n-1} U_n^{(2)}(\pi_2 g_{C_{s2n}}) \right|^2 &\lesssim \frac{(n-2)^2 \mathbb{E} G_{C_{s2n}}^2}{(n-1)^4} = O(\frac{1}{n^2h_x^{2d}h_r^2h}) = o(\frac{1}{nh}), \\
    \mathbb{E} \sup_{(w,x)\in\Omega} \left| \frac{n-2}{n-1} U_n^{(3)}(\pi_3 g_{C_{s2n}}) \right|^2 &\lesssim \frac{n \mathbb{E} G_{C_{s2n}}^2}{(n-1)^4} = O(\frac{1}{n^3h_x^{2d}h_r^2h}) = o(\frac{1}{nh}). 
\end{align*}
Then it can be concluded that $\frac{1}{n-1} C_{s1n}(w,x)$ and $\frac{n-2}{n-1} C_{s2n}(w,x)$ are uniformly bounded by $\frac{1}{\sqrt{nh}}$ over $(w,x)\in\Omega$. Furthermore, $\sup_{(w,x)\in\Omega} \left| C_{sn}(w,x) \right| = o_p(\frac{1}{\sqrt{nh}})$ for $s=5,6,7$.

Combining all the results, 
\begin{align*}
    \hat{U}_{n}(w,x) - U_{n}(w,x) =& C_{1n}(w,x) - C_{2n}(w,x) + o_p(\frac{1}{\sqrt{nh}}) \\
    =& \frac{1}{n} \sum_{j=1}^n \mathrm{i}w e^{\mathrm{i}x'X_j} \kappa(w,X_j) \left( Y_j - D_j\tau(X_j) \right) K_{h_r}(R_j)\delta_j + o_p(\frac{1}{\sqrt{nh}}).
\end{align*}

\subsubsection{The local linear estimator}
By the expansion of $\hat{a}/\hat{b}$, $\hat\tau(X) - \tau(X)$ can be decomposed as 
\begin{align*}
    \hat\tau(X_i) - \tau(X_i) =& \frac{\Delta_{\hat\mu}(X_i)}{\Delta_p(X_i)} - \tau(X_i) \frac{\Delta_{\hat{p}}(X_i)}{\Delta_p(X_i)} - \frac{\Delta_{\hat\mu}(X_i)-\Delta_\mu(X_i)}{\Delta_p(X_i)} \frac{\Delta_{\hat{p}}(X_i)-\Delta_p(X_i)}{\Delta_p(X_i)} \\&+ \tau(X_i) \left( \frac{\Delta_{\hat{p}}(X_i) - \Delta_p(X_i)}{\Delta_p(X_i)}\right)^2 + \text{high-order terms}.
\end{align*}
For $\Delta_{\hat\mu}(X_i) - \Delta_\mu(X_i)$ and $\Delta_{\hat{p}}(X_i)-\Delta_p(X_i)$, the estimation effect of the local linear estimator can be further expanded as
\begin{align*}
    & \Delta_{\hat{m}_Z}(X_i) - \Delta_{m_Z}(X_i) \\
    =& e_0' \left( B_{n+}^{-1}(X_i) A_{Zn+}(X_i) - B_{n-}^{-1}(X_i) A_{Zn-}(X_i) \right) - e_0' \left( B_+^{-1}(X_i) A_{Z+}(X_i) - B_-^{-1}(X_i) A_{Z-}(X_i) \right) \\
    &+ e_0' \left( B_+^{-1}(X_i) A_{Z+}(X_i) - B_-^{-1}(X_i) A_{Z-}(X_i) \right) - (m_Z(X_i, 0^+) - m_Z(X_i, 0^-)) \\
    =& e_0' \left( B_{n+}^{-1}(X_i) A_{Zn+}(X_i) - B_+^{-1}(X_i) A_{Z+}(X_i) \right) - e_0' \left( B_{n-}^{-1}(X_i) A_{Zn-}(X_i) - B_-^{-1}(X_i) A_{Z-}(X_i) \right) \\
    &+ (e_0' B_+^{-1}(X_i) A_{Z+}(X_i) - m_Z(X_i, 0^+)) - ( e_0' B_-^{-1}(X_i) A_{Z-}(X_i) - m_Z(X_i, 0^-))
\end{align*} 
for a generic variable $Z$ representing $Y$ or $D$. By Lemma S.A2, the last two terms in the last equation are of order $h_r^2$. The first two terms in the last equation can be expanded by the expression of $B_n^{-1}A_n$, where the leading terms are $$B^{-1} A_n - B^{-1} B_n B^{-1}A - B^{-1} (B_n-B) B^{-1} (A_n-A) + B^{-1} (B_n-B) B^{-1} (B_n-B) B^{-1}A.$$ 

Following the same decomposition as that for the local constant estimator, the main term of $\hat{U}_n(w,x)- U_n(w,x)$ can be denoted as $C_n(w,x) = C_{1n}(w,x) - C_{2n}(w,x) - C_{3n}(w,x) + C_{4n}(w,x) - C_{5n}(w,x) - C_{6n}(w,x) + C_{7n}(w,x)$, where the components are 
\begin{align*}
    C_{1n}(w,x) =& \frac{\mathrm{i}w}{n} \sum_{i=1}^n e^{\mathrm{i} (wW_i+x'X_i)} (1-D_i) \frac{\Delta_{\hat\mu}(X_i)}{\Delta_p(X_i)} K_h(R_i) \delta_i \\
    C_{2n}(w,x) =& \frac{\mathrm{i}w}{n} \sum_{i=1}^n e^{\mathrm{i} (wW_i+x'X_i)} (1-D_i) \tau(X_i) \frac{\Delta_{\hat{p}}(X_i)}{\Delta_p(X_i)} K_h(R_i) \delta_i \\
    C_{3n}(w,x) =& \frac{\mathrm{i}w}{n} \sum_{i=1}^n e^{\mathrm{i} (wW_i+x'X_i)} (1-D_i) \frac{\Delta_{\hat\mu}(X_i)-\Delta_\mu(X_i)}{\Delta_p(X_i)} \frac{\Delta_{\hat{p}}(X_i)-\Delta_p(X_i)}{\Delta_p(X_i)} K_h(R_i) \delta_i \\
    C_{4n}(w,x) =& \frac{\mathrm{i}w}{n} \sum_{i=1}^n e^{\mathrm{i} (wW_i+x'X_i)} (1-D_i) \tau(X_i) \left( \frac{\Delta_{\hat{p}}(X_i) - \Delta_p(X_i)}{\Delta_p(X_i)}\right)^2 K_h(R_i) \delta_i \\
    C_{5n}(w,x) =& \frac{w^2}{2n} \sum_{i=1}^n e^{\mathrm{i} (wW_i+x'X_i)} (1-D_i) \left( \frac{\Delta_{\hat\mu}(X_i)-\Delta_\mu(X_i)}{\Delta_p(X_i)} \right)^2 K_h(R_i) \delta_i \\
    C_{6n}(w,x) =& \frac{w^2}{2n} \sum_{i=1}^n e^{\mathrm{i} (wW_i+x'X_i)} (1-D_i) \tau^2(X_i) \left( \frac{\Delta_{\hat{p}}(X_i) - \Delta_p(X_i)}{\Delta_p(X_i)}\right)^2 K_h(R_i) \delta_i \\
    C_{7n}(w,x) =& \frac{w^2}{n} \sum_{i=1}^n e^{\mathrm{i} (wW_i+x'X_i)} (1-D_i) \frac{\Delta_{\hat\mu}(X_i)-\Delta_\mu(X_i)}{\Delta_p(X_i)} \frac{\Delta_{\hat{p}}(X_i)-\Delta_p(X_i)}{\Delta_p(X_i)} \tau(X_i) K_h(R_i) \delta_i.
\end{align*}

With Lemma S.A2, by expanding $\Delta_{\hat\mu}(X_i) - \Delta_\mu(X_i)$, $C_{1n}(w,x)$ can be expressed as $C_{1n}(w,x) = C_{1n+}(w,x) - C_{1n-}(w,x) + O_p(h_r^2)$, where
\begin{align*}
    C_{1n+}(w,x) =& \frac{\mathrm{i}w}{n} \sum_{i=1}^n e^{\mathrm{i} (wW_i+x'X_i)} \frac{1-D_i}{\Delta_p(X_i)} e_0' \left( B_{n+}^{-1}(X_i)A_{Yn+}(X_i) - B_+^{-1}(X_i)A_{Y+}(X_i) \right) K_h(R_i) \delta_i, \\    
    C_{1n-}(w,x) =& \frac{\mathrm{i}w}{n} \sum_{i=1}^n e^{\mathrm{i} (wW_i+x'X_i)} \frac{1-D_i}{\Delta_p(X_i)} e_0' \left( B_{n-}^{-1}(X_i)A_{Yn-}(X_i) - B_-^{-1}(X_i)A_{Y-}(X_i) \right) K_h(R_i) \delta_i.
\end{align*}
The leading terms of $C_{1n+}(w,x)$ can be written as $C_{11n+}(w,x) - C_{12n+}(w,x) - C_{13n+}(w,x) + C_{14n+}(w,x)$, where 
\begin{align*}
    C_{11n+}(w,x) =& \frac{\mathrm{i}w}{n} \sum_{i=1}^n e^{\mathrm{i} (wW_i+x'X_i)} \frac{1-D_i}{\Delta_p(X_i)} e_0' B_+^{-1}(X_i) A_{Yn+}(X_i) K_h(R_i) \delta_i, \\
    C_{12n+}(w,x) =& \frac{\mathrm{i}w}{n} \sum_{i=1}^n e^{\mathrm{i} (wW_i+x'X_i)} \frac{1-D_i}{\Delta_p(X_i)} e_0' B_+^{-1}(X_i) B_{n+}(X_i) B_+^{-1}(X_i) A_{Y+}(X_i) K_h(R_i) \delta_i, \\
    C_{13n+}(w,x) =& \frac{\mathrm{i}w}{n} \sum_{i=1}^n e^{\mathrm{i} (wW_i+x'X_i)} \frac{1-D_i}{\Delta_p(X_i)} e_0' B_+^{-1}(X_i) \left( B_{n+}(X_i) - B_+(X_i) \right) \\& B_+^{-1}(X_i) \left( A_{Yn+}(X_i) - A_{Y+}(X_i) \right) K_h(R_i) \delta_i, \\
    C_{14n+}(w,x) =& \frac{\mathrm{i}w}{n} \sum_{i=1}^n e^{\mathrm{i} (wW_i+x'X_i)} \frac{1-D_i}{\Delta_p(X_i)} e_0' B_+^{-1}(X_i) \left( B_{n+}(X_i) - B_+(X_i) \right) \\& B_+^{-1}(X_i) \left( B_{n+}(X_i) - B_+(X_i) \right) B_+^{-1}(X_i) A_{Y+}(X_i) K_h(R_i) \delta_i.
\end{align*}
where $C_{11n+}(w,x)$ and $C_{12n+}(w,x)$ are second-order $U$-processes, $C_{13n+}(w,x)$ and $C_{14n+}(w,x)$ are combinations of second-order and third-order $U$-processes derived from the second-order term in the expansion of $\hat{a}/\hat{b}$.
The Hoeffding decompositions of $U$-processes $C_{11n+}(w,x)$ and $C_{12n+}(w,x)$ are
\begin{align*}
   C_{11n+}(w,x) =& \frac{1}{n} \sum_{i=1}^n \mathbb{E}(C_{11n+,ij}|\mathcal{F}_i) + \frac{1}{n} \sum_{j=1}^n \mathbb{E}(C_{11n+,ij}|\mathcal{F}_j) - \mathbb{E}C_{11n+,ij} + U_n^{(2)}(\pi_2 g_{C_{11n+}}), \\
   C_{12n+}(w,x) =& \frac{1}{n} \sum_{i=1}^n \mathbb{E}(C_{12n+,ij}|\mathcal{F}_i) + \frac{1}{n} \sum_{j=1}^n \mathbb{E}(C_{12n+,ij}|\mathcal{F}_j) - \mathbb{E}C_{12n+,ij} + U_n^{(2)}(\pi_2 g_{C_{12n+}}), 
\end{align*}
where the elements are
\begin{align*}
   C_{11n+,ij} =& \mathrm{i}w e^{\mathrm{i} (wW_i+x'X_i)} \frac{1-D_i}{\Delta_p(X_i)} e_0' B_+^{-1}(X_i) \mathbf{r}(R_j)Y_j K_{h_x}(X_j-X_i) K_{h_r}(R_j)\mathbf{1}_+(R_j) K_h(R_i) \delta_i, \\
   C_{12n+,ij} =& \mathrm{i}w e^{\mathrm{i} (wW_i+x'X_i)} \frac{1-D_i}{\Delta_p(X_i)} e_0' B_+^{-1}(X_i) \mathbf{r}(R_j)\mathbf{r}(R_j)' B_+^{-1}(X_i) A_{Y+}(X_i) \\& K_{h_x}(X_j-X_i) K_{h_r}(R_j)\mathbf{1}_+(R_j) K_h(R_i) \delta_i.
\end{align*}
The conditional expectation $\mathbb{E}(C_{11n+,ij}|\mathcal{F}_i)$ is
$$ \mathbb{E}(C_{11n+,ij}|\mathcal{F}_i) = \mathrm{i}w e^{\mathrm{i} (wW_i+x'X_i)} \frac{1-D_i}{\Delta_p(X_i)} e_0' B_+^{-1}(X_i) K_h(R_i) \delta_i \mathbb{E}\left[ \mathbf{r}(R_j) Y_j K_{h_x}(X_j-X_i) K_{h_r}(R_j)\mathbf{1}_+(R_j) |X_i \right]. $$
By Lemma S.A2, $$ \mathbb{E}(C_{11n+,ij}|\mathcal{F}_i) = \mathrm{i}w e^{\mathrm{i} (wW_i+x'X_i)} \frac{1-D_i}{\Delta_p(X_i)} e_0' B_+^{-1}(X_i) A_{Y+}(X_i) K_h(R_i) \delta_i + O_p(h_x^l). $$
Also by Lemma S.A2, the conditional expectation term $\mathbb{E}(C_{12n+,ij}|\mathcal{F}_i)$ is 
\begin{align*}
    \mathbb{E}(C_{12n+,ij}|\mathcal{F}_i) =& \mathrm{i}w e^{\mathrm{i} (wW_i+x'X_i)} \frac{1-D_i}{\Delta_p(X_i)} e_0' B_+^{-1}(X_i) \mathbb{E}\left[ \mathbf{r}(R_j) \mathbf{r}(R_j)' K_{h_x}(X_j-X_i) K_{h_r}(R_j)\mathbf{1}_+(R_j) |X_i \right] \\& B_+^{-1}(X_i) A_{Y+}(X_i) K_h(R_i) \delta_i \\
    =& \mathrm{i}w e^{\mathrm{i} (wW_i+x'X_i)} \frac{1-D_i}{\Delta_p(X_i)} e_0' B_+^{-1}(X_i) B_+(X_i) B_+^{-1}(X_i) A_{Y+}(X_i) K_h(R_i) \delta_i + O_p(h_x^l) \\
    =& \mathrm{i}w e^{\mathrm{i} (wW_i+x'X_i)} \frac{1-D_i}{\Delta_p(X_i)} e_0' B_+^{-1}(X_i) A_{Y+}(X_i) K_h(R_i) \delta_i + O_p(h_x^l) .
\end{align*}
Note that $\mathbb{E}(C_{11n+,ij}|\mathcal{F}_i) - \mathbb{E}(C_{12n+,ij}|\mathcal{F}_i) = O_p(h_x^l)$ and $\mathbb{E}C_{11n+,ij} - \mathbb{E}C_{12n+,ij} = O(h_x^l)$. 
By Lemma S.A2, the conditional expectations $\mathbb{E}(C_{11n+,ij}|\mathcal{F}_j)$ and $\mathbb{E}(C_{12n+,ij}|\mathcal{F}_j)$ are
\begin{align*}
    & \mathbb{E}(C_{11n+,ij}|\mathcal{F}_j) \\
    =& \mathrm{i}w \mathbb{E}\left[ \frac{e^{\mathrm{i} x'X_i}}{\Delta_p(X_i)} e_0' B_+^{-1}(X_i) (1-D_i)e^{\mathrm{i}wW_i} K_{h_x}(X_j-X_i) K_h(R_i) \delta_i \bigg| X_j \right] \mathbf{r}(R_j) Y_j  K_{h_r}(R_j)\mathbf{1}_+(R_j) \\
    =& \mathrm{i}w e^{\mathrm{i}x'X_j} \kappa(w,X_j) \frac{f(X_j, 0^+)}{2} e_0' B_+^{-1}(X_j) \mathbf{r}(R_j) Y_j K_{h_r}(R_j)\mathbf{1}_+(R_j) + O_p(h) + O_p(h_x^l); \\
    & \mathbb{E}(C_{12n+,ij}|\mathcal{F}_j) \\
    =& \mathrm{i}w \mathbb{E} \bigg[ \frac{e^{\mathrm{i} x'X_i}}{\Delta_p(X_i)} e_0' B_+^{-1}(X_i) \mathbf{r}(R_j)\mathbf{r}(R_j)' B_+^{-1}(X_i) A_{Y+}(X_i) (1-D_i)e^{\mathrm{i}wW_i} K_{h_x}(X_j-X_i) K_h(R_i) \delta_i \bigg| X_j, R_j \bigg] \\& K_{h_r} (R_j) \mathbf{1}_+(R_j) \\
    =& \mathrm{i}w e^{\mathrm{i}x'X_j} \kappa(w,X_j) \frac{f(X_j, 0^+)}{2} e_0' B_+^{-1}(X_j) \mathbf{r}(R_j) \mathbf{r}(R_j)' B_+^{-1}(X_j) A_{Y+}(X_j) K_{h_r}(R_j)\mathbf{1}_+(R_j) + O_p(h) + O_p(h_x^l).
\end{align*}
For the high-order terms in the Hoeffding decompositions of $U$-processes $C_{11n+}(w,x)$ and $C_{12n+}(w,x)$, $U_n^{(2)}(\pi_2 g_{C_{11n+}})$ and $U_n^{(2)}(\pi_2 g_{C_{12n+}})$, the symmetrized functions $g_{C_{11n+}} := \frac{1}{2} (C_{11n+,ij} + C_{11n+,ji})$ and $g_{C_{12n+}} := \frac{1}{2} (C_{12n+,ij} + C_{12n+,ji})$ are of VC-type and have envelopes $G_{C_{11n+}} = \sup_{(w,x)\in\Omega}|g_{C_{11n+}}|$ and $G_{C_{12n+}} = \sup_{(w,x)\in\Omega}|g_{C_{12n+}}|$ with bounded second moments $G_{C_{11n+}}^2 \le \sup_{(w,x)\in\Omega} \left|C_{11n+,ij}\right|^2$ and $G_{C_{12n+}}^2 \le \sup_{(w,x)\in\Omega} \left|C_{12n+,ij}\right|^2$. Given Assumption B5:
\begin{align*}
    \mathbb{E} G_{C_{11n+}}^2 \le& \mathbb{E} \sup_{(w,x)\in\Omega} \bigg| \mathrm{i}w e^{\mathrm{i} (wW_i+x'X_i)} \frac{1-D_i}{\Delta_p(X_i)} e_0' B_+^{-1}(X_i) \mathbf{r}(R_j) Y_j K_{h_x}(X_j-X_i) K_{h_r}(R_j)\mathbf{1}_+(R_j) K_h(R_i) \delta_i \bigg|^2 \\
    \lesssim& \mathbb{E} \left[ e_0' B_+^{-1}(X_i) \mathbf{r}(R_j) Y_j K_{h_x}(X_j-X_i) K_{h_r}(R_j) K_h(R_i) \right]^2 = O(\frac{1}{hh_rh_x^d}). \\
    \mathbb{E} G_{C_{12n+}}^2 \le& \mathbb{E} \sup_{(w,x)\in\Omega} \bigg| \mathrm{i}w e^{\mathrm{i} (wW_i+x'X_i)} \frac{1-D_i}{\Delta_p(X_i)} e_0' B_+^{-1}(X_i) \mathbf{r}(R_j) \mathbf{r}(R_j)' \\& B_+^{-1}(X_i) A_{Y+}(X_i) K_{h_x}(X_j-X_i) K_{h_r}(R_j)\mathbf{1}_+(R_j) K_h(R_i) \delta_i \bigg|^2 \\
    \lesssim& \mathbb{E} \left[ e_0' B_+^{-1}(X_i) \mathbf{r}(R_j) \mathbf{r}(R_j)' B_+^{-1}(X_i) A_{Y+}(X_i) K_{h_x}(X_j-X_i) K_{h_r}(R_j) K_h(R_i) \right]^2 = O(\frac{1}{hh_rh_x^d}).
\end{align*}
Following Proposition 4 in \cite{Delgado2001}:
\begin{align*}
    \mathbb{E} \sup_{(w,x)\in\Omega} \left| U_n^{(2)}(\pi_2 g_{C_{11n+}}) \right|^2 &\lesssim \frac{\mathbb{E} G_{C_{11n+}}^2}{(n-1)^2} = O(\frac{1}{n^2hh_rh_x^d}) = o(\frac{1}{nh}). \\
    \mathbb{E} \sup_{(w,x)\in\Omega} \left| U_n^{(2)}(\pi_2 g_{C_{12n+}}) \right|^2 &\lesssim \frac{\mathbb{E} G_{C_{12n+}}^2}{(n-1)^2} = O(\frac{1}{n^2hh_rh_x^d}) = o(\frac{1}{nh}). 
\end{align*}
Thus, $U_n^{(2)}(\pi_2 g_{C_{11n+}})$ and $U_n^{(2)}(\pi_2 g_{C_{12n+}})$ are uniformly bounded by $\frac{1}{\sqrt{nh}}$. 

Summarizing all results of $C_{11n+}(w,x)$ and $C_{12n+}(w,x)$, it follows that
\begin{align*}    
    C_{11n+}(w,x) - C_{12n+}(w,x) =& \frac{\mathrm{i}w}{2n} \sum_{j=1}^n e^{\mathrm{i}x'X_j} \kappa(w,X_j) f(X_j, 0) e_0' B_+^{-1}(X_j) \mathbf{r}(R_j) \\& \left( Y_j - \mathbf{r}(R_j)' B_+^{-1}(X_j) A_{Y+}(X_j) \right) K_{h_r}(R_j)\mathbf{1}_+(R_j) + o_p(\frac{1}{\sqrt{nh}}). 
\end{align*}
Let $Y_j = \mathbb{E}(Y_j|X_j,R_j) + u_{Yj}$, where $u_{Yj}$ is the zero-mean nonparametric residual with finite variance. Then the term $Y_j - \mathbf{r}(R_j)' B_+^{-1}(X_j) A_{Y+}(X_j)$ becomes 
\begin{align*}
    Y_j - \mathbf{r}(R_j)' B_+^{-1}(X_j) A_{Y+}(X_j) 
    =& u_{Yj} + \mu(X_j, R_j) - \left( \mu(X_j, 0^+) + R_j \frac{\partial}{\partial r} \mu(X_j, 0^+) \right) + O_p(h_r^2) \\
    =& u_{Yj} + \frac{R_j^2}{2} \frac{\partial^2}{\partial r^2} \mu(X_j, 0^+) (1 + o_p(1)) = u_{Yj} + O_p(h_r^2).
\end{align*}
Then $C_{11n+}(w,x) - C_{12n+}(w,x)$ can be written as $$ C_{11n+}(w,x) - C_{12n+}(w,x) = \frac{\mathrm{i}w}{2n} \sum_{j=1}^n e^{\mathrm{i}x'X_j} \kappa(w,X_j) f(X_j, 0) e_0' B_+^{-1}(X_j) \mathbf{r}(R_j) u_{Yj} K_{h_r}(R_j)\mathbf{1}_+(R_j) + o_p(\frac{1}{\sqrt{nh}}). $$

For the high-order terms $C_{13n+}(w,x)$ and $C_{14n+}(w,x)$, they can be expressed as $C_{13n+}(w,x) = \frac{1}{n-1} C_{131n+}(w,x) + \frac{n-2}{n-1} C_{132n+}(w,x)$ and $C_{14n+}(w,x) = \frac{1}{n-1} C_{141n+}(w,x) + \frac{n-2}{n-1} C_{142n+}(w,x)$, where the components $C_{131n+}(w,x)$, $C_{132n+}(w,x)$, $C_{141n+}(w,x)$ and $C_{142n+}(w,x)$ are $U$-processes with Hoeffding decompositions
\begin{align*}
    C_{131n+}(w,x) =& U_n^{(2)}(g_{C_{131n+}}) = \mathbb{E} C_{131n+,ij} + 2U_n^{(1)}(\pi_1 g_{C_{131n+}}) + U_n^{(2)}(\pi_2 g_{C_{131n+}}), \\
    C_{132n+}(w,x) =& U_n^{(3)}(g_{C_{132n+}}) = \frac{1}{n} \sum_{i=1}^n \mathbb{E} \left( C_{132n+,ijk} | \mathcal{F}_i \right) + \frac{1}{n} \sum_{j=1}^n \mathbb{E} \left( C_{132n+,ijk} | \mathcal{F}_j \right) + \frac{1}{n} \sum_{k=1}^n \mathbb{E} \left( C_{132n+,ijk} | \mathcal{F}_k \right) \\ &-2 \mathbb{E} C_{132n+,ijk} + 3 U_n^{(2)} (\pi_2 g_{C_{132n+}}) + U_n^{(3)} (\pi_3 g_{C_{132n+}}); \\
    C_{141n+}(w,x) =& U_n^{(2)}(g_{C_{141n+}}) = \mathbb{E} C_{141n+,ij} + 2U_n^{(1)}(\pi_1 g_{C_{141n+}}) + U_n^{(2)}(\pi_2 g_{C_{141n+}}), \\
    C_{142n+}(w,x) =& U_n^{(3)}(g_{C_{142n+}}) = \frac{1}{n} \sum_{i=1}^n \mathbb{E} \left( C_{142n+,ijk} | \mathcal{F}_i \right) + \frac{1}{n} \sum_{j=1}^n \mathbb{E} \left( C_{142n+,ijk} | \mathcal{F}_j \right) + \frac{1}{n} \sum_{k=1}^n \mathbb{E} \left( C_{142n+,ijk} | \mathcal{F}_k \right) \\ &-2 \mathbb{E} C_{142n+,ijk} + 3 U_n^{(2)} (\pi_2 g_{C_{142n+}}) + U_n^{(3)} (\pi_3 g_{C_{142n+}}).
\end{align*}
The elements of these $U$-processes are
\begin{align*}
    C_{131n+,ij} =& \mathrm{i}w e^{\mathrm{i} (wW_i+x'X_i)} \frac{1-D_i}{\Delta_p(X_i)} e_0' B_+^{-1}(X_i) \left[ \mathbf{r}(R_j) \mathbf{r}(R_j)' K_{h_x}(X_j-X_i) K_{h_r}(R_j) \mathbf{1}_+(R_j) - B_+(X_i) \right] \\& B_+^{-1}(X_i) \left[ \mathbf{r}(R_j) Y_j K_{h_x}(X_j-X_i) K_{h_r}(R_j) \mathbf{1}_+(R_j) - A_{Y+}(X_i) \right] K_h(R_i) \delta_i, \\
    C_{132n+,ijk} =& \mathrm{i}w e^{\mathrm{i} (wW_i+x'X_i)} \frac{1-D_i}{\Delta_p(X_i)} e_0' B_+^{-1}(X_i) \left[ \mathbf{r}(R_j) \mathbf{r}(R_j)' K_{h_x}(X_j-X_i) K_{h_r}(R_j) \mathbf{1}_+(R_j) - B_+(X_i) \right] \\& B_+^{-1}(X_i) \left[ \mathbf{r}(R_k) Y_k K_{h_x}(X_k-X_i) K_{h_r}(R_k) \mathbf{1}_+(R_k) - A_{Y+}(X_i) \right] K_h(R_i) \delta_i, \\
    C_{141n+,ij} =& \mathrm{i}w e^{\mathrm{i} (wW_i+x'X_i)} \frac{1-D_i}{\Delta_p(X_i)} e_0' B_+^{-1}(X_i) \left[ \mathbf{r}(R_j) \mathbf{r}(R_j)' K_{h_x}(X_j-X_i) K_{h_r}(R_j) \mathbf{1}_+(R_j) - B_+(X_i) \right] \\& B_+^{-1}(X_i) \left[ \mathbf{r}(R_j) \mathbf{r}(R_j)' K_{h_x}(X_j-X_i) K_{h_r}(R_j) \mathbf{1}_+(R_j) - B_+(X_i) \right] B_+^{-1}(X_i) A_{Y+}(X_i) K_h(R_i) \delta_i, \\
    C_{142n+,ijk} =& \mathrm{i}w e^{\mathrm{i} (wW_i+x'X_i)} \frac{1-D_i}{\Delta_p(X_i)} e_0' B_+^{-1}(X_i) \left[ \mathbf{r}(R_j) \mathbf{r}(R_j)' K_{h_x}(X_j-X_i) K_{h_r}(R_j) \mathbf{1}_+(R_j) - B_+(X_i) \right] \\& B_+^{-1}(X_i) \left[ \mathbf{r}(R_k) \mathbf{r}(R_k)' K_{h_x}(X_k-X_i) K_{h_r}(R_k) \mathbf{1}_+(R_k) - B_+(X_i) \right] B_+^{-1}(X_i) A_{Y+}(X_i) K_h(R_i) \delta_i.
\end{align*}

For $C_{131n+}(w,x)$ and $C_{141n+}(w,x)$, the total expectations of the elements are $\mathbb{E} C_{131n+,ij} = O(\frac{1}{h_xh_r})$ and $\mathbb{E} C_{141n+,ij} = O(\frac{1}{h_xh_r})$. Then $\frac{1}{n-1} \mathbb{E} C_{131n+,ij} = O(\frac{1}{nh_xh_r}) = o(\frac{1}{\sqrt{nh}})$ and $\frac{1}{n-1} \mathbb{E} C_{141n+,ij} = O(\frac{1}{nh_xh_r}) = o(\frac{1}{\sqrt{nh}})$. The conditional expectation terms in the Hoeffding decompositions are
\begin{align*}
    \mathbb{E} \left( C_{132n+,ijk} | \mathcal{F}_i \right) =& \mathrm{i}w e^{\mathrm{i} (wW_i+x'X_i)} \frac{1-D_i}{\Delta_p(X_i)} e_0' B_+^{-1}(X_i) K_h(R_i) \delta_i \\
    & \mathbb{E} \left[ \mathbf{r}(R_j) \mathbf{r}(R_j)' K_{h_x}(X_j-X_i) K_{h_r}(R_j) \mathbf{1}_+(R_j) - B_+(X_i) |X_i \right] \\& B_+^{-1}(X_i) \mathbb{E} \left[ \mathbf{r}(R_k) Y_k K_{h_x}(X_k-X_i) K_{h_r}(R_k) \mathbf{1}_+(R_k) - A_{Y+}(X_i) |X_i \right] \\
    =& \mathrm{i}w e^{\mathrm{i} (wW_i+x'X_i)} \frac{1-D_i}{\Delta_p(X_i)} e_0' B_+^{-1}(X_i) K_h(R_i) \delta_i O_p(h_x^l) B_+^{-1}(X_i) O_p(h_x^l) = O_p(h_x^{2l}) ,\\
    \mathbb{E} \left( C_{132n+,ijk} | \mathcal{F}_j \right) =& \mathbb{E} [\mathbb{E} \left( C_{132n+,ijk} | \mathcal{F}_i, \mathcal{F}_j \right) | \mathcal{F}_j] = \mathrm{i}w e^{\mathrm{i} (wW_i+x'X_i)} \frac{1-D_i}{\Delta_p(X_i)} e_0' B_+^{-1}(X_i) K_h(R_i) \delta_i \\
    & \mathbb{E} \left[ \mathbf{r}(R_j) \mathbf{r}(R_j)' K_{h_x}(X_j-X_i) K_{h_r}(R_j) \mathbf{1}_+(R_j) - B_+(X_i) |X_i \right]  = O_p(h_x^l) ,\\
    \mathbb{E} \left( C_{132n+,ijk} | \mathcal{F}_k \right) =& \mathbb{E} [\mathbb{E} \left( C_{132n+,ijk} | \mathcal{F}_i, \mathcal{F}_k \right) | \mathcal{F}_k] = \mathrm{i}w e^{\mathrm{i} (wW_i+x'X_i)} \frac{1-D_i}{\Delta_p(X_i)} e_0' B_+^{-1}(X_i) K_h(R_i) \delta_i \\
    & O_p(h_x^l) B_+^{-1}(X_i) \mathbb{E} \left[ \mathbf{r}(R_k) Y_k K_{h_x}(X_k-X_i) K_{h_r}(R_k) \mathbf{1}_+(R_k) - A_{Y+}(X_i) |X_i \right] = O_p(h_x^l) ;\\
    \mathbb{E} \left( C_{142n+,ijk} | \mathcal{F}_i \right) =& \mathrm{i}w e^{\mathrm{i} (wW_i+x'X_i)} \frac{1-D_i}{\Delta_p(X_i)} e_0' B_+^{-1}(X_i) K_h(R_i) \delta_i \\
    & \mathbb{E} \left[ \mathbf{r}(R_j) \mathbf{r}(R_j)' K_{h_x}(X_j-X_i) K_{h_r}(R_j) \mathbf{1}_+(R_j) - B_+(X_i) |X_i \right] B_+^{-1}(X_i) \\
    & \mathbb{E} \left[ \mathbf{r}(R_k) \mathbf{r}(R_k)' K_{h_x}(X_k-X_i) K_{h_r}(R_k) \mathbf{1}_+(R_k) - B_+(X_i) |X_i \right] B_+^{-1}(X_i) A_{Y+}(X_i) \\
    =&  \mathrm{i}w e^{\mathrm{i} (wW_i+x'X_i)} \frac{1-D_i}{\Delta_p(X_i)} e_0' B_+^{-1}(X_i) O_p(h_x^l) B_+^{-1}(X_i) O_p(h_x^l) B_+^{-1}(X_i) A_{Y+}(X_i) K_h(R_i) \delta_i \\=& O_p(h_x^{2l}) ,\\
    \mathbb{E} \left( C_{142n+,ijk} | \mathcal{F}_j \right) =& \mathbb{E} [\mathbb{E} \left( C_{142n+,ijk} | \mathcal{F}_i, \mathcal{F}_j \right) | \mathcal{F}_j] = \mathrm{i}w e^{\mathrm{i} (wW_i+x'X_i)} \frac{1-D_i}{\Delta_p(X_i)} e_0' B_+^{-1}(X_i) K_h(R_i) \delta_i \\
    & \mathbb{E} \left[ \mathbf{r}(R_j) \mathbf{r}(R_j)' K_{h_x}(X_j-X_i) K_{h_r}(R_j) \mathbf{1}_+(R_j) - B_+(X_i) |X_i \right] B_+^{-1}(X_i) \\
    & O_p(h_x^l) B_+^{-1}(X_i) A_{Y+}(X_i) = O_p(h_x^l) ,\\
    \mathbb{E} \left( C_{142n+,ijk} | \mathcal{F}_k \right) =& \mathbb{E} [\mathbb{E} \left( C_{142n+,ijk} | \mathcal{F}_i, \mathcal{F}_k \right) | \mathcal{F}_k] = \mathrm{i}w e^{\mathrm{i} (wW_i+x'X_i)} \frac{1-D_i}{\Delta_p(X_i)} e_0' B_+^{-1}(X_i) K_h(R_i) \delta_i \\
    & O_p(h_x^l) B_+^{-1}(X_i) \mathbb{E} \left[ \mathbf{r}(R_k) \mathbf{r}(R_k)' K_{h_x}(X_k-X_i) K_{h_r}(R_k) \mathbf{1}_+(R_k) - B_+(X_i) |X_i \right] \\& B_+^{-1}(X_i) A_{Y+}(X_i) = O_p(h_x^l).
\end{align*}

The second-order moments of the envelopes $G_{C_{131n+}}$ and $G_{C_{141n+}}$ are bounded by 
\begin{align*}
    \mathbb{E} G_{C_{131n+}}^2 \le& \mathbb{E} \sup_{(w,x)\in\Omega} \bigg| \mathrm{i}w e^{\mathrm{i} (wW_i+x'X_i)} \frac{1-D_i}{\Delta_p(X_i)} e_0' B_+^{-1}(X_i) [ \mathbf{r}(R_j) \mathbf{r}(R_j)' K_{h_x}(X_j-X_i) K_{h_r}(R_j) \mathbf{1}_+(R_j) \\&- B_+(X_i)]  B_+^{-1}(X_i) \left[ \mathbf{r}(R_j) Y_j K_{h_x}(X_j-X_i) K_{h_r}(R_j) \mathbf{1}_+(R_j) - A_{Y+}(X_i) \right] K_h(R_i) \delta_i \bigg|^2 \\
    \lesssim& O(\frac{1}{h_x^{3d}h_r^3h}). \\
    \mathbb{E} G_{C_{141n+}}^2 \le& \mathbb{E} \sup_{(w,x)\in\Omega} \bigg| \mathrm{i}w e^{\mathrm{i} (wW_i+x'X_i)} \frac{1-D_i}{\Delta_p(X_i)} e_0' B_+^{-1}(X_i) [ \mathbf{r}(R_j) \mathbf{r}(R_j)' K_{h_x}(X_j-X_i) K_{h_r}(R_j) \mathbf{1}_+(R_j) \\&- B_+(X_i) ] B_+^{-1}(X_i) \left[ \mathbf{r}(R_j) \mathbf{r}(R_j)' K_{h_x}(X_j-X_i) K_{h_r}(R_j) \mathbf{1}_+(R_j) - B_+(X_i) \right] \\& B_+^{-1}(X_i) A_{Y+}(X_i) K_h(R_i) \delta_i \bigg|^2 \lesssim O(\frac{1}{h_x^{3d}h_r^3h}).
\end{align*}
Following Proposition 4 in \cite{Delgado2001}, the first- and second-order terms in the $U$-process are uniformly bounded since
\begin{align*}
    \mathbb{E} \sup_{(w,x)\in\Omega} \left\Vert \frac{2}{n-1} U_n^{(1)}(\pi_1 g_{C_{131n+}}) \right\Vert^2 &\lesssim \frac{\mathbb{E} G_{C_{131n+}}^2}{n(n-1)^2} = O(\frac{1}{n^3h_x^{3d}h_r^3h}) = o(\frac{1}{nh}). \\
    \mathbb{E} \sup_{(w,x)\in\Omega} \left\Vert \frac{1}{n-1} U_n^{(2)}(\pi_2 g_{C_{131n+}}) \right\Vert^2 &\lesssim \frac{\mathbb{E} G_{C_{131n+}}^2}{(n-1)^4} = O(\frac{1}{n^4h_x^{3d}h_r^3h}) = o(\frac{1}{nh}); \\
    \mathbb{E} \sup_{(w,x)\in\Omega} \left\Vert \frac{2}{n-1} U_n^{(1)}(\pi_1 g_{C_{141n+}}) \right\Vert^2 &\lesssim \frac{\mathbb{E} G_{C_{141n+}}^2}{n(n-1)^2} = O(\frac{1}{n^3h_x^{3d}h_r^3h}) = o(\frac{1}{nh}). \\
    \mathbb{E} \sup_{(w,x)\in\Omega} \left\Vert \frac{1}{n-1} U_n^{(2)}(\pi_2 g_{C_{141n+}}) \right\Vert^2 &\lesssim \frac{\mathbb{E} G_{C_{141n+}}^2}{(n-1)^4} = O(\frac{1}{n^4h_x^{3d}h_r^3h}) = o(\frac{1}{nh}).
\end{align*}

The second-order moments of envelopes $G_{C_{132n+}}$ and $G_{C_{142n+}}$ are bounded by
\begin{align*}
    \mathbb{E} G_{C_{132n+}}^2 \le& \mathbb{E} \sup_{(w,x)\in\Omega} \bigg| \mathrm{i}w e^{\mathrm{i} (wW_i+x'X_i)} \frac{1-D_i}{\Delta_p(X_i)} e_0' B_+^{-1}(X_i) [ \mathbf{r}(R_j) \mathbf{r}(R_j)' K_{h_x}(X_j-X_i) K_{h_r}(R_j) \mathbf{1}_+(R_j) \\&- B_+(X_i)]  B_+^{-1}(X_i) \left[ \mathbf{r}(R_k) Y_k K_{h_x}(X_k-X_i) K_{h_r}(R_k) \mathbf{1}_+(R_k) - A_{Y+}(X_i) \right] K_h(R_i) \delta_i \bigg|^2 \\
    \lesssim& O(\frac{1}{h_x^{2d}h_r^2h}). \\
    \mathbb{E} G_{C_{142n+}}^2 \le& \mathbb{E} \sup_{(w,x)\in\Omega} \bigg| \mathrm{i}w e^{\mathrm{i} (wW_i+x'X_i)} \frac{1-D_i}{\Delta_p(X_i)} e_0' B_+^{-1}(X_i) [ \mathbf{r}(R_j) \mathbf{r}(R_j)' K_{h_x}(X_j-X_i) K_{h_r}(R_j) \mathbf{1}_+(R_j) \\&- B_+(X_i) ] B_+^{-1}(X_i) \left[ \mathbf{r}(R_k) \mathbf{r}(R_k)' K_{h_x}(X_k-X_i) K_{h_r}(R_k) \mathbf{1}_+(R_k) - B_+(X_i) \right] \\& B_+^{-1}(X_i) A_{Y+}(X_i) K_h(R_i) \delta_i \bigg|^2 \lesssim O(\frac{1}{h_x^{2d}h_r^2h}).
\end{align*}
Following Proposition 4 in \cite{Delgado2001}, the second- and third-order terms in the $U$-process are uniformly bounded since
\begin{align*}
    \mathbb{E} \sup_{(w,x)\in\Omega} \left\Vert 3 \frac{n-2}{n-1} U_n^{(2)}(\pi_2 g_{C_{132n+}}) \right\Vert^2 &\lesssim \frac{(n-2)^2 \mathbb{E} G_{C_{132n+}}^2}{(n-1)^4} = O(\frac{1}{n^2h_x^{2d}h_r^2h}) = o(\frac{1}{nh}), \\
    \mathbb{E} \sup_{(w,x)\in\Omega} \left\Vert \frac{n-2}{n-1} U_n^{(3)}(\pi_3 g_{C_{132n+}}) \right\Vert^2 &\lesssim \frac{n \mathbb{E} G_{C_{132n+}}^2}{(n-1)^4} = O(\frac{1}{n^3h_x^{2d}h_r^2h}) = o(\frac{1}{nh}); \\
    \mathbb{E} \sup_{(w,x)\in\Omega} \left\Vert 3 \frac{n-2}{n-1} U_n^{(2)}(\pi_2 g_{C_{142n+}}) \right\Vert^2 &\lesssim \frac{(n-2)^2 \mathbb{E} G_{C_{142n+}}^2}{(n-1)^4} = O(\frac{1}{n^2h_x^{2d}h_r^2h}) = o(\frac{1}{nh}), \\
    \mathbb{E} \sup_{(w,x)\in\Omega} \left\Vert \frac{n-2}{n-1} U_n^{(3)}(\pi_3 g_{C_{142n+}}) \right\Vert^2 &\lesssim \frac{n \mathbb{E} G_{C_{142n+}}^2}{(n-1)^4} = O(\frac{1}{n^3h_x^{2d}h_r^2h}) = o(\frac{1}{nh}).
\end{align*}
Then it can be concluded that $C_{13n+}(w,x)$ and $C_{14n+}(w,x)$ are uniformly bounded by $\frac{1}{\sqrt{nh}}$. 
Summarizing all the results about $C_{11n+}(w,x)$--$C_{14n+}(w,x)$, $C_{1n+}(w,x)$ can be written as $$ C_{1n+}(w,x) = \frac{\mathrm{i}w}{2n} \sum_{j=1}^n e^{\mathrm{i}x'X_j} \kappa(w,X_j) f(X_j, 0) e_0' B_+^{-1}(X_j) \mathbf{r}(R_j) u_{Yj} K_{h_r}(R_j)\mathbf{1}_+(R_j) + o_p(\frac{1}{\sqrt{nh}}). $$ Its counterpart on the negative side, $C_{1n-}(w,x)$ can be shown to equal $$ C_{1n-}(w,x) = \frac{\mathrm{i}w}{2n} \sum_{j=1}^n e^{\mathrm{i}x'X_j} \kappa(w,X_j) f(X_j, 0) e_0' B_-^{-1}(X_j) \mathbf{r}(R_j) u_{Yj} K_{h_r}(R_j)\mathbf{1}_-(R_j) + o_p(\frac{1}{\sqrt{nh}}) $$ by the same arguments as for $C_{1n+}(w,x)$ by replacing the positive signal with the negative signal in matrix functional $B$, vector functional $A$, and scalar indicator function $\mathbf{1}_\pm(R)$. Then 
\begin{align*}
    C_{1n}(w,x) =& \frac{\mathrm{i}w}{2n} \sum_{j=1}^n e^{\mathrm{i}x'X_j} \kappa(w,X_j) f(X_j, 0) u_{Yj} K_{h_r}(R_j) \\& e_0' \left( B_+^{-1}(X_j) \mathbf{1}_+(R_j) - B_-^{-1}(X_j) \mathbf{1}_-(R_j) \right) \mathbf{r}(R_j) + o_p(\frac{1}{\sqrt{nh}}).
\end{align*} 

Using the same decomposition as for $C_{1n}(w,x)$, and replacing $Y_j$ with $D_j\tau(X_j)$ in the linear representation, it can be shown that 
\begin{align*}
    C_{2n}(w,x) =& \frac{\mathrm{i}w}{2n} \sum_{j=1}^n e^{\mathrm{i}x'X_j} \kappa(w,X_j) f(X_j, 0) \tau(X_j) u_{Dj} K_{h_r}(R_j) \\& e_0' \left( B_+^{-1}(X_j) \mathbf{1}_+(R_j) - B_-^{-1}(X_j) \mathbf{1}_-(R_j) \right) \mathbf{r}(R_j) + o_p(\frac{1}{\sqrt{nh}}), 
\end{align*} 
where $u_{Dj} = D_j - p(X_j,R_j)$ is the conditional residual of $D_j$.

For $C_{3n}(w,x)$, multiplying the decompositions of $(\Delta_{\hat\mu}(X_i)-\Delta_\mu(X_i))$ and $(\Delta_{\hat{p}}(X_i)-\Delta_p(X_i))$ yields
\begin{align*}
    & C_{3n}(w,x) \\
    =& \frac{\mathrm{i}w}{n} \sum_{i=1}^n e^{\mathrm{i}(wW_i + x'X_i)} \frac{(1-D_i)}{\Delta_p^2(X_i)} \bigg[ e_0' \left( B_{n+}^{-1}(X_i) A_{Yn+}(X_i) - B_{n-}^{-1}(X_i) A_{Yn-}(X_i) \right) - \Delta_\mu(X_i) \bigg] \\ 
    & \bigg[ e_0' \left( B_{n+}^{-1}(X_i) A_{Dn+}(X_i) - B_{n-}^{-1}(X_i) A_{Dn-}(X_i) \right) - \Delta_p(X_i) \bigg] K_h(R_i) \delta_i \\
    =& \frac{\mathrm{i}w}{n} \sum_{i=1}^n e^{\mathrm{i}(wW_i + x'X_i)} \frac{(1-D_i)}{\Delta_p^2(X_i)} K_h(R_i) \delta_i \\
    & \bigg[ e_0' \left( B_{n+}^{-1}(X_i) A_{Yn+}(X_i) - B_{n-}^{-1}(X_i) A_{Yn-}(X_i) \right) - e_0' \left( B_+^{-1}(X_i) A_{Y+}(X_i) - B_-^{-1}(X_i) A_{Y-}(X_i) \right) \\
    &+ e_0' \left( B_+^{-1}(X_i) A_{Y+}(X_i) - B_-^{-1}(X_i) A_{Y-}(X_i) \right) - \left( \mu(X_i,0^+) - \mu(X_i,0^-) \right)\bigg] \\
    & \bigg[ e_0' \left( B_{n+}^{-1}(X_i) A_{Dn+}(X_i) - B_{n-}^{-1}(X_i) A_{Dn-}(X_i) \right) - e_0' \left( B_+^{-1}(X_i) A_{D+}(X_i) - B_-^{-1}(X_i) A_{D-}(X_i) \right) \\
    &+ e_0' \left( B_+^{-1}(X_i) A_{D+}(X_i) - B_-^{-1}(X_i) A_{D-}(X_i) \right) - \left( p(X_i,0^+) - p(X_i,0^-) \right)\bigg] \\
    =& \frac{\mathrm{i}w}{n} \sum_{i=1}^n e^{\mathrm{i}(wW_i + x'X_i)} \frac{(1-D_i)}{\Delta_p^2(X_i)} K_h(R_i) \delta_i \\
    & e_0' \bigg[ \left( B_{n+}^{-1}(X_i) A_{Yn+}(X_i) - B_{n-}^{-1}(X_i) A_{Yn-}(X_i) \right) - \left( B_+^{-1}(X_i) A_{Y+}(X_i) - B_-^{-1}(X_i) A_{Y-}(X_i) \right) \bigg] \\
    & e_0' \bigg[ \left( B_{n+}^{-1}(X_i) A_{Dn+}(X_i) - B_{n-}^{-1}(X_i) A_{Dn-}(X_i) \right) - \left( B_+^{-1}(X_i) A_{D+}(X_i) - B_-^{-1}(X_i) A_{D-}(X_i) \right) \bigg] + O_p(h_r^2).
\end{align*}
The last equality comes from the property of the local linear estimator that the approximation error for the conditional mean function $m_Z(X_i,0^\pm)$ is $O_p(h_r^2)$, which is uniformly bounded by $\frac{1}{\sqrt{nh}}$. Applying the expansion of $B_n^{-1}A_n$, $C_{3n}(w,x)$ can be further expressed as
\begin{align*}
    & C_{3n}(w,x) = \frac{\mathrm{i}w}{n} \sum_{i=1}^n e^{\mathrm{i}(wW_i + x'X_i)} \frac{(1-D_i)}{\Delta_p^2(X_i)} K_h(R_i) \delta_i \\
    & e_0' \bigg[ B_+^{-1}(X_i) \left( A_{Yn+}(X_i) -  A_{Y+}(X_i) \right) - B_+^{-1}(X_i) \left( B_{n+}(X_i) -B_+(X_i) \right) B_+^{-1}(X_i) A_{Y+}(X_i) \\
    &- B_-^{-1}(X_i) \left( A_{Yn-}(X_i) -  A_{Y-}(X_i) \right) + B_-^{-1}(X_i) \left( B_{n-}(X_i) -B_-(X_i) \right) B_-^{-1}(X_i) A_{Y-}(X_i) \bigg] \\
    & e_0' \bigg[ B_+^{-1}(X_i) \left( A_{Dn+}(X_i) -  A_{D+}(X_i) \right) - B_+^{-1}(X_i) \left( B_{n+}(X_i) -B_+(X_i) \right) B_+^{-1}(X_i) A_{D+}(X_i) \\
    &- B_-^{-1}(X_i) \left( A_{Dn-}(X_i) -  A_{D-}(X_i) \right) + B_-^{-1}(X_i) \left( B_{n-}(X_i) -B_-(X_i) \right) B_-^{-1}(X_i) A_{D-}(X_i) \bigg] + o_p(\frac{1}{\sqrt{nh}}).
\end{align*}
Combining these terms in the multiplication, $C_{3n}(w,x)$ can be written as 
\begin{align*}
    C_{3n}(w,x) =& C_{301n}(w,x) - C_{302n}(w,x) - C_{303n}(w,x) + C_{304n}(w,x) \\
    &- C_{305n}(w,x) + C_{306n}(w,x) + C_{307n}(w,x) - C_{308n}(w,x) \\
    &- C_{309n}(w,x) + C_{310n}(w,x) + C_{311n}(w,x) - C_{312n}(w,x) \\
    &+ C_{313n}(w,x) - C_{314n}(w,x) - C_{315n}(w,x) + C_{316n}(w,x) + o_p(\frac{1}{\sqrt{nh}}),
\end{align*}
where the first component is 
\begin{align*}
    & C_{301n}(w,x) \\
    =& \frac{\mathrm{i}w}{n} \sum_{i=1}^n e^{\mathrm{i}(wW_i + x'X_i)} \frac{(1-D_i)}{\Delta_p^2(X_i)} K_h(R_i) \delta_i e_0' B_+^{-1}(X_i) \left( A_{Yn+}(X_i) -  A_{Y+}(X_i) \right) e_0' B_+^{-1}(X_i) \\& \left( A_{Dn+}(X_i) -  A_{D+}(X_i) \right) \\
    =& \frac{\mathrm{i}w}{n(n-1)^2} \sum_{i=1}^n \sum_{j\ne i} \sum_{k\ne i} e^{\mathrm{i}(wW_i + x'X_i)} \frac{(1-D_i)}{\Delta_p^2(X_i)} e_0' B_+^{-1}(X_i) \\& \left[ \mathbf{r}(R_j) Y_j K_{h_x}(X_j-X_i) K_{h_r}(R_j) \mathbf{1}_+(R_j) - A_{Y+}(X_i) \right]  e_0' B_+^{-1}(X_i) \\& \left[ \mathbf{r}(R_k) D_k K_{h_x}(X_k-X_i) K_{h_r}(R_k) \mathbf{1}_+(R_k) - A_{D+}(X_i) \right] K_h(R_i) \delta_i, 
\end{align*}
which can be written as $$C_{301n}(w,x) = \frac{1}{n-1} C_{3011n}(w,x) + \frac{n-2}{n-1} C_{3012n}(w,x), $$ where $C_{3011n}(w,x) = U_n^{(2)}(C_{3011n,ij})$ is a second-order $U$-process and $C_{3012n}(w,x) = U_n^{(3)}(C_{3012n,ijk})$ is a third-order $U$-process with elements 
\begin{align*}
    C_{3011n,ij} =& \mathrm{i}w e^{\mathrm{i}(wW_i + x'X_i)} \frac{(1-D_i)}{\Delta_p^2(X_i)} e_0' B_+^{-1}(X_i) \left[ \mathbf{r}(R_j) Y_j K_{h_x}(X_j-X_i) K_{h_r}(R_j) \mathbf{1}_+(R_j) - A_{Y+}(X_i) \right] \\& e_0' B_+^{-1}(X_i) \left[ \mathbf{r}(R_j) D_j K_{h_x}(X_j-X_i) K_{h_r}(R_j) \mathbf{1}_+(R_j) - A_{D+}(X_i) \right] K_h(R_i) \delta_i, \\
    C_{3012n,ijk} =& \mathrm{i}w e^{\mathrm{i}(wW_i + x'X_i)} \frac{(1-D_i)}{\Delta_p^2(X_i)} e_0' B_+^{-1}(X_i) \left[ \mathbf{r}(R_j) Y_j K_{h_x}(X_j-X_i) K_{h_r}(R_j) \mathbf{1}_+(R_j) - A_{Y+}(X_i) \right] \\& e_0' B_+^{-1}(X_i) \left[ \mathbf{r}(R_k) D_k K_{h_x}(X_k-X_i) K_{h_r}(R_k) \mathbf{1}_+(R_k) - A_{D+}(X_i) \right] K_h(R_i) \delta_i, 
\end{align*}
The Hoeffding decompositions of the $U$-processes $C_{3011n}(w,x)$ and $C_{3012n}(w,x)$ are
\begin{align*}
    C_{3011n}(w,x) =& \mathbb{E} C_{3011n,ij} + 2U_n^{(1)}(\pi_1 g_{C_{3011n}}) + U_n^{(2)}(\pi_2 g_{C_{3011n}}), \\
    C_{3012n}(w,x) =& \frac{1}{n} \sum_{i=1}^n \mathbb{E} \left( C_{3012n,ijk} | \mathcal{F}_i \right) + \frac{1}{n} \sum_{j=1}^n \mathbb{E} \left( C_{3012n,ijk} | \mathcal{F}_j \right) + \frac{1}{n} \sum_{k=1}^n \mathbb{E} \left( C_{3012n,ijk} | \mathcal{F}_k \right) \\ &-2 \mathbb{E} C_{3012n,ijk} + 3 U_n^{(2)} (\pi_2 g_{C_{3012n}}) + U_n^{(3)} (\pi_3 g_{C_{3012n}}).    
\end{align*}
The expectation in the Hoeffding decomposition of $C_{3011n}(w,x)$ is $\mathbb{E} C_{3011n,ij} = O(\frac{1}{h_x^dh_r})$ and $\frac{1}{n-1} \mathbb{E} C_{3011n,ij} = O(\frac{1}{nh_x^dh_r}) = o(\frac{1}{\sqrt{nh}})$. 
The conditional expectation terms in the Hoeffding decomposition of $C_{3012n}(w,x)$ are
\begin{align*}
    & \mathbb{E} \left( C_{3012n,ijk} | \mathcal{F}_i \right) = O_p(h_x^{2l}); \\
    & \mathbb{E} \left( C_{3012n,ijk} | \mathcal{F}_j \right) = \mathbb{E} \left( \mathbb{E} (C_{3012n,ijk}|\mathcal{F}_i, \mathcal{F}_j) | \mathcal{F}_j \right) = O_p(h_x^l); \\
    & \mathbb{E} \left( C_{3012n,ijk} | \mathcal{F}_k \right) = \mathbb{E} \left( \mathbb{E} (C_{3012n,ijk}|\mathcal{F}_i, \mathcal{F}_k) | \mathcal{F}_k \right) = O_p(h_x^l).
\end{align*}
Then $\mathbb{E} C_{3012n,ijk} = \mathbb{E}[\mathbb{E}(C_{3012n,ijk}|\mathcal{F}_i)] = O(h_x^{2l}) + O(h_r^2) = o(\frac{1}{\sqrt{nh}})$. 
For the high-order terms in the Hoeffding decompositions, we investigate the second moments of their envelopes, $\mathbb{E}G_{C_{3011n}}^2 \lesssim O(h_x^{-3d} h_r^{-3} h^{-1})$ and $\mathbb{E}G_{C_{3012n}}^2 \lesssim O(h_x^{-2d} h_r^{-2} h^{-1})$.
Following Proposition 4 in \cite{Delgado2001}, the high-order terms in the Hoeffding decompositions are uniformly bounded since
\begin{align*}
    \mathbb{E} \sup_{(w,x)\in\Omega} \left| \frac{2}{n-1} U_n^{(1)}(\pi_1 g_{C_{3011n}}) \right|^2 &\lesssim \frac{\mathbb{E} G_{C_{3011n}}^2}{n(n-1)^2} = O(\frac{1}{n^3h_x^{3d}h_r^3h}) = o(\frac{1}{nh}), \\
    \mathbb{E} \sup_{(w,x)\in\Omega} \left| \frac{1}{n-1} U_n^{(2)}(\pi_2 g_{C_{3011n}}) \right|^2 &\lesssim \frac{\mathbb{E} G_{C_{3011n}}^2}{(n-1)^4} = O(\frac{1}{n^4h_x^{3d}h_r^3h}) = o(\frac{1}{nh}); \\
    \mathbb{E} \sup_{(w,x)\in\Omega} \left| 3 \frac{n-2}{n-1} U_n^{(2)}(\pi_2 g_{C_{3012n}}) \right|^2 &\lesssim \frac{(n-2)^2 \mathbb{E} G_{C_{3012n}}^2}{(n-1)^4} = O(\frac{1}{n^2h_x^{2d}h_r^2h}) = o(\frac{1}{nh}), \\
    \mathbb{E} \sup_{(w,x)\in\Omega} \left| \frac{n-2}{n-1} U_n^{(3)}(\pi_3 g_{C_{3012n}}) \right|^2 &\lesssim \frac{n \mathbb{E} G_{C_{3012n}}^2}{(n-1)^4} = O(\frac{1}{n^3h_x^{2d}h_r^2h}) = o(\frac{1}{nh}). 
\end{align*}
Then it can be concluded that $\frac{1}{n-1} C_{3011n}(w,x)$ and $\frac{n-2}{n-1} C_{3012n}(w,x)$ are uniformly bounded by $\frac{1}{\sqrt{nh}}$ over $(w,x)\in\Omega$. Furthermore, $\sup_{(w,x)\in\Omega}$ $|C_{301n}(w,x)| = o_p(\frac{1}{\sqrt{nh}})$. 

Other components can be similarly written as a combination of second-order and third-order $U$-processes $$C_{sn}(w,x) = \frac{1}{n-1} C_{s1n}(w,x) + \frac{n-2}{n-1} C_{s2n}(w,x),$$ whose Hoeffding decompositions are
\begin{align*}
    C_{s1n}(w,x) =& \mathbb{E} C_{s1n,ij} + 2U_n^{(1)}(\pi_1 g_{C_{s1n}}) + U_n^{(2)}(\pi_2 g_{C_{s1n}}), \\
    C_{s2n}(w,x) =& \frac{1}{n} \sum_{i=1}^n \mathbb{E} \left( C_{s2n,ijk} | \mathcal{F}_i \right) + \frac{1}{n} \sum_{j=1}^n \mathbb{E} \left( C_{s2n,ijk} | \mathcal{F}_j \right) + \frac{1}{n} \sum_{k=1}^n \mathbb{E} \left( C_{s2n,ijk} | \mathcal{F}_k \right) \\ &-2 \mathbb{E} C_{s2n,ijk} + 3 U_n^{(2)} (\pi_2 g_{C_{s2n}}) + U_n^{(3)} (\pi_3 g_{C_{s2n}}).
\end{align*}
for $s = 302, \cdots, 316$.
The expectation in the Hoeffding decomposition of $C_{s1n}(w,x)$ is $\mathbb{E} C_{s1n,ij} = O(\frac{1}{h_x^dh_r})$ and $\frac{1}{n-1} \mathbb{E} C_{s1n,ij} = O(\frac{1}{nh_x^dh_r}) = o(\frac{1}{\sqrt{nh}})$. The conditional expectation terms in the Hoeffding decomposition of $C_{s2n}(w,x)$ are
\begin{align*}
    & \mathbb{E} \left( C_{s2n,ijk} | \mathcal{F}_i \right) = O_p(h_x^{2l}); \\
    & \mathbb{E} \left( C_{s2n,ijk} | \mathcal{F}_j \right) = \mathbb{E} \left( \mathbb{E} (C_{s2n,ijk}|\mathcal{F}_i, \mathcal{F}_j) | \mathcal{F}_j \right) = O_p(h_x^l); \\
    & \mathbb{E} \left( C_{s2n,ijk} | \mathcal{F}_k \right) = \mathbb{E} \left( \mathbb{E} (C_{s2n,ijk}|\mathcal{F}_i, \mathcal{F}_k) | \mathcal{F}_k \right) = O_p(h_x^l).
\end{align*}
Then $\mathbb{E} C_{s2n,ijk} = \mathbb{E}[\mathbb{E}(C_{s2n,ijk}|\mathcal{F}_i)] = O(h_x^{2l}) + O(h_r^2) = o(\frac{1}{\sqrt{nh}})$. 
For the high-order terms in the Hoeffding decompositions, we investigate the second moments of their envelopes, $\mathbb{E}G_{C_{s1n}}^2 \lesssim O(h_x^{-3d} h_r^{-3} h^{-1})$ and $\mathbb{E}G_{C_{s2n}}^2 \lesssim O(h_x^{-2d} h_r^{-2} h^{-1})$.
Following Proposition 4 in \cite{Delgado2001}, the high-order terms in the Hoeffding decompositions are uniformly bounded since
\begin{align*}
    \mathbb{E} \sup_{(w,x)\in\Omega} \left| \frac{2}{n-1} U_n^{(1)}(\pi_1 g_{C_{s1n}}) \right|^2 &\lesssim \frac{\mathbb{E} G_{C_{s1n}}^2}{n(n-1)^2} = O(\frac{1}{n^3h_x^{3d}h_r^3h}) = o(\frac{1}{nh}), \\
    \mathbb{E} \sup_{(w,x)\in\Omega} \left| \frac{1}{n-1} U_n^{(2)}(\pi_2 g_{C_{s1n}}) \right|^2 &\lesssim \frac{\mathbb{E} G_{C_{s1n}}^2}{(n-1)^4} = O(\frac{1}{n^4h_x^{3d}h_r^3h}) = o(\frac{1}{nh}); \\
    \mathbb{E} \sup_{(w,x)\in\Omega} \left| 3 \frac{n-2}{n-1} U_n^{(2)}(\pi_2 g_{C_{s2n}}) \right|^2 &\lesssim \frac{(n-2)^2 \mathbb{E} G_{C_{s2n}}^2}{(n-1)^4} = O(\frac{1}{n^2h_x^{2d}h_r^2h}) = o(\frac{1}{nh}), \\
    \mathbb{E} \sup_{(w,x)\in\Omega} \left| \frac{n-2}{n-1} U_n^{(3)}(\pi_3 g_{C_{s2n}}) \right|^2 &\lesssim \frac{n \mathbb{E} G_{C_{s2n}}^2}{(n-1)^4} = O(\frac{1}{n^3h_x^{2d}h_r^2h}) = o(\frac{1}{nh}). 
\end{align*}
Then it can be concluded that $\frac{1}{n-1} C_{s1n}(w,x)$ and $\frac{n-2}{n-1} C_{s2n}(w,x)$ are uniformly bounded by $\frac{1}{\sqrt{nh}}$ over $(w,x)\in\Omega$ for any subscript $s$. Furthermore, $\sup_{(w,x)\in\Omega}$ $|C_{sn}(w,x)| = o_p(\frac{1}{\sqrt{nh}})$ for $s=302, \cdots, 316$. Then $\sup_{(w,x)\in\Omega}$ $|C_{3n}(w,x)| = o_p(\frac{1}{\sqrt{nh}})$.

$C_{4n}(w,x)$--$C_{7n}(w,x)$ can be decomposed in the same way by expanding the multiplication of the interaction and quadratic terms among $\Delta_{\hat\mu}(X_i) - \Delta_\mu(X_i)$ and $\Delta_{\hat{p}}(X_i)-\Delta_p(X_i)$ yielding the same results that $\sup_{(w,x)\in\Omega}$ $|C_{sn}(w,x)| = o_p(\frac{1}{\sqrt{nh}})$ for $s=4,5,6,7$.
To summarize, $\hat{U}_n(w,x) - U_n(w,x)$ can be rewritten as $$\hat{U}_n(w,x) - U_n(w,x) = C_{1n}(w,x) - C_{2n}(w,x), $$ which is further expressed as
\begin{align*}
    \hat{U}_n(w,x) - U_n(w,x) =& \frac{\mathrm{i}w}{2n} \sum_{j=1}^n e^{\mathrm{i}x'X_j} \kappa(w,X_j) f(X_j, 0) \left( u_{Yj} - u_{Dj} \tau(X_j) \right) K_{h_r}(R_j) \\& e_0' \left( B_+^{-1}(X_j) \mathbf{1}_+(R_j) - B_-^{-1}(X_j) \mathbf{1}_-(R_j) \right) \mathbf{r}(R_j) + o_p(\frac{1}{\sqrt{nh}}).
\end{align*}

\subsection{Proof of Theorem \ref{thm1}}
Let $U_{1n}(w,x) = \frac{1}{n} \sum_{i=1}^n f_{1i}(w,x)$ and $U_{2n}(w,x) = \frac{1}{n} \sum_{i=1}^n f_{2i}(w,x)$, where
\begin{align*}
    f_{1i}(w,x) =& e^{\mathrm{i} \left( x'X_i + wW_i \right)} K_h(R_i) \delta_i + \mathrm{i} w e^{\mathrm{i} x'X_i} \kappa(w,X_i) (Y_i-D_i\tau(X_i)) K_{h_r}(R_i) \delta_i. \\
    f_{2i}(w,x) =& e^{\mathrm{i} \left( x'X_i + wW_i \right)} K_h(R_i) \delta_i + \frac{\mathrm{i}w}{2} e^{\mathrm{i} x'X_i} \kappa(w,X_i) f(X_i,0) (u_{Yi} - u_{Di}\tau(X_i)) K_{h_r}(R_i) \\
    & e_0' \left( \mathbf{1}_+(R_i)B_+^{-1}(X_i) - \mathbf{1}_-(R_i) B_-^{-1}(X_i) \right) \mathbf{r}(R_i).
\end{align*}
The conditional expectations of functions $f_{1i}$ and $f_{2i}$ given $X_i$ and $R_i$ are
\begin{align*}
    & \mathbb{E}( f_{1i}(w,x) | X_i, R_i ) \\
    =& e^{\mathrm{i} x'X_i} \mathbb{E}(e^{\mathrm{i} wW_i}|X_i, R_i) K_h(R_i) \delta_i + \mathrm{i} w e^{\mathrm{i} x'X_i} \kappa(w,X_i) \left( \mathbb{E}(Y_i|X_i, R_i) - \mathbb{E}(D_i|X_i,R_i) \tau(X_i) \right) K_{h_r}(R_i) \delta_i, \\
    =& e^{\mathrm{i} x'X_i} \varphi_{W|XR}(w|X_i, R_i) K_h(R_i) \delta_i + \mathrm{i} w e^{\mathrm{i} x'X_i} \kappa(w,X_i) \left( \mu(X_i,R_i) - p(X_i,R_i) \tau(X_i) \right) K_{h_r}(R_i) \delta_i. \\
    & \mathbb{E}( f_{2i}(w,x) | X_i, R_i ) \\=& e^{\mathrm{i} x'X_i} \varphi_{W|XR}(w|X_i, R_i) K_h(R_i) \delta_i + \frac{\mathrm{i}w}{2} e_0' \left( \mathbf{1}_+(R_i)B_+^{-1}(X_i) - \mathbf{1}_-(R_i) B_-^{-1}(X_i) \right) \mathbf{r}(R_i) \\
    & e^{\mathrm{i} x'X_i} \kappa(w,X_i) f(X_i,0) \left( \mathbb{E}(u_{Yi}|X_i, R_i) - \mathbb{E}(u_{Di}|X_i,R_i) \tau(X_i) \right) K_{h_r}(R_i) \\
    =& e^{\mathrm{i} x'X_i} \varphi_{W|XR}(w|X_i, R_i) K_h(R_i) \delta_i.
\end{align*}
The second equation of $\mathbb{E}(f_{2i}(w,x) | X_i, R_i)$ comes from the zero-mean residuals $u_{Yi}$ and $u_{Di}$. 
Note that from the proof of Lemma S.A1, $$ \mathbb{E}\left[ e^{\mathrm{i} x'X_i} \varphi_{W|XR}(w|X_i, R_i) K_h(R_i) \delta_i | X_i\right] = \frac{1}{2} e^{\mathrm{i} x'X_i} f_{R|X}(0|X_i) \Delta_\varphi(w,X_i) + O_p(h) $$
with a slight abuse of notation, $\Delta_\varphi(w,X_i) = \varphi_{W|XR}(w|X_i, 0^+)-\varphi_{W|XR}(w|X_i, 0^-)$ here.
Then $\mathbb{E}(f_{2i}(w,x)|X_i) = \frac{1}{2} e^{\mathrm{i} x'X_i} f_{R|X}(0|X_i) \Delta_\varphi(w,X_i) + O_p(h)$ and 
\begin{align*}
    \mathbb{E}( f_{1i}(w,x) | X_i) 
    =& \frac{1}{2} e^{\mathrm{i} x'X_i} f_{R|X}(0|X_i) \bigg[  \Delta_\varphi(w,X_i) + O_p(h) \bigg] - \mathrm{i} w e^{\mathrm{i} x'X_i} \kappa(w,X_i) \\& \bigg[ \frac{1}{2} \left[ \left( \mu(X_i, 0^+) - p(X_i, 0^+) \tau(X_i) \right) - \left( \mu(X_i, 0^-) - p(X_i, 0^-) \tau(X_i) \right) \right] + O_p(h_r) \bigg] \\
    =& \frac{1}{2} e^{\mathrm{i} x'X_i} f_{R|X}(0|X_i) \Delta_\varphi(w,X_i) + O_p(h) + O_p(h_r) .
\end{align*}
The second equality holds because $\left( \mu(X_i, 0^+) - p(X_i, 0^+) \tau(X_i) \right) - \left( \mu(X_i, 0^-) - p(X_i, 0^-) \tau(X_i) \right) = \Delta_\mu(X_i) - \Delta_p(X_i) \tau(X_i) =0$.

Denote $g_i(w,x) = \frac{1}{2} e^{\mathrm{i} x'X_i} f_{R|X}(0|X_i) \Delta_\varphi(w,X_i)$. The difference between $g_i(w,x)$ and $\mathbb{E}(f_{1i}(w,x)|X_i)$ or $\mathbb{E}(f_{2i}(w,x)|X_i)$ is asymptotically negligible under Assumption B1 or B1'.
It is straightforward to show that $f$ and $g$ are VC-type with a square-integrable envelope. Then the sums and differences of VC-type functions are also VC-type with a square-integrable envelope. Therefore, $f\in \left\{ f_{wx}: (w,x)\in\Omega \right\}$ and $g \in \left\{ g_{wx}: (w,x)\in\Omega \right\}$ are Donsker function classes. 

When Assumption B1 holds for $\hat{U}_n(w,x)$ with the local constant estimator $\hat\tau(X_i)$, $\sqrt{nh}\hat{U}_n(w,x)$ can be written as $$ \sqrt{nh}\hat{U}_n(w,x) = \sqrt{\frac{h}{n}} \sum_{i=1}^n f_{1i}(w,x) + o_p(1) = \sqrt{\frac{h}{n}} \sum_{i=1}^n g_i(w,x) + \sqrt{\frac{h}{n}} \sum_{i=1}^n g_{1i}(w,x) + o_p(1), $$ where $g_{1i}(w,x) = f_{1i}(w,x) - g_i(w,x)$.

When Assumption B1' holds for $\hat{U}_n(w,x)$ with the local linear estimator $\hat\tau(X_i)$, $\sqrt{nh} \hat{U}_n(w,x)$ can be written as $$ \sqrt{nh} \hat{U}_n(w,x) = \sqrt{\frac{h}{n}} \sum_{i=1}^n f_{2i}(w,x) + o_p(1) = \sqrt{\frac{h}{n}} \sum_{i=1}^n g_i(w,x) + \sqrt{\frac{h}{n}} \sum_{i=1}^n g_{2i}(w,x) + o_p(1), $$ where $g_{2i}(w,x) = f_{2i}(w,x) - g_i(w,x)$.

Furthermore, when $h=o(h_r)$ holds while Assumption B1 or B1' holds, the estimation effect is asymptotically negligible regardless of whether the local constant estimator or the local linear estimator is applied to calculate $\hat{U}_n(w,x)$. In this case, let $U_n(w,x) = \frac{1}{n} \sum_{i=1}^n f_{0i}(w,x)$, where $f_{0i}(w,x) = e^{\mathrm{i} \left( x'X_i + wW_i \right)} K_h(R_i) \delta_i$ and $\sqrt{nh}\hat{U}_n(w,x)$ can be written as $$ \sqrt{nh}\hat{U}_n(w,x) = \sqrt{\frac{h}{n}} \sum_{i=1}^n f_{0i}(w,x) + o_p(1) = \sqrt{\frac{h}{n}} \sum_{i=1}^n g_i(w,x) + \sqrt{\frac{h}{n}} \sum_{i=1}^n g_{0i}(w,x) + o_p(1), $$ where $g_{0i}(w,x) = f_{0i}(w,x) - g_i(w,x) = e^{\mathrm{i} \left( x'X_i + wW_i \right)} K_h(R_i) \delta_i - \frac{1}{2} e^{\mathrm{i} x'X_i} f_{R|X}(0|X_i) \Delta_\varphi(w,X_i)$.

Under $\mathbb{H}_0$, $ \varphi_{W|XR}(w|X_i, 0^+) - \varphi_{W|XR}(w|X_i, 0^-)=0$ and $g_i(w,x)=0$ a.s. By the Donsker theorem, $\sqrt{\frac{h}{n}}\sum_{i=1}^n g_{1i}(w,x)$ and $\sqrt{\frac{h}{n}}\sum_{i=1}^n g_{2i}(w,x)$ converge weakly: $$ \sqrt{\frac{h}{n}}\sum_{i=1}^n g_{1i}(w,x) \Longrightarrow U_{\infty,1}(w,x) \quad\quad\quad \sqrt{\frac{h}{n}}\sum_{i=1}^n g_{2i}(w,x) \Longrightarrow U_{\infty,2}(w,x), $$ where $U_{\infty,1}(w,x)$ and $U_{\infty,2}(w,x)$ are zero-mean Gaussian processes. 
Their covariance kernels can be derived by the convergence of the covariance of $\sqrt{\frac{h}{n}}\sum_{i=1}^n g_{1i}(w,x)$ and $\sqrt{\frac{h}{n}}\sum_{i=1}^n g_{2i}(w,x)$.
For the function $g_{1i}$, the covariance of $\sqrt{\frac{h}{n}} \sum_{i=1}^n g_{1i}(w,x)$ is 
\begin{align*}
    & \operatorname{Cov} \left( \sqrt{\frac{h}{n}} \sum_{i=1}^n g_{1i}(w_1,x_1), \sqrt{\frac{h}{n}} \sum_{i=1}^n g_{1i}(w_2,x_2) \right) = h \mathbb{E} \left[ g_{1i}(w_1, x_1) \bar{g}_{1i}(w_2,x_2) \right] \\
    =& h \mathbb{E} \bigg[ e^{\mathrm{i} (x_1-x_2)'X} \bigg\{ e^{\mathrm{i} (w_1-w_2)W} K_h^2(R) \delta^2 + w_1w_2 \kappa(w_1, X) \kappa(-w_2, X) ( Y-D\tau(X) )^2 K_{h_r}^2(R) \delta^2 \\
    &+ \mathrm{i} \bigg( w_1 e^{-\mathrm{i} w_2W} \kappa(w_1, X) - w_2 e^{\mathrm{i} w_1W} \kappa(-w_2, X) \bigg) (Y-D\tau(X)) K_h(R) K_{h_r}(R) \delta^2 \bigg\} \bigg],
\end{align*}
which converges to 
\begin{align*}
    & \mathcal{K}_1(w_1, x_1; w_2, x_2) = R_K \int e^{\mathrm{i} (x_1-x_2)'x} \varphi_{W|XR}(w_1-w_2|x, 0) f_{XR}(x, 0) dx \\
    &+ \mathrm{i}w_1 \int K(v)K(v/c)dv \int e^{\mathrm{i} (x_1-x_2)'x} \mathbb{E} \left[ e^{-\mathrm{i} w_2W} (Y-D\tau(X)) |X=x, R=0 \right] \kappa(w_1, x) f_{XR}(x, 0) dx \\
    &- \mathrm{i}w_2 \int K(v)K(v/c)dv \int e^{\mathrm{i} (x_1-x_2)'x} \mathbb{E} \left[ e^{\mathrm{i} w_1W} (Y-D\tau(X)) |X=x, R=0 \right] \kappa(-w_2, x) f_{XR}(x, 0) dx \\
    &+ c w_1w_2 R_K \int e^{\mathrm{i} (x_1-x_2)'x} \kappa(w_1, x) \kappa(-w_2, x) \mathbb{E} \left[ (Y-D\tau(X))^2 |X=x, R=0 \right] f_{XR}(x, 0) dx
\end{align*}
with $c=\lim_{n\to\infty} \frac{h}{h_r}$ and $R_K = \int K^2(v)dv$.
Note that when $c=0$, $\int K(v)K(v/c)dv=0$ because the kernel function is defined on the compact support. In this case, $\mathcal{K}_1$ reduces to $$\mathcal{K}_0(w_1, x_1; w_2, x_2) := R_K \int e^{\mathrm{i} (x_1-x_2)'x} \varphi_{W|XR}(w_1-w_2|x, 0) f_{XR}(x, 0) dx.$$
The covariance of $\sqrt{\frac{h}{n}} \sum_{i=1}^n g_{2i}(w,x)$ is 
\begin{align*}
    & \operatorname{Cov} \left( \sqrt{\frac{h}{n}} \sum_{i=1}^n g_{2i}(w_1,x_1), \sqrt{\frac{h}{n}} \sum_{i=1}^n g_{2i}(w_2,x_2) \right) = h \mathbb{E} \left[ g_{2i}(w_1, x_1) \bar{g}_{2i}(w_2,x_2) \right] \\
    =& h \mathbb{E} \bigg[ e^{\mathrm{i} (x_1-x_2)'X} \bigg\{ e^{\mathrm{i} (w_1-w_2)W} K_h^2(R) \delta^2 + \frac{w_1w_2}{4} \kappa(w_1, X) \kappa(-w_2, X) (u_Y-u_D\tau(X))^2 K_{h_r}^2(R) f^2(X,0) \\
    & \left( e_0' \left( \mathbf{1}_+(R)B_+^{-1}(X) - \mathbf{1}_-(R) B_-^{-1}(X) \right) \mathbf{r}(R) \right)^2 \\
    &+ \frac{\mathrm{i}}{2} \left( w_1  e^{-\mathrm{i} w_2W} \kappa(w_1,X) - w_2  e^{\mathrm{i} w_1W} \kappa(-w_2,X) \right) (u_Y-u_D\tau(X)) K_h(R) K_{h_r}(R) \delta \\
    & e_0' \left( \mathbf{1}_+(R)B_+^{-1}(X) - \mathbf{1}_-(R) B_-^{-1}(X) \right) \mathbf{r}(R) f(X,0)
    \bigg\} \bigg],
\end{align*}
which converges to 
\begin{align*}
    & \mathcal{K}_2(w_1, x_1; w_2, x_2) = R_K \int e^{\mathrm{i} (x_1-x_2)'x} \varphi_{W|XR}(w_1-w_2|x, 0) f_{XR}(x,0)dx \\
    &+ \frac{\mathrm{i} w_1}{4} \int K(v)K(v/c)dv \int e^{\mathrm{i} (x_1-x_2)'x} \kappa(w_1, x) \bigg[ \mathbb{E} \left( e^{-\mathrm{i} w_2W} (u_Y-u_D\tau(X)) \mid X=x, R=0^+ \right) \\& e_0' B_+^{-1}(x) \mathbf{r}(0) + \mathbb{E} \left( e^{-\mathrm{i} w_2W} (u_Y-u_D\tau(X)) \mid X=x, R=0^- \right) e_0' B_-^{-1}(x) \mathbf{r}(0) \bigg] f_{XR}^2(x,0)dx \\
    &- \frac{\mathrm{i} w_2}{4} \int K(v)K(v/c)dv \int e^{\mathrm{i} (x_1-x_2)'x} \kappa(-w_2, x) \bigg[ \mathbb{E} \left( e^{\mathrm{i} w_1W} (u_Y-u_D\tau(X)) \mid X=x, R=0^+ \right) \\& e_0' B_+^{-1}(x) \mathbf{r}(0) + \mathbb{E} \left( e^{\mathrm{i} w_1W} (u_Y-u_D\tau(X)) \mid X=x, R=0^- \right) e_0' B_-^{-1}(x) \mathbf{r}(0) \bigg] f_{XR}^2(x,0)dx \\
    &+ c \frac{w_1w_2}{4} \frac{R_K}{2} \int e^{\mathrm{i} (x_1-x_2)'x} \kappa(w_1, x) \kappa(-w_2, x) \bigg[ \mathbb{E}\left( (u_Y-u_D\tau(X))^2 \mid X=x, R=0^+ \right) \\& \left( e_0' B_+^{-1}(x) \mathbf{r}(0) \right)^2 + \mathbb{E}\left( (u_Y-u_D\tau(X))^2 \mid X=x, R=0^- \right) \left( e_0' B_-^{-1}(x) \mathbf{r}(0) \right)^2 \bigg] f_{XR}^3(x,0)dx.
\end{align*}
Note that when $c=0$, $\int K(v)K(v/c)dv=0$ and $\mathcal{K}_2$ reduces to $\mathcal{K}_0$.

When $h=o(h_r)$, $\sqrt{\frac{h}{n}}\sum_{i=1}^n f_{0i}(w,x) \Longrightarrow U_{\infty,0}(w,x), $ where $U_{\infty,0}(w,x)$ is a zero-mean Gaussian process, and its covariance kernel can be derived by the convergence of the covariance of $\sqrt{\frac{h}{n}}\sum_{i=1}^n f_{0i}(w,x)$:
$$\operatorname{Cov} \left( \sqrt{\frac{h}{n}} \sum_{i=1}^n g_{0i}(w_1,x_1), \sqrt{\frac{h}{n}} \sum_{i=1}^n g_{0i}(w_2,x_2) \right) = h \mathbb{E} \left[ e^{\mathrm{i} (x_1-x_2)'X_i} e^{\mathrm{i} (w_1-w_2)W_i} K_h^2(R_i) \delta_i^2 \right], $$
which converges to $\mathcal{K}_0(w_1, x_1; w_2, x_2)$. 

\subsection{Proof of Theorem \ref{thm2}}
Under $\mathbb{H}_1$, it can be derived that
\begin{align*}
    & \mathbb{E} \left[ e^{\mathrm{i}x'X} \varphi_{W|XR}(w|X,R) K_h(R) \delta \right] = \mathbb{E} \left[ \mathbb{E} \left[ e^{\mathrm{i}x'X} \varphi_{W|XR}(w|X,R) K_h(R) \delta \mid X \right] \right] \\
    =& \mathbb{E} \left[ e^{\mathrm{i} x'X} \Delta_\varphi(w,X) f_{R|X}(0|X) + O_p(h) \right] \\
    =& \int e^{\mathrm{i} x'\bar{x}} \Delta_\varphi(w,\bar{x}) f_{XR}(\bar{x},0) d\bar{x} + O(h) =: \Gamma(w,x) + O(h), 
\end{align*}
where $\Gamma(w,x) = \int e^{\mathrm{i} x'\bar{x}} \Delta_\varphi(w,\bar{x}) f_{XR}(\bar{x},0) d\bar{x} \ne 0$. For $\hat{U}_n(w,x)$ with the local constant and local linear estimators,
\begin{align*}
    & \sup_{(w,x)\in\Omega} \left| \hat{U}_n(w,x) - \Gamma(w,x) \right| \\
    \le& \sup_{(w,x)\in\Omega} \left| \hat{U}_n(w,x) - U_{1n}(w,x) \right| + \sup_{(w,x)\in\Omega} \left| U_{1n}(w,x) - \mathbb{E} f_{1i}(w,x) \right| + \sup_{(w,x)\in\Omega} \left| \mathbb{E} f_{1i}(w,x) - \Gamma(w,x) \right| \\=& o_p(1). 
\end{align*}
From Lemma \ref{Lemma1}, the first term is the estimation effect of order $o_p(\frac{1}{\sqrt{nh}})$.
The second term is $o_p(1)$ by the ULLN: the envelope of the main part satisfies $\mathbb{E}\sup_{(w,x)\in\Omega}|e^{\mathrm{i}(x'X_i+wW_i)}K_h(R_i)\delta_i|\le\mathbb{E}|K_h(R)|<\infty$, while the estimation-effect term in $f_{1i}$ is $o_p(1)$ uniformly on $\Omega$ by Lemma \ref{Lemma1} and has mean $O(h_r)=o(1)$ because $\Delta_\mu-\Delta_p\tau=0$.
The third term is $o(1)$ by the display above, and the same argument applies to $U_{2n}(w,x)$ and $U_{0n}(w,x)$.
Then $\hat{U}_n(w,x)$ converges to a nonzero expectation uniformly on $\Omega$ and furthermore $\sqrt{nh}\hat{U}_n(w,x)$ diverges as $n\to\infty$ under $\mathbb{H}_1$ for $j=1,2,0$.

\subsection{Proof of Theorem \ref{thm3}}
Under $\mathbb{H}_{1n}$, $g_i(w,x) = \frac{1}{\sqrt{nh}} e^{\mathrm{i} x'X_i} f_{R|X}(0|X_i) \lambda(w,X_i)$ and $$\sqrt{\frac{h}{n}} \sum_{i=1}^n g_i(w,x) = \frac{1}{n} \sum_{i=1}^n e^{\mathrm{i} x'X_i} f_{R|X}(0|X_i) \lambda(w,X_i), $$ which converges to $\Lambda(w,x) = \mathbb{E} \left[ \lambda(w,X) e^{\mathrm{i} x'X} f_{R|X}(0|X) \right] = \int e^{\mathrm{i} x'\bar{x}} \lambda(w,\bar{x}) f_{XR}(\bar{x},0) d\bar{x} \ne 0$ uniformly on $\Omega$ by ULLN since $\mathbb{E} \sup_{(w,x)\in\Omega} \left| e^{\mathrm{i} x'X} f_{R|X}(0|X) \lambda(w,X) \right| < \infty $. 

When Assumption B1 holds for $\hat{U}_n(w,x)$ with the local constant estimator $\hat\tau(X_i)$, $$ \sqrt{nh}\hat{U}_n(w,x) = \sqrt{\frac{h}{n}} \sum_{i=1}^n g_{1i}(w,x) + \Lambda(w,x) + o_p(1) \Longrightarrow U_{\infty,1}(w,x) + \Lambda(w,x). $$
The covariance of $\sqrt{\frac{h}{n}} \sum_{i=1}^n g_{1i}(w,x)$ is
\begin{align*}
    & \operatorname{Cov} \left( \sqrt{\frac{h}{n}} \sum_{i=1}^n g_{1i}(w_1,x_1), \sqrt{\frac{h}{n}} \sum_{i=1}^n g_{1i}(w_2,x_2) \right) = h \mathbb{E} \left[ g_{1i}(w_1, x_1) \bar{g}_{1i}(w_2,x_2) \right] \\
    =& h \mathbb{E}\bigg[ e^{\mathrm{i} (x_1-x_2)'X} \left\{ e^{\mathrm{i} (w_1-w_2)W} K_h^2(R) \delta^2 + w_1w_2 \kappa(w_1, X) \kappa(-w_2, X) ( Y-D\tau(X) )^2 K_{h_r}^2(R) \delta^2 \right. \\
    &+ \mathrm{i} \bigg( w_1 e^{-\mathrm{i} w_2W} \kappa(w_1, X) - w_2 e^{\mathrm{i} w_1W} \kappa(-w_2, X) \bigg) (Y-D\tau(X)) K_h(R) K_{h_r}(R) \delta^2 \\
    &+ \mathrm{i} \frac{f_{R|X}(0|X)}{\sqrt{nh}} \bigg( w_2 \kappa(-w_2, X) \lambda(w_1, X) - w_1 \kappa(w_1, X) \lambda(-w_2, X) \bigg) \left(Y-D\tau(X)\right) K_{h_r}(R) \delta \\
    &- \frac{f_{R|X}(0|X)}{\sqrt{nh}} \left( \lambda(w_1,X) e^{-\mathrm{i} w_2W} + \lambda(-w_2,X) e^{\mathrm{i} w_1W} \right) K_h(R) \delta + \frac{f^2_{R|X}(0|X)}{nh} \lambda(w_1,X) \lambda(-w_2,X) \bigg\} \bigg],
\end{align*}
which converges to $\mathcal{K}_1(w_1, x_1; w_2, x_2)$ since the last three terms with $\frac{1}{\sqrt{nh}}$ and $\frac{1}{nh}$ approach 0 as $n\to\infty$. Then $\sqrt{nh} \hat{U}_n(w,x)$ converges to the Gaussian process $U_{\infty,1}(w,x)$ plus a nonzero shift term $\Lambda(w,x)$.

When Assumption B1' holds for $\hat{U}_n(w,x)$ with the local linear estimator $\hat\tau(X_i)$, $$ \sqrt{nh}\hat{U}_n(w,x) = \sqrt{\frac{h}{n}} \sum_{i=1}^n g_{2i}(w,x) + \Lambda(w,x) + o_p(1) \Longrightarrow U_{\infty,2}(w,x) + \Lambda(w,x). $$
The covariance of $\sqrt{\frac{h}{n}} \sum_{i=1}^n g_{2i}(w,x)$ is
\begin{align*}
    & \operatorname{Cov} \left( \sqrt{\frac{h}{n}} \sum_{i=1}^n g_{2i}(w_1,x_1), \sqrt{\frac{h}{n}} \sum_{i=1}^n g_{2i}(w_2,x_2) \right) = h \mathbb{E} \left[ g_{2i}(w_1, x_1) \bar{g}_{2i}(w_2,x_2) \right] \\
    =&  h \mathbb{E} \bigg[ e^{\mathrm{i} (x_1-x_2)'X} \bigg\{ e^{\mathrm{i} (w_1-w_2)W} K_h^2(R) \delta^2 + \frac{w_1w_2}{4} \kappa(w_1, X) \kappa(-w_2, X) (u_Y-u_D\tau(X))^2 K_{h_r}^2(R) f^2(X,0) \\
    & \left( e_0' \left( \mathbf{1}_+(R)B_+^{-1}(X) - \mathbf{1}_-(R) B_-^{-1}(X) \right) \mathbf{r}(R) \right)^2 + \frac{\mathrm{i}}{2} \left( w_1  e^{-\mathrm{i} w_2W} \kappa(w_1,X) - w_2  e^{\mathrm{i} w_1W} \kappa(-w_2,X) \right) \\
    & e_0' \left( \mathbf{1}_+(R)B_+^{-1}(X) - \mathbf{1}_-(R) B_-^{-1}(X) \right) \mathbf{r}(R) f(X,0) (u_Y-u_D\tau(X)) K_h(R) K_{h_r}(R) \delta \\ 
    &- \frac{f_{R|X}(0|X)}{\sqrt{nh}} \left( \lambda(w_1,X) e^{-\mathrm{i} w_2W} + \lambda(-w_2,X) e^{\mathrm{i} w_1W} \right) K_h(R) \delta + \frac{f^2_{R|X}(0|X)}{nh} \lambda(w_1,X) \lambda(-w_2,X) \bigg\} \bigg],
\end{align*}
which converges to $\mathcal{K}_2(w_1, x_1; w_2, x_2)$ since the last two terms with $\frac{1}{\sqrt{nh}}$ and $\frac{1}{nh}$ approach 0 as $n\to\infty$. Then $\sqrt{nh} \hat{U}_n(w,x)$ converges to the Gaussian process $U_{\infty,2}(w,x)$ plus a nonzero shift term $\Lambda(w,x)$.

Moreover, if $h=o(h_r)$ and Assumption B1 or B1' holds, $$ \sqrt{nh}\hat{U}_n(w,x) = \sqrt{\frac{h}{n}} \sum_{i=1}^n g_{0i}(w,x) + \Lambda(w,x) + o_p(1) \Longrightarrow U_{\infty,0}(w,x) + \Lambda(w,x). $$
The covariance of $\sqrt{\frac{h}{n}} \sum_{i=1}^n g_{0i}(w,x)$ is
\begin{align*}
    & \operatorname{Cov} \left( \sqrt{\frac{h}{n}} \sum_{i=1}^n g_{0i}(w_1,x_1), \sqrt{\frac{h}{n}} \sum_{i=1}^n g_{0i}(w_2,x_2) \right) = h \mathbb{E} \left[ g_{0i}(w_1, x_1) \bar{g}_{0i}(w_2,x_2) \right] \\
    =& h \mathbb{E} \bigg[ e^{\mathrm{i} (x_1-x_2)'X} \bigg\{ e^{\mathrm{i} (w_1-w_2)W} K_h^2(R) \delta^2 - \frac{f_{R|X}(0|X)}{\sqrt{nh}} \left( \lambda(w_1,X) e^{-\mathrm{i} w_2W} + \lambda(-w_2,X) e^{\mathrm{i} w_1W} \right) K_h(R) \delta \\
    &+ \frac{f^2_{R|X}(0|X)}{nh} \lambda(w_1,X) \lambda(-w_2,X) \bigg\} \bigg],
\end{align*}
which converges to $\mathcal{K}_0(w_1, x_1; w_2, x_2)$ because the last two terms with $\frac{1}{\sqrt{nh}}$ and $\frac{1}{nh}$ approach 0 as $n\to\infty$. Then $\sqrt{nh} \hat{U}_n(w,x)$ converges to the Gaussian process $U_{\infty,0}(w,x)$ plus a nonzero shift term $\Lambda(w,x)$.

\subsection{Proof of Corollary \ref{corollary}}
\begin{proof}
Lemma \ref{Lemma1} shows that $\sup_{(w,x)\in\Omega} \left| \sqrt{nh}\hat{U}_n(w,x) - \sqrt{\frac{h}{n}} \sum_{i=1}^n f_{ji}(w,x) \right| = o_p(1)$ for $j=0,1,2$ in each case described before under $\mathbb{H}_0$, which means that $\sqrt{nh}\hat{U}_n(w,x)$ converges to Gaussian processes $U_{\infty,j}(w,x)$ uniformly on $\Omega$ under $\mathbb{H}_0$ for $j=0,1,2$.
The continuous mapping theorem can be applied to deduce the convergence of the statistics. Two simple examples are introduced in this paper.
The Kolmogorov--Smirnov (KS) statistic is $$ \mathrm{KS}_n = \sup_{(w,x)\in\Omega} \max\left\{ \left| \mathrm{Re} (\sqrt{nh} \hat{U}_n(w,x))\right|, \left| \mathrm{Im} (\sqrt{nh} \hat{U}_n(w,x))\right| \right\}. $$
For the Cram\'{e}r--von Mises (CvM) statistic, a fast-decaying weighting function $G(w,x)$ can be used to derive the expression of the CvM statistic. 
One example we used in this paper is an exponential weighting function: $G(w,x) = \int_{-\infty}^w \int_{-\infty}^x e^{-\frac{1}{2}(\bar{w}^2+ \lVert \bar{x} \rVert ^2)} d\bar{w}d\bar{x}$ with $g(w,x) = \frac{\partial^2 G(w,x)}{\partial w \partial x'} = e^{-\frac{1}{2}(w^2 + \lVert x \rVert^2)}$, which has finite total mass. 
For every $\varepsilon>0$, there exists a compact set $\Omega \subset \mathbb{R}^{d+1}$ such that $|\int_{\mathbb{R}^{d+1} \setminus \Omega} G(dw,dx)| \le \varepsilon$. Theorem~\ref{thm1} is understood to hold on every compact subset of $\mathbb{R}^{d+1}$.
The expression of the CvM statistic with this weighting function is 
\begin{align*}
    \mathrm{CvM}_n &= \int \left| \sqrt{nh} \hat{U}_n(w,x) \right|^2 G(dw, dx) \\
    &= \int \frac{1}{nh} \sum_{i=1}^n \sum_{j=1}^n e^{\mathrm{i} (w\hat{W}_i+x'X_i)} K(\frac{R_i}{h}) \delta_i e^{-\mathrm{i}(w\hat{W}_j+x'X_j)} K(\frac{R_j}{h}) \delta_j g(w,x) dwdx \\
    &= \frac{1}{nh} \sum_{i=1}^n \sum_{j=1}^n  K(\frac{R_i}{h})  K(\frac{R_j}{h}) \delta_i \delta_j \int 
    e^{\mathrm{i} w(\hat{W}_i-\hat{W}_j) - \frac{1}{2} w^2}dw \int e^{\mathrm{i} x'(X_i-X_j) - \frac{1}{2} \lVert x \rVert^2} dx \\
    &= \frac{1}{nh} \sum_{i=1}^n \sum_{j=1}^n e^{-\frac{(\hat{W}_i-\hat{W}_j)^2 + \lVert X_i-X_j \rVert^2}{2}} K(\frac{R_i}{h}) K(\frac{R_j}{h}) \delta_i \delta_j,
\end{align*}
where the constant factor $(2\pi)^{\frac{d+1}{2}}$ from the Gaussian integrals is absorbed into the weighting measure $G(dw,dx)$ and cancels out in the test statistic; the expression above is therefore proportional to the $L_2$-norm and equivalent for testing purposes.

Their convergence can be obtained by the continuous mapping theorem (CMT) (Theorem 1.3.6 in \cite{vdv1996}). 
For CvM-type statistics,  
\begin{align*}
    \mathrm{CvM}_n &= \int \left| \sqrt{nh} \hat{U}_n(w,x) \right|^2 G(dw, dx) \\
    &= \int_{\Omega} \left| \sqrt{nh} \hat{U}_n(w,x) \right|^2 G(dw, dx) + \int_{\mathbb{R}^{d+1} \setminus \Omega} \left| \sqrt{nh} \hat{U}_n(w,x) \right|^2 G(dw, dx).
\end{align*}
For the first part, Theorem~\ref{thm1} and the continuous mapping theorem for the map $f\mapsto\int_\Omega\lvert f\rvert^2 G(dw,dx)$ on $\ell^\infty(\Omega)$ give $$\int_{\Omega} \left| \sqrt{nh} \hat{U}_n(w,x) \right|^2 G(dw, dx) = \int_{\Omega} \left| U_{\infty,j}(w,x) \right|^2 G(dw,dx) + o_p(1). $$
For the tail, given any fixed small $\varepsilon>0$, there exists a compact $\Omega$ such that $$ \int_{\mathbb{R}^{d+1} \setminus \Omega} G(dw, dx) \le \varepsilon. $$ 
Note that $\left| e^{\mathrm{i}(wW+x'X)} \right|\le 1$ and $\left| e^{\mathrm{i}x'X} \left( e^{\mathrm{i}w\hat{W}} - e^{\mathrm{i}wW}\right) \right|\le 2$ still hold off $\Omega$. 
Then $$\sup_{(w,x)} \mathbb{E} \left| \sqrt{nh} U_n(w,x) \right|^2 \lesssim h^{-1} \mathbb{E} \left[ K(R/h) \right]^2 <\infty$$ and $$\sup_{(w,x)} \mathbb{E} \left| \sqrt{nh} \left( \hat{U}_n - U_n \right) (w,x) \right|^2 \lesssim 4 h^{-1} \mathbb{E} \left[ K(R/h) \right]^2 <\infty.$$
By the Cauchy–Schwarz inequality, 
\begin{align*}
    \sup_{(w,x)} \mathbb{E} \left| \sqrt{nh} \hat{U}_n(w,x) \right|^2 \le 2\sup_{(w,x)} \mathbb{E} \left| \sqrt{nh} U_n(w,x) \right|^2 + 2 \sup_{(w,x)} \mathbb{E} \left| \sqrt{nh} \left( \hat{U}_n - U_n \right) (w,x) \right|^2 \le C < \infty 
\end{align*}
for a constant C independent of $n$. 
Then for every $\varepsilon>0$, one can choose a compact set $\Omega$ such that $$\mathbb{E} \int_{\mathbb{R}^{d+1} \setminus \Omega} \left| \sqrt{nh} \hat{U}_n(w,x) \right|^2 G(dw, dx) \le C \int_{\mathbb{R}^{d+1} \setminus \Omega} G(dw, dx) \le C\varepsilon. $$ 
Then Markov's inequality yields $$ \int_{\mathbb{R}^{d+1} \setminus \Omega} \left| \sqrt{nh} \hat{U}_n(w,x) \right|^2 G(dw,dx) = O_p(\varepsilon). $$
The same bound holds for the limiting Gaussian process $U_{\infty,j}$, which is well-defined through its covariance kernel.

Fixing such an $\Omega$ and applying the continuous mapping theorem gives
\begin{align*}
    \mathrm{CvM}_n &= \int_{\Omega} \left| \sqrt{nh} \hat{U}_n(w,x) \right|^2 G(dw, dx) + \int_{\mathbb{R}^{d+1} \setminus \Omega} \left| \sqrt{nh} \hat{U}_n(w,x) \right|^2 G(dw, dx) \\
    &= \int_\Omega \left| U_{\infty,j}(w,x) \right|^2 G(dw,dx) + o_p(1) + O_p(\varepsilon).
\end{align*}
Letting $\varepsilon\downarrow 0$, $$ \mathrm{CvM}_n \overset{d}{\rightarrow}
\int\left| U_{\infty,j}(w,x) \right|^2 G(dw,dx). $$
For KS-type statistics, by the continuous mapping theorem, for $j=0,1,2$, 
\begin{align*}
  \mathrm{KS}_n =& \sup_{(w,x)\in\Omega} \max \left\{ \left| \mathrm{Re} \left(\sqrt{nh} \hat{U}_n(w,x) \right) \right|, \left| \mathrm{Im} \left(\sqrt{nh} \hat{U}_n(w,x) \right) \right| \right\} \\
  =& \sup_{(w,x)\in\Omega} \max \left\{ \left| \mathrm{Re} \left( U_{\infty,j}(w,x) \right) \right|, \left| \mathrm{Im} \left( U_{\infty,j}(w,x) \right) \right| \right\} + o_p(1). 
\end{align*}
Hence, $\mathrm{KS}_n \overset{d}{\rightarrow} \sup_{(w,x)\in\Omega} \max \left\{ \left| \mathrm{Re} \left( U_{\infty,j}(w,x) \right) \right|, \left| \mathrm{Im} \left( U_{\infty,j}(w,x) \right) \right| \right\}$ and $\mathrm{CvM}_n \overset{d}{\rightarrow} \int \left| U_{\infty,j}(w,x) \right|^2 \\ G(dw,dx)$.
\end{proof}

\subsection{Proof of Lemma \ref{Lemma2}}
\subsubsection{(i)}
The difference between $\hat{U}_n^\ast(w,x)$ and $U_n^\ast(w,x)$ is
\begin{equation}
\begin{aligned}
    & \hat{U}_n^\ast(w,x) - U_n^\ast(w,x) \\
    =& \frac{1}{n} \sum_{i=1}^n V_i \bigg( e^{\mathrm{i} w\hat{W}_i} - e^{\mathrm{i} wW_i} \bigg) e^{\mathrm{i} x'X_i} K_h(R_i) \delta_i \\
    &+ \frac{\mathrm{i}w}{n} \sum_{i=1}^n V_i e^{\mathrm{i} x'X_i} \bigg( \hat\kappa(w, X_i) (Y_i-D_i \hat\tau(X_i)) - \kappa(w, X_i) (Y_i-D_i\tau(X_i)) \bigg)  K_{h_r}(R_i) \delta_i \\=:& \text{I} + \text{II}. 
\label{eq1}
\end{aligned}
\end{equation}
Term I of \eqref{eq1} can be rewritten as the sum of two parts using a Taylor expansion of $e^{\mathrm{i} w\hat{W}_i}$ at the point $W_i$: 
\begin{align*}
    \text{I} =& \frac{\mathrm{i}w}{n} \sum_{i=1}^n V_i e^{\mathrm{i} (wW_i+x'X_i)} (1-D_i) ( \hat\tau(X_i) - \tau(X_i) )K_h(R_i) \delta_i \\
    &- \frac{w^2}{2n} \sum_{i=1}^n V_i e^{\mathrm{i} (wW_i+x'X_i)} (1-D_i) (\hat\tau(X_i) - \tau(X_i))^2 K_h(R_i) \delta_i \\
    &- \frac{\mathrm{i}w^3}{6n} \sum_{i=1}^n V_i e^{\mathrm{i} (wW_i+x'X_i)} (1-D_i) (\hat\tau(X_i) - \tau(X_i))^3 K_h(R_i) \delta_i, 
\end{align*}
where $C_{1n}^\ast(w,x) = \frac{\mathrm{i}w}{n} \sum_{i=1}^n V_i e^{\mathrm{i} (wW_i+x'X_i)} (1-D_i) (\hat\tau(X_i) - \tau(X_i)) K_h(R_i) \delta_i - \frac{w^2}{2n} \sum_{i=1}^n V_i e^{\mathrm{i} (wW_i+x'X_i)}$ $(1-D_i) (\hat\tau(X_i) - \tau(X_i))^2 K_h(R_i) \delta_i$. The last term in the expression is of smaller order than the second term of $C_{1n}^\ast(w,x)$.
The asymptotic behavior of $C_{1n}^\ast(w,x)$ can be analyzed following the proof of Lemma \ref{Lemma1}. The slight difference is that some of the corresponding conditional expectations and total expectations in the Hoeffding decompositions are exactly 0 due to the independence of the zero-mean multipliers $\{V_i\}_{i=1}^n$ from other variables. 

With these results, the leading terms in $C_{1n}^\ast(w,x)$ can be decomposed as $C_{1n}^\ast(w,x) = C_{11n}^\ast(w,x) - C_{12n}^\ast(w,x) - C_{13n}^\ast(w,x) + C_{14n}^\ast(w,x) - C_{15n}^\ast(w,x) - C_{16n}^\ast(w,x) + C_{17n}^\ast(w,x)$, where the elements are 
\begin{align*}
    C_{11n}^\ast(w,x) =& \frac{\mathrm{i}w}{n} \sum_{i=1}^n V_i e^{\mathrm{i} (wW_i+x'X_i)} (1-D_i) \frac{\hat f(X_i,0) \Delta_{\hat\mu}(X_i)}{\frac{1}{2} f(X_i,0) \Delta_p(X_i)} K_h(R_i) \delta_i, \\
    C_{12n}^\ast(w,x) =& \frac{\mathrm{i}w}{n} \sum_{i=1}^n V_i e^{\mathrm{i} (wW_i+x'X_i)} (1-D_i) \frac{\hat f(X_i,0) \Delta_{\hat p}(X_i)}{\frac{1}{2} f(X_i,0) \Delta_p(X_i)} \tau(X_i) K_h(R_i) \delta_i, \\
    C_{13n}^\ast(w,x) =& \frac{\mathrm{i}w}{n} \sum_{i=1}^n V_i e^{\mathrm{i} (wW_i+x'X_i)} (1-D_i) \left( \frac{\hat{f}(X_i,0) \Delta_{\hat\mu}(X_i)}{\frac{1}{2} f(X_i,0) \Delta_p(X_i)} - \tau(X_i) \right) \left( \frac{\hat{f}(X_i,0) \Delta_{\hat{p}}(X_i)}{\frac{1}{2} f(X_i,0) \Delta_p(X_i)} - 1 \right) \\& K_h(R_i) \delta_i, \\
    C_{14n}^\ast(w,x) =& \frac{\mathrm{i}w}{n} \sum_{i=1}^n V_i e^{\mathrm{i} (wW_i+x'X_i)} (1-D_i) \left( \frac{\hat{f}(X_i,0) \Delta_{\hat{p}}(X_i)}{\frac{1}{2} f(X_i,0) \Delta_p(X_i)} - 1 \right)^2 \tau(X_i) K_h(R_i) \delta_i, \\
    C_{15n}^\ast(w,x) =& \frac{w^2}{2n} \sum_{i=1}^n V_i e^{\mathrm{i} (wW_i+x'X_i)} (1-D_i) \left( \frac{\hat{f}(X_i,0) \Delta_{\hat\mu}(X_i)}{\frac{1}{2} f(X_i,0) \Delta_p(X_i)} - \tau(X_i) \right)^2 K_h(R_i) \delta_i, \\
    C_{16n}^\ast(w,x) =& \frac{w^2}{2n} \sum_{i=1}^n V_i e^{\mathrm{i} (wW_i+x'X_i)} (1-D_i) \left( \frac{\hat{f}(X_i,0) \Delta_{\hat{p}}(X_i)}{\frac{1}{2} f(X_i,0) \Delta_p(X_i)} - 1 \right)^2 \tau^2(X_i) K_h(R_i) \delta_i, \\
    C_{17n}^\ast(w,x) =& \frac{w^2}{n} \sum_{i=1}^n V_i e^{\mathrm{i} (wW_i+x'X_i)} (1-D_i) \left( \frac{\hat{f}(X_i,0) \Delta_{\hat\mu}(X_i)}{\frac{1}{2} f(X_i,0) \Delta_p(X_i)} - \tau(X_i) \right) \left( \frac{\hat{f}(X_i,0) \Delta_{\hat{p}}(X_i)}{\frac{1}{2} f(X_i,0) \Delta_p(X_i)} - 1 \right) \\& \tau(X_i) K_h(R_i) \delta_i.
\end{align*}
The first two terms $C_{11n}^\ast(w,x)$ and $C_{12n}^\ast(w,x)$ are second-order $U$-processes. The other terms $C_{13n}^\ast(w,x)$--$C_{17n}^\ast(w,x)$ are combinations of second-order $U$-processes and third-order $U$-processes. $C_{11n}^\ast(w,x)$ and $C_{12n}^\ast(w,x)$ have Hoeffding decompositions
\begin{align*}
    C_{11n}^\ast(w,x) =& \frac{1}{n}\sum_{i=1}^n \mathbb{E} (C_{11n,ij}^\ast|\mathcal{F}_i^\ast) + \frac{1}{n}\sum_{j=1}^n \mathbb{E} (C_{11n,ij}^\ast|\mathcal{F}_j^\ast) - \mathbb{E} C_{11n,ij}^\ast + U_n^{(2)}(\pi_2 g_{C_{11n}^\ast}), \\
    C_{12n}^\ast(w,x) =& \frac{1}{n}\sum_{i=1}^n \mathbb{E}(C_{12n,ij}^\ast|\mathcal{F}_i^\ast) + \frac{1}{n}\sum_{j=1}^n \mathbb{E} (C_{12n,ij}^\ast|\mathcal{F}_j^\ast) - \mathbb{E} C_{12n,ij}^\ast + U_n^{(2)}(\pi_2 g_{C_{12n}^\ast}), 
\end{align*}
where the elements are 
\begin{align*}
    C_{11n,ij}^\ast =& \mathrm{i}w V_i e^{\mathrm{i} (wW_i+x'X_i)} \frac{1-D_i}{\frac{1}{2} f(X_i,0) \Delta_p(X_i)} Y_j K_{h_x}(X_j-X_i) K_{h_r}(R_j)\delta_j K_h(R_i) \delta_i, \\
    C_{12n,ij}^\ast =& \mathrm{i}w V_i e^{\mathrm{i} (wW_i+x'X_i)} \frac{\tau(X_i) (1-D_i)}{\frac{1}{2} f(X_i,0) \Delta_p(X_i)} D_j K_{h_x}(X_j-X_i) K_{h_r}(R_j)\delta_j K_h(R_i) \delta_i. 
\end{align*}
The conditional expectations given $\mathcal{F}_j$ and total expectations are exactly 0, i.e., $\mathbb{E}(C_{\cdot n,ij}^\ast |\mathcal{F}_j^\ast) = \mathbb{E}C_{\cdot n,ij}^\ast = 0$. Then the decompositions can be simplified as 
$$ C_{11n}^\ast(w,x) - C_{12n}^\ast(w,x) = \frac{1}{n}\sum_{i=1}^n \mathbb{E} (C_{11n,ij}^\ast|\mathcal{F}_i^\ast) + U_n^{(2)}(\pi_2 g_{C_{11n}^\ast}) - \frac{1}{n}\sum_{i=1}^n \mathbb{E}(C_{12n,ij}^\ast|\mathcal{F}_i^\ast) - U_n^{(2)}(\pi_2 g_{C_{12n}^\ast}), $$ where the conditional expectations are
\begin{align*}
    \mathbb{E}(C_{11n,ij}^\ast|\mathcal{F}_i^\ast) =& \mathrm{i}w V_i e^{\mathrm{i} (wW_i+x'X_i)} \frac{1-D_i}{\frac{1}{2} f(X_i,0) \Delta_p(X_i)} K_h(R_i) \delta_i \mathbb{E}\left[ Y_j K_{h_x}(X_j-X_i) K_{h_r}(R_j)\delta_j |X_i \right], \\
    \mathbb{E}(C_{12n,ij}^\ast|\mathcal{F}_i^\ast) =& \mathrm{i}w V_i e^{\mathrm{i} (wW_i+x'X_i)} \frac{\tau(X_i)(1-D_i)}{\frac{1}{2} f(X_i,0) \Delta_p(X_i)} K_h(R_i) \delta_i \mathbb{E}\left[ D_j K_{h_x}(X_j-X_i) K_{h_r}(R_j)\delta_j | X_i \right] .
\end{align*}
By Lemma S.A1, they are
\begin{align*}
    \mathbb{E}(C_{11n,ij}^\ast|\mathcal{F}_i^\ast) =& \mathrm{i}w V_i e^{\mathrm{i} (wW_i+x'X_i)} (1-D_i) \tau(X_i) K_h(R_i) \delta_i + O_p(h_r) + O_p(h_x^l) ,\\
    \mathbb{E}(C_{12n,ij}^\ast|\mathcal{F}_i^\ast) =& \mathrm{i}w V_i e^{\mathrm{i} (wW_i+x'X_i)} (1-D_i) \tau(X_i) K_h(R_i) \delta_i + O_p(h_r) + O_p(h_x^l) .
\end{align*}
Then $\frac{1}{n} \sum_{i=1}^n \mathbb{E}(C_{11n,ij}^\ast|\mathcal{F}_i^\ast) - \frac{1}{n} \sum_{i=1}^n \mathbb{E}(C_{12n,ij}^\ast|\mathcal{F}_i^\ast) = O_p(h_r) + O_p(h_x^l)$, which is uniformly bounded by $\frac{1}{\sqrt{nh}}$. For the high-order residual terms in the Hoeffding decompositions, $U_n^{(2)}(\pi_2 g_{C_{11n}^\ast})$ and $U_n^{(2)}(\pi_2 g_{C_{12n}^\ast})$ are constructed on the symmetrized functions $g_{C_{11n}^\ast} := \frac{1}{2} (C_{11n,ij}^\ast + C_{11n,ji}^\ast)$ and $g_{C_{12n}^\ast} := \frac{1}{2} (C_{12n,ij}^\ast + C_{12n,ji}^\ast)$ with envelopes $G_{C_{11n}^\ast} = \sup_{(w,x)\in\Omega} |g_{C_{11n}^\ast}|$ and $G_{C_{12n}^\ast} = \sup_{(w,x)\in\Omega} |g_{C_{12n}^\ast}|$. Then the second moments of the envelopes are bounded by $\mathbb{E} \sup_{(w,x)\in\Omega} |C_{11n,ij}^\ast|^2$ and $\mathbb{E} \sup_{(w,x)\in\Omega} |C_{12n,ij}^\ast|^2$:
\begin{align*}
    \mathbb{E} G_{C_{11n}^\ast}^2 \le& \mathbb{E} \sup_{(w,x)\in\Omega} \bigg| \mathrm{i}w V_i e^{\mathrm{i} (wW_i+x'X_i)} \frac{1-D_i}{\frac{1}{2} f(X_i,0) \Delta_p(X_i)} Y_j K_{h_x}(X_j-X_i) K_{h_r}(R_j) K_h(R_i) \delta_i \bigg|^2 \\
    \lesssim& \mathbb{E} \left[ Y_j K_{h_x}(X_j-X_i) K_{h_r}(R_j) K_h(R_i) \right]^2 = O(\frac{1}{hh_rh_x^d}); \\
    \mathbb{E} G_{C_{12n}^\ast}^2 \le& \mathbb{E} \sup_{(w,x)\in\Omega} \bigg| \mathrm{i}w V_i e^{\mathrm{i} (wW_i+x'X_i)} \frac{\tau(X_i) (1-D_i)}{\frac{1}{2} f(X_i,0) \Delta_p(X_i)} D_j K_{h_x}(X_j-X_i) K_{h_r}(R_j)\delta_j K_h(R_i) \delta_i \bigg|^2 \\
    \lesssim& \mathbb{E} \left[ \tau(X_i) K_{h_x}(X_j-X_i) K_{h_r}(R_j) K_h(R_i) \right]^2 = O(\frac{1}{hh_rh_x^d}). 
\end{align*}
Following Proposition 4 in \cite{Delgado2001},
\begin{align*}
    \mathbb{E} \sup_{(w,x)\in\Omega} \left| U_n^{(2)}(\pi_2 g_{C_{11n}^\ast}) \right|^2 &\lesssim \frac{\mathbb{E} G_{C_{11n}^\ast}^2}{(n-1)^2} = O(\frac{1}{n^2hh_rh_x^d}) = o(\frac{1}{nh}). \\
    \mathbb{E} \sup_{(w,x)\in\Omega} \left| U_n^{(2)}(\pi_2 g_{C_{12n}^\ast}) \right|^2 &\lesssim \frac{\mathbb{E} G_{C_{12n}^\ast}^2}{(n-1)^2} = O(\frac{1}{n^2hh_rh_x^d}) = o(\frac{1}{nh}). 
\end{align*}
Then $U_n^{(2)}(\pi_2 g_{C_{11n}^\ast})$ and $U_n^{(2)}(\pi_2 g_{C_{12n}^\ast})$ are uniformly bounded by $\frac{1}{\sqrt{nh}}$ over $\Omega$ and $C_{11n}^\ast(w,x) - C_{12n}^\ast(w,x) = o_p(\frac{1}{\sqrt{nh}})$.

Term II in \eqref{eq1} can be rewritten as 
\begin{equation}
\begin{aligned}
    \text{II} =& \frac{\mathrm{i}w}{n} \sum_{i=1}^n V_i e^{\mathrm{i} x'X_i} (\hat\kappa(w, X_i)-\kappa(w, X_i)) (Y_i-D_i\tau(X_i)) K_{h_r}(R_i) \delta_i \\
    &+ \frac{\mathrm{i}w}{n} \sum_{i=1}^n V_i D_i \kappa(w, X_i) (\hat\tau(X_i) - \tau(X_i)) e^{\mathrm{i} x'X_i} K_{h_r}(R_i) \delta_i \\
    &+ \frac{\mathrm{i}w}{n} \sum_{i=1}^n V_i e^{\mathrm{i} x'X_i} D_i (\hat\kappa(w, X_i)-\kappa(w, X_i)) (\hat\tau(X_i) - \tau(X_i)) K_{h_r}(R_i) \delta_i. \\
    =& C_{2n}^\ast(w,x) + C_{3n}^\ast(w,x) + C_{4n}^\ast(w,x) .
\label{II}
\end{aligned}
\end{equation}
The first term, $C_{2n}^\ast(w,x)$, arises from estimating the function $\kappa$ with $$\hat\kappa(w,X_i) = \frac{\frac{1}{n-1} \sum_{j\ne i} e^{\mathrm{i}w\hat{W}_j} (1-D_j) K_{h_x}(X_j-X_i) K_{h_r}(R_j) \delta_j}{\frac{1}{n-1} \sum_{j\ne i} D_j K_{h_x}(X_j-X_i) K_{h_r}(R_j) \delta_j} = \frac{\hat{f}(X_i, 0) \Delta_{\hat\varphi}(w, X_i)}{\hat{f}(X_i, 0) \Delta_{\hat{p}}(X_i)}.$$ 
Using the expansion of $\hat{a}/\hat{b}$, $\hat\kappa(w,X_i) - \kappa(w,X_i)$ can be expanded as 
\begin{align*}
    \hat\kappa(w,X_i) - \kappa(w,X_i) =& \frac{\hat{f}(X_i,0) \Delta_{\hat\varphi}(w,X_i)}{\frac{1}{2} f(X_i,0) \Delta_p(X_i)} - \kappa(w,X_i) \frac{\hat{f}(X_i,0) \Delta_{\hat{p}}(X_i)}{\frac{1}{2} f(X_i,0) \Delta_p(X_i)} \\
    &- \left( \frac{\hat{f}(X_i,0) \Delta_{\hat\varphi}(w,X_i)}{\frac{1}{2} f(X_i,0) \Delta_p(X_i)} - \kappa(w,X_i) \right) \left( \frac{\hat{f}(X_i,0) \Delta_{\hat{p}}(X_i)}{\frac{1}{2} f(X_i,0) \Delta_p(X_i)} - 1 \right) \\
    &+ \kappa(w,X_i) \left( \frac{\hat{f}(X_i,0) \Delta_{\hat{p}}(X_i)}{\frac{1}{2} f(X_i,0) \Delta_p(X_i)} -1 \right)^2 - \hat\kappa(w,X_i) \left( \frac{\hat{f}(X_i,0) \Delta_{\hat{p}}(X_i)}{\frac{1}{2} f(X_i,0) \Delta_p(X_i)} -1 \right)^3 \\
    &+ \left( \frac{\hat{f}(X_i,0) \Delta_{\hat\varphi}(w,X_i)}{\frac{1}{2} f(X_i,0) \Delta_p(X_i)} - \kappa(w,X_i) \right) \left( \frac{\hat{f}(X_i,0) \Delta_{\hat{p}}(X_i)}{\frac{1}{2} f(X_i,0) \Delta_p(X_i)} -1 \right)^2,
\end{align*}
where the last two terms are asymptotically negligible compared to the first four terms.
Then the leading terms of $C_{2n}^\ast(w,x)$ can be expanded as $C_{2n}^\ast(w,x) = C_{21n}^\ast(w,x) - C_{22n}^\ast(w,x) - C_{23n}^\ast(w,x) + C_{24n}^\ast(w,x)$, where the components are
\begin{align*}
    C_{21n}^\ast(w,x) =& \frac{\mathrm{i}w}{n} \sum_{i=1}^n V_i e^{\mathrm{i} x'X_i} \frac{\hat{f}(X_i, 0) \Delta_{\hat\varphi}(w, X_i)}{\frac{1}{2} f(X_i,0) \Delta_p(X_i)} (Y_i-D_i\tau(X_i)) K_{h_r}(R_i) \delta_i \\
    =& \frac{\mathrm{i}w}{n(n-1)} \sum_{i\ne j} V_i e^{\mathrm{i} x'X_i} \frac{Y_i-D_i\tau(X_i)}{\frac{1}{2} f(X_i,0) \Delta_p(X_i)} e^{\mathrm{i}w\hat{W}_j} (1-D_j) K_{h_x}(X_j-X_i) K_{h_r}(R_j) \delta_j K_{h_r}(R_i) \delta_i, 
\end{align*}
\begin{align*}
    C_{22n}^\ast(w,x) =& \frac{\mathrm{i}w}{n} \sum_{i=1}^n V_i e^{\mathrm{i} x'X_i} \kappa(w,X_i) \frac{\hat{f}(X_i, 0) \Delta_{\hat{p}}(X_i)}{\frac{1}{2} f(X_i,0) \Delta_p(X_i)} (Y_i-D_i\tau(X_i)) K_{h_r}(R_i) \delta_i \\
    =& \frac{\mathrm{i}w}{n(n-1)} \sum_{i\ne j} V_i e^{\mathrm{i} x'X_i} \kappa(w,X_i) \frac{Y_i-D_i\tau(X_i)}{\frac{1}{2} f(X_i,0) \Delta_p(X_i)} D_j K_{h_x}(X_j-X_i) K_{h_r}(R_j) \delta_j K_{h_r}(R_i) \delta_i, \\
    C_{23n}^\ast(w,x) =& \frac{\mathrm{i}w}{n(n-1)^2} \sum_{i\ne j, i\ne k} V_i e^{\mathrm{i} x'X_i} \left( \frac{(1-D_j) e^{\mathrm{i}w\hat{W}_j} K_{h_x}(X_j-X_i) K_{h_r}(R_j) \delta_j}{\frac{1}{2} f(X_i,0) \Delta_p(X_i)} - \kappa(w,X_i) \right) \\& \left( \frac{D_k K_{h_x}(X_k-X_i) K_{h_r}(R_k) \delta_k}{\frac{1}{2} f(X_i,0) \Delta_p(X_i)} - 1 \right) (Y_i-D_i\tau(X_i)) K_{h_r}(R_i) \delta_i, \\
    C_{24n}^\ast(w,x) =& \frac{\mathrm{i}w}{n(n-1)^2} \sum_{i\ne j, i\ne k} V_i e^{\mathrm{i} x'X_i} \kappa(w,X_i) \left( \frac{D_j K_{h_x}(X_j-X_i) K_{h_r}(R_j) \delta_j}{\frac{1}{2} f(X_i,0) \Delta_p(X_i)} - 1 \right) \\& \left( \frac{D_k K_{h_x}(X_k-X_i) K_{h_r}(R_k) \delta_k}{\frac{1}{2} f(X_i,0) \Delta_p(X_i)} - 1 \right) (Y_i-D_i\tau(X_i)) K_{h_r}(R_i) \delta_i.
\end{align*}
For $C_{21n}^\ast(w,x)$, due to the existence of $\hat{W}_j$, we separate it into the infeasible term $$C_{211n}^\ast(w,x) =\frac{\mathrm{i}w}{n(n-1)} \sum_{i\ne j} V_i e^{\mathrm{i} x'X_i} \frac{Y_i-D_i\tau(X_i)}{\frac{1}{2} f(X_i,0) \Delta_p(X_i)} e^{\mathrm{i}wW_j} (1-D_j) K_{h_x}(X_j-X_i) K_{h_r}(R_j) \delta_j K_{h_r}(R_i) \delta_i $$ and the estimation-effect term of $W_j$ as $$\frac{\mathrm{i}w}{n(n-1)} \sum_{i\ne j} V_i e^{\mathrm{i} x'X_i} \frac{Y_i-D_i\tau(X_i)}{\frac{1}{2} f(X_i,0) \Delta_p(X_i)} \left( e^{\mathrm{i}w\hat{W}_j} - e^{\mathrm{i}wW_j} \right) (1-D_j) K_{h_x}(X_j-X_i) K_{h_r}(R_j) \delta_j K_{h_r}(R_i) \delta_i. $$
By a Taylor expansion, the estimation effect of $e^{\mathrm{i}w\hat{W}_j}$ can be further expanded as $$\frac{(\mathrm{i}w)^2}{n(n-1)} \sum_{i\ne j} V_i e^{\mathrm{i} x'X_i} \frac{Y_i-D_i\tau(X_i)}{\frac{1}{2} f(X_i,0) \Delta_p(X_i)} e^{\mathrm{i}wW_j} (1-D_j)^2 (\hat\tau(X_j)-\tau(X_j)) K_{h_x}(X_j-X_i) K_{h_r}(R_j) \delta_j K_{h_r}(R_i) \delta_i$$ plus $$\frac{(\mathrm{i}w)^3}{2n(n-1)} \sum_{i\ne j} V_i e^{\mathrm{i} x'X_i} \frac{Y_i-D_i\tau(X_i)}{\frac{1}{2} f(X_i,0) \Delta_p(X_i)} e^{\mathrm{i}w\bar{W}_j} (1-D_j)^3 (\hat\tau(X_j)-\tau(X_j))^2 K_{h_x}(X_j-X_i) K_{h_r}(R_j) \delta_j K_{h_r}(R_i) \delta_i, $$
where the first leading term is defined as
\begin{align*}
    C_{212n}^\ast(w,x) =& \frac{(\mathrm{i}w)^2}{n(n-1)} \sum_{i\ne j} V_i e^{\mathrm{i} x'X_i} \frac{Y_i-D_i\tau(X_i)}{\frac{1}{2} f(X_i,0) \Delta_p(X_i)} e^{\mathrm{i}wW_j} (1-D_j)^2 (\hat\tau(X_j)-\tau(X_j)) K_{h_x}(X_j-X_i) \\& K_{h_r}(R_j) \delta_j K_{h_r}(R_i) \delta_i, 
\end{align*} 
and the latter term is of smaller order than $C_{212n}^\ast(w,x)$.

$U$-process $C_{211n}^\ast(w,x) = U_n^{(2)}(C_{211n,ij}^\ast)$ with element $C_{211n,ij}^\ast = \mathrm{i}w V_i e^{\mathrm{i} x'X_i} \frac{Y_i-D_i\tau(X_i)}{\frac{1}{2} f(X_i,0) \Delta_p(X_i)} e^{\mathrm{i}wW_j}$ \\ $(1-D_j) K_{h_x}(X_j-X_i) K_{h_r}(R_j) \delta_j K_{h_r}(R_i) \delta_i$ has Hoeffding decomposition $$C_{211n}^\ast(w,x) = \frac{1}{n} \sum_{i=1}^n \mathbb{E}(C_{211n,ij}^\ast|\mathcal{F}_i^\ast) + U_n^{(2)}(\pi_2 g_{C_{211n}^\ast}), $$
where the conditional expectation term is 
\begin{align*}
    \mathbb{E}(C_{211n,ij}^\ast|\mathcal{F}_i^\ast) =& \mathrm{i}w V_i e^{\mathrm{i} x'X_i} \frac{Y_i-D_i\tau(X_i)}{\frac{1}{2} f(X_i,0) \Delta_p(X_i)} \mathbb{E}[ e^{\mathrm{i}wW_j} (1-D_j) K_{h_x}(X_j-X_i) K_{h_r}(R_j) \delta_j |X_i] K_{h_r}(R_i) \delta_i, \\
    =& \mathrm{i}w V_i e^{\mathrm{i} x'X_i} \kappa(w,X_i) (Y_i-D_i\tau(X_i)) K_{h_r}(R_i) \delta_i + O_p(h_r) + O_p(h_x^l).
\end{align*}
Then $C_{211n}^\ast(w,x)$ can be written as $$ C_{211n}^\ast(w,x) = \frac{1}{n} \sum_{i=1}^n \mathrm{i}w V_i e^{\mathrm{i} x'X_i} \kappa(w,X_i) (Y_i-D_i\tau(X_i)) K_{h_r}(R_i) \delta_i + U_n^{(2)}(\pi_2 g_{C_{211n}^\ast}) + O_p(h_r) + O_p(h_x^l). $$ 
    
For the high-order term $U_n^{(2)}(\pi_2 g_{C_{211n}^\ast})$ in the Hoeffding decomposition, $g_{C_{211n}^\ast} := \frac{1}{2} (C_{211n,ij}^\ast + C_{211n,ji}^\ast)$ is the symmetrized function with envelope $G_{C_{211n}^\ast} = \sup_{(w,x)\in\Omega} |g_{C_{211n}^\ast}|$. The second moment of the envelope is bounded by $\mathbb{E} \sup_{(w,x)\in\Omega} |C_{211n,ij}^\ast|^2$:
\begin{align*}
    \mathbb{E} G_{C_{211n}^\ast}^2 \le& \mathbb{E} \sup_{(w,x)\in\Omega} \bigg| V_i e^{\mathrm{i} x'X_i} \frac{Y_i-D_i\tau(X_i)}{\frac{1}{2} f(X_i,0) \Delta_p(X_i)} e^{\mathrm{i}wW_j} (1-D_j) K_{h_x}(X_j-X_i) K_{h_r}(R_j) \delta_j K_{h_r}(R_i) \delta_i \bigg|^2 \\
    \lesssim& \mathbb{E} \left[ (Y_i-D_i\tau(X_i)) K_{h_x}(X_j-X_i) K_{h_r}(R_j) K_{h_r}(R_i) \right]^2 = O(\frac{1}{h_r^2h_x^d}).
\end{align*}
Following Proposition 4 in \cite{Delgado2001}, $$\mathbb{E} \sup_{(w,x)\in\Omega} \left| U_n^{(2)}(\pi_2 g_{C_{211n}^\ast}) \right|^2 \lesssim \frac{\mathbb{E} G_{C_{211n}^\ast}^2}{(n-1)^2} = O(\frac{1}{n^2h_r^2h_x^d}) = o(\frac{1}{nh}). $$ Then $U_n^{(2)}(\pi_2 g_{C_{211n}^\ast})$ is uniformly bounded by $\frac{1}{\sqrt{nh}}$ over $(w,x)\in\Omega$ and $C_{211n}^\ast(w,x)$ can be written as $$ C_{211n}^\ast(w,x) = \frac{\mathrm{i}w}{n} \sum_{i=1}^n V_i e^{\mathrm{i} x'X_i} \kappa(w,X_i) (Y_i-D_i\tau(X_i)) K_{h_r}(R_i) \delta_i + o_p(\frac{1}{\sqrt{nh}}).$$

Expanding $\hat\tau(X_j)-\tau(X_j)$, $C_{212n}^\ast(w,x)$ can be written as
\begin{align*}
    C_{212n}^\ast(w,x) =& \frac{(\mathrm{i}w)^2}{n(n-1)} \sum_{i\ne j} V_i e^{\mathrm{i} x'X_i} \frac{Y_i-D_i\tau(X_i)}{\frac{1}{2} f(X_i,0) \Delta_p(X_i)} e^{\mathrm{i}wW_j} (1-D_j) \left( \frac{\hat{f}(X_j,0) \Delta_{\hat\mu}(X_j)}{\frac{1}{2} f(X_j,0) \Delta_p(X_j)} \right. \\&- \left. \tau(X_j) \frac{\hat{f}(X_j,0) \Delta_{\hat{p}}(X_j)}{\frac{1}{2} f(X_j,0) \Delta_p(X_j)} \right) K_{h_x}(X_j-X_i) K_{h_r}(R_j) \delta_j K_{h_r}(R_i) \delta_i + o_p(\frac{1}{\sqrt{nh}}), 
\end{align*} 
where the high-order terms in the expansion of $\hat\tau(X_i) - \tau(X_i)$ are of smaller order than the first-order terms in the expansion. Further expanding $\hat{f}\Delta_{\hat\mu}$ and $\hat{f}\Delta_{\hat{p}}$ in the numerator, $C_{212n}^\ast(w,x)$ can be decomposed as 
\begin{align*}
    C_{212n}^\ast(w,x) =& \frac{1}{n-1} C_{2121n}^\ast(w,x) + \frac{n-2}{n-1} C_{2122n}^\ast(w,x) - \frac{n-2}{n-1} C_{2123n}^\ast(w,x) + o_p(\frac{1}{\sqrt{nh}}),
\end{align*} 
where the components $C_{2121n}^\ast(w,x) = U_n^{(2)} (C_{2121n,ij}^\ast)$, $C_{2122n}^\ast(w,x) = U_n^{(3)} (C_{2122n,ijk}^\ast)$, $C_{2123n}^\ast(w,x) = U_n^{(3)} (C_{2123n,ijk}^\ast)$ are second- and third-order $U$-processes with elements 
\begin{align*}
    C_{2121n,ij}^\ast =& (\mathrm{i}w)^2 V_i e^{\mathrm{i} x'X_i} \frac{Y_i-D_i\tau(X_i)}{\frac{1}{2} f(X_i,0) \Delta_p(X_i)} \frac{e^{\mathrm{i}wW_j} (1-D_j)}{\frac{1}{2} f(X_j,0) \Delta_p(X_j)} (Y_i-D_i\tau(X_j)) \\& K_{h_x}^2(X_i-X_j) K_{h_r}(R_j) \delta_j K_{h_r}^2(R_i) \delta_i^2 \\
    C_{2122n,ijk}^\ast =& (\mathrm{i}w)^2 V_i e^{\mathrm{i} x'X_i} \frac{Y_i-D_i\tau(X_i)}{\frac{1}{2} f(X_i,0) \Delta_p(X_i)} \frac{e^{\mathrm{i}wW_j} (1-D_j)}{\frac{1}{2} f(X_j,0) \Delta_p(X_j)} Y_k \\& K_{h_x}(X_k-X_j) K_{h_x}(X_j-X_i) K_{h_r}(R_k) \delta_k K_{h_r}(R_j) \delta_j K_{h_r}(R_i) \delta_i \\
    C_{2123n,ijk}^\ast =& (\mathrm{i}w)^2 V_i e^{\mathrm{i} x'X_i} \frac{Y_i-D_i\tau(X_i)}{\frac{1}{2} f(X_i,0) \Delta_p(X_i)} \frac{e^{\mathrm{i}wW_j} (1-D_j)}{\frac{1}{2} f(X_j,0) \Delta_p(X_j)} \tau(X_j) D_k \\& K_{h_x}(X_k-X_j) K_{h_x}(X_j-X_i) K_{h_r}(R_k) \delta_k K_{h_r}(R_j) \delta_j K_{h_r}(R_i) \delta_i.
\end{align*}
Their Hoeffding decompositions are
\begin{align*}
     C_{2121n}^\ast(w,x) =& 2U_n^{(1)} (\pi_1 C_{2121n,ij}^\ast) + U_n^{(2)} (\pi_2 C_{2121n,ij}^\ast), \\
     C_{2122n}^\ast(w,x) =& \frac{1}{n} \sum_{i=1}^n \mathbb{E}(C_{2122n,ijk}^\ast |\mathcal{F}_i^\ast) + 3U_n^{(2)} (\pi_2 C_{2122n,ijk}^\ast) + U_n^{(3)} (\pi_3 C_{2122n,ijk}^\ast), \\ 
     C_{2123n}^\ast(w,x) =& \frac{1}{n} \sum_{i=1}^n \mathbb{E}(C_{2123n,ijk}^\ast |\mathcal{F}_i^\ast) + 3U_n^{(2)} (\pi_2 C_{2123n,ijk}^\ast) + U_n^{(3)} (\pi_3 C_{2123n,ijk}^\ast),  
\end{align*}
For the $U$-process $C_{2121n}^\ast(w,x)$, $g_{C_{2121n}^\ast} := \frac{1}{2} (C_{2121n,ij}^\ast + C_{2121n,ji}^\ast)$ is the symmetrized function with envelope $G_{C_{2121n}^\ast} = \sup_{(w,x)\in\Omega} |g_{C_{2121n}^\ast}|$, whose second moment is bounded by $\mathbb{E} \sup_{(w,x)\in\Omega} |C_{2121n,ij}^\ast|^2$:
\begin{align*}
    \mathbb{E} G_{C_{2121n}^\ast}^2 \le& \mathbb{E} \sup_{(w,x)\in\Omega} \bigg| (\mathrm{i}w)^2 V_i e^{\mathrm{i} x'X_i} \frac{Y_i-D_i\tau(X_i)}{\frac{1}{2} f(X_i,0) \Delta_p(X_i)} \frac{e^{\mathrm{i}wW_j} (1-D_j)}{\frac{1}{2} f(X_j,0) \Delta_p(X_j)} (Y_i-D_i\tau(X_j)) \\& K_{h_x}^2(X_i-X_j) K_{h_r}(R_j) \delta_j K_{h_r}^2(R_i) \delta_i^2 \bigg|^2 \\
    \lesssim& \mathbb{E} \left[ (Y_i-D_i\tau(X_i)) (Y_i-D_i\tau(X_j)) K_{h_x}^2(X_j-X_i) K_{h_r}(R_j) K_{h_r}^2(R_i) \right]^2 = O(\frac{1}{h_r^4h_x^{3d}}).
\end{align*}
Following Proposition 4 in \cite{Delgado2001}, 
\begin{align*}
    \mathbb{E} \sup_{(w,x)\in\Omega} \left| \frac{2}{n-1} U_n^{(1)} (\pi_1 C_{2121n,ij}^\ast) \right|^2 \lesssim& \frac{\mathbb{E} G_{C_{2121n}^\ast}^2}{n(n-1)^2} = O(\frac{1}{n^3h_r^4h_x^{3d}}) = o(\frac{1}{nh}). \\
    \mathbb{E} \sup_{(w,x)\in\Omega} \left| \frac{1}{n-1} U_n^{(2)} (\pi_2 C_{2121n,ij}^\ast) \right|^2 \lesssim& \frac{\mathbb{E} G_{C_{2121n}^\ast}^2}{(n-1)^4} = O(\frac{1}{n^4h_r^4h_x^{3d}}) = o(\frac{1}{nh}). 
\end{align*}
Then $\frac{2}{n-1} U_n^{(1)} (\pi_1 C_{2121n,ij}^\ast)$ and $\frac{1}{n-1} U_n^{(2)} (\pi_2 C_{2121n,ij}^\ast)$ are uniformly bounded by $\frac{1}{\sqrt{nh}}$ over $(w,x)\in\Omega$ and $\frac{1}{n-1} C_{2121n}^\ast(w,x)$ is $o_p(\frac{1}{\sqrt{nh}})$ uniformly over $(w,x)\in\Omega$.
The conditional expectation terms in the decompositions of $C_{2122n}^\ast(w,x)$ and $C_{2123n}^\ast(w,x)$ are
\begin{align*}
    & \mathbb{E}(C_{2122n,ijk}^\ast | \mathcal{F}_i^\ast) = \mathbb{E} \left[ \mathbb{E}(C_{2122n,ijk}^\ast | \mathcal{F}_{ij}^\ast) |\mathcal{F}_i^\ast \right] \\
    =& (\mathrm{i}w)^2 V_i e^{\mathrm{i} x'X_i} \frac{Y_i-D_i\tau(X_i)}{\frac{1}{2} f(X_i,0) \Delta_p(X_i)} \mathbb{E} \left[ e^{\mathrm{i}wW_j} (1-D_j) K_{h_x}(X_j-X_i) K_{h_r}(R_j) \delta_j | X_j \right] K_{h_r}(R_i) \delta_i \\& (1+O_p(h_r)+O_p(h_x^l)) \\
    =& (\mathrm{i}w)^2 V_i e^{\mathrm{i} x'X_i} (Y_i-D_i\tau(X_i)) \kappa(w,X_i) K_{h_r}(R_i) \delta_i (1+O_p(h_r)+O_p(h_x^l)) .\\
    & \mathbb{E}(C_{2123n,ijk}^\ast|\mathcal{F}_i^\ast) = \mathbb{E} \left[ \mathbb{E}(C_{2123n,ijk}^\ast | \mathcal{F}_{ij}^\ast) |\mathcal{F}_i^\ast \right] \\
    =& (\mathrm{i}w)^2 V_i e^{\mathrm{i} x'X_i} \frac{Y_i-D_i\tau(X_i)}{\frac{1}{2} f(X_i,0) \Delta_p(X_i)} \mathbb{E} \left[ e^{\mathrm{i}wW_j} (1-D_j) \tau(X_j) K_{h_x}(X_j-X_i) K_{h_r}(R_j) \delta_j | X_j \right] K_{h_r}(R_i) \delta_i \\& (1+O_p(h_r)+O_p(h_x^l)) \\
    =& (\mathrm{i}w)^2 V_i e^{\mathrm{i} x'X_i} (Y_i-D_i\tau(X_i)) \kappa(w,X_i) \tau(X_i) K_{h_r}(R_i) \delta_i (1+O_p(h_r)+O_p(h_x^l)).
\end{align*}

For the high-order terms $U_n^{(2)}(\pi_2 g_{C_{2122n}^\ast})$, $U_n^{(2)}(\pi_2 g_{C_{2123n}^\ast})$, $U_n^{(3)}(\pi_3 g_{C_{2122n}^\ast})$ and $U_n^{(3)}(\pi_3 g_{C_{2123n}^\ast})$ in the Hoeffding decompositions, $g_{C_{2122n}^\ast} := \frac{1}{3!} \sum_{[3!]} C_{2122n,ijk}^\ast$ and $g_{C_{2123n}^\ast} := \frac{1}{3!} \sum_{[3!]} C_{2123n,ijk}^\ast$ are the symmetrized functions with envelopes $G_{C_{2122n}^\ast} = \sup_{(w,x)\in\Omega} |g_{C_{2122n}^\ast}|$ and $G_{C_{2123n}^\ast} = \sup_{(w,x)\in\Omega} |g_{C_{2123n}^\ast}|$, whose second moments are bounded by $\mathbb{E} \sup_{(w,x)\in\Omega} |C_{2122n,ij}^\ast|^2$ and $\mathbb{E} \sup_{(w,x)\in\Omega} |C_{2123n,ij}^\ast|^2$:
\begin{align*}
    \mathbb{E} G_{C_{2122n}^\ast}^2 \le& \mathbb{E} \sup_{(w,x)\in\Omega} \bigg| (\mathrm{i}w)^2 V_i e^{\mathrm{i} x'X_i} \frac{Y_i-D_i\tau(X_i)}{\frac{1}{2} f(X_i,0) \Delta_p(X_i)} \frac{e^{\mathrm{i}wW_j} (1-D_j)}{\frac{1}{2} f(X_j,0) \Delta_p(X_j)} Y_k K_{h_x}(X_k-X_j) \\& K_{h_x}(X_j-X_i) K_{h_r}(R_k) \delta_k K_{h_r}(R_j) \delta_j K_{h_r}(R_i) \delta_i \bigg|^2 \\
    \lesssim& \mathbb{E} \left[ (Y_i-D_i\tau(X_i)) Y_k K_{h_x}(X_k-X_j) K_{h_x}(X_j-X_i) K_{h_r}(R_k) K_{h_r}(R_j) K_{h_r}(R_i) \right]^2 \\=& O(\frac{1}{h_r^3h_x^{2d}}). 
\end{align*}
\begin{align*}
    \mathbb{E} G_{C_{2123n}^\ast}^2 \le& \mathbb{E} \sup_{(w,x)\in\Omega} \bigg| (\mathrm{i}w)^2 V_i e^{\mathrm{i} x'X_i} \frac{Y_i-D_i\tau(X_i)}{\frac{1}{2} f(X_i,0) \Delta_p(X_i)} \frac{e^{\mathrm{i}wW_j} (1-D_j)}{\frac{1}{2} f(X_j,0) \Delta_p(X_j)} \tau(X_j) D_k K_{h_x}(X_k-X_j) \\& K_{h_x}(X_j-X_i) K_{h_r}(R_k) \delta_k K_{h_r}(R_j) \delta_j K_{h_r}(R_i) \delta_i \bigg|^2 \\
    \lesssim& \mathbb{E} \left[ (Y_i-D_i\tau(X_i)) \tau(X_j) K_{h_x}(X_k-X_j) K_{h_x}(X_j-X_i) K_{h_r}(R_k) K_{h_r}(R_j) K_{h_r}(R_i) \right]^2 \\=& O(\frac{1}{h_r^3h_x^{2d}}).
\end{align*}
Following Proposition 4 in \cite{Delgado2001}, 
\begin{align*}
    \mathbb{E} \sup_{(w,x)\in\Omega} \left| 3\frac{n-2}{n-1} U_n^{(2)}(\pi_2 g_{C_{2122n}^\ast}) \right|^2 \lesssim \frac{\mathbb{E} G_{C_{2122n}^\ast}^2}{(n-1)^2} = O(\frac{1}{n^2h_r^3h_x^{2d}}) = o(\frac{1}{nh}). \\
    \mathbb{E} \sup_{(w,x)\in\Omega} \left| 3\frac{n-2}{n-1} U_n^{(2)}(\pi_2 g_{C_{2123n}^\ast}) \right|^2 \lesssim \frac{\mathbb{E} G_{C_{2123n}^\ast}^2}{(n-1)^2} = O(\frac{1}{n^2h_r^3h_x^{2d}}) = o(\frac{1}{nh}). \\
    \mathbb{E} \sup_{(w,x)\in\Omega} \left| \frac{n-2}{n-1} U_n^{(3)}(\pi_3 g_{C_{2122n}^\ast}) \right|^2 \lesssim \frac{n\mathbb{E} G_{C_{2122n}^\ast}^2}{(n-1)^4} = O(\frac{1}{n^3h_r^3h_x^{2d}}) = o(\frac{1}{nh}). \\
    \mathbb{E} \sup_{(w,x)\in\Omega} \left| \frac{n-2}{n-1} U_n^{(3)}(\pi_3 g_{C_{2123n}^\ast}) \right|^2 \lesssim \frac{n\mathbb{E} G_{C_{2123n}^\ast}^2}{(n-1)^4} = O(\frac{1}{n^3h_r^3h_x^{2d}}) = o(\frac{1}{nh}). 
\end{align*}
Then $U_n^{(3)}(\pi_3 g_{C_{2122n}^\ast})$ and $U_n^{(3)}(\pi_3 g_{C_{2123n}^\ast})$ are uniformly bounded by $\frac{1}{\sqrt{nh}}$ over $(w,x)\in\Omega$. 
Summarizing the above results, $\sup_{(w,x)\in\Omega} \left| C_{212n}^\ast(w,x) \right| = o_p(\frac{1}{\sqrt{nh}})$ and $$ C_{21n}^\ast(w,x) = \frac{\mathrm{i}w}{n} \sum_{i=1}^n V_i e^{\mathrm{i} x'X_i} \kappa(w,X_i) (Y_i-D_i\tau(X_i)) K_{h_r}(R_i) \delta_i + o_p(\frac{1}{\sqrt{nh}}).$$

Following steps similar to those in the proof of Lemma \ref{Lemma1}, $U$-process $C_{22n}^\ast(w,x) = U_n^{(2)}(C_{22n,ij}^\ast)$ with element $C_{22n,ij}^\ast = V_i e^{\mathrm{i} x'X_i} \kappa(w,X_i) \frac{Y_i-D_i\tau(X_i)}{\frac{1}{2} f(X_i,0) \Delta_p(X_i)} D_j K_{h_x}(X_j-X_i) K_{h_r}(R_j) \delta_j K_{h_r}(R_i) \delta_i$ has Hoeffding decomposition $$C_{22n}^\ast(w,x) = \frac{1}{n} \sum_{i=1}^n \mathbb{E}(C_{22n,ij}^\ast|\mathcal{F}_i^\ast) + U_n^{(2)}(\pi_2 g_{C_{22n}^\ast}), $$
where the conditional expectation term is 
\begin{align*}
    \mathbb{E}(C_{22n,ij}^\ast|\mathcal{F}_i^\ast) =& \mathrm{i}w V_i e^{\mathrm{i} x'X_i} \frac{Y_i-D_i\tau(X_i)}{\frac{1}{2} f(X_i,0) \Delta_p(X_i)} \mathbb{E}[ e^{\mathrm{i}wW_j} (1-D_j) K_{h_x}(X_j-X_i) K_{h_r}(R_j) \delta_j |X_i] K_{h_r}(R_i) \delta_i, \\
    =& \mathrm{i}w V_i e^{\mathrm{i} x'X_i} \kappa(w,X_i) (Y_i-D_i\tau(X_i)) K_{h_r}(R_i) \delta_i + O_p(h_r) + O_p(h_x^l).
\end{align*}
Then $C_{22n}^\ast(w,x)$ can be written as $$ C_{22n}^\ast(w,x) = \frac{1}{n} \sum_{i=1}^n \mathrm{i}w V_i e^{\mathrm{i} x'X_i} \kappa(w,X_i) (Y_i-D_i\tau(X_i)) K_{h_r}(R_i) \delta_i + U_n^{(2)}(\pi_2 g_{C_{22n}^\ast}) + O_p(h_r) + O_p(h_x^l). $$ 
For the high-order term $U_n^{(2)}(\pi_2 g_{C_{22n}^\ast})$ in the Hoeffding decomposition, $g_{C_{22n}^\ast} := \frac{1}{2} (C_{22n,ij}^\ast + C_{22n,ji}^\ast)$ is the symmetrized function with envelope $G_{C_{22n}^\ast} = \sup_{(w,x)\in\Omega} |g_{C_{22n}^\ast}|$. Then the second moment of the envelope is bounded by $\mathbb{E} \sup_{(w,x)\in\Omega} |C_{22n,ij}^\ast|^2$:
\begin{align*}
    \mathbb{E} G_{C_{22n}^\ast}^2 \le& \mathbb{E} \sup_{(w,x)\in\Omega} \bigg| V_i e^{\mathrm{i} x'X_i} \kappa(w,X_i) \frac{Y_i-D_i\tau(X_i)}{\frac{1}{2} f(X_i,0) \Delta_p(X_i)} D_j K_{h_x}(X_j-X_i) K_{h_r}(R_j) \delta_j K_{h_r}(R_i) \delta_i \bigg|^2 \\
    \lesssim& \mathbb{E} \left[ (Y_i-D_i\tau(X_i)) K_{h_x}(X_j-X_i) K_{h_r}(R_j) K_{h_r}(R_i) \right]^2 = O(\frac{1}{hh_rh_x^d}).
\end{align*}
Following Proposition 4 in \cite{Delgado2001}, $$\mathbb{E} \sup_{(w,x)\in\Omega} \left| U_n^{(2)}(\pi_2 g_{C_{22n}^\ast}) \right|^2 \lesssim \frac{\mathbb{E} G_{C_{22n}^\ast}^2}{(n-1)^2} = O(\frac{1}{n^2hh_rh_x^d}) = o(\frac{1}{nh}). $$ Then $U_n^{(2)}(\pi_2 g_{C_{22n}^\ast})$ is uniformly bounded by $\frac{1}{\sqrt{nh}}$ over $(w,x)\in\Omega$ and $C_{22n}^\ast(w,x)$ can be written as $$ C_{22n}^\ast(w,x) = \frac{1}{n} \sum_{i=1}^n V_i e^{\mathrm{i} x'X_i} \kappa(w,X_i) (Y_i-D_i\tau(X_i)) K_{h_r}(R_i) \delta_i + o_p(\frac{1}{\sqrt{nh}}).$$
The leading term of $C_{22n}^\ast(w,x)$ is exactly equal to that of $C_{211n}^\ast(w,x)$, and $C_{211n}^\ast(w,x) - C_{22n}^\ast(w,x)$ is uniformly bounded by $\frac{1}{\sqrt{nh}}$. To summarize, $\sup_{(w,x)\in\Omega} \left| C_{2n}^\ast(w,x) \right| = o_p(\frac{1}{\sqrt{nh}})$.

For the third term in the decomposition of $C_{2n}^\ast(w,x)$, it can be further decomposed as $C_{23n}^\ast(w,x) = \frac{1}{n-1} C_{231n}^\ast(w,x) + \frac{n-2}{n-1} C_{232n}^\ast(w,x)$
\begin{align*}
    C_{231n}^\ast(w,x) =& \frac{\mathrm{i}w}{n(n-1)} \sum_{i\ne j, i\ne k} V_i e^{\mathrm{i} x'X_i} \left( \frac{(1-D_j) e^{\mathrm{i}w\hat{W}_j} K_{h_x}(X_j-X_i) K_{h_r}(R_j) \delta_j}{\frac{1}{2} f(X_i,0) \Delta_p(X_i)} - \kappa(w,X_i) \right) \\& \left( \frac{D_j K_{h_x}(X_j-X_i) K_{h_r}(R_j) \delta_j}{\frac{1}{2} f(X_i,0) \Delta_p(X_i)} - 1 \right) (Y_i-D_i\tau(X_i)) K_{h_r}(R_i) \delta_i, \\
    C_{232n}^\ast(w,x) =& \frac{\mathrm{i}w}{n(n-1)(n-2)} \sum_{i\ne j, i\ne k} V_i e^{\mathrm{i} x'X_i} \left( \frac{(1-D_j) e^{\mathrm{i}w\hat{W}_j} K_{h_x}(X_j-X_i) K_{h_r}(R_j) \delta_j}{\frac{1}{2} f(X_i,0) \Delta_p(X_i)} - \kappa(w,X_i) \right) \\& \left( \frac{D_k K_{h_x}(X_k-X_i) K_{h_r}(R_k) \delta_k}{\frac{1}{2} f(X_i,0) \Delta_p(X_i)} - 1 \right) (Y_i-D_i\tau(X_i)) K_{h_r}(R_i) \delta_i. 
\end{align*}

Due to the existence of $\hat{W}_j$, we can separate the estimation effect and divide $C_{231n}^\ast(w,x)$ and $C_{232n}^\ast(w,x)$ into several leading parts:
\begin{align*}
    C_{2311n}^\ast(w,x) =& \frac{\mathrm{i}w}{n(n-1)} \sum_{i\ne j} V_i e^{\mathrm{i} x'X_i} \left( \frac{(1-D_j) e^{\mathrm{i}wW_j} K_{h_x}(X_j-X_i) K_{h_r}(R_j) \delta_j}{\frac{1}{2} f(X_i,0) \Delta_p(X_i)} - \kappa(w,X_i) \right) \\& \left( \frac{D_j K_{h_x}(X_j-X_i) K_{h_r}(R_j) \delta_j}{\frac{1}{2} f(X_i,0) \Delta_p(X_i)} - 1 \right) (Y_i-D_i\tau(X_i)) K_{h_r}(R_i) \delta_i, \\
    C_{2312n}^\ast(w,x) =& \frac{(\mathrm{i}w)^2}{n(n-1)} \sum_{i\ne j} V_i e^{\mathrm{i} x'X_i} \frac{(1-D_j) e^{\mathrm{i}wW_j} K_{h_x}(X_j-X_i) K_{h_r}(R_j) \delta_j}{\frac{1}{2} f(X_i,0) \Delta_p(X_i)} (\hat\tau(X_j)-\tau(X_j)) \\& \left( \frac{D_j K_{h_x}(X_j-X_i) K_{h_r}(R_j) \delta_j}{\frac{1}{2} f(X_i,0) \Delta_p(X_i)} - 1 \right) (Y_i-D_i\tau(X_i)) K_{h_r}(R_i) \delta_i, \\
    C_{2321n}^\ast(w,x) =& \frac{\mathrm{i}w}{n(n-1)(n-2)} \sum_{i\ne j\ne k} V_i e^{\mathrm{i} x'X_i} \left( \frac{(1-D_j) e^{\mathrm{i}wW_j} K_{h_x}(X_j-X_i) K_{h_r}(R_j) \delta_j}{\frac{1}{2} f(X_i,0) \Delta_p(X_i)} - \kappa(w,X_i) \right) \\& \left( \frac{D_k K_{h_x}(X_k-X_i) K_{h_r}(R_k) \delta_k}{\frac{1}{2} f(X_i,0) \Delta_p(X_i)} - 1 \right) (Y_i-D_i\tau(X_i)) K_{h_r}(R_i) \delta_i, \\ 
    C_{2322n}^\ast(w,x) =& \frac{(\mathrm{i}w)^2}{n(n-1)(n-2)} \sum_{i\ne j\ne k} V_i e^{\mathrm{i} x'X_i} \frac{(1-D_j) e^{\mathrm{i}wW_j} K_{h_x}(X_j-X_i) K_{h_r}(R_j) \delta_j}{\frac{1}{2} f(X_i,0) \Delta_p(X_i)} \\& (\hat\tau(X_j)-\tau(X_j)) \left( \frac{D_k K_{h_x}(X_k-X_i) K_{h_r}(R_k) \delta_k}{\frac{1}{2} f(X_i,0) \Delta_p(X_i)} - 1 \right) (Y_i-D_i\tau(X_i)) K_{h_r}(R_i) \delta_i. 
\end{align*}
The high-order terms in the Taylor expansions are of smaller order than $C_{2312n}^\ast(w,x)$ and $C_{2332n}^\ast(w,x)$.

$C_{2311n}^\ast(w,x)$ has Hoeffding decomposition $C_{2311n}^\ast(w,x) = 2U_n^{(1)}(\pi_1 g_{C_{2311n}^\ast}) + U_n^{(2)}(\pi_2 g_{C_{2311n}^\ast})$. $g_{C_{2311n}^\ast} := \frac{1}{2} (C_{2311n,ij}^\ast + C_{2311n,ji}^\ast)$ is the symmetrized function with envelope $G_{C_{2311n}^\ast} = \sup_{(w,x)\in\Omega} |g_{C_{2311n}^\ast}|$, whose second moment is bounded by $\mathbb{E} \sup_{(w,x)\in\Omega} |C_{2311n,ij}^\ast|^2$:
\begin{align*}
    \mathbb{E} G_{C_{2311n}^\ast}^2 \le& \mathbb{E} \sup_{(w,x)\in\Omega} \bigg| \mathrm{i}w V_i e^{\mathrm{i} x'X_i} \left( \frac{(1-D_j) e^{\mathrm{i}wW_j} K_{h_x}(X_j-X_i) K_{h_r}(R_j) \delta_j}{\frac{1}{2} f(X_i,0) \Delta_p(X_i)} - \kappa(w,X_i) \right) \\& \left( \frac{D_j K_{h_x}(X_j-X_i) K_{h_r}(R_j) \delta_j}{\frac{1}{2} f(X_i,0) \Delta_p(X_i)} - 1 \right) (Y_i-D_i\tau(X_i)) K_{h_r}(R_i) \delta_i \bigg|^2 \\
    \lesssim& \mathbb{E} \left[ (Y_i-D_i\tau(X_i)) K_{h_x}(X_j-X_i) K_{h_r}(R_j) K_{h_r}(R_i) \right]^2 \\&+ \mathbb{E} \left[ (Y_i-D_i\tau(X_i)) K_{h_r}(R_i) \right]^2 = O(\frac{1}{h_r^2h_x^d}) + O(\frac{1}{h_r}) = O(\frac{1}{h_r^2h_x^d}).
\end{align*}
Following Proposition 4 in \cite{Delgado2001}, 
\begin{align*}
    \mathbb{E} \sup_{(w,x)\in\Omega} \left| \frac{2}{n-1} U_n^{(1)} (\pi_1 C_{2311n,ij}^\ast) \right|^2 \lesssim& \frac{\mathbb{E} G_{C_{2311n}^\ast}^2}{n(n-1)^2} = O(\frac{1}{n^3h_r^2h_x^d}) = o(\frac{1}{nh}). \\
    \mathbb{E} \sup_{(w,x)\in\Omega} \left| \frac{1}{n-1} U_n^{(2)} (\pi_2 C_{2311n,ij}^\ast) \right|^2 \lesssim& \frac{\mathbb{E} G_{C_{2311n}^\ast}^2}{(n-1)^4} = O(\frac{1}{n^4h_r^2h_x^d}) = o(\frac{1}{nh}). 
\end{align*}
Then $\frac{2}{n-1} U_n^{(1)} (\pi_1 C_{2311n,ij}^\ast)$ and $\frac{1}{n-1} U_n^{(2)} (\pi_2 C_{2311n,ij}^\ast)$ are uniformly bounded by $\frac{1}{\sqrt{nh}}$ over $(w,x)\in\Omega$ and $\frac{1}{n-1} C_{2311n}^\ast(w,x)$ is $o_p(\frac{1}{\sqrt{nh}})$ uniformly over $(w,x)\in\Omega$.

For $C_{2312n}^\ast(w,x)$, note that $D_j (1-D_j) =0 $ for the binary variable $D_j$. Then $C_{2312n}^\ast(w,x)$ can be further written as
\begin{align*}
    &C_{2312n}^\ast(w,x) \\
    =& \frac{(\mathrm{i}w)^2}{n(n-1)} \sum_{i\ne j} V_i e^{\mathrm{i} x'X_i} \frac{(1-D_j) e^{\mathrm{i}wW_j} K_{h_x}(X_j-X_i) K_{h_r}(R_j) \delta_j}{\frac{1}{2} f(X_i,0) \Delta_p(X_i)} (\hat\tau(X_j)-\tau(X_j)) \\& \left( \frac{D_j K_{h_x}(X_j-X_i) K_{h_r}(R_j) \delta_j}{\frac{1}{2} f(X_i,0) \Delta_p(X_i)} - 1 \right) (Y_i-D_i\tau(X_i)) K_{h_r}(R_i) \delta_i, \\
    =& \frac{-(\mathrm{i}w)^2}{n(n-1)} \sum_{i\ne j} V_i e^{\mathrm{i} x'X_i} \frac{(1-D_j) e^{\mathrm{i}wW_j} K_{h_x}(X_j-X_i) K_{h_r}(R_j) \delta_j}{\frac{1}{2} f(X_i,0) \Delta_p(X_i)} (\hat\tau(X_j)-\tau(X_j)) (Y_i-D_i\tau(X_i)) K_{h_r}(R_i) \delta_i, 
\end{align*}
By expanding $\hat\tau(X_j) - \tau(X_j)$ to the first-order terms, $C_{2312n}^\ast(w,x)$ can be expressed as $C_{2312n}^\ast(w,x) = \frac{1}{n-1} C_{23121n}^\ast(w,x) - \frac{n-2}{n-1} C_{23122n}^\ast(w,x)$ where $C_{23121n}^\ast(w,x) = U_n^{(2)}(g_{C_{23121n}^\ast})$ is a second-order $U$-process. $C_{23122n}^\ast(w,x) = U_n^{(3)}(g_{C_{23122n}^\ast})$ is a third-order $U$-process. Their Hoeffding decompositions and elements are:
\begin{align*}
    C_{23121n}^\ast(w,x) =& 2U_n^{(1)}(\pi_1 g_{C_{23121n}^\ast}) + U_n^{(2)}(\pi_2 g_{C_{23121n}^\ast}) ,\\
    C_{23122n}^\ast(w,x) =& 3U_n^{(1)}(\pi_1 g_{C_{23122n}^\ast}) + 3U_n^{(2)}(\pi_2 g_{C_{23122n}^\ast}) + U_n^{(3)}(\pi_3 g_{C_{23122n}^\ast}). \\ 
    C_{23121n,ij}^\ast =& (\mathrm{i}w)^2 V_i e^{\mathrm{i} x'X_i} \frac{(1-D_j) e^{\mathrm{i}wW_j} K_{h_x}(X_j-X_i) K_{h_r}(R_j) \delta_j}{\frac{1}{2} f(X_i,0) \Delta_p(X_i)} \left( \frac{D_j K_{h_x}(X_j-X_i) K_{h_r}(R_j) \delta_j}{\frac{1}{2} f(X_i,0) \Delta_p(X_i)} - 1 \right) \\& \frac{K_{h_x}(X_i-X_j) K_{h_r}(R_i) \delta_i}{\frac{1}{2} f(X_j,0) \Delta_p(X_j)} (Y_i-D_i\tau(X_i)) (Y_i-D_i\tau(X_j)) K_{h_r}(R_i) \delta_i, \\
    C_{23122n,ijk}^\ast =& (\mathrm{i}w)^2 V_i e^{\mathrm{i} x'X_i} \frac{(1-D_j) e^{\mathrm{i}wW_j} K_{h_x}(X_j-X_i) K_{h_r}(R_j) \delta_j}{\frac{1}{2} f(X_i,0) \Delta_p(X_i)} \\& \frac{K_{h_x}(X_k-X_j) K_{h_r}(R_k) \delta_k}{\frac{1}{2} f(X_j,0) \Delta_p(X_j)} (Y_i-D_i\tau(X_i)) (Y_k-D_k\tau(X_j)) K_{h_r}(R_i) \delta_i.    
\end{align*}

For the $U$-process $C_{23121n}^\ast(w,x)$, $g_{C_{23121n}^\ast} := \frac{1}{2} (C_{23121n,ij}^\ast + C_{23121n,ji}^\ast)$ is the symmetrized function with envelope $G_{C_{23121n}^\ast} = \sup_{(w,x)\in\Omega} |g_{C_{23121n}^\ast}|$, whose second moment is bounded by $\mathbb{E} \sup_{(w,x)\in\Omega} |C_{23121n,ij}^\ast|^2$:
\begin{align*}
    \mathbb{E} G_{C_{23121n}^\ast}^2 \le& \mathbb{E} \sup_{(w,x)\in\Omega} \bigg| (\mathrm{i}w)^2 V_i e^{\mathrm{i} x'X_i} \frac{(1-D_j) e^{\mathrm{i}wW_j} K_{h_x}(X_j-X_i) K_{h_r}(R_j) \delta_j}{\frac{1}{2} f(X_i,0) \Delta_p(X_i)} \frac{K_{h_x}(X_i-X_j) K_{h_r}(R_i) \delta_i}{\frac{1}{2} f(X_j,0) \Delta_p(X_j)} \\& (Y_i-D_i\tau(X_i)) (Y_i-D_i\tau(X_j)) K_{h_r}(R_i) \delta_i \bigg|^2 \\
    \lesssim& \mathbb{E} \left[ (Y_i-D_i\tau(X_i)) (Y_i-D_i\tau(X_j)) K_{h_x}^2(X_j-X_i) K_{h_r}(R_j) K_{h_r}^2(R_i) \right]^2 = O(\frac{1}{h_r^4h_x^{3d}}).
\end{align*}
Following Proposition 4 in \cite{Delgado2001}, 
\begin{align*}
    \mathbb{E} \sup_{(w,x)\in\Omega} \left| \frac{2}{(n-1)^2} U_n^{(1)} (\pi_1 C_{23121n,ij}^\ast) \right|^2 \lesssim& \frac{\mathbb{E} G_{C_{23121n}^\ast}^2}{n(n-1)^4} = O(\frac{1}{n^5h_r^4h_x^{3d}}) = o(\frac{1}{nh}). \\
    \mathbb{E} \sup_{(w,x)\in\Omega} \left| \frac{1}{(n-1)^2} U_n^{(2)} (\pi_2 C_{23121n,ij}^\ast) \right|^2 \lesssim& \frac{\mathbb{E} G_{C_{23121n}^\ast}^2}{(n-1)^6} = O(\frac{1}{n^6h_r^4h_x^{3d}}) = o(\frac{1}{nh}). 
\end{align*}
Then $\frac{2}{(n-1)^2} U_n^{(1)} (\pi_1 C_{23121n,ij}^\ast)$ and $\frac{1}{(n-1)^2} U_n^{(2)} (\pi_2 C_{23121n,ij}^\ast)$ are uniformly bounded by $\frac{1}{\sqrt{nh}}$ over $(w,x)\in\Omega$ and $\frac{1}{(n-1)^2} C_{23121n}^\ast(w,x)$ is $o_p(\frac{1}{\sqrt{nh}})$ uniformly over $(w,x)\in\Omega$.

For $C_{23122n}^\ast$, the symmetrized function $g_{C_{23122n}^\ast} := \frac{1}{3!} \sum_{[3!]} (C_{23122n,ijk}^\ast )$ has envelope $G_{C_{23122n}^\ast} = \sup_{(w,x)\in\Omega} |g_{C_{23122n}^\ast}|$, whose second moment is bounded by $\mathbb{E} \sup_{(w,x)\in\Omega} |C_{23122n,ij}^\ast|^2$:
\begin{align*}
    \mathbb{E} G_{C_{23122n}^\ast}^2 \le& \mathbb{E} \sup_{(w,x)\in\Omega} \bigg| (\mathrm{i}w)^2 V_i e^{\mathrm{i} x'X_i} \frac{(1-D_j) e^{\mathrm{i}wW_j} K_{h_x}(X_j-X_i) K_{h_r}(R_j) \delta_j}{\frac{1}{2} f(X_i,0) \Delta_p(X_i)} \frac{K_{h_x}(X_k-X_j) K_{h_r}(R_k) \delta_k}{\frac{1}{2} f(X_j,0) \Delta_p(X_j)} \\& (Y_i-D_i\tau(X_i)) (Y_k-D_k\tau(X_j)) K_{h_r}(R_i) \delta_i \bigg|^2 \\
    \lesssim& \mathbb{E} \left[ (Y_i-D_i\tau(X_i)) (Y_k-D_k\tau(X_j)) K_{h_x}(X_j-X_i) K_{h_x}(X_k-X_j) K_{h_r}(R_j) K_{h_r}(R_i) K_{h_r}(R_k) \right]^2 \\=& O(\frac{1}{h_r^3h_x^{2d}}).
\end{align*}
Following Proposition 4 in \cite{Delgado2001}, 
\begin{align*}
    \mathbb{E} \sup_{(w,x)\in\Omega} \left| 3\frac{(n-2)}{(n-1)^2} U_n^{(1)} (\pi_1 g_{C_{23122n}^\ast}) \right|^2 \lesssim& \frac{(n-2)^2 \mathbb{E} G_{C_{23122n}^\ast}^2}{n(n-1)^4} = O(\frac{1}{n^3h_r^3h_x^{2d}}) = o(\frac{1}{nh}). \\
    \mathbb{E} \sup_{(w,x)\in\Omega} \left| 3\frac{(n-2)}{(n-1)^2} U_n^{(2)} (\pi_2 g_{C_{23122n}^\ast}) \right|^2 \lesssim& \frac{\mathbb{E} G_{C_{23122n}^\ast}^2}{(n-1)^4} = O(\frac{1}{n^4h_r^3h_x^{2d}}) = o(\frac{1}{nh}). \\
    \mathbb{E} \sup_{(w,x)\in\Omega} \left| \frac{(n-2)}{(n-1)^2} U_n^{(3)} (\pi_3 g_{C_{23122n}^\ast}) \right|^2 \lesssim& \frac{n \mathbb{E} G_{C_{23122n}^\ast}^2}{(n-1)^6} = O(\frac{1}{n^5h_r^3h_x^{2d}}) = o(\frac{1}{nh}). 
\end{align*}
Then $\frac{(n-2)}{(n-1)^2} U_n^{(2)} (\pi_2 C_{23122n,ij}^\ast)$ and $\frac{(n-2)}{(n-1)^2} U_n^{(3)} (\pi_3 C_{23122n,ij}^\ast)$ are uniformly bounded by $\frac{1}{\sqrt{nh}}$ over $(w,x)\in\Omega$ and $\frac{1}{(n-1)^2} C_{23122n}^\ast(w,x)$ is $o_p(\frac{1}{\sqrt{nh}})$ uniformly over $(w,x)\in\Omega$. Furthermore, $\frac{1}{n-1} C_{231n}^\ast(w,x)$ is uniformly bounded by $\frac{1}{\sqrt{nh}}$ over $\Omega$.

For the third-order $U$-process $C_{2321n}^\ast(w,x)$, its Hoeffding decomposition is $$C_{2321n}^\ast(w,x) = \frac{1}{n} \sum_{i=1}^n \mathbb{E}(C_{2321n,ijk}^\ast|\mathcal{F}_i^\ast) + 3U_n^{(2)}(\pi_2 g_{C_{2321n}^\ast}) + U_n^{(3)}(\pi_3 g_{C_{2321n}^\ast}), $$ where the conditional expectation term is $$ \mathbb{E}(C_{2321n,ijk}^\ast|\mathcal{F}_i^\ast) = \mathrm{i}w V_i e^{\mathrm{i} x'X_i} \frac{m_\kappa(w,X_i)}{\frac{1}{2} f(X_i,0) \Delta_p(X_i)} \frac{m_p(X_i)}{\frac{1}{2} f(X_i,0) \Delta_p(X_i)} (Y_i-D_i\tau(X_i)) K_{h_r}(R_i) \delta_i = O_p(h_x^{2l}). $$
For the high-order terms $U_n^{(2)}(\pi_2 g_{C_{2321n}^\ast})$ and $U_n^{(3)}(\pi_3 g_{C_{2321n}^\ast})$ in the Hoeffding decomposition, $g_{C_{2321n}^\ast} := \frac{1}{3!} \sum_{[3!]} C_{2321n,ijk}^\ast$ is the symmetrized function with envelope $G_{C_{2321n}^\ast} = \sup_{(w,x)\in\Omega} |g_{C_{2321n}^\ast}|$, whose second moment is bounded by $\mathbb{E} \sup_{(w,x)\in\Omega} |C_{2321n,ij}^\ast|^2$:
\begin{align*}
    \mathbb{E} G_{C_{2321n}^\ast}^2 \le& \mathbb{E} \sup_{(w,x)\in\Omega} \bigg| \mathrm{i}wV_i e^{\mathrm{i} x'X_i} \left( \frac{(1-D_j) e^{\mathrm{i}wW_j} K_{h_x}(X_j-X_i) K_{h_r}(R_j) \delta_j}{\frac{1}{2} f(X_i,0) \Delta_p(X_i)} - \kappa(w,X_i) \right) \\& \left( \frac{D_k K_{h_x}(X_k-X_i) K_{h_r}(R_k) \delta_k}{\frac{1}{2} f(X_i,0) \Delta_p(X_i)} - 1 \right) (Y_i-D_i\tau(X_i)) K_{h_r}(R_i) \delta_i \bigg|^2 \\
    \lesssim& \mathbb{E} \left[ (Y_i-D_i\tau(X_i)) K_{h_x}(X_k-X_j) K_{h_x}(X_j-X_i) K_{h_r}(R_k) K_{h_r}(R_j) K_{h_r}(R_i) \right]^2 = O(\frac{1}{h_r^3h_x^{2d}}). 
\end{align*}
Following Proposition 4 in \cite{Delgado2001}, 
\begin{align*}
    \mathbb{E} \sup_{(w,x)\in\Omega} \left| 3\frac{n-2}{n-1} U_n^{(2)}(\pi_2 g_{C_{2321n}^\ast}) \right|^2 \lesssim \frac{\mathbb{E} G_{C_{2321n}^\ast}^2}{(n-1)^2} = O(\frac{1}{n^2h_r^3h_x^{2d}}) = o(\frac{1}{nh}). \\
    \mathbb{E} \sup_{(w,x)\in\Omega} \left| \frac{n-2}{n-1} U_n^{(3)}(\pi_3 g_{C_{2321n}^\ast}) \right|^2 \lesssim \frac{n\mathbb{E} G_{C_{2321n}^\ast}^2}{(n-1)^4} = O(\frac{1}{n^3h_r^3h_x^{2d}}) = o(\frac{1}{nh}). 
\end{align*}
Then $U_n^{(2)}(\pi_2 g_{C_{2321n}^\ast})$ and $U_n^{(3)}(\pi_3 g_{C_{2321n}^\ast})$ are uniformly bounded by $\frac{1}{\sqrt{nh}}$ over $(w,x)\in\Omega$. Furthermore, $C_{2321n}^\ast(w,x)$ is uniformly bounded by $\frac{1}{\sqrt{nh}}$ over $(w,x)\in\Omega$.

For $C_{2322n}^\ast(w,x)$, by expanding $\hat\tau(X_j)-\tau(X_j)$ to the first order, it can be decomposed as $C_{2322n}^\ast(w,x) = \frac{1}{n-1} C_{23221n}^\ast(w,x) + \frac{1}{n-1} C_{23222n}^\ast(w,x) + \frac{n-2}{n-1} C_{23223n}^\ast(w,x)$, where $C_{23221n}^\ast(w,x) = U_n^{(3)}(g_{C_{23221n}^\ast})$ and $C_{23222n}^\ast(w,x) = U_n^{(3)}(g_{C_{23222n}^\ast})$ are two third-order $U$-processes and $C_{23223n}^\ast(w,x) = U_n^{(4)}(g_{C_{23223n}^\ast})$ is a fourth-order $U$-process with elements:
\begin{align*}
    C_{23221n,ijk}^\ast =& (\mathrm{i}w)^2 V_i e^{\mathrm{i} x'X_i} \frac{(1-D_j) e^{\mathrm{i}wW_j} K_{h_x}(X_j-X_i) K_{h_r}(R_j) \delta_j}{\frac{1}{2} f(X_i,0) \Delta_p(X_i)} \frac{K_{h_x}(X_i-X_j) K_{h_r}(R_i) \delta_i}{\frac{1}{2} f(X_j,0) \Delta_p(X_j)} \\& \left( \frac{D_k K_{h_x}(X_k-X_i) K_{h_r}(R_k) \delta_k}{\frac{1}{2} f(X_i,0) \Delta_p(X_i)} - 1 \right) (Y_i-D_i\tau(X_i)) (Y_i-D_i\tau(X_j)) K_{h_r}(R_i) \delta_i, \\
    C_{23222n,ijk}^\ast =& (\mathrm{i}w)^2 V_i e^{\mathrm{i} x'X_i} \frac{(1-D_j) e^{\mathrm{i}wW_j} K_{h_x}(X_j-X_i) K_{h_r}(R_j) \delta_j}{\frac{1}{2} f(X_i,0) \Delta_p(X_i)} \frac{K_{h_x}(X_k-X_j) K_{h_r}(R_k) \delta_k}{\frac{1}{2} f(X_j,0) \Delta_p(X_j)} \\& \left( \frac{D_k K_{h_x}(X_k-X_i) K_{h_r}(R_k) \delta_k}{\frac{1}{2} f(X_i,0) \Delta_p(X_i)} - 1 \right) (Y_i-D_i\tau(X_i)) (Y_k-D_k\tau(X_j)) K_{h_r}(R_i) \delta_i, \\
    C_{23223n,ijkl}^\ast =& (\mathrm{i}w)^2 V_i e^{\mathrm{i} x'X_i} \frac{(1-D_j) e^{\mathrm{i}wW_j} K_{h_x}(X_j-X_i) K_{h_r}(R_j) \delta_j}{\frac{1}{2} f(X_i,0) \Delta_p(X_i)} \frac{K_{h_x}(X_l-X_j) K_{h_r}(R_l) \delta_l}{\frac{1}{2} f(X_j,0) \Delta_p(X_j)} \\& \left( \frac{D_k K_{h_x}(X_k-X_i) K_{h_r}(R_k) \delta_k}{\frac{1}{2} f(X_i,0) \Delta_p(X_i)} - 1 \right) (Y_i-D_i\tau(X_i)) (Y_l-D_l\tau(X_j)) K_{h_r}(R_i) \delta_i,
\end{align*}
Their Hoeffding decompositions are
\begin{align*}
    C_{23221n}^\ast(w,x) =& 3U_n^{(1)}(\pi_1 g_{C_{23221n}^\ast}) + 3U_n^{(2)}(\pi_2 g_{C_{23221n}^\ast}) + U_n^{(3)}(\pi_3 g_{C_{23221n}^\ast}), \\ 
    C_{23222n}^\ast(w,x) =& 3U_n^{(1)}(\pi_1 g_{C_{23222n}^\ast}) + 3U_n^{(2)}(\pi_2 g_{C_{23222n}^\ast}) + U_n^{(3)}(\pi_3 g_{C_{23222n}^\ast}), \\ 
    C_{23223n}^\ast(w,x) =& \frac{1}{n} \sum_{i=1}^n \mathbb{E}(C_{23223n,ijk}^\ast|\mathcal{F}_i^\ast) + \frac{1}{n(n-1)} \sum_{i\ne j} \mathbb{E}(C_{23223n,ijk}^\ast|\mathcal{F}_i^\ast, \mathcal{F}_j^\ast) \\&+ \frac{1}{n(n-1)} \sum_{i\ne k} \mathbb{E}(C_{23223n,ijk}^\ast|\mathcal{F}_i^\ast, \mathcal{F}_k^\ast) + \frac{1}{n(n-1)} \sum_{i\ne l} \mathbb{E}(C_{23223n,ijk}^\ast|\mathcal{F}_i^\ast, \mathcal{F}_l^\ast) \\&+ 6U_n^{(2)}(\pi_2 g_{C_{23223n}^\ast}) + 4U_n^{(3)}(\pi_3 g_{C_{23223n}^\ast}) + U_n^{(4)}(\pi_4 g_{C_{23223n}^\ast}).
\end{align*}
where the conditional expectations are $\mathbb{E}(C_{23223n,ijk}^\ast|\mathcal{F}_i^\ast, \mathcal{F}_j^\ast) = O_p(h_x^{2l})$, $\mathbb{E}(C_{23223n,ijk}^\ast|\mathcal{F}_i^\ast, \mathcal{F}_k^\ast) = O_p(h_x^l)$, $\mathbb{E}(C_{23223n,ijk}^\ast|\mathcal{F}_i^\ast, \mathcal{F}_l^\ast) = O_p(h_x^l)$ and their sums are uniformly bounded by $\frac{1}{\sqrt{nh}}$ on $\Omega$.  
For the high-order terms, the envelopes have second-order moments bounded by
\begin{align*}
    \mathbb{E}G_{C_{23221n}^\ast}^2 \lesssim& \mathbb{E}[K_{h_x}^2(X_j-X_i)K_{h_x}(X_k-X_i)  K_{h_r}^2(R_i) K_{h_r}(R_j) K_{h_r}(R_k)]^2 = O(\frac{1}{h_x^{4d}h_r^5}). \\
    \mathbb{E}G_{C_{23222n}^\ast}^2 \lesssim& \mathbb{E}[K_{h_x}(X_j-X_i)K_{h_x}(X_k-X_i)K_{h_x}(X_j-X_k) K_{h_r}^2(R_i) K_{h_r}(R_j) K_{h_r}^2(R_k)]^2 = O(\frac{1}{h_x^{4d}h_r^5}). \\
    \mathbb{E}G_{C_{23223n}^\ast}^2 \lesssim& \mathbb{E}[K_{h_x}(X_j-X_i)K_{h_x}(X_k-X_i)K_{h_x}(X_j-X_l) K_{h_r}(R_i) K_{h_r}(R_j) K_{h_r}(R_k) K_{h_r}(R_l)]^2 \\=& O(\frac{1}{h_x^{3d}h_r^4}). 
\end{align*}
Following Proposition 4 in \cite{Delgado2001}, 
\begin{align*}
    \mathbb{E} \sup_{(w,x)\in\Omega} \left| \frac{3}{n-1} U_n^{(1)} (\pi_1 g_{C_{23221n}^\ast}) \right|^2 \lesssim& \frac{\mathbb{E} G_{C_{23221n}^\ast}^2}{n(n-1)^2} = O(\frac{1}{n^3h_x^{4d}h_r^5}) = o(\frac{1}{nh}). \\
    \mathbb{E} \sup_{(w,x)\in\Omega} \left| \frac{3}{n-1} U_n^{(2)} (\pi_2 g_{C_{23221n}^\ast}) \right|^2 \lesssim& \frac{\mathbb{E} G_{C_{23221n}^\ast}^2}{(n-1)^4} = O(\frac{1}{n^4h_x^{4d}h_r^5}) = o(\frac{1}{nh}). \\
    \mathbb{E} \sup_{(w,x)\in\Omega} \left| \frac{1}{n-1} U_n^{(3)} (\pi_3 g_{C_{23221n}^\ast}) \right|^2 \lesssim& \frac{n \mathbb{E} G_{C_{23221n}^\ast}^2}{(n-1)^6} = O(\frac{1}{n^5h_x^{4d}h_r^5}) = o(\frac{1}{nh}). \\
    \mathbb{E} \sup_{(w,x)\in\Omega} \left| \frac{3}{n-1} U_n^{(1)} (\pi_1 g_{C_{23222n}^\ast}) \right|^2 \lesssim& \frac{\mathbb{E} G_{C_{23222n}^\ast}^2}{n(n-1)^2} = O(\frac{1}{n^3h_x^{4d}h_r^5}) = o(\frac{1}{nh}). \\
    \mathbb{E} \sup_{(w,x)\in\Omega} \left| \frac{3}{n-1} U_n^{(2)} (\pi_2 g_{C_{23222n}^\ast}) \right|^2 \lesssim& \frac{\mathbb{E} G_{C_{23222n}^\ast}^2}{(n-1)^4} = O(\frac{1}{n^4h_x^{4d}h_r^5}) = o(\frac{1}{nh}). \\
    \mathbb{E} \sup_{(w,x)\in\Omega} \left| \frac{1}{n-1} U_n^{(3)} (\pi_3 g_{C_{23222n}^\ast}) \right|^2 \lesssim& \frac{n \mathbb{E} G_{C_{23222n}^\ast}^2}{(n-1)^6} = O(\frac{1}{n^5h_x^{4d}h_r^5}) = o(\frac{1}{nh}). \\
    \mathbb{E} \sup_{(w,x)\in\Omega} \left| 4U_n^{(3)} (\pi_3 g_{C_{23223n}^\ast}) \right|^2 \lesssim& \frac{n \mathbb{E} G_{C_{23223n}^\ast}^2}{(n-1)^4} = O(\frac{1}{n^3h_x^{3d}h_r^4}) = o(\frac{1}{nh}). \\
    \mathbb{E} \sup_{(w,x)\in\Omega} \left| U_n^{(4)} (\pi_4 g_{C_{23223n}^\ast}) \right|^2 \lesssim& \frac{\mathbb{E} G_{C_{23223n}^\ast}^2}{(n-1)^4} = O(\frac{1}{n^4h_x^{3d}h_r^4}) = o(\frac{1}{nh}). 
\end{align*}
Then all components in $C_{2322n}^\ast(w,x)$ are uniformly bounded by $\frac{1}{\sqrt{nh}}$ on $\Omega$. Combining all the results, $C_{23n}^\ast(w,x)$ is uniformly bounded by $\frac{1}{\sqrt{nh}}$ on $\Omega$. 

For the $U$-process $C_{24n}^\ast(w,x)$, its Hoeffding decomposition is $$C_{24n}^\ast(w,x) = \frac{1}{n} \sum_{i=1}^n \mathbb{E}(C_{24n,ijk}^\ast|\mathcal{F}_i) + 3U_n^{(2)}(\pi_2 g_{C_{24n}^\ast}) + U_n^{(3)}(\pi_3 g_{C_{24n}^\ast}).$$
The conditional expectation term is $$ \mathbb{E} (C_{24n,ijk}^\ast | \mathcal{F}_i^\ast) = (\mathrm{i}w)^2 V_i e^{\mathrm{i} x'X_i} \kappa(w,X_i) \frac{m_p^2(X_i) K_{h_r}(R_i) \delta_i}{\left( \frac{1}{2} f(X_i,0) \Delta_p(X_i) \right)^2} = O_p(h_x^{2l}) + O_p(h_r^2). $$ 
The second-order moment of envelope $G_{C_{24n}^\ast}$ is bounded by
\begin{align*}
    \mathbb{E} G_{C_{24n}^\ast}^2 \le& \mathbb{E} \sup_{(w,x)\in\Omega} \bigg| (\mathrm{i}w)^2 V_i e^{\mathrm{i} x'X_i} \kappa(w,X_i) \left( \frac{D_j K_{h_x}(X_j-X_i) K_{h_r}(R_j) \delta_j}{\frac{1}{2} f(X_i,0) \Delta_p(X_i)} - 1 \right) \\& \left( \frac{D_k K_{h_x}(X_k-X_i) K_{h_r}(R_k) \delta_k}{\frac{1}{2} f(X_i,0) \Delta_p(X_i)} - 1 \right) (Y_i-D_i\tau(X_i)) K_{h_r}(R_i) \delta_i \bigg|^2 \\
    \lesssim& \mathbb{E} \left| Y_j K_{h_x}(X_j-X_i) K_{h_x}(X_k-X_i) K_{h_r}(R_j) K_{h_r}(R_k) K_{h_r}(R_i) \right|^2 = O(h_x^{-2d} h_r^{-3}).
\end{align*}
Following Proposition 4 in \cite{Delgado2001}, 
\begin{align*}
    \mathbb{E} \sup_{(w,x)\in\Omega} \left\Vert 3 U_n^{(2)}(\pi_2 g_{C_{24n}^\ast}) \right\Vert^2 &\lesssim \frac{\mathbb{E} G_{C_{24n}^\ast}^2}{(n-1)^2} = O(\frac{1}{n^2h_x^{2d}h_r^3}) = o(\frac{1}{nh}), \\
    \mathbb{E} \sup_{(w,x)\in\Omega} \left\Vert U_n^{(3)}(\pi_3 g_{C_{24n}^\ast}) \right\Vert^2 &\lesssim \frac{n \mathbb{E} G_{C_{24n}^\ast}^2}{(n-1)^2(n-2)^2} = O(\frac{1}{n^3h_x^{2d}h_r^3}) = o(\frac{1}{nh}).
\end{align*}
Then it can be concluded that $\sup_{(w,x)\in\Omega} \left\Vert C_{24n}^\ast(w,x) \right\Vert = o_p\left( (nh)^{-1/2} \right)$. To summarize, $C_{2n}^\ast(w,x)$ is uniformly bounded by $\frac{1}{\sqrt{nh}}$ on $\Omega$.

The third term, $C_{3n}^\ast(w,x) = \frac{\mathrm{i}w}{n} \sum_{i=1}^n V_i e^{\mathrm{i} x'X_i} D_i \kappa(w, X_i) (\hat\tau(X_i) - \tau(X_i)) K_{h_r}(R_i) \delta_i$ can be analyzed by decomposing $\hat\tau(X_i) - \tau(X_i)$ with the expansion of $\hat{a}/\hat{b}$, where the leading terms are
\begin{align*}
    & \frac{\hat{f}(X_i,0) \Delta_{\hat\mu}(X_i)}{\frac{1}{2} f(X_i,0) \Delta_p(X_i)} - \tau(X_i) \frac{\hat{f}(X_i,0) \Delta_{\hat p}(X_i)}{\frac{1}{2} f(X_i,0) \Delta_p(X_i)} \\&- \left( \frac{\hat{f}(X_i,0) \Delta_{\hat\mu}(X_i)}{\frac{1}{2} f(X_i,0) \Delta_p(X_i)} - \tau(X_i) \right) \left( \frac{\hat{f}(X_i,0) \Delta_{\hat{p}}(X_i)}{\frac{1}{2} f(X_i,0) \Delta_p(X_i)} - 1 \right) + \tau(X_i) \left( \frac{\hat{f}(X_i,0) \Delta_{\hat{p}}(X_i)}{\frac{1}{2} f(X_i,0) \Delta_p(X_i)} - 1 \right)^2.
\end{align*}
Then the leading term of $C_{3n}^\ast(w,x)$ can be written as
$C_{31n}^\ast(w,x) - C_{32n}^\ast(w,x) - C_{33n}^\ast(w,x) + C_{34n}^\ast(w,x)$, where the components are
\begin{align*}
    C_{31n}^\ast(w,x) =& \frac{\mathrm{i}w}{n} \sum_{i=1}^n V_i e^{\mathrm{i} x'X_i} D_i \kappa(w,X_i) \frac{\hat f(X_i,0) \Delta_{\hat\mu}(X_i)}{\frac{1}{2} f(X_i,0) \Delta_p(X_i)} K_{h_r}(R_i) \delta_i, \\
    C_{32n}^\ast(w,x) =& \frac{\mathrm{i}w}{n} \sum_{i=1}^n V_i e^{\mathrm{i} x'X_i} D_i \kappa(w,X_i) \frac{\hat f(X_i,0) \Delta_{\hat p}(X_i)}{\frac{1}{2} f(X_i,0) \Delta_p(X_i)} \tau(X_i) K_{h_r}(R_i) \delta_i, \\  
    C_{33n}^\ast(w,x) =& \frac{\mathrm{i}w}{n} \sum_{i=1}^n V_i e^{\mathrm{i} x'X_i} D_i \kappa(w,X_i) \left( \frac{\hat{f}(X_i,0) \Delta_{\hat\mu}(X_i)}{\frac{1}{2} f(X_i,0) \Delta_p(X_i)} - \tau(X_i) \right) \left( \frac{\hat{f}(X_i,0) \Delta_{\hat{p}}(X_i)}{\frac{1}{2} f(X_i,0) \Delta_p(X_i)} - 1 \right) K_{h_r}(R_i) \delta_i, \\  
    C_{34n}^\ast(w,x) =& \frac{\mathrm{i}w}{n} \sum_{i=1}^n V_i e^{\mathrm{i} x'X_i} D_i \kappa(w,X_i) \left( \frac{\hat{f}(X_i,0) \Delta_{\hat{p}}(X_i)}{\frac{1}{2} f(X_i,0) \Delta_p(X_i)} - 1 \right)^2 \tau(X_i) K_{h_r}(R_i) \delta_i.
\end{align*}
The first two $U$-processes $C_{31n}^\ast(w,x)$ and $C_{32n}^\ast(w,x)$ have Hoeffding decompositions
$$C_{31n}^\ast(w,x) = \frac{1}{n}\sum_{i=1}^n \mathbb{E} (C_{31n,ij}^\ast |\mathcal{F}_i^\ast) + U_n^{(2)}(\pi_2 g_{C_{31n}^\ast}) \, \text{ and } \, C_{32n}^\ast(w,x) = \frac{1}{n}\sum_{i=1}^n \mathbb{E}(C_{32n,ij}^\ast |\mathcal{F}_i^\ast) + U_n^{(2)}(\pi_2 g_{C_{32n}^\ast}),$$ where the elements and conditional expectations are
\begin{align*}
    C_{31n,ij}^\ast =& \mathrm{i}w V_i e^{\mathrm{i} x'X_i} \frac{D_i \kappa(w,X_i)}{\frac{1}{2} f(X_i,0) \Delta_p(X_i)} Y_j K_{h_x}(X_j-X_i) K_{h_r}(R_j)\delta_j K_{h_r}(R_i) \delta_i, \\
    C_{32n,ij}^\ast =& \mathrm{i}w V_i e^{\mathrm{i} x'X_i} \frac{D_i \kappa(w,X_i) \tau(X_i)}{\frac{1}{2} f(X_i,0) \Delta_p(X_i)} D_j K_{h_x}(X_j-X_i) K_{h_r}(R_j)\delta_j K_{h_r}(R_i) \delta_i. \\
    \mathbb{E}(C_{31n,ij}^\ast|\mathcal{F}_i^\ast) =& \mathrm{i}w V_i e^{\mathrm{i} x'X_i} \frac{D_i \kappa(w,X_i)}{\frac{1}{2} f(X_i,0) \Delta_p(X_i)} K_{h_r}(R_i) \delta_i \mathbb{E}\left[ Y_j K_{h_x}(X_j-X_i) K_{h_r}(R_j)\delta_j |X_i \right] \\
    =& \mathrm{i}w V_i e^{\mathrm{i} x'X_i} D_i \kappa(w,X_i)\tau(X_i) K_{h_r}(R_i) \delta_i \left( 1 + O_p(h_r) + O_p(h_x^l) \right), \\
    \mathbb{E}(C_{32n,ij}^\ast|\mathcal{F}_i^\ast) =& \mathrm{i}w V_i e^{\mathrm{i} x'X_i} \frac{D_i \kappa(w,X_i) \tau(X_i)}{\frac{1}{2} f(X_i,0) \Delta_p(X_i)} K_{h_r}(R_i) \delta_i \mathbb{E}\left[ D_j K_{h_x}(X_j-X_i) K_{h_r}(R_j)\delta_j | X_i \right] \\
    =& \mathrm{i}w V_i e^{\mathrm{i} x'X_i} D_i \kappa(w,X_i)\tau(X_i) K_{h_r}(R_i) \delta_i \left( 1 + O_p(h_r) + O_p(h_x^l) \right).
\end{align*}
Then $\frac{1}{n} \sum_{i=1}^n \mathbb{E}(C_{31n,ij}^\ast|\mathcal{F}_i^\ast) - \frac{1}{n} \sum_{i=1}^n \mathbb{E}(C_{32n,ij}^\ast |\mathcal{F}_i^\ast) = O_p(h_r) + O_p(h_x^l) = o_p(\frac{1}{\sqrt{nh}})$. For the high-order residual terms in the Hoeffding decompositions, $U_n^{(2)}(\pi_2 g_{C_{31n}^\ast})$ and $U_n^{(2)}(\pi_2 g_{C_{32n}^\ast})$ are constructed on the symmetrized functions $g_{C_{31n}^\ast} := \frac{1}{2} (C_{31n,ij}^\ast + C_{31n,ji}^\ast)$ and $g_{C_{32n}^\ast} := \frac{1}{2} (C_{32n,ij}^\ast + C_{32n,ji}^\ast)$ with envelopes $G_{C_{31n}^\ast} = \sup_{(w,x)\in\Omega} |g_{C_{31n}^\ast}|$ and $G_{C_{32n}^\ast} = \sup_{(w,x)\in\Omega} |g_{C_{32n}^\ast}|$. The second moments of the envelopes are bounded by $\mathbb{E} \sup_{(w,x)\in\Omega} |C_{31n,ij}^\ast|^2$ and $\mathbb{E} \sup_{(w,x)\in\Omega} |C_{32n,ij}^\ast|^2$:
\begin{align*}
    \mathbb{E} G_{C_{31n}^\ast}^2 \le& \mathbb{E} \sup_{(w,x)\in\Omega} \bigg| \mathrm{i}w V_i e^{\mathrm{i} x'X_i} \frac{D_i \kappa(w,X_i)}{\frac{1}{2} f(X_i,0) \Delta_p(X_i)} Y_j K_{h_x}(X_j-X_i) K_{h_r}(R_j)\delta_j K_{h_r}(R_i) \delta_i \bigg|^2 \\
    \lesssim& \mathbb{E} \left[ Y_j K_{h_x}(X_j-X_i) K_{h_r}(R_j) K_{h_r}(R_i) \right]^2 = O(\frac{1}{hh_rh_x^d}).\\
    \mathbb{E} G_{C_{32n}^\ast}^2 \le& \mathbb{E} \sup_{(w,x)\in\Omega} \bigg| \mathrm{i}w V_i e^{\mathrm{i} x'X_i} \frac{D_i \kappa(w,X_i) \tau(X_i)}{\frac{1}{2} f(X_i,0) \Delta_p(X_i)} D_j K_{h_x}(X_j-X_i) K_{h_r}(R_j)\delta_j K_{h_r}(R_i) \delta_i \bigg|^2 \\
    \lesssim& \mathbb{E} \left[ \tau(X_i) K_{h_x}(X_j-X_i) K_{h_r}(R_j) K_{h_r}(R_i) \right]^2 = O(\frac{1}{hh_rh_x^d}). 
\end{align*}
Following Proposition 4 in \cite{Delgado2001},
\begin{align*}
    \mathbb{E} \sup_{(w,x)\in\Omega} \left| U_n^{(2)}(\pi_2 g_{C_{31n}^\ast}) \right|^2 &\lesssim \frac{\mathbb{E} G_{C_{31n}^\ast}^2}{(n-1)^2} = O(\frac{1}{n^2hh_rh_x^d}) = o(\frac{1}{nh}). \\
    \mathbb{E} \sup_{(w,x)\in\Omega} \left| U_n^{(2)}(\pi_2 g_{C_{32n}^\ast}) \right|^2 &\lesssim \frac{\mathbb{E} G_{C_{32n}^\ast}^2}{(n-1)^2} = O(\frac{1}{n^2hh_rh_x^d}) = o(\frac{1}{nh}).
\end{align*}
Then $U_n^{(2)}(\pi_2 g_{C_{31n}^\ast})$ and $U_n^{(2)}(\pi_2 g_{C_{32n}^\ast})$ are uniformly bounded by $\frac{1}{\sqrt{nh}}$ over $(w,x)\in\Omega$ and $C_{3n}^\ast(w,x)$ is uniformly bounded by $\frac{1}{\sqrt{nh}}$ over $(w,x)\in\Omega$.

For the third term in the leading terms of $C_{3n}^\ast(w,x)$, $C_{33n}^\ast(w,x)$, $$C_{33n}^\ast(w,x) = \frac{1}{n-1} C_{331n}^\ast(w,x) + \frac{n-2}{n-1} C_{332n}^\ast(w,x), $$ where $C_{331n}^\ast(w,x) = U_n^{(2)}(C_{331n,ij}^\ast)$ is a second-order $U$-process and $C_{332n}^\ast(w,x) = U_n^{(3)}(C_{332n,ijk}^\ast)$ is a third-order $U$-process. The elements of the two $U$-processes are 
\begin{align*}
    C_{331n,ij}^\ast =& \mathrm{i}w V_i e^{\mathrm{i} x'X_i} D_i \kappa(w,X_i) \left( \frac{Y_j K_{h_x}(X_j-X_i) K_{h_r}(R_j) \delta_j}{\frac{1}{2} f(X_i,0) \Delta_p(X_i)} - \tau(X_i) \right) \\& \left( \frac{D_j K_{h_x}(X_j-X_i) K_{h_r}(R_j) \delta_j}{\frac{1}{2} f(X_i,0) \Delta_p(X_i)} -1 \right) K_{h_r}(R_i) \delta_i, \\
    C_{332n,ijk}^\ast =& \mathrm{i}w V_i e^{\mathrm{i} x'X_i} D_i \kappa(w,X_i) \left( \frac{Y_j K_{h_x}(X_j-X_i) K_{h_r}(R_j) \delta_j}{\frac{1}{2} f(X_i,0) \Delta_p(X_i)} - \tau(X_i) \right) \\& \left( \frac{D_k K_{h_x}(X_k-X_i) K_{h_r}(R_k) \delta_k}{\frac{1}{2} f(X_i,0) \Delta_p(X_i)} -1 \right)  K_{h_r}(R_i) \delta_i. 
\end{align*}

The Hoeffding decomposition of $U$-process $C_{331n}^\ast(w,x)$ is $$ C_{331n}^\ast(w,x) = \mathbb{E} C_{331n,ij}^\ast + 2U_n^{(1)}(\pi_1 g_{C_{331n}^\ast}) + U_n^{(2)}(\pi_2 g_{C_{331n}^\ast}) = 2U_n^{(1)}(\pi_1 g_{C_{331n}^\ast}) + U_n^{(2)}(\pi_2 g_{C_{331n}^\ast}). $$ The expectation $\mathbb{E} C_{331n,ij}^\ast =0$ because the multipliers $\{V_i\}_{i=1}^n$ have mean zero. The second-order moment of envelope $G_{C_{331n}^\ast}$ is bounded by 
\begin{align*}
    \mathbb{E} G_{C_{331n}^\ast}^2 \le& \mathbb{E} \sup_{(w,x)\in\Omega} \left| \mathrm{i}w V_i e^{\mathrm{i} x'X_i} D_i \kappa(w,X_i) \left( \frac{Y_j K_{h_x}(X_j-X_i) K_{h_r}(R_j) \delta_j}{\frac{1}{2} f(X_i,0) \Delta_p(X_i)} - \tau(X_i) \right) \right. \\& \left. \left( \frac{D_j K_{h_x}(X_j-X_i) K_{h_r}(R_j) \delta_j}{\frac{1}{2} f(X_i,0) \Delta_p(X_i)} -1 \right) K_{h_r}(R_i) \delta_i \right|^2 \\
    \lesssim& \mathbb{E} \left| Y_j K_{h_x}^2(X_j-X_i) K_{h_r}^2(R_j) K_{h_r}(R_i) \right|^2 + \mathbb{E} \left| Y_j K_{h_x}(X_j-X_i) K_{h_r}(R_j) K_{h_r}(R_i) \right|^2 \\&+ \mathbb{E} \left| K_{h_x}(X_j-X_i) K_{h_r}(R_j) K_{h_r}(R_i) \right|^2 + \mathbb{E} \left| \tau(X_i) K_{h_r}(R_i) \right|^2 \\
    =& O(h_x^{-3d} h_r^{-4}) + O(h_x^{-d} h_r^{-2}) + O(h^{-1}) = O(h_x^{-3d} h_r^{-4}).
\end{align*}
Following Proposition 4 in \cite{Delgado2001}, the first- and second-order terms in the $U$-process are uniformly bounded since
\begin{align*}
    \mathbb{E} \sup_{(w,x)\in\Omega} \left\Vert \frac{2}{n-1} U_n^{(1)}(\pi_1 g_{C_{331n}^\ast}) \right\Vert^2 &\lesssim \frac{\mathbb{E} G_{C_{331n}^\ast}^2}{n(n-1)^2} = O(\frac{1}{n^3h_x^{3d}h_r^4}) = o(\frac{1}{nh}). \\
    \mathbb{E} \sup_{(w,x)\in\Omega} \left\Vert \frac{1}{n-1} U_n^{(2)}(\pi_2 g_{C_{331n}^\ast}) \right\Vert^2 &\lesssim \frac{\mathbb{E} G_{C_{331n}^\ast}^2}{(n-1)^4} = O(\frac{1}{n^4h_x^{3d}h_r^4}) = o(\frac{1}{nh}).
\end{align*}
Then it can be concluded that $\sup_{(w,x)\in\Omega} \left\Vert \frac{C_{331n}^\ast(w,x)}{n-1} \right\Vert = o_p\left( (nh)^{-1/2} \right)$. 

The Hoeffding decomposition of $U$-process $C_{332n}^\ast(w,x)$ is 
\begin{align*}
    C_{332n}^\ast(w,x) =& \frac{1}{n} \sum_{i=1}^n \mathbb{E} \left( C_{332n,ijk}^\ast | \mathcal{F}_i \right) + \frac{1}{n} \sum_{j=1}^n \mathbb{E} \left( C_{332n,ijk}^\ast | \mathcal{F}_j \right) + \frac{1}{n} \sum_{k=1}^n \mathbb{E} \left( C_{332n,ijk}^\ast | \mathcal{F}_k \right) \\ &-2 \mathbb{E} C_{332n,ijk}^\ast + 3 U_n^{(2)} (\pi_2 g_{C_{332n}^\ast}) + U_n^{(3)} (\pi_3 g_{C_{332n}^\ast}) \\
    =& \frac{1}{n} \sum_{i=1}^n \mathbb{E} \left( C_{332n,ijk}^\ast | \mathcal{F}_i \right) + 3 U_n^{(2)} (\pi_2 g_{C_{332n}^\ast}) + U_n^{(3)} (\pi_3 g_{C_{332n}^\ast}).
\end{align*}
The conditional expectation term is 
\begin{align*}
    & \mathbb{E} \left( C_{332n,ijk}^\ast | \mathcal{F}_i^\ast \right) = \mathrm{i}w V_i e^{\mathrm{i} x'X_i} D_i \kappa(w,X_i) \mathbb{E} \left[ \left( \frac{Y_j K_{h_x}(X_j-X_i) K_{h_r}(R_j) \delta_j}{\frac{1}{2} f(X_i,0) \Delta_p(X_i)} - \tau(X_i) \right) \right. \\& \left. \left( \frac{D_k K_{h_x}(X_k-X_i) K_{h_r}(R_k) \delta_k}{\frac{1}{2} f(X_i,0) \Delta_p(X_i)} -1 \right) \bigg|X_i \right] K_{h_r}(R_i) \delta_i \\
    =& \mathrm{i}w V_i e^{\mathrm{i} x'X_i} D_i \kappa(w,X_i) \frac{m_\mu(X_i) m_p(X_i)}{\left( \frac{1}{2} f(X_i,0) \Delta_p(X_i) \right)^2} K_{h_r}(R_i) \delta_i \\=& [O_p(h_x^l) + O_p(h_r)] [O_p(h_x^l) + O_p(h_r)] = O_p(h_x^{2l}) + O_p(h_r^2). 
\end{align*}

The second-order moment of envelope $G_{C_{332n}^\ast}$ is bounded by
\begin{align*}
    \mathbb{E} G_{C_{332n}^\ast}^2 \le& \mathbb{E} \sup_{(w,x)\in\Omega} \left| \mathrm{i}w V_i e^{\mathrm{i} x'X_i} D_i \kappa(w,X_i) \left( \frac{Y_j K_{h_x}(X_j-X_i) K_{h_r}(R_j) \delta_j}{\frac{1}{2} f(X_i,0) \Delta_p(X_i)} - \tau(X_i) \right) \right. \\& \left. \left( \frac{D_k K_{h_x}(X_k-X_i) K_{h_r}(R_k) \delta_k}{\frac{1}{2} f(X_i,0) \Delta_p(X_i)} -1 \right)  K_{h_r}(R_i) \delta_i \right|^2 \\
    \lesssim& \mathbb{E} \left| Y_j K_{h_x}(X_j-X_i) K_{h_x}(X_k-X_i) K_{h_r}(R_j) K_{h_r}(R_k) K_{h_r}(R_i) \right|^2 + \mathbb{E} \left| \tau(X_i) K_{h_r}(R_i) \right|^2 \\
    &+ \mathbb{E} \left| Y_j K_{h_x}(X_j-X_i) K_{h_r}(R_j) K_{h_r}(R_i) \right|^2 + \mathbb{E} \left| K_{h_x}(X_k-X_i) K_{h_r}(R_k) K_{h_r}(R_i) \right|^2 \\
    =& O(h_x^{-2d} h_r^{-3}) + O(h_x^{-d} h_r^{-2}) + O(h^{-1}) = O(h_x^{-2d} h_r^{-3}).
\end{align*}
Following Proposition 4 in \cite{Delgado2001}, the second- and third-order terms in the $U$-process are uniformly bounded since
\begin{align*}
    \mathbb{E} \sup_{(w,x)\in\Omega} \left\Vert 3 \frac{n-2}{n-1} U_n^{(2)}(\pi_2 g_{C_{332n}^\ast}) \right\Vert^2 &\lesssim \frac{(n-2)^2 \mathbb{E} G_{C_{332n}^\ast}^2}{(n-1)^4} = O(\frac{1}{n^2h_x^{2d}h_r^3}) = o(\frac{1}{nh}), \\
    \mathbb{E} \sup_{(w,x)\in\Omega} \left\Vert \frac{n-2}{n-1} U_n^{(3)}(\pi_3 g_{C_{332n}^\ast}) \right\Vert^2 &\lesssim \frac{n \mathbb{E} G_{C_{332n}^\ast}^2}{(n-1)^4} = O(\frac{1}{n^3h_x^{2d}h_r^3}) = o(\frac{1}{nh}).
\end{align*}
Then it can be concluded that $\sup_{(w,x)\in\Omega} \left\Vert \frac{n-2}{n-1} C_{332n}^\ast(w,x) \right\Vert = o_p\left( (nh)^{-1/2} \right)$. Furthermore, $\sup_{(w,x)\in\Omega} \left\Vert C_{33n}^\ast(w,x) \right\Vert = o_p\left( (nh)^{-1/2} \right)$.

For the fourth term in the leading terms of $C_{3n}^\ast(w,x)$, $C_{34n}^\ast(w,x)$, $$C_{34n}^\ast(w,x) = \frac{1}{n-1} C_{341n}^\ast(w,x) + \frac{n-2}{n-1} C_{342n}^\ast(w,x), $$ where $C_{341n}^\ast(w,x) = U_n^{(2)}(C_{341n,ij}^\ast)$ is a second-order $U$-process and $C_{342n}^\ast(w,x) = U_n^{(3)}(C_{342n,ijk}^\ast)$ is a third-order $U$-process. The elements of the two $U$-processes are 
\begin{align*}
    C_{341n,ij}^\ast =& \mathrm{i}w V_i e^{\mathrm{i} x'X_i} D_i \kappa(w,X_i) \left( \frac{D_j K_{h_x}(X_j-X_i) K_{h_r}(R_j) \delta_j}{\frac{1}{2} f(X_i,0) \Delta_p(X_i)} -1 \right)^2 \tau(X_i) K_{h_r}(R_i) \delta_i, \\
    C_{342n,ijk}^\ast =& \mathrm{i}w V_i e^{\mathrm{i} x'X_i} D_i \kappa(w,X_i) \left( \frac{D_j K_{h_x}(X_j-X_i) K_{h_r}(R_j) \delta_j}{\frac{1}{2} f(X_i,0) \Delta_p(X_i)} -1 \right) \\& \left( \frac{D_k K_{h_x}(X_k-X_i) K_{h_r}(R_k) \delta_k}{\frac{1}{2} f(X_i,0) \Delta_p(X_i)} -1 \right) \tau(X_i) K_{h_r}(R_i) \delta_i. 
\end{align*}

The Hoeffding decomposition of $U$-process $C_{341n}^\ast(w,x)$ is $$ C_{341n}^\ast(w,x) = 2U_n^{(1)}(\pi_1 g_{C_{341n}^\ast}) + U_n^{(2)}(\pi_2 g_{C_{341n}^\ast}). $$ The expectation $\mathbb{E} C_{341n,ij}^\ast =0$ because the multipliers $\{V_i\}_{i=1}^n$ are zero-mean. The second-order moment of envelope $G_{C_{341n}^\ast}$ is bounded by 
\begin{align*}
    \mathbb{E} G_{C_{341n}^\ast}^2 \le& \mathbb{E} \sup_{(w,x)\in\Omega} \left| \mathrm{i}w V_i e^{\mathrm{i} x'X_i} D_i \kappa(w,X_i) \left( \frac{D_j K_{h_x}(X_j-X_i) K_{h_r}(R_j) \delta_j}{\frac{1}{2} f(X_i,0) \Delta_p(X_i)} -1 \right)^2 \tau(X_i) K_{h_r}(R_i) \delta_i \right|^2 \\
    \lesssim& \mathbb{E} \left| \tau(X_i) K_{h_x}^2(X_j-X_i) K_{h_r}^2(R_j) K_{h_r}(R_i) \right|^2 + \mathbb{E} \left| \tau(X_i) K_{h_x}(X_j-X_i) K_{h_r}(R_j) K_{h_r}(R_i) \right|^2 \\&+ \mathbb{E} \left| \tau(X_i) K_{h_r}(R_i) \right|^2 \\
    =& O(h_x^{-3d} h_r^{-4}) + O(h_x^{-d} h_r^{-2}) + O(h^{-1}) = O(h_x^{-3d} h_r^{-4}).
\end{align*}
Following Proposition 4 in \cite{Delgado2001}, the first- and second-order terms in the $U$-process are uniformly bounded since
\begin{align*}
    \mathbb{E} \sup_{(w,x)\in\Omega} \left\Vert \frac{2}{n-1} U_n^{(1)}(\pi_1 g_{C_{341n}^\ast}) \right\Vert^2 &\lesssim \frac{\mathbb{E} G_{C_{341n}^\ast}^2}{n(n-1)^2} = O(\frac{1}{n^3h_x^{3d}h_r^4}) = o(\frac{1}{nh}). \\
    \mathbb{E} \sup_{(w,x)\in\Omega} \left\Vert \frac{1}{n-1} U_n^{(2)}(\pi_2 g_{C_{341n}^\ast}) \right\Vert^2 &\lesssim \frac{\mathbb{E} G_{C_{341n}^\ast}^2}{(n-1)^4} = O(\frac{1}{n^4h_x^{3d}h_r^4}) = o(\frac{1}{nh}).
\end{align*}
Then it can be concluded that $\sup_{(w,x)\in\Omega} \left\Vert \frac{C_{341n}^\ast(w,x)}{n-1} \right\Vert = o_p\left( (nh)^{-1/2} \right)$. 

The Hoeffding decomposition of $U$-process $C_{342n}^\ast(w,x)$ is $$C_{342n}^\ast(w,x) = \frac{1}{n} \sum_{i=1}^n \mathbb{E} \left( C_{342n,ijk}^\ast | \mathcal{F}_i \right) + 3 U_n^{(2)} (\pi_2 g_{C_{342n}^\ast}) + U_n^{(3)} (\pi_3 g_{C_{342n}^\ast}).$$

The conditional expectation term is $$ \mathbb{E} \left( C_{342n,ijk}^\ast | \mathcal{F}_i^\ast \right) = V_i e^{\mathrm{i} x'X_i} \frac{D_i \kappa(w,X_i) \tau(X_i)}{\left( \frac{1}{2} f(X_i,0) \Delta_p(X_i) \right)^2} m_p^2(X_i) K_{h_r}(R_i) \delta_i = O_p(h_x^{2l}) + O_p(h_r^2). $$
The second-order moment of envelope $G_{C_{342n}^\ast}$ is bounded by
\begin{align*}
    \mathbb{E} G_{C_{342n}^\ast}^2 \le& \mathbb{E} \sup_{(w,x)\in\Omega} \left| V_i e^{\mathrm{i} x'X_i} D_i \kappa(w,X_i) \left( \frac{D_j K_{h_x}(X_j-X_i) K_{h_r}(R_j) \delta_j}{\frac{1}{2} f(X_i,0) \Delta_p(X_i)} -1 \right) \right. \\& \left. \left( \frac{D_k K_{h_x}(X_k-X_i) K_{h_r}(R_k) \delta_k}{\frac{1}{2} f(X_i,0) \Delta_p(X_i)} -1 \right) \tau(X_i) K_{h_r}(R_i) \delta_i \right|^2 \\
    \lesssim& \mathbb{E} \left| \tau(X_i) K_{h_x}(X_j-X_i) K_{h_x}(X_k-X_i) K_{h_r}(R_j) K_{h_r}(R_k) K_{h_r}(R_i) \right|^2 = O(h_x^{-2d} h_r^{-3}).
\end{align*}
Following Proposition 4 in \cite{Delgado2001}, the second- and third-order terms in the $U$-process are uniformly bounded since
\begin{align*}
    \mathbb{E} \sup_{(w,x)\in\Omega} \left\Vert 3 \frac{n-2}{n-1} U_n^{(2)}(\pi_2 g_{C_{342n}^\ast}) \right\Vert^2 &\lesssim \frac{(n-2)^2 \mathbb{E} G_{C_{342n}^\ast}^2}{(n-1)^4} = O(\frac{1}{n^2h_x^{2d}h_r^3}) = o(\frac{1}{nh}), \\
    \mathbb{E} \sup_{(w,x)\in\Omega} \left\Vert \frac{n-2}{n-1} U_n^{(3)}(\pi_3 g_{C_{342n}^\ast}) \right\Vert^2 &\lesssim \frac{n \mathbb{E} G_{C_{342n}^\ast}^2}{(n-1)^4} = O(\frac{1}{n^3h_x^{2d}h_r^3}) = o(\frac{1}{nh}).
\end{align*}
Then it can be concluded that $\sup_{(w,x)\in\Omega} \left\Vert \frac{n-2}{n-1} C_{342n}^\ast(w,x) \right\Vert = o_p\left( (nh)^{-1/2} \right)$. Furthermore, $\sup_{(w,x)\in\Omega} \left\Vert C_{34n}^\ast(w,x) \right\Vert = o_p\left( (nh)^{-1/2} \right)$ and $\sup_{(w,x)\in\Omega} \left\Vert C_{3n}^\ast(w,x) \right\Vert = o_p\left( (nh)^{-1/2} \right)$.

For the last term in \eqref{eq1}, by expanding $\hat\kappa-\kappa$ and $\hat\tau-\tau$, its leading terms can be expressed as
\begin{align*}
    & C_{4n}^\ast(w,x) \\=& \frac{\mathrm{i}w}{n} \sum_{i=1}^n V_i e^{\mathrm{i} x'X_i} D_i \left( \frac{\hat{f}(X_i,0) \Delta_{\hat\varphi}(w,X_i)}{\frac{1}{2} f(X_i,0) \Delta_p(X_i)} - \kappa(w,X_i) \right) \left( \frac{\hat{f}(X_i,0) \Delta_{\hat\mu}(X_i)}{\frac{1}{2} f(X_i,0) \Delta_p(X_i)} - \tau(X_i) \right) K_{h_r}(R_i) \delta_i \\
    &- \frac{\mathrm{i}w}{n} \sum_{i=1}^n V_i e^{\mathrm{i} x'X_i} D_i \left( \frac{\hat{f}(X_i,0) \Delta_{\hat\varphi}(w,X_i)}{\frac{1}{2} f(X_i,0) \Delta_p(X_i)} - \kappa(w,X_i) \right) \left( \frac{\hat{f}(X_i,0) \Delta_{\hat{p}}(X_i)}{\frac{1}{2} f(X_i,0) \Delta_p(X_i)} - 1 \right) \tau(X_i) K_{h_r}(R_i) \delta_i \\
    &- \frac{\mathrm{i}w}{n} \sum_{i=1}^n V_i e^{\mathrm{i} x'X_i} D_i \kappa(w,X_i) \left( \frac{\hat{f}(X_i,0) \Delta_{\hat{p}}(X_i)}{\frac{1}{2} f(X_i,0) \Delta_p(X_i)} - 1 \right) \left( \frac{\hat{f}(X_i,0) \Delta_{\hat\mu}(X_i)}{\frac{1}{2} f(X_i,0) \Delta_p(X_i)} - \tau(X_i) \right) K_{h_r}(R_i) \delta_i \\
    &+ \frac{\mathrm{i}w}{n} \sum_{i=1}^n V_i e^{\mathrm{i} x'X_i} D_i \kappa(w,X_i) \left( \frac{\hat{f}(X_i,0) \Delta_{\hat{p}}(X_i)}{\frac{1}{2} f(X_i,0) \Delta_p(X_i)} - 1 \right) \left( \frac{\hat{f}(X_i,0) \Delta_{\hat{p}}(X_i)}{\frac{1}{2} f(X_i,0) \Delta_p(X_i)} - 1 \right) \tau(X_i) K_{h_r}(R_i) \delta_i \\
    =& C_{41n}^\ast(w,x) - C_{42n}^\ast(w,x) - C_{43n}^\ast(w,x) + C_{44n}^\ast(w,x).    
\end{align*}
The high-order interaction terms in the expansion of $\hat\kappa-\kappa$ and $\hat\tau-\tau$ are asymptotically smaller than the terms $C_{41n}^\ast(w,x)$--$C_{44n}^\ast(w,x)$. 
Note that $C_{41n}^\ast(w,x)$ and $C_{42n}^\ast(w,x)$ have a similar structure to $C_{23n}^\ast(w,x)$, where the difference in the components is that $D_i \left( \frac{\hat{f}(X_i,0) \Delta_{\hat\mu}(X_i)}{\frac{1}{2} f(X_i,0) \Delta_p(X_i)} - \tau(X_i) \right) $ in $C_{41n}^\ast(w,x)$ is replaced by $\left( \frac{\hat{f}(X_i,0) \Delta_{\hat{p}}(X_i)}{\frac{1}{2} f(X_i,0) \Delta_p(X_i)} - 1 \right) (Y_i-D_i\tau(X_i))$ and $D_i\tau(X_i)$ in $C_{42n}^\ast(w,x)$ is replaced by $Y_i-D_i\tau(X_i)$ in $C_{23n}^\ast(w,x)$. Therefore, the same argument as for $C_{23n}^\ast(w,x)$ shows that $C_{41n}^\ast(w,x)$ and $C_{42n}^\ast(w,x)$ are uniformly bounded by $\frac{1}{\sqrt{nh}}$ on $\Omega$. $C_{43n}^\ast(w,x)$ and $C_{44n}^\ast(w,x)$ are exactly the same as $C_{33n}^\ast(w,x)$ and $C_{34n}^\ast(w,x)$, which are uniformly bounded by $\frac{1}{\sqrt{nh}}$. Summarizing the aforementioned results, $\sup_{(w,x)\in\Omega}|C_{4n}^\ast(w,x)| = o_p(\frac{1}{\sqrt{nh}})$.

\subsubsection{(ii)}
The difference between $\hat{U}_n^\ast(w,x)$ and $U_n^\ast(w,x)$ can be expressed by a Taylor expansion of the function $e^{\mathrm{i} w\hat{W}_i}$ at the point $W_i$: 
\begin{align*}
    \hat{U}_n^\ast(w,x) - U_n^\ast(w,x)
    =& \mathrm{i}w\frac{1}{n} \sum_{i=1}^n V_i e^{\mathrm{i} (wW_i+x'X_i)} (1-D_i) ( \hat\tau(X_i) - \tau(X_i) ) K_h(R_i) \delta_i \\
    -& \frac{w^2}{2n} \sum_{i=1}^n V_i e^{\mathrm{i} (wW_i+x'X_i)} (1-D_i) (\hat\tau(X_i) - \tau(X_i))^2 K_h(R_i) \delta_i \\
    -& \frac{\mathrm{i}w^3}{6n} \sum_{i=1}^n V_i e^{\mathrm{i} (w\bar{W}_i+x'X_i)} (1-D_i) (\hat\tau(X_i) - \tau(X_i))^3 K_h(R_i) \delta_i.
\end{align*}
The last component is of smaller order than the first two terms. The first two terms are denoted as
\begin{align*}
    C_n^\ast(w,x) =& \mathrm{i}w\frac{1}{n} \sum_{i=1}^n V_i e^{\mathrm{i} (wW_i+x'X_i)} (1-D_i) ( \hat\tau(X_i) - \tau(X_i) ) K_h(R_i) \delta_i \\
    &- \frac{w^2}{2n} \sum_{i=1}^n V_i e^{\mathrm{i} (wW_i+x'X_i)} (1-D_i) (\hat\tau(X_i) - \tau(X_i))^2 K_h(R_i) \delta_i.
\end{align*}

Following the same decomposition as that for the local constant estimator, the main term of $\hat{U}_n^\ast(w,x)- U_n^\ast(w,x)$ can be denoted as $C_n^\ast(w,x) = C_{1n}^\ast(w,x) - C_{2n}^\ast(w,x) - C_{3n}^\ast(w,x) + C_{4n}^\ast(w,x) - C_{5n}^\ast(w,x)$, where the components are 
\begin{align*}
    C_{1n}^\ast(w,x) =& \frac{\mathrm{i}w}{n} \sum_{i=1}^n V_i e^{\mathrm{i} (wW_i+x'X_i)} (1-D_i) \frac{\Delta_{\hat\mu}(X_i)}{\Delta_p(X_i)} K_h(R_i) \delta_i \\
    C_{2n}^\ast(w,x) =& \frac{\mathrm{i}w}{n} \sum_{i=1}^n V_i e^{\mathrm{i} (wW_i+x'X_i)} (1-D_i) \tau(X_i) \frac{\Delta_{\hat{p}}(X_i)}{\Delta_p(X_i)} K_h(R_i) \delta_i \\
    C_{3n}^\ast(w,x) =& \frac{\mathrm{i}w}{n} \sum_{i=1}^n V_i e^{\mathrm{i} (wW_i+x'X_i)} (1-D_i) \frac{\Delta_{\hat\mu}(X_i)-\Delta_\mu(X_i)}{\Delta_p(X_i)} \frac{\Delta_{\hat{p}}(X_i)-\Delta_p(X_i)}{\Delta_p(X_i)} \left( 1 + \mathrm{i}w \tau(X_i) \right) K_h(R_i) \delta_i \\
    C_{4n}^\ast(w,x) =& \frac{\mathrm{i}w}{n} \sum_{i=1}^n V_i e^{\mathrm{i} (wW_i+x'X_i)} (1-D_i) \left( 1 + \frac{\mathrm{i}w}{2} \tau(X_i) \right) \tau(X_i) \left( \frac{\Delta_{\hat{p}}(X_i) - \Delta_p(X_i)}{\Delta_p(X_i)}\right)^2 K_h(R_i) \delta_i \\
    C_{5n}^\ast(w,x) =& \frac{w^2}{2n} \sum_{i=1}^n V_i e^{\mathrm{i} (wW_i+x'X_i)} (1-D_i) \left( \frac{\Delta_{\hat\mu}(X_i)-\Delta_\mu(X_i)}{\Delta_p(X_i)} \right)^2 K_h(R_i) \delta_i.
\end{align*}

With Lemma S.A2, by expanding $\Delta_{\hat\mu}(X_i) - \Delta_\mu(X_i)$, $C_{1n}^\ast(w,x)$ can be expressed as $C_{1n}^\ast(w,x) = C_{1n+}^\ast(w,x) - C_{1n-}^\ast(w,x) + O_p(h_r^2)$, where
\begin{align*}
    C_{1n+}^\ast(w,x) =& \frac{\mathrm{i}w}{n} \sum_{i=1}^n V_i e^{\mathrm{i} (wW_i+x'X_i)} \frac{1-D_i}{\Delta_p(X_i)} e_0' \left( B_{n+}^{-1}(X_i)A_{Yn+}(X_i) - B_+^{-1}(X_i)A_{Y+}(X_i) \right) K_h(R_i) \delta_i, \\    
    C_{1n-}^\ast(w,x) =& \frac{\mathrm{i}w}{n} \sum_{i=1}^n V_i e^{\mathrm{i} (wW_i+x'X_i)} \frac{1-D_i}{\Delta_p(X_i)} e_0' \left( B_{n-}^{-1}(X_i)A_{Yn-}(X_i) - B_-^{-1}(X_i)A_{Y-}(X_i) \right) K_h(R_i) \delta_i.
\end{align*}
The leading terms of $C_{1n+}^\ast(w,x)$ can be written as $C_{11n+}^\ast(w,x) - C_{12n+}^\ast(w,x) - C_{13n+}^\ast(w,x) + C_{14n+}^\ast(w,x)$, where 
\begin{align*}
    C_{11n+}^\ast(w,x) =& \frac{\mathrm{i}w}{n} \sum_{i=1}^n V_i e^{\mathrm{i} (wW_i+x'X_i)} \frac{1-D_i}{\Delta_p(X_i)} e_0' B_+^{-1}(X_i) A_{Yn+}(X_i) K_h(R_i) \delta_i, \\
    C_{12n+}^\ast(w,x) =& \frac{\mathrm{i}w}{n} \sum_{i=1}^n V_i e^{\mathrm{i} (wW_i+x'X_i)} \frac{1-D_i}{\Delta_p(X_i)} e_0' B_+^{-1}(X_i) B_{n+}(X_i) B_+^{-1}(X_i) A_{Y+}(X_i) K_h(R_i) \delta_i, \\
    C_{13n+}^\ast(w,x) =& \frac{\mathrm{i}w}{n} \sum_{i=1}^n V_i e^{\mathrm{i} (wW_i+x'X_i)} \frac{1-D_i}{\Delta_p(X_i)} e_0' B_+^{-1}(X_i) \left( B_{n+}(X_i) - B_+(X_i) \right) \\& B_+^{-1}(X_i) \left( A_{Yn+}(X_i) - A_{Y+}(X_i) \right) K_h(R_i) \delta_i, \\
    C_{14n+}^\ast(w,x) =& \frac{\mathrm{i}w}{n} \sum_{i=1}^n V_i e^{\mathrm{i} (wW_i+x'X_i)} \frac{1-D_i}{\Delta_p(X_i)} e_0' B_+^{-1}(X_i) \left( B_{n+}(X_i) - B_+(X_i) \right) \\& B_+^{-1}(X_i) \left( B_{n+}(X_i) - B_+(X_i) \right) B_+^{-1}(X_i) A_{Y+}(X_i) K_h(R_i) \delta_i.
\end{align*}
$C_{11n+}^\ast(w,x)$ and $C_{12n+}^\ast(w,x)$ are second-order $U$-processes, $C_{13n+}^\ast(w,x)$ and $C_{14n+}^\ast(w,x)$ are combinations of the second-order and third-order $U$-processes derived from the second-order term in the expansion of $\hat{a}/\hat{b}$.
The Hoeffding decompositions of $U$-processes $C_{11n+}^\ast(w,x)$ and $C_{12n+}^\ast(w,x)$ are
\begin{align*}
   C_{11n+}^\ast(w,x) =& \frac{1}{n} \sum_{i=1}^n \mathbb{E}(C_{11n+,ij}^\ast|\mathcal{F}_i^\ast) + \frac{1}{n} \sum_{j=1}^n \mathbb{E}(C_{11n+,ij}^\ast|\mathcal{F}_j^\ast) - \mathbb{E}C_{11n+,ij}^\ast + U_n^{(2)}(\pi_2 g_{C_{11n+}^\ast}), \\
   C_{12n+}^\ast(w,x) =& \frac{1}{n} \sum_{i=1}^n \mathbb{E}(C_{12n+,ij}^\ast|\mathcal{F}_i^\ast) + \frac{1}{n} \sum_{j=1}^n \mathbb{E}(C_{12n+,ij}^\ast|\mathcal{F}_j^\ast) - \mathbb{E}C_{12n+,ij}^\ast + U_n^{(2)}(\pi_2 g_{C_{12n+}^\ast}), 
\end{align*}
where the terms conditional on $\mathcal{F}_j^\ast$ and the total expectations are exactly zero due to the zero-mean multipliers $\{V_i\}_{i=1}^n$. The Hoeffding decompositions become \begin{align*}
    C_{11n+}^\ast(w,x) =& \frac{1}{n} \sum_{i=1}^n \mathbb{E}(C_{11n+,ij}^\ast|\mathcal{F}_i^\ast) + U_n^{(2)}(\pi_2 g_{C_{11n+}^\ast}) + O_p(h_x^l), \\
    C_{12n+}^\ast(w,x) =& \frac{1}{n} \sum_{i=1}^n \mathbb{E}(C_{12n+,ij}^\ast|\mathcal{F}_i^\ast) + U_n^{(2)}(\pi_2 g_{C_{12n+}^\ast}) + O_p(h_x^l), 
\end{align*}
where the elements are
\begin{align*}
   C_{11n+,ij}^\ast =& \mathrm{i}w V_i e^{\mathrm{i} (wW_i+x'X_i)} \frac{1-D_i}{\Delta_p(X_i)} e_0' B_+^{-1}(X_i) \mathbf{r}(R_j)Y_j K_{h_x}(X_j-X_i) K_{h_r}(R_j)\mathbf{1}_+(R_j) K_h(R_i) \delta_i, \\
   C_{12n+,ij}^\ast =& \mathrm{i}w V_i e^{\mathrm{i} (wW_i+x'X_i)} \frac{1-D_i}{\Delta_p(X_i)} e_0' B_+^{-1}(X_i) \mathbf{r}(R_j)\mathbf{r}(R_j)' B_+^{-1}(X_i) A_{Y+}(X_i) \\& K_{h_x}(X_j-X_i) K_{h_r}(R_j)\mathbf{1}_+(R_j) K_h(R_i) \delta_i;
\end{align*}
and the conditional expectations $\mathbb{E}(C_{11n+,ij}^\ast|\mathcal{F}_i^\ast)$ and $\mathbb{E}(C_{12n+,ij}^\ast|\mathcal{F}_i^\ast)$ are
\begin{align*}
    \mathbb{E}(C_{11n+,ij}^\ast|\mathcal{F}_i^\ast) =& \mathrm{i}w V_i e^{\mathrm{i} (wW_i+x'X_i)} \frac{1-D_i}{\Delta_p(X_i)} e_0' B_+^{-1}(X_i) K_h(R_i) \delta_i \mathbb{E}\left[ \mathbf{r}(R_j) Y_j K_{h_x}(X_j-X_i) K_{h_r}(R_j)\mathbf{1}_+(R_j) |X_i \right], \\
    =& \mathrm{i}w V_i e^{\mathrm{i} (wW_i+x'X_i)} \frac{1-D_i}{\Delta_p(X_i)} e_0' B_+^{-1}(X_i) A_{Y+}(X_i) K_h(R_i) \delta_i + O_p(h_x^l), \\
    \mathbb{E}(C_{12n+,ij}^\ast|\mathcal{F}_i^\ast) =& \mathrm{i}w V_i e^{\mathrm{i} (wW_i+x'X_i)} \frac{1-D_i}{\Delta_p(X_i)} e_0' B_+^{-1}(X_i) \mathbb{E}\left[ \mathbf{r}(R_j) \mathbf{r}(R_j)' K_{h_x}(X_j-X_i) K_{h_r}(R_j)\mathbf{1}_+(R_j) |X_i \right] \\& B_+^{-1}(X_i) A_{Y+}(X_i) K_h(R_i) \delta_i, \\
    =& \mathrm{i}w V_i e^{\mathrm{i} (wW_i+x'X_i)} \frac{1-D_i}{\Delta_p(X_i)} e_0' B_+^{-1}(X_i) A_{Y+}(X_i) K_h(R_i) \delta_i + O_p(h_x^l),
\end{align*}
For the high-order terms in the Hoeffding decompositions of $U$-processes $C_{11n+}^\ast(w,x)$ and $C_{12n+}^\ast(w,x)$, $U_n^{(2)}(\pi_2 g_{C_{11n+}^\ast})$ and $U_n^{(2)}(\pi_2 g_{C_{12n+}^\ast})$, the symmetrized functions $g_{C_{11n+}^\ast} := \frac{1}{2} (C_{11n+,ij}^\ast + C_{11n+,ji}^\ast)$ and $g_{C_{12n+}^\ast} := \frac{1}{2} (C_{12n+,ij}^\ast + C_{12n+,ji}^\ast)$ are of VC-type and have envelopes $G_{C_{11n+}^\ast} = \sup_{(w,x)\in\Omega}|g_{C_{11n+}^\ast}|$ and $G_{C_{12n+}^\ast} = \sup_{(w,x)\in\Omega}|g_{C_{12n+}^\ast}|$ with bounded second moments $G_{C_{11n+}^\ast}^2 \le \sup_{(w,x)\in\Omega} \left|C_{11n+,ij}^\ast\right|^2$ and $G_{C_{12n+}^\ast}^2 \le \sup_{(w,x)\in\Omega} \left|C_{12n+,ij}^\ast\right|^2$. Given Assumption B5:
\begin{align*}
    \mathbb{E} G_{C_{11n+}^\ast}^2 \le& \mathbb{E} \sup_{(w,x)\in\Omega} \bigg| \mathrm{i}w V_i e^{\mathrm{i} (wW_i+x'X_i)} \frac{1-D_i}{\Delta_p(X_i)} e_0' B_+^{-1}(X_i) \mathbf{r}(R_j) Y_j K_{h_x}(X_j-X_i) K_{h_r}(R_j)\mathbf{1}_+(R_j) K_h(R_i) \delta_i \bigg|^2 \\
    \lesssim& \mathbb{E} \left[ e_0' B_+^{-1}(X_i) \mathbf{r}(R_j) Y_j K_{h_x}(X_j-X_i) K_{h_r}(R_j) K_h(R_i) \right]^2 = O(\frac{1}{hh_rh_x^d}). \\
    \mathbb{E} G_{C_{12n+}^\ast}^2 \le& \mathbb{E} \sup_{(w,x)\in\Omega} \bigg| \mathrm{i}w V_i e^{\mathrm{i} (wW_i+x'X_i)} \frac{1-D_i}{\Delta_p(X_i)} e_0' B_+^{-1}(X_i) \mathbf{r}(R_j) \mathbf{r}(R_j)' \\& B_+^{-1}(X_i) A_{Y+}(X_i) K_{h_x}(X_j-X_i) K_{h_r}(R_j)\mathbf{1}_+(R_j) K_h(R_i) \delta_i \bigg|^2 \\
    \lesssim& \mathbb{E} \left[ e_0' B_+^{-1}(X_i) \mathbf{r}(R_j) \mathbf{r}(R_j)' B_+^{-1}(X_i) A_{Y+}(X_i) K_{h_x}(X_j-X_i) K_{h_r}(R_j) K_h(R_i) \right]^2 = O(\frac{1}{hh_rh_x^d}).
\end{align*}
Following Proposition 4 in \cite{Delgado2001}:
\begin{align*}
    \mathbb{E} \sup_{(w,x)\in\Omega} \left| U_n^{(2)}(\pi_2 g_{C_{11n+}^\ast}) \right|^2 &\lesssim \frac{\mathbb{E} G_{C_{11n+}^\ast}^2}{(n-1)^2} = O(\frac{1}{n^2hh_rh_x^d}) = o(\frac{1}{nh}). \\
    \mathbb{E} \sup_{(w,x)\in\Omega} \left| U_n^{(2)}(\pi_2 g_{C_{12n+}^\ast}) \right|^2 &\lesssim \frac{\mathbb{E} G_{C_{12n+}^\ast}^2}{(n-1)^2} = O(\frac{1}{n^2hh_rh_x^d}) = o(\frac{1}{nh}). 
\end{align*}
Thus, $U_n^{(2)}(\pi_2 g_{C_{11n+}^\ast})$ and $U_n^{(2)}(\pi_2 g_{C_{12n+}^\ast})$ are uniformly bounded by $\frac{1}{\sqrt{nh}}$. 
The results for $C_{11n+}^\ast(w,x)$ and $C_{12n+}^\ast(w,x)$ imply that $C_{11n+}^\ast(w,x) - C_{12n+}^\ast(w,x) = O_p(h_x^l) = o_p(\frac{1}{\sqrt{nh}})$.  

For the high-order terms $C_{13n+}^\ast(w,x) = \frac{1}{n-1} C_{131n+}^\ast(w,x) + \frac{n-2}{n-1} C_{132n+}^\ast(w,x)$ and $C_{14n+}^\ast(w,x) = \frac{1}{n-1} C_{141n+}^\ast(w,x) + \frac{n-2}{n-1} C_{142n+}^\ast(w,x)$, where the components $C_{131n+}^\ast(w,x)$, $C_{132n+}^\ast(w,x)$, $C_{141n+}^\ast(w,x)$ and $C_{142n+}^\ast(w,x)$ are $U$-processes with Hoeffding decompositions
\begin{align*}
    C_{131n+}^\ast(w,x) =& U_n^{(2)}(g_{C_{131n+}^\ast}) = 2U_n^{(1)}(\pi_1 g_{C_{131n+}^\ast}) + U_n^{(2)}(\pi_2 g_{C_{131n+}^\ast}), \\
    C_{132n+}^\ast(w,x) =& U_n^{(3)}(g_{C_{132n+}^\ast}) = \frac{1}{n} \sum_{i=1}^n \mathbb{E} \left( C_{132n+,ijk}^\ast | \mathcal{F}_i^\ast \right) + 3 U_n^{(2)} (\pi_2 g_{C_{132n+}^\ast}) + U_n^{(3)} (\pi_3 g_{C_{132n+}^\ast}); \\
    C_{141n+}^\ast(w,x) =& U_n^{(2)}(g_{C_{141n+}^\ast}) = 2U_n^{(1)}(\pi_1 g_{C_{141n+}^\ast}) + U_n^{(2)}(\pi_2 g_{C_{141n+}^\ast}), \\
    C_{142n+}^\ast(w,x) =& U_n^{(3)}(g_{C_{142n+}^\ast}) = \frac{1}{n} \sum_{i=1}^n \mathbb{E} \left( C_{142n+,ijk}^\ast | \mathcal{F}_i^\ast \right) + 3 U_n^{(2)} (\pi_2 g_{C_{142n+}^\ast}) + U_n^{(3)} (\pi_3 g_{C_{142n+}^\ast}).
\end{align*}
The elements of these $U$-processes are
\begin{align*}
    C_{131n+,ij}^\ast =& \mathrm{i}w V_i e^{\mathrm{i} (wW_i+x'X_i)} \frac{1-D_i}{\Delta_p(X_i)} e_0' B_+^{-1}(X_i) \left[ \mathbf{r}(R_j) \mathbf{r}(R_j)' K_{h_x}(X_j-X_i) K_{h_r}(R_j) \mathbf{1}_+(R_j) - B_+(X_i) \right] \\& B_+^{-1}(X_i) \left[ \mathbf{r}(R_j) Y_j K_{h_x}(X_j-X_i) K_{h_r}(R_j) \mathbf{1}_+(R_j) - A_{Y+}(X_i) \right] K_h(R_i) \delta_i, \\
    C_{132n+,ijk}^\ast =& \mathrm{i}w V_i e^{\mathrm{i} (wW_i+x'X_i)} \frac{1-D_i}{\Delta_p(X_i)} e_0' B_+^{-1}(X_i) \left[ \mathbf{r}(R_j) \mathbf{r}(R_j)' K_{h_x}(X_j-X_i) K_{h_r}(R_j) \mathbf{1}_+(R_j) - B_+(X_i) \right] \\& B_+^{-1}(X_i) \left[ \mathbf{r}(R_k) Y_k K_{h_x}(X_k-X_i) K_{h_r}(R_k) \mathbf{1}_+(R_k) - A_{Y+}(X_i) \right] K_h(R_i) \delta_i, \\
    C_{141n+,ij}^\ast =& \mathrm{i}w V_i e^{\mathrm{i} (wW_i+x'X_i)} \frac{1-D_i}{\Delta_p(X_i)} e_0' B_+^{-1}(X_i) \left[ \mathbf{r}(R_j) \mathbf{r}(R_j)' K_{h_x}(X_j-X_i) K_{h_r}(R_j) \mathbf{1}_+(R_j) - B_+(X_i) \right] \\& B_+^{-1}(X_i) \left[ \mathbf{r}(R_j) \mathbf{r}(R_j)' K_{h_x}(X_j-X_i) K_{h_r}(R_j) \mathbf{1}_+(R_j) - B_+(X_i) \right] B_+^{-1}(X_i) A_{Y+}(X_i) K_h(R_i) \delta_i, \\
    C_{142n+,ijk}^\ast =& \mathrm{i}w V_i e^{\mathrm{i} (wW_i+x'X_i)} \frac{1-D_i}{\Delta_p(X_i)} e_0' B_+^{-1}(X_i) \left[ \mathbf{r}(R_j) \mathbf{r}(R_j)' K_{h_x}(X_j-X_i) K_{h_r}(R_j) \mathbf{1}_+(R_j) - B_+(X_i) \right] \\& B_+^{-1}(X_i) \left[ \mathbf{r}(R_k) \mathbf{r}(R_k)' K_{h_x}(X_k-X_i) K_{h_r}(R_k) \mathbf{1}_+(R_k) - B_+(X_i) \right] B_+^{-1}(X_i) A_{Y+}(X_i) K_h(R_i) \delta_i.
\end{align*}
The conditional expectation terms in the Hoeffding decompositions are
\begin{align*}
    \mathbb{E} \left( C_{132n+,ijk}^\ast | \mathcal{F}_i^\ast \right) =& \mathrm{i}w V_i e^{\mathrm{i} (wW_i+x'X_i)} \frac{1-D_i}{\Delta_p(X_i)} e_0' B_+^{-1}(X_i) K_h(R_i) \delta_i O_p(h_x^l) B_+^{-1}(X_i) O_p(h_x^l) = O_p(h_x^{2l}) ;\\
    \mathbb{E} \left( C_{142n+,ijk}^\ast | \mathcal{F}_i^\ast \right) =& \mathrm{i}w V_i e^{\mathrm{i} (wW_i+x'X_i)} \frac{1-D_i}{\Delta_p(X_i)} e_0' B_+^{-1}(X_i) O_p(h_x^l) B_+^{-1}(X_i) O_p(h_x^l) B_+^{-1}(X_i) A_{Y+}(X_i) K_h(R_i) \delta_i \\=& O_p(h_x^{2l}).
\end{align*}

The second-order moments of the envelopes $G_{C_{131n+}}$ and $G_{C_{141n+}}$ are bounded by 
\begin{align*}
    \mathbb{E} G_{C_{131n+}^\ast}^2 \le& \mathbb{E} \sup_{(w,x)\in\Omega} \bigg| \mathrm{i}w V_i e^{\mathrm{i} (wW_i+x'X_i)} \frac{1-D_i}{\Delta_p(X_i)} e_0' B_+^{-1}(X_i) [ \mathbf{r}(R_j) \mathbf{r}(R_j)' K_{h_x}(X_j-X_i) K_{h_r}(R_j) \mathbf{1}_+(R_j) \\&- B_+(X_i)]  B_+^{-1}(X_i) \left[ \mathbf{r}(R_j) Y_j K_{h_x}(X_j-X_i) K_{h_r}(R_j) \mathbf{1}_+(R_j) - A_{Y+}(X_i) \right] K_h(R_i) \delta_i \bigg|^2 \\
    \lesssim& O(\frac{1}{h_x^{3d}h_r^3h}). \\
    \mathbb{E} G_{C_{141n+}^\ast}^2 \le& \mathbb{E} \sup_{(w,x)\in\Omega} \bigg| \mathrm{i}w V_i e^{\mathrm{i} (wW_i+x'X_i)} \frac{1-D_i}{\Delta_p(X_i)} e_0' B_+^{-1}(X_i) [ \mathbf{r}(R_j) \mathbf{r}(R_j)' K_{h_x}(X_j-X_i) K_{h_r}(R_j) \mathbf{1}_+(R_j) \\&- B_+(X_i) ] B_+^{-1}(X_i) \left[ \mathbf{r}(R_j) \mathbf{r}(R_j)' K_{h_x}(X_j-X_i) K_{h_r}(R_j) \mathbf{1}_+(R_j) - B_+(X_i) \right] \\& B_+^{-1}(X_i) A_{Y+}(X_i) K_h(R_i) \delta_i \bigg|^2 \lesssim O(\frac{1}{h_x^{3d}h_r^3h}).
\end{align*}
Following Proposition 4 in \cite{Delgado2001}, the first- and second-order terms in the $U$-process are uniformly bounded since
\begin{align*}
    \mathbb{E} \sup_{(w,x)\in\Omega} \left\Vert \frac{2}{n-1} U_n^{(1)}(\pi_1 g_{C_{131n+}^\ast}) \right\Vert^2 &\lesssim \frac{\mathbb{E} G_{C_{131n+}^\ast}^2}{n(n-1)^2} = O(\frac{1}{n^3h_x^{3d}h_r^3h}) = o(\frac{1}{nh}). \\
    \mathbb{E} \sup_{(w,x)\in\Omega} \left\Vert \frac{1}{n-1} U_n^{(2)}(\pi_2 g_{C_{131n+}^\ast}) \right\Vert^2 &\lesssim \frac{\mathbb{E} G_{C_{131n+}^\ast}^2}{(n-1)^4} = O(\frac{1}{n^4h_x^{3d}h_r^3h}) = o(\frac{1}{nh}); \\
    \mathbb{E} \sup_{(w,x)\in\Omega} \left\Vert \frac{2}{n-1} U_n^{(1)}(\pi_1 g_{C_{141n+}^\ast}) \right\Vert^2 &\lesssim \frac{\mathbb{E} G_{C_{141n+}^\ast}^2}{n(n-1)^2} = O(\frac{1}{n^3h_x^{3d}h_r^3h}) = o(\frac{1}{nh}). \\
    \mathbb{E} \sup_{(w,x)\in\Omega} \left\Vert \frac{1}{n-1} U_n^{(2)}(\pi_2 g_{C_{141n+}^\ast}) \right\Vert^2 &\lesssim \frac{\mathbb{E} G_{C_{141n+}^\ast}^2}{(n-1)^4} = O(\frac{1}{n^4h_x^{3d}h_r^3h}) = o(\frac{1}{nh}).
\end{align*}

The second-order moments of envelopes $G_{C_{132n+}^\ast}$ and $G_{C_{142n+}^\ast}$ are bounded by
\begin{align*}
    \mathbb{E} G_{C_{132n+}^\ast}^2 \le& \mathbb{E} \sup_{(w,x)\in\Omega} \bigg| \mathrm{i}w V_i e^{\mathrm{i} (wW_i+x'X_i)} \frac{1-D_i}{\Delta_p(X_i)} e_0' B_+^{-1}(X_i) [ \mathbf{r}(R_j) \mathbf{r}(R_j)' K_{h_x}(X_j-X_i) K_{h_r}(R_j) \mathbf{1}_+(R_j) \\&- B_+(X_i)]  B_+^{-1}(X_i) \left[ \mathbf{r}(R_k) Y_k K_{h_x}(X_k-X_i) K_{h_r}(R_k) \mathbf{1}_+(R_k) - A_{Y+}(X_i) \right] K_h(R_i) \delta_i \bigg|^2 \\
    \lesssim& O(\frac{1}{h_x^{2d}h_r^2h}). 
\end{align*}
\begin{align*}
    \mathbb{E} G_{C_{142n+}^\ast}^2 \le& \mathbb{E} \sup_{(w,x)\in\Omega} \bigg| \mathrm{i}w V_i e^{\mathrm{i} (wW_i+x'X_i)} \frac{1-D_i}{\Delta_p(X_i)} e_0' B_+^{-1}(X_i) [ \mathbf{r}(R_j) \mathbf{r}(R_j)' K_{h_x}(X_j-X_i) K_{h_r}(R_j) \mathbf{1}_+(R_j) \\&- B_+(X_i) ] B_+^{-1}(X_i) \left[ \mathbf{r}(R_k) \mathbf{r}(R_k)' K_{h_x}(X_k-X_i) K_{h_r}(R_k) \mathbf{1}_+(R_k) - B_+(X_i) \right] \\& B_+^{-1}(X_i) A_{Y+}(X_i) K_h(R_i) \delta_i \bigg|^2 \lesssim O(\frac{1}{h_x^{2d}h_r^2h}).
\end{align*}
Following Proposition 4 in \cite{Delgado2001}, the second- and third-order terms in the $U$-process are uniformly bounded since
\begin{align*}
    \mathbb{E} \sup_{(w,x)\in\Omega} \left\Vert 3 \frac{n-2}{n-1} U_n^{(2)}(\pi_2 g_{C_{132n+}^\ast}) \right\Vert^2 &\lesssim \frac{(n-2)^2 \mathbb{E} G_{C_{132n+}^\ast}^2}{(n-1)^4} = O(\frac{1}{n^2h_x^{2d}h_r^2h}) = o(\frac{1}{nh}), \\
    \mathbb{E} \sup_{(w,x)\in\Omega} \left\Vert \frac{n-2}{n-1} U_n^{(3)}(\pi_3 g_{C_{132n+}^\ast}) \right\Vert^2 &\lesssim \frac{n \mathbb{E} G_{C_{132n+}^\ast}^2}{(n-1)^4} = O(\frac{1}{n^3h_x^{2d}h_r^2h}) = o(\frac{1}{nh}); \\
    \mathbb{E} \sup_{(w,x)\in\Omega} \left\Vert 3 \frac{n-2}{n-1} U_n^{(2)}(\pi_2 g_{C_{142n+}^\ast}) \right\Vert^2 &\lesssim \frac{(n-2)^2 \mathbb{E} G_{C_{142n+}^\ast}^2}{(n-1)^4} = O(\frac{1}{n^2h_x^{2d}h_r^2h}) = o(\frac{1}{nh}), \\
    \mathbb{E} \sup_{(w,x)\in\Omega} \left\Vert \frac{n-2}{n-1} U_n^{(3)}(\pi_3 g_{C_{142n+}^\ast}) \right\Vert^2 &\lesssim \frac{n \mathbb{E} G_{C_{142n+}^\ast}^2}{(n-1)^4} = O(\frac{1}{n^3h_x^{2d}h_r^2h}) = o(\frac{1}{nh}).
\end{align*}
Then it can be concluded that $C_{13n+}^\ast(w,x)$ and $C_{14n+}^\ast(w,x)$ are uniformly bounded by $\frac{1}{\sqrt{nh}}$. 
Summarizing all the results about $C_{11n+}^\ast(w,x)$--$C_{14n+}^\ast(w,x)$, $\sup_{(w,x)\in\Omega}|C_{1n+}^\ast(w,x)|=o_p(\frac{1}{\sqrt{nh}})$. By the same investigation procedure, the counterpart on the negative side, $C_{1n-}^\ast(w,x)$ can be shown to equal $o_p(\frac{1}{\sqrt{nh}})$ uniformly on $\Omega$ by the same argument as for $C_{1n+}^\ast(w,x)$ by replacing the positive signal with the negative signal in matrix functional $B$, vector functional $A$, and scalar indicator function $\mathbf{1}_\pm(R)$. Then $\sup_{(w,x)\in\Omega}|C_{1n}^\ast(w,x)|=o_p(\frac{1}{\sqrt{nh}})$. The same argument with $Y_j$ replaced by $D_j\tau(X_i)$ in the raw processes shows that $\sup_{(w,x)\in\Omega}|C_{2n}^\ast(w,x)|=o_p(\frac{1}{\sqrt{nh}})$.

For $C_{3n}^\ast(w,x)$--$C_{5n}^\ast(w,x)$, multiplying the interaction terms of $\Delta_{\hat\mu}(X_i) - \Delta_\mu(X_i)$ and $\Delta_{\hat{p}}(X_i)-\Delta_p(X_i)$ shows that the approximation error terms $(e_0' B_+^{-1}(X_i) A_{Z+}(X_i) - m_Z(X_i, 0^+))$ and $( e_0' B_-^{-1}(X_i) A_{Z-}(X_i) - m_Z(X_i, 0^-))$ are of order $h_r^2$, which is $o(\frac{1}{\sqrt{nh}})$.
Any component in the leading terms can be written as $$C_{sn\pm}^\ast(w,x) = \frac{1}{n-1} C_{s1n\pm}^\ast(w,x) + \frac{n-2}{n-1} C_{s2n\pm}^\ast(w,x), $$ where $C_{s1n\pm}^\ast(w,x) = U_n^{(2)}(C_{s1n\pm,ij}^\ast)$ is a second-order $U$-process and $C_{s2n\pm}^\ast(w,x) = U_n^{(3)}(C_{s2n\pm,ijk}^\ast)$ is a third-order $U$-process. They have the same structure as the aforementioned processes $C_{131n+}^\ast(w,x)$ and $C_{132n+}^\ast(w,x)$ investigated in detail above. Then the asymptotic property of processes $C_{s1n\pm}^\ast(w,x)$ and $C_{s2n\pm}^\ast(w,x)$ can be investigated in the same way.

The Hoeffding decompositions of the $U$-processes $C_{s1n\pm}(w,x)$ and $C_{s2n\pm}(w,x)$ are
\begin{align*}
    C_{s1n\pm}^\ast(w,x) =& 2U_n^{(1)}(\pi_1 g_{C_{s1n\pm}^\ast}) + U_n^{(2)}(\pi_2 g_{C_{s1n\pm}^\ast}), \\
    C_{s2n\pm}^\ast(w,x) =& \frac{1}{n} \sum_{i=1}^n \mathbb{E} \left( C_{s2n\pm,ijk}^\ast | \mathcal{F}_i^\ast \right) + 3 U_n^{(2)} (\pi_2 g_{C_{s2n\pm}^\ast}) + U_n^{(3)} (\pi_3 g_{C_{s2n\pm}^\ast}), 
\end{align*}
where the conditional expectation terms in the Hoeffding decomposition of $C_{s2n\pm}^\ast(w,x)$ are $\mathbb{E} \left( C_{s2n\pm,ijk}^\ast | \mathcal{F}_i^\ast \right) = O_p(h_x^{2l})$.

For the high-order terms in the Hoeffding decompositions, we investigate the second moments of their envelopes, $\mathbb{E}G_{C_{s1n\pm}^\ast}^2 \lesssim O(h_x^{-3d} h_r^{-3} h^{-1})$ and $\mathbb{E}G_{C_{s2n\pm}^\ast}^2 \lesssim O(h_x^{-2d} h_r^{-2} h^{-1})$.
Following Proposition 4 in \cite{Delgado2001}, the high-order terms in the Hoeffding decompositions are uniformly bounded by
\begin{align*}
    \mathbb{E} \sup_{(w,x)\in\Omega} \left| \frac{2}{n-1} U_n^{(1)}(\pi_1 g_{C_{s1n\pm}^\ast}) \right|^2 &\lesssim \frac{\mathbb{E} G_{C_{s1n\pm}^\ast}^2}{n(n-1)^2} = O(\frac{1}{n^3h_x^{3d}h_r^3h}) = o(\frac{1}{nh}), \\
    \mathbb{E} \sup_{(w,x)\in\Omega} \left| \frac{1}{n-1} U_n^{(2)}(\pi_2 g_{C_{s1n\pm}^\ast}) \right|^2 &\lesssim \frac{\mathbb{E} G_{C_{s1n\pm}^\ast}^2}{(n-1)^4} = O(\frac{1}{n^4h_x^{3d}h_r^3h}) = o(\frac{1}{nh}); \\
    \mathbb{E} \sup_{(w,x)\in\Omega} \left| 3 \frac{n-2}{n-1} U_n^{(2)}(\pi_2 g_{C_{s2n\pm}^\ast}) \right|^2 &\lesssim \frac{(n-2)^2 \mathbb{E} G_{C_{s2n\pm}^\ast}^2}{(n-1)^4} = O(\frac{1}{n^2h_x^{2d}h_r^2h}) = o(\frac{1}{nh}), \\
    \mathbb{E} \sup_{(w,x)\in\Omega} \left| \frac{n-2}{n-1} U_n^{(3)}(\pi_3 g_{C_{s2n\pm}^\ast}) \right|^2 &\lesssim \frac{n \mathbb{E} G_{C_{s2n\pm}^\ast}^2}{(n-1)^4} = O(\frac{1}{n^3h_x^{2d}h_r^2h}) = o(\frac{1}{nh}). 
\end{align*}
Then it can be concluded that $\frac{1}{n-1} C_{s1n\pm}^\ast(w,x)$ and $\frac{n-2}{n-1} C_{s2n\pm}^\ast(w,x)$ are uniformly bounded by $\frac{1}{\sqrt{nh}}$ over $(w,x)\in\Omega$ for any subscript $s$ in the components of $C_{3n}^\ast(w,x)$--$C_{5n}^\ast(w,x)$. Furthermore, $\sup_{(w,x)\in\Omega}$ $|C_{sn}^\ast(w,x)| = o_p(\frac{1}{\sqrt{nh}})$ for $s=3,4,5$. In sum, $$ \sup_{(w,x)\in\Omega} \left| \hat{U}_n^\ast(w,x) - U_n^\ast(w,x) \right| = o_p(\frac{1}{\sqrt{nh}}).$$

\subsection{Proof of Theorem \ref{thm4}}
Let $U_{1n}^\ast(w,x) = \frac{1}{n} \sum_{i=1}^n V_i f_{1i}(w,x)$ and $U_{2n}^\ast(w,x) = \frac{1}{n} \sum_{i=1}^n V_i f_{2i}(w,x)$.
When Assumption B1 holds, $\hat{U}_n^\ast(w,x)$ with the local constant estimator $\hat\tau(X_i)$ can be written as $$ \sqrt{nh} \hat{U}_n^\ast(w,x) = \sqrt{\frac{h}{n}} \sum_{i=1}^n V_i g_{1i}(w,x) + \sqrt{\frac{h}{n}} \sum_{i=1}^n V_i g_i(w,x) + o_p(1). $$
Suppose Assumption B1 holds with the local constant estimator $\hat\tau(X_i)$ in $\hat{U}_n^\ast(w,x)$ or Assumption B1' holds with the local linear estimator $\hat\tau(X_i)$ in $\hat{U}_n^\ast(w,x)$. 
When $h=o(h_r)$ holds, $\hat{U}_n^\ast(w,x)$ can be expressed as $$ \sqrt{nh} \hat{U}_n^\ast(w,x) = \sqrt{\frac{h}{n}} \sum_{i=1}^n V_i g_{0i}(w,x) + \sqrt{\frac{h}{n}} \sum_{i=1}^n V_i g_i(w,x) + o_p(1). $$

\subsubsection{Under the null}
Under $\mathbb{H}_0$, $g_i(w,x)=0$. By the multiplier central limit theorem, the bootstrap version $\hat{U}_n^\ast(w,x)$ converges to the same Gaussian processes $U_{\infty,1}(w,x)$ and $U_{\infty,0}(w,x)$, respectively, as $\sqrt{\frac{h}{n}} \sum_{i=1}^n g_{1i}(w,x)$ and $\sqrt{\frac{h}{n}} \sum_{i=1}^n g_{0i}(w,x)$ do. Then it can be concluded that under $\mathbb{H}_0$, when Assumption B1 holds for the local constant estimator $\hat\tau(X_i)$, $$ \sqrt{nh} \hat{U}_n^\ast(w,x) \Longrightarrow_\ast U_{\infty,1}(w,x), $$ 
Suppose Assumption B1 holds with the local constant estimator $\hat\tau(X_i)$ in $\hat{U}_n^\ast(w,x)$ or Assumption B1' holds with the local linear estimator $\hat\tau(X_i)$ in $\hat{U}_n^\ast(w,x)$. 
When $h=o(h_r)$ holds, $\hat{U}_n^\ast(w,x)$ can be expressed as $$ \sqrt{nh} \hat{U}_n^\ast(w,x) \Longrightarrow_\ast U_{\infty,0}(w,x). $$ 

\subsubsection{Under the fixed alternative}
Under $\mathbb{H}_1$, $g_i(w,x) = e^{\mathrm{i} x'X_i} f_{R|X}(0|X_i) ( \varphi_{W|XR}(w|X_i,0^+) - \varphi_{W|XR}(w|X_i,0^-) )$. Under different conditions, the common component in the expression of $\sqrt{nh} \hat{U}_n^\ast(w,x)$ is $\sqrt{\frac{h}{n}} \sum_{i=1}^n V_i g_i(w,x)$.
Since $g_i(w,x)$ is uniformly bounded, $\frac{1}{\sqrt{n}} \sum_{i=1}^n V_i g_i(w,x)$ converges to a zero-mean Gaussian process and is $O_p(1)$. Then $\sqrt{\frac{h}{n}} \sum_{i=1}^n V_i g_i(w,x) = \sqrt{h} \frac{1}{\sqrt{n}} \sum_{i=1}^n V_i g_i(w,x) = \sqrt{h} O_p(1) = o_p(1)$.

When Assumption B1 holds for $\hat{U}_n^\ast(w,x)$ with the local constant estimator $\hat\tau(X_i)$, $$ \sqrt{nh} \hat{U}_n^\ast(w,x) = \sqrt{\frac{h}{n}} \sum_{i=1}^n V_i g_{1i}(w,x) + o_p(1) \Longrightarrow_\ast U_{\infty,1}(w,x), $$
When $h=o(h_r)$ holds while Assumption B1 holds with the local constant estimator $\hat\tau(X_i)$ in $\hat{U}_n^\ast(w,x)$ or Assumption B1' holds with the local linear estimator $\hat\tau(X_i)$ in $\hat{U}_n^\ast(w,x)$, $\hat{U}_n^\ast(w,x)$ can be expressed as $$ \sqrt{nh} \hat{U}_n^\ast(w,x) = \sqrt{\frac{h}{n}} \sum_{i=1}^n V_i g_{0i}(w,x) + o_p(1) \Longrightarrow_\ast U_{\infty,0}(w,x). $$

$U_{\infty,1}(w,x)$ is a zero-mean Gaussian process with covariance $\mathcal{K}_1(w_1, x_1; w_2, x_2)$, which is the same as the limiting process of $\sqrt{nh}\hat{U}_n(w,x)$ under $\mathbb{H}_0$ with the local constant estimator $\hat\tau(X_i)$ given Assumption B1.

$U_{\infty,0}(w,x)$ is another zero-mean Gaussian process with covariance $\mathcal{K}_0(w_1, x_1; w_2, x_2)$, which is the same as the limiting process of $\sqrt{nh}\hat{U}_n(w,x)$ under $\mathbb{H}_0$ with the local constant estimator $\hat\tau(X_i)$ given Assumption B1 and $h=o(h_r)$, or with the local linear estimator $\hat\tau(X_i)$ given Assumption B1' and $h=o(h_r)$.

\subsubsection{Under local alternatives}

Under $\mathbb{H}_{1n}$, $g_i(w,x) = \frac{1}{\sqrt{nh}} e^{\mathrm{i} x'X_i} f_{R|X}(0|X_i) \lambda(w,X_i)$. The common component in $\sqrt{nh} \hat{U}_n^\ast(w,x)$ is $$ \sqrt{\frac{h}{n}} \sum_{i=1}^n V_i g_i(w,x) = \frac{1}{n} \sum_{i=1}^n V_i e^{\mathrm{i} x'X_i} f_{R|X}(0|X_i) \lambda(w,X_i). $$
Given $\mathbb{E} \sup_{(w,x)\in\Omega} \left| V_i \lambda(w,X_i) e^{\mathrm{i} x'X_i} f_{R|X}(0|X_i) \right| <\infty $, it can be deduced that $\sqrt{\frac{h}{n}} \sum_{i=1}^n V_i g_i(w,x)$ is $o_p(1)$ uniformly over $(w,x)\in\Omega$ by ULLN:
\begin{align*}
& \sup_{(w,x)\in\Omega} \left| \frac{1}{n} \sum_{i=1}^n V_i \lambda(w,X_i) e^{\mathrm{i} x'X_i} f_{R|X}(0|X_i) \right| \\ =& \sup_{(w,x)\in\Omega} \left| \frac{1}{n} \sum_{i=1}^n V_i \lambda(w,X_i) e^{\mathrm{i} x'X_i} f_{R|X}(0|X_i) - \mathbb{E} \left[ V_i \lambda(w,X_i) e^{\mathrm{i} x'X_i} f_{R|X}(0|X_i) \right] \right| = o_p(1).    
\end{align*}
The equality follows from the fact that the multiplier $V_i$ satisfies $$\mathbb{E} \left[ V_i \lambda(w,X_i) e^{\mathrm{i} x'X_i} f_{R|X}(0|X_i) \right] = \mathbb{E} V_i \mathbb{E} \left[ \lambda(w,X_i) e^{\mathrm{i} x'X_i} f_{R|X}(0|X_i) \right] = 0. $$
When Assumption B1 holds for $\hat{U}_n^\ast(w,x)$ with the local constant estimator $\hat\tau(X_i)$, $$ \sqrt{nh} \hat{U}_n^\ast(w,x) = \sqrt{\frac{h}{n}} \sum_{i=1}^n V_i g_{1i}(w,x) + o_p(1) \Longrightarrow_\ast U_{\infty,1}(w,x), $$
When $h=o(h_r)$ holds while Assumption B1 holds with the local constant estimator $\hat\tau(X_i)$ in $\hat{U}_n(w,x)$ or Assumption B1' holds with the local linear estimator $\hat\tau(X_i)$ in $\hat{U}_n(w,x)$, $\hat{U}_n^\ast(w,x)$ can be expressed as $$ \sqrt{nh} \hat{U}_n^\ast(w,x) = \sqrt{\frac{h}{n}} \sum_{i=1}^n V_i g_{0i}(w,x) + o_p(1) \Longrightarrow_\ast U_{\infty,0}(w,x). $$

$U_{\infty,1}(w,x)$ and $U_{\infty,0}(w,x)$ are zero-mean Gaussian processes with covariance $\mathcal{K}_1(w_1, x_1; w_2, x_2)$ and $\mathcal{K}_0(w_1, x_1; w_2, x_2)$, which are the same as the limiting processes of $\sqrt{nh}\hat{U}_n(w,x)$ and $\sqrt{nh}\hat{U}_n^\ast(w,x)$ under the null.

\subsection{Supplementary Numerical Study} 
This section reports supplementary numerical experiments assessing the finite-sample performance of the proposed test with the local linear estimator of the RD treatment effect $\hat\tau(X_i)$. 
The DGPs are the same as those in Section \ref{sec:simulation}, adapted from \cite{Hsu2019, Hsu2021}. 
The estimation methods are the same as those in Section \ref{sec:simulation} and are consistent with the recommendations in \cite{Imbens2012a}, \cite{Calonico2014}, \cite{Calonico2019} and \cite{Hsu2019}. 
Bandwidth selection is also based on \cite{Calonico2014} and \cite{Calonico2019} but is slightly different from that in the main text. The estimating bandwidth $h_x=h_r$ is modified by an under-smoothing factor $n^{1/5-1/{k_1}}$, while the testing bandwidth $h$ is set as $h = h_r \times n^{1/5-1/{k_2}}$, so that $h=o(h_r)$. Other simulation settings are the same as those in Section \ref{sec:simulation}. 
Tables \ref{tab:DGP_null_ll} and \ref{tab:DGP_alt_ll} report the empirical size and power, respectively.

\begin{table}[H]
    \centering
    \caption{Rejection rates under DGPs 1--3 (local linear estimator)}
    \label{tab:DGP_null_ll}
    \setlength{\extrarowheight}{-1pt}
    \begin{adjustbox}{max width=\textwidth, max height=\textheight}
    \begin{tabular}{@{}ccccccccccccccccccccc@{}}
        \toprule
        \multirow{3}{*}{DGP} & \multirow{3}{*}{$k_1$} & $k_2$ & \multicolumn{6}{c}{3} & \multicolumn{6}{c}{3.25} & \multicolumn{6}{c}{3.5} \\
        \cmidrule(lr){4-9} \cmidrule(lr){10-15} \cmidrule(lr){16-21}
        & & type & \multicolumn{3}{c}{KS} & \multicolumn{3}{c}{CvM} & \multicolumn{3}{c}{KS} & \multicolumn{3}{c}{CvM} & \multicolumn{3}{c}{KS} & \multicolumn{3}{c}{CvM} \\
        \cmidrule(lr){4-6} \cmidrule(lr){7-9} \cmidrule(lr){10-12} \cmidrule(lr){13-15} \cmidrule(lr){16-18} \cmidrule(lr){19-21}
        & & $n$ & 1\% & 5\% & 10\% & 1\% & 5\% & 10\% & 1\% & 5\% & 10\% & 1\% & 5\% & 10\% & 1\% & 5\% & 10\% & 1\% & 5\% & 10\% \\
        \midrule
        \multirow{12}{*}{1} & \multirow{4}{*}{3.5} & 500 & 0.010 & 0.050 & 0.091 & 0.009 & 0.047 & 0.092 & 0.011 & 0.049 & 0.092 & 0.009 & 0.050 & 0.092 & 0.013 & 0.051 & 0.096 & 0.012 & 0.051 & 0.093 \\
         & & 1000 & 0.015 & 0.071 & 0.130 & 0.014 & 0.065 & 0.129 & 0.013 & 0.075 & 0.132 & 0.012 & 0.073 & 0.137 & 0.014 & 0.072 & 0.134 & 0.015 & 0.070 & 0.137 \\
         & & 2000 & 0.015 & 0.056 & 0.112 & 0.014 & 0.059 & 0.112 & 0.014 & 0.059 & 0.112 & 0.016 & 0.061 & 0.108 & 0.012 & 0.058 & 0.113 & 0.010 & 0.060 & 0.109 \\
         & & 4000 & 0.008 & 0.049 & 0.107 & 0.007 & 0.046 & 0.103 & 0.005 & 0.055 & 0.107 & 0.005 & 0.048 & 0.108 & 0.007 & 0.052 & 0.107 & 0.009 & 0.048 & 0.103 \\
        \cmidrule{2-21}
         & \multirow{4}{*}{3.75} & 500 & 0.009 & 0.048 & 0.092 & 0.008 & 0.048 & 0.096 & 0.011 & 0.048 & 0.094 & 0.011 & 0.051 & 0.098 & 0.011 & 0.051 & 0.099 & 0.010 & 0.053 & 0.093 \\
         & & 1000 & 0.014 & 0.076 & 0.130 & 0.014 & 0.072 & 0.134 & 0.014 & 0.070 & 0.134 & 0.014 & 0.072 & 0.135 & 0.018 & 0.074 & 0.144 & 0.016 & 0.073 & 0.147 \\
         & & 2000 & 0.016 & 0.057 & 0.108 & 0.016 & 0.057 & 0.107 & 0.012 & 0.055 & 0.104 & 0.010 & 0.056 & 0.107 & 0.012 & 0.055 & 0.110 & 0.010 & 0.061 & 0.115 \\
         & & 4000 & 0.007 & 0.046 & 0.108 & 0.005 & 0.050 & 0.104 & 0.007 & 0.051 & 0.110 & 0.007 & 0.047 & 0.102 & 0.011 & 0.053 & 0.101 & 0.007 & 0.049 & 0.107 \\
        \cmidrule{2-21}
         & \multirow{4}{*}{4} & 500 & 0.009 & 0.049 & 0.098 & 0.009 & 0.053 & 0.103 & 0.008 & 0.053 & 0.098 & 0.010 & 0.054 & 0.094 & 0.007 & 0.056 & 0.100 & 0.009 & 0.061 & 0.095 \\
         & & 1000 & 0.012 & 0.067 & 0.135 & 0.013 & 0.071 & 0.131 & 0.018 & 0.068 & 0.141 & 0.017 & 0.068 & 0.145 & 0.017 & 0.068 & 0.148 & 0.019 & 0.065 & 0.147 \\
         & & 2000 & 0.013 & 0.056 & 0.108 & 0.015 & 0.058 & 0.107 & 0.011 & 0.055 & 0.110 & 0.010 & 0.058 & 0.113 & 0.009 & 0.056 & 0.116 & 0.010 & 0.060 & 0.117 \\
         & & 4000 & 0.005 & 0.050 & 0.107 & 0.005 & 0.049 & 0.103 & 0.011 & 0.051 & 0.105 & 0.008 & 0.050 & 0.103 & 0.011 & 0.054 & 0.108 & 0.010 & 0.051 & 0.107 \\
        \midrule
        \multirow{12}{*}{2} & \multirow{4}{*}{3.5} & 500 & 0.008 & 0.053 & 0.108 & 0.008 & 0.051 & 0.097 & 0.009 & 0.053 & 0.099 & 0.009 & 0.049 & 0.090 & 0.008 & 0.051 & 0.095 & 0.010 & 0.047 & 0.090 \\
         & & 1000 & 0.014 & 0.067 & 0.125 & 0.013 & 0.071 & 0.124 & 0.013 & 0.065 & 0.125 & 0.010 & 0.073 & 0.132 & 0.011 & 0.066 & 0.135 & 0.012 & 0.071 & 0.137 \\
         & & 2000 & 0.014 & 0.053 & 0.118 & 0.013 & 0.056 & 0.115 & 0.012 & 0.060 & 0.113 & 0.010 & 0.059 & 0.113 & 0.012 & 0.058 & 0.114 & 0.009 & 0.060 & 0.112 \\
         & & 4000 & 0.011 & 0.054 & 0.124 & 0.010 & 0.054 & 0.120 & 0.011 & 0.053 & 0.112 & 0.012 & 0.051 & 0.112 & 0.012 & 0.052 & 0.114 & 0.012 & 0.053 & 0.112 \\
        \cmidrule{2-21}
         & \multirow{4}{*}{3.75} & 500 & 0.008 & 0.055 & 0.097 & 0.008 & 0.049 & 0.096 & 0.008 & 0.053 & 0.098 & 0.010 & 0.048 & 0.101 & 0.008 & 0.052 & 0.096 & 0.008 & 0.049 & 0.099 \\
         & & 1000 & 0.014 & 0.068 & 0.122 & 0.013 & 0.069 & 0.127 & 0.010 & 0.069 & 0.136 & 0.009 & 0.067 & 0.137 & 0.014 & 0.074 & 0.140 & 0.014 & 0.068 & 0.143 \\
         & & 2000 & 0.011 & 0.062 & 0.113 & 0.010 & 0.061 & 0.111 & 0.011 & 0.059 & 0.115 & 0.010 & 0.063 & 0.112 & 0.012 & 0.056 & 0.117 & 0.009 & 0.060 & 0.114 \\
         & & 4000 & 0.009 & 0.049 & 0.118 & 0.012 & 0.047 & 0.113 & 0.012 & 0.052 & 0.115 & 0.012 & 0.050 & 0.113 & 0.011 & 0.057 & 0.117 & 0.012 & 0.054 & 0.112 \\
        \cmidrule{2-21}
         & \multirow{4}{*}{4} & 500 & 0.010 & 0.054 & 0.103 & 0.009 & 0.055 & 0.099 & 0.007 & 0.054 & 0.104 & 0.007 & 0.056 & 0.103 & 0.009 & 0.058 & 0.102 & 0.009 & 0.058 & 0.095 \\
         & & 1000 & 0.012 & 0.066 & 0.131 & 0.010 & 0.073 & 0.132 & 0.014 & 0.070 & 0.134 & 0.013 & 0.065 & 0.138 & 0.012 & 0.074 & 0.144 & 0.013 & 0.069 & 0.142 \\
         & & 2000 & 0.011 & 0.063 & 0.117 & 0.012 & 0.063 & 0.114 & 0.011 & 0.058 & 0.113 & 0.010 & 0.061 & 0.110 & 0.013 & 0.057 & 0.117 & 0.013 & 0.057 & 0.118 \\
         & & 4000 & 0.012 & 0.052 & 0.111 & 0.011 & 0.050 & 0.111 & 0.010 & 0.055 & 0.119 & 0.010 & 0.053 & 0.111 & 0.011 & 0.060 & 0.116 & 0.011 & 0.059 & 0.115 \\
        \midrule
        \multirow{12}{*}{3} & \multirow{4}{*}{3.5} & 500 & 0.008 & 0.051 & 0.096 & 0.010 & 0.052 & 0.094 & 0.009 & 0.050 & 0.099 & 0.010 & 0.046 & 0.096 & 0.008 & 0.048 & 0.094 & 0.009 & 0.050 & 0.095 \\
         & & 1000 & 0.014 & 0.069 & 0.123 & 0.013 & 0.067 & 0.125 & 0.014 & 0.068 & 0.124 & 0.012 & 0.065 & 0.130 & 0.012 & 0.070 & 0.137 & 0.013 & 0.070 & 0.136 \\
         & & 2000 & 0.015 & 0.056 & 0.116 & 0.013 & 0.061 & 0.112 & 0.013 & 0.064 & 0.114 & 0.013 & 0.065 & 0.114 & 0.009 & 0.062 & 0.118 & 0.013 & 0.062 & 0.111 \\
         & & 4000 & 0.007 & 0.061 & 0.124 & 0.009 & 0.059 & 0.118 & 0.008 & 0.057 & 0.119 & 0.011 & 0.054 & 0.119 & 0.008 & 0.055 & 0.125 & 0.007 & 0.051 & 0.120 \\
        \cmidrule{2-21}
         & \multirow{4}{*}{3.75} & 500 & 0.009 & 0.050 & 0.098 & 0.008 & 0.046 & 0.100 & 0.009 & 0.049 & 0.101 & 0.008 & 0.045 & 0.101 & 0.009 & 0.049 & 0.098 & 0.008 & 0.050 & 0.100 \\
         & & 1000 & 0.015 & 0.069 & 0.119 & 0.012 & 0.068 & 0.125 & 0.011 & 0.064 & 0.130 & 0.012 & 0.067 & 0.135 & 0.012 & 0.067 & 0.140 & 0.016 & 0.066 & 0.134 \\
         & & 2000 & 0.011 & 0.063 & 0.110 & 0.012 & 0.065 & 0.110 & 0.010 & 0.065 & 0.117 & 0.011 & 0.065 & 0.111 & 0.011 & 0.057 & 0.120 & 0.007 & 0.058 & 0.119 \\
         & & 4000 & 0.008 & 0.055 & 0.122 & 0.009 & 0.053 & 0.124 & 0.007 & 0.053 & 0.122 & 0.010 & 0.048 & 0.122 & 0.006 & 0.061 & 0.115 & 0.006 & 0.057 & 0.118 \\
        \cmidrule{2-21}
        & \multirow{4}{*}{4} & 500 & 0.010 & 0.054 & 0.099 & 0.009 & 0.052 & 0.104 & 0.007 & 0.052 & 0.098 & 0.008 & 0.054 & 0.101 & 0.011 & 0.055 & 0.099 & 0.010 & 0.056 & 0.093 \\
        & & 1000 & 0.013 & 0.066 & 0.135 & 0.010 & 0.065 & 0.134 & 0.013 & 0.068 & 0.138 & 0.014 & 0.068 & 0.133 & 0.014 & 0.070 & 0.142 & 0.017 & 0.067 & 0.143 \\
        & & 2000 & 0.011 & 0.061 & 0.112 & 0.011 & 0.063 & 0.109 & 0.011 & 0.059 & 0.115 & 0.008 & 0.058 & 0.115 & 0.013 & 0.053 & 0.116 & 0.011 & 0.054 & 0.114 \\
        & & 4000 & 0.008 & 0.055 & 0.120 & 0.009 & 0.053 & 0.115 & 0.008 & 0.059 & 0.118 & 0.006 & 0.056 & 0.115 & 0.006 & 0.057 & 0.120 & 0.005 & 0.055 & 0.119 \\
        \bottomrule
    \end{tabular}
    \end{adjustbox}
\end{table}

Table \ref{tab:DGP_null_ll} reports the empirical sizes under $\mathbb{H}_0$ for DGPs 1--3. The results show that the proposed test has accurate size across different values of the under-smoothing factor $k_1$ for estimating bandwidth $h_r$. The reported rejection rates converge to the nominal levels as the sample size grows, demonstrating the size control of the proposed test. 

\begin{table}[H]
    \centering
    \caption{Rejection rates under DGPs 4--6 (local linear estimator)}
    \label{tab:DGP_alt_ll}
    \setlength{\extrarowheight}{-1pt}
    \begin{adjustbox}{max width=\textwidth, max height=\textheight}
    \begin{tabular}{@{}ccccccccccccccccccccc@{}}
        \toprule
        \multirow{3}{*}{DGP} & \multirow{3}{*}{$k_1$} & $k_2$ & \multicolumn{6}{c}{3} & \multicolumn{6}{c}{3.25} & \multicolumn{6}{c}{3.5} \\
        \cmidrule(lr){4-9} \cmidrule(lr){10-15} \cmidrule(lr){16-21}
        & & type & \multicolumn{3}{c}{KS} & \multicolumn{3}{c}{CvM} & \multicolumn{3}{c}{KS} & \multicolumn{3}{c}{CvM} & \multicolumn{3}{c}{KS} & \multicolumn{3}{c}{CvM} \\
        \cmidrule(lr){4-6} \cmidrule(lr){7-9} \cmidrule(lr){10-12} \cmidrule(lr){13-15} \cmidrule(lr){16-18} \cmidrule(lr){19-21}
        & & $n$ & 1\% & 5\% & 10\% & 1\% & 5\% & 10\% & 1\% & 5\% & 10\% & 1\% & 5\% & 10\% & 1\% & 5\% & 10\% & 1\% & 5\% & 10\% \\
        \midrule
        \multirow{12}{*}{4} & \multirow{4}{*}{3.5} & 500 & 0.010 & 0.056 & 0.111 & 0.004 & 0.051 & 0.098 & 0.008 & 0.058 & 0.118 & 0.006 & 0.055 & 0.108 & 0.010 & 0.066 & 0.143 & 0.010 & 0.054 & 0.111 \\
         & & 1000 & 0.024 & 0.118 & 0.208 & 0.013 & 0.069 & 0.131 & 0.034 & 0.133 & 0.246 & 0.014 & 0.078 & 0.139 & 0.048 & 0.169 & 0.286 & 0.018 & 0.085 & 0.162 \\
         & & 2000 & 0.117 & 0.319 & 0.458 & 0.031 & 0.103 & 0.192 & 0.169 & 0.396 & 0.528 & 0.035 & 0.125 & 0.232 & 0.244 & 0.465 & 0.602 & 0.041 & 0.149 & 0.260 \\
         & & 4000 & 0.348 & 0.617 & 0.731 & 0.058 & 0.190 & 0.338 & 0.466 & 0.710 & 0.802 & 0.082 & 0.252 & 0.400 & 0.574 & 0.775 & 0.846 & 0.109 & 0.310 & 0.506 \\
        \cmidrule{2-21}
         & \multirow{4}{*}{3.75} & 500 & 0.010 & 0.063 & 0.137 & 0.004 & 0.049 & 0.099 & 0.013 & 0.072 & 0.154 & 0.006 & 0.057 & 0.114 & 0.015 & 0.093 & 0.187 & 0.007 & 0.058 & 0.120 \\
         & & 1000 & 0.038 & 0.165 & 0.284 & 0.015 & 0.080 & 0.154 & 0.057 & 0.205 & 0.348 & 0.020 & 0.091 & 0.170 & 0.089 & 0.271 & 0.409 & 0.022 & 0.100 & 0.193 \\
         & & 2000 & 0.199 & 0.449 & 0.592 & 0.039 & 0.136 & 0.239 & 0.303 & 0.540 & 0.667 & 0.049 & 0.177 & 0.284 & 0.391 & 0.620 & 0.738 & 0.066 & 0.205 & 0.359 \\
         & & 4000 & 0.520 & 0.754 & 0.835 & 0.089 & 0.263 & 0.418 & 0.638 & 0.823 & 0.887 & 0.115 & 0.339 & 0.547 & 0.729 & 0.865 & 0.914 & 0.167 & 0.428 & 0.645 \\
        \cmidrule{2-21}
         & \multirow{4}{*}{4} & 500 & 0.014 & 0.092 & 0.173 & 0.005 & 0.060 & 0.111 & 0.014 & 0.119 & 0.198 & 0.008 & 0.063 & 0.117 & 0.021 & 0.127 & 0.243 & 0.013 & 0.068 & 0.129 \\
         & & 1000 & 0.070 & 0.225 & 0.372 & 0.019 & 0.096 & 0.171 & 0.099 & 0.298 & 0.457 & 0.024 & 0.111 & 0.193 & 0.153 & 0.374 & 0.523 & 0.035 & 0.122 & 0.230 \\
         & & 2000 & 0.311 & 0.571 & 0.691 & 0.053 & 0.171 & 0.290 & 0.428 & 0.658 & 0.778 & 0.066 & 0.214 & 0.363 & 0.511 & 0.740 & 0.817 & 0.090 & 0.272 & 0.449 \\
         & & 4000 & 0.636 & 0.838 & 0.911 & 0.111 & 0.336 & 0.538 & 0.751 & 0.901 & 0.938 & 0.174 & 0.436 & 0.667 & 0.828 & 0.935 & 0.958 & 0.230 & 0.568 & 0.756 \\
        \midrule
        \multirow{12}{*}{5} & \multirow{4}{*}{3.5} & 500 & 0.009 & 0.058 & 0.124 & 0.004 & 0.062 & 0.108 & 0.008 & 0.067 & 0.129 & 0.002 & 0.058 & 0.110 & 0.010 & 0.074 & 0.136 & 0.006 & 0.061 & 0.121 \\
         & & 1000 & 0.015 & 0.110 & 0.188 & 0.012 & 0.072 & 0.131 & 0.023 & 0.129 & 0.223 & 0.015 & 0.085 & 0.140 & 0.034 & 0.150 & 0.254 & 0.018 & 0.086 & 0.147 \\
         & & 2000 & 0.079 & 0.241 & 0.379 & 0.033 & 0.096 & 0.176 & 0.122 & 0.295 & 0.442 & 0.036 & 0.116 & 0.198 & 0.175 & 0.361 & 0.511 & 0.039 & 0.141 & 0.232 \\
         & & 4000 & 0.280 & 0.539 & 0.672 & 0.053 & 0.199 & 0.340 & 0.399 & 0.644 & 0.752 & 0.082 & 0.252 & 0.419 & 0.483 & 0.703 & 0.813 & 0.107 & 0.316 & 0.505 \\
        \cmidrule{2-21}
         & \multirow{4}{*}{3.75} & 500 & 0.012 & 0.068 & 0.124 & 0.004 & 0.062 & 0.110 & 0.013 & 0.077 & 0.139 & 0.008 & 0.063 & 0.118 & 0.014 & 0.087 & 0.169 & 0.014 & 0.062 & 0.125 \\
         & & 1000 & 0.030 & 0.142 & 0.242 & 0.015 & 0.084 & 0.143 & 0.046 & 0.180 & 0.293 & 0.019 & 0.092 & 0.153 & 0.063 & 0.210 & 0.334 & 0.024 & 0.100 & 0.174 \\
         & & 2000 & 0.133 & 0.352 & 0.496 & 0.040 & 0.123 & 0.214 & 0.194 & 0.430 & 0.593 & 0.050 & 0.151 & 0.261 & 0.259 & 0.534 & 0.672 & 0.053 & 0.195 & 0.331 \\
         & & 4000 & 0.420 & 0.686 & 0.797 & 0.084 & 0.268 & 0.435 & 0.536 & 0.762 & 0.857 & 0.120 & 0.339 & 0.550 & 0.645 & 0.829 & 0.889 & 0.164 & 0.443 & 0.653 \\
        \cmidrule{2-21}
         & \multirow{4}{*}{4} & 500 & 0.014 & 0.079 & 0.142 & 0.007 & 0.060 & 0.110 & 0.016 & 0.096 & 0.163 & 0.010 & 0.059 & 0.120 & 0.021 & 0.106 & 0.195 & 0.013 & 0.064 & 0.127 \\
         & & 1000 & 0.048 & 0.183 & 0.309 & 0.021 & 0.094 & 0.165 & 0.079 & 0.231 & 0.361 & 0.027 & 0.110 & 0.180 & 0.097 & 0.286 & 0.443 & 0.035 & 0.115 & 0.207 \\
         & & 2000 & 0.207 & 0.463 & 0.618 & 0.050 & 0.155 & 0.266 & 0.290 & 0.565 & 0.705 & 0.060 & 0.200 & 0.337 & 0.386 & 0.653 & 0.785 & 0.082 & 0.229 & 0.427 \\
         & & 4000 & 0.540 & 0.793 & 0.878 & 0.114 & 0.332 & 0.539 & 0.676 & 0.861 & 0.914 & 0.164 & 0.443 & 0.663 & 0.760 & 0.902 & 0.939 & 0.224 & 0.559 & 0.759 \\
        \midrule
        \multirow{12}{*}{6} & \multirow{4}{*}{3.5} & 500 & 0.009 & 0.065 & 0.105 & 0.006 & 0.056 & 0.099 & 0.012 & 0.063 & 0.118 & 0.007 & 0.058 & 0.105 & 0.015 & 0.068 & 0.135 & 0.009 & 0.063 & 0.112 \\
         & & 1000 & 0.020 & 0.096 & 0.170 & 0.018 & 0.077 & 0.131 & 0.023 & 0.109 & 0.201 & 0.009 & 0.074 & 0.140 & 0.026 & 0.128 & 0.236 & 0.014 & 0.071 & 0.147 \\
         & & 2000 & 0.052 & 0.174 & 0.293 & 0.022 & 0.088 & 0.157 & 0.071 & 0.233 & 0.368 & 0.026 & 0.099 & 0.179 & 0.100 & 0.289 & 0.450 & 0.031 & 0.107 & 0.203 \\
         & & 4000 & 0.143 & 0.379 & 0.566 & 0.032 & 0.140 & 0.231 & 0.220 & 0.501 & 0.658 & 0.042 & 0.160 & 0.284 & 0.303 & 0.598 & 0.734 & 0.057 & 0.203 & 0.360 \\
        \cmidrule{2-21}
         & \multirow{4}{*}{3.75} & 500 & 0.013 & 0.066 & 0.125 & 0.006 & 0.049 & 0.099 & 0.014 & 0.070 & 0.129 & 0.010 & 0.057 & 0.104 & 0.012 & 0.083 & 0.157 & 0.010 & 0.062 & 0.113 \\
         & & 1000 & 0.021 & 0.114 & 0.211 & 0.011 & 0.080 & 0.139 & 0.025 & 0.146 & 0.265 & 0.016 & 0.078 & 0.154 & 0.031 & 0.180 & 0.312 & 0.016 & 0.083 & 0.161 \\
         & & 2000 & 0.076 & 0.248 & 0.408 & 0.027 & 0.098 & 0.193 & 0.112 & 0.328 & 0.485 & 0.033 & 0.117 & 0.216 & 0.173 & 0.406 & 0.576 & 0.043 & 0.141 & 0.255 \\
         & & 4000 & 0.232 & 0.531 & 0.699 & 0.040 & 0.173 & 0.294 & 0.331 & 0.652 & 0.783 & 0.061 & 0.220 & 0.384 & 0.458 & 0.734 & 0.845 & 0.086 & 0.275 & 0.468 \\
        \cmidrule{2-21}
         & \multirow{4}{*}{4} & 500 & 0.014 & 0.074 & 0.125 & 0.009 & 0.057 & 0.099 & 0.014 & 0.084 & 0.143 & 0.009 & 0.061 & 0.104 & 0.017 & 0.095 & 0.161 & 0.010 & 0.063 & 0.115 \\
         & & 1000 & 0.029 & 0.155 & 0.268 & 0.016 & 0.087 & 0.146 & 0.038 & 0.191 & 0.320 & 0.017 & 0.085 & 0.166 & 0.053 & 0.248 & 0.387 & 0.023 & 0.096 & 0.189 \\
         & & 2000 & 0.108 & 0.328 & 0.496 & 0.031 & 0.123 & 0.224 & 0.185 & 0.420 & 0.604 & 0.048 & 0.143 & 0.256 & 0.252 & 0.527 & 0.692 & 0.051 & 0.176 & 0.315 \\
         & & 4000 & 0.324 & 0.645 & 0.797 & 0.063 & 0.213 & 0.374 & 0.458 & 0.757 & 0.863 & 0.085 & 0.272 & 0.477 & 0.585 & 0.832 & 0.919 & 0.117 & 0.367 & 0.575 \\
        \bottomrule
    \end{tabular}
    \end{adjustbox}
\end{table}

Table \ref{tab:DGP_alt_ll} demonstrates the empirical power under the alternative $\mathbb{H}_1$. In all three cases, the reported rejection rates increase as the sample size grows. This shows that the proposed test can effectively detect the unobserved heterogeneity for suitable bandwidth choices.

\let\oldthebibliography\thebibliography
\renewenvironment{thebibliography}[1]{%
  \oldthebibliography{#1}%
  \linespread{1.1}\selectfont   
}{%
  \endlist
}

\setlength{\bibsep}{4pt plus 1pt minus 1pt}
\putbib
\end{bibunit}

\end{document}